\documentclass[11pt,a4paper]{article}
\usepackage[left=0.75in, right=0.75in, top=0.75in, bottom=0.75in]{geometry}
\usepackage[utf8]{inputenc}
\usepackage[T1]{fontenc}
\usepackage{amsmath}
\usepackage{amsthm}
\usepackage{amsfonts}
\usepackage{amssymb}
\usepackage{graphicx}
\usepackage{setspace}
\usepackage{dsfont}
\usepackage{natbib}
\usepackage{booktabs}
\usepackage{multirow}
\usepackage{bm}
\usepackage{xcolor}
\usepackage[colorlinks,allcolors=blue]{hyperref}
\usepackage[normalem]{ulem}
\usepackage{paralist}

\newcommand{\vect}[1]{\textrm{vec}\left(#1\right)}

\newcommand{\unvect}[1]{\textrm{unvec}\left(#1\right)}
\newcommand{\dg}[1]{\mathrm{dg}\left(#1\right)}
\newcommand{\tr}[1]{\textrm{tr}\left(#1\right)}

\newcommand{\RN}[1]{\textup{\uppercase\expandafter{\romannumeral#1}}}

\newtheorem{assumption}{Assumption}
\newtheorem{proposition}{Proposition}
\newtheorem{lemma}{Lemma}

\newcounter{subassumption}[assumption]
\renewcommand{\thesubassumption}{(\textit{\roman{subassumption}})}
\makeatletter
\renewcommand{\p@subassumption}{\theassumption}
\makeatother
\newcommand{\subass}{ 
	\refstepcounter{subassumption}%
	\item[\thesubassumption]	
}

\begin{document}

	\title{\bf 
	Estimation of large approximate dynamic matrix  factor models based on the EM algorithm and Kalman filtering
	}	
		\author{Matteo Barigozzi 
		\footnote{Department of Economics,
						University of Bologna, 
						E-mail: \url{matteo.barigozzi@unibo.it}} 
	 \and Luca Trapin 
	 \footnote{Department of Statistical Sciences,
		 	University of Bologna, 
		 	E-mail: \url{luca.trapin@unibo.it}}
			}
	
	\date{\today}
	
	\maketitle
	
	\begin{abstract}
		 This paper considers an approximate dynamic matrix factor model that accounts for the time series nature of the data by explicitly modelling the time evolution of the factors. We study estimation of the model parameters based on the Expectation Maximization (EM) algorithm, implemented jointly with the Kalman smoother which gives estimates of the factors. 
		We establish the consistency of the estimated loadings and factor matrices as the sample size $T$ and the matrix dimensions $p_1$ and $p_2$ diverge to infinity. 
		We then extend this approach to: (a) the case of arbitrary patterns of missing data and (b) the presence of common stochastic trends. The finite sample properties of the estimators are assessed through a large simulation study and two applications on: (i) a financial dataset of volatility proxies and (ii) a macroeconomic dataset covering the main euro area countries. 
		
	\end{abstract}

	\section{Introduction}
	Matrix-variate time series data are becoming increasingly popular in economics and finance. For example, when forecasting regional specific economic activity \citep{chernis2020three}, investigating the dynamics of international trade flows \citep{chen2022modeling}, measuring of financial connectedness \citep{billio2021matrix}. 
	This has stimulated the development of high-dimensional methods to analyze matrix time series data, including matrix autoregressive models \citep{chen2021autoregressive,hsu2021matrix,billio2023bayesian}, matrix panel regression models \citep{kapetanios2021estimation}, and matrix factor models \citep{wang2019factor,yu2021projected,chen2021statistical,xu2024quasi,yu2024dynamic}.

	In this paper, we study a matrix factor model for a $p_1\times p_2$ zero-mean matrix-valued stationary process $\left\{\mathbf{Y}_t\right\}$, with latent factors following  a Matrix Autoregressive (MAR) model of order $P$, i.e., for $t\in \mathbb Z$, 
	\begin{align}
	\mathbf{Y}_t &= \mathbf{R} \mathbf{F}_t \mathbf{C}' + \mathbf{E}_{t} , 
	\label{eq:FM}\\
	\mathbf{F}_t &= \sum_{\ell=1}^P \mathbf{A}_\ell \mathbf{F}_{t-\ell}\mathbf{B}'_\ell + \mathbf{U}_t.
	\label{eq:MAR}
	\end{align}	
	In \eqref{eq:FM}, $\mathbf{F}_t$ is a $k_1 \times k_2$ matrix of latent factors with $k_1,k_2<\min(p_1,p_2)$, $\mathbf{R}$ and $\mathbf{C}$ are $p_1\times k_1$ and $p_2\times k_2$ matrices of unknown row and column loadings, $\mathbf{E}_t$  is a $p_1\times p_2$ matrix of idiosyncratic components with $p_1\times p_1$ row covariance matrix $\mathbf{H}$ and $p_2\times p_2$ column covariance matrix $\mathbf{K}$. In \eqref{eq:MAR}, $\mathbf{A}_\ell$ and $\mathbf{B}_\ell$,  $\ell=1,\dots,P$, are $k_1\times k_1$ and $k_2\times k_2$ matrices of autoregressive parameters, and $\textbf{U}_t$  is a $k_1\times k_2$ matrix of innovations from a  matrix-variate distribution with $k_1\times k_1$ row covariance matrix $\textbf{P}$ and $k_2\times k_2$ column covariance matrix $\textbf{Q}$. The processes  $\left\{\mathbf{F}_t\right\}$ and  $\left\{\mathbf{E}_t\right\}$ are assumed to be uncorrelated (at all leads and lags).

In our setting the idiosyncratic components are allowed to be correlated both across rows and columns, i.e., in general $\mathbf{H}$ and $\mathbf{K}$ are allowed to be full matrices, and we say that the factor model is  {\it approximate}. Furthermore, the model is {\it dynamic} since the factors are autocorrelated as specified by the MAR in \eqref{eq:MAR}, and, moreover, we also allow the idiosyncratic components to be autocorrelated, although no explicit model for their dynamics is introduced. Therefore, we call a model defined by \eqref{eq:FM}-\eqref{eq:MAR} an approximate dynamic matrix factor model (DMFM). It combines the matrix factor model, as formulated by  \citet{yu2021projected} and \citet{chen2021statistical}, and the MAR proposed by 
\citet{chen2021autoregressive}. A DMFM has also been considered by \citet{yu2024dynamic} (see Section \ref{se:lit} for a detailed comparison with our work).


In this paper, we  propose a new estimator of the factor loading matrices and factor matrices of the DMFM, implemented via the Expectation Maximization (EM) algorithm jointly with the Kalman smoother.  We prove consistency of the spaces spanned by the estimated loadings and by the factors as $\min(p_1,p_2,T)\to\infty$. 

	We argue that accounting for factor dynamics via the Kalman smoother, thus considering joint estimation of all parameters and the factors, is particularly convenient as it allows the user to  impose {\it a priori} restrictions on the models' parameters and/or dynamics, construct counterfactual scenarios, conditional forecasts, obtain now-casts, and deal with missing values due to different sampling frequencies or plain unavailability of the data (see, e.g., the applications in \citealp{banbura2014maximum} and \citealp{banbura2015conditional} in the case of vector time series). Furthermore, we also show that, thanks to the use of the Kalman filter, this approach is also particularly convenient to handle the case in which the data is driven by common stochastic trends, i.e., when (some of) the factors are $I(1)$ \citep[see, e.g., the applications in][in the case of vector time series]{barigozzi2023measuring}.

	
	
	Similarly to the vector case, the DMFM can be identified only in the limit $p_1,p_2\to\infty$ due to its approximate structure. That is to say that the numbers of factors $k_1$ and $k_2$ can be consistently estimated only when both dimensions grow large. This is what allows one to disentangle the factor driven component from the idiosyncratic one. However, this forces us to work in a high-dimensional setting. This makes joint Maximum Likelihood estimation of all the parameters and the factors in \eqref{eq:FM}-\eqref{eq:MAR} a hard if not unfeasible task due to the large number of parameters we need to estimate, which is $O( (p_1^2+p_2^2)T)$ (all autocovariances of the factors and idiosyncratic components), and due to the lack of a closed form solution. 
	
The estimation approach we consider has two main features which allow us to solve both problems. First, it is based on a mis-specified likelihood where the idiosyncratic components are treated as if they were uncorrelated. This reduces the number of parameters to be estimated to $O( p_1+p_2)$. Second, it is an iterative approach where, in a first step, for given parameters we estimate the factors via the Kalman smoother, and, in a second step, for given factors we estimate all parameters by maximizing the expected likelihood conditional on the factors. This allows us to derive a closed form expression for all estimators. 

Our approach is the generalization of the approach proposed by \citet{doz2012quasi} for the vector case.  However, such generalization is non-trivial, indeed, in the present matrix time series setting, we need,  at each iteration of the EM algorithm, to jointly estimate the two matrices of loadings, $\mathbf R$ and $\mathbf C$, which depend on each other (the same goes for the row and column 
 idiosyncratic covariances, $\mathbf H$ and $\mathbf K$, the
MAR coefficients, $\mathbf A$ and $\mathbf B$,  and the MAR innovation covariances, $\mathbf P$ and $\mathbf Q$). Respecting this bilinear structure requires modifying the algorithm accordingly and makes the derivation of the asymptotic properties more challenging.

	Finally, we show the potential of the proposed approach through  two applications. First, we analyze a matrix times series containing various volatility proxies for many stocks. Since not all proxies are available for all stocks, we show how to adapt the EM algorithm to deal with missing values and we then produce volatility forecasts for all stocks. Second, we analyze a matrix of time series of real macroeconomic variables of various Euro Area countries, which are clearly driven by few common trends.
		



	The rest of the paper is organized as follows. Section \ref{se:lit} discusses related works.
	Section \ref{sec:QML} presents the estimator obtained via the EM algorithm. 
	Section \ref{sec:asymp} presents the assumptions and the consistency results. 
	Sections \ref{subsec:banbura} and \ref{subsec:unitroot} explain how to extend the EM algorithm in presence of missing data and/or common stochastic trends.
	Section \ref{sec:Sim} studies the finite sample properties of the EM algorithm through Monte Carlo simulations. 
	Section \ref{sec:App} presents two real data applications on variance proxies of financial assets and on macroeconomic indicators of the Euro Area. Appendix \ref{ap:notation} contains all notation, as well as relevant results on matrix operations. Appendix \ref{app:EM} contains details on the EM updates. Appendix \ref{app:proofs} contains all proofs; Appendix \ref{sec:gianni} explains how to identify $I(1)$ and $I(0)$ factors in the case of $I(1)$ data;
	Appendix \ref{app:sim} contains additional simulation results.
	
	\section{Related literature}\label{se:lit}
	

	There exist many works considering  estimation only  of the matrix factor model in \eqref{eq:FM}, thus without explicitly accounting for the factors' dynamics. First, \citet{wang2019factor} introduce the class of large matrix factor models under the assumption of serially uncorrelated idiosyncratic components, and
	propose to estimate the loadings by means of eigenvectors of a long-run covariance matrix (see also \citealp{chen2020constrained}). Second,
	 \citet{yu2021projected} and \citet{chen2021statistical}  extend this approach to the case of possibly autocorrelated idiosyncratic components, and  propose two different generalizations to the matrix setting of the Principal Component (PC) estimators typically used in the vector case. Both these work consider also methods for determining the number of factors (see also  \citealp{he2023one}, and \citealp{han2022rank}, for alternative methods). In a similar setting, \cite{gao2023denoising} consider estimating in the case of idiosyncratic components containing weak signals.
	 Last, \citet{yuan2023two} and \cite{xu2024quasi} consider QML estimation of two different specifications of a matrix factor model.  

	To the best of our knowledge only \cite{yu2024dynamic} consider a DMFM as specified by \eqref{eq:FM}-\eqref{eq:MAR}. However, our work differs in several aspects. First, we consider joint estimation of factors and parameters of the model, while  they consider a two-step approach where first the loadings and the factors are estimated and then a MAR is estimated on the factors. Second, we allow the idiosyncratic components to be serially correlated, while they impose a different factor structure with time independent idiosyncratic components. Third, we derive the asymptotic properties of the  factors 
	estimated via the Kalman smoother, while they do not study the asymptotic properties of such estimator, although entertaining the possibility of retrieving the factors via filtering.  As a last difference, we also study our estimation approach in presence of arbitrary patterns of missing data or stochastic trends.

	Our work is also related to three other strands of the literature. First, the idea of considering a misspecified likelihood in factor analysis to make its maximization more treatable dates back to \citet{tipping1999probabilistic} who, in a vector context, treated the idiosyncratic components as i.i.d.. This idea was then extended by \citet{doz2012quasi} and \citet{bai2016maximum} to the case of high-dimensional vectors of time series having serially and cross-sectionally correlated idiosyncratic components. In particular,  \citet{doz2012quasi} explicitly model the factors dynamics.
	
	Second, there exist many factor model approaches for handling missing values in high-dimensional vector time series. On the one hand, \citet{banbura2014maximum} propose an EM-based approach which we generalize to the matrix setting in this paper. On the other hand, there are a few  approaches based on various modifications of standard PC analysis, see, e.g., the recent works by \citet{xiong2023large} and \citet{cahan2023factor}. Finally, \citet{cen2024tensor} consider a PC based approach for the tensor case, which includes the matrix case.
	
	Third, in the case of $I(1)$ vector time series, estimation of factor models via PC has been studied in a few works either under the assumption of stationary idiosyncratic components, which can be serially uncorrelated \citep{zhang2019identifying} or autocorrelated \citep{bai2004estimating}, or when allowing for $I(1)$ idiosyncratic components \citep{bai2004panic,barigozzi2021large}.  Recently, \citet{chen2025inference} considered estimation via PC methods for matrix time series with $I(1)$ and $I(0)$ factors and stationary idiosyncratic components.

	\section{Estimation of the Dynamic Matrix Factor Model}
	\label{sec:QML}
	\paragraph{The log-likelihood.}
	Let consider a DMFM as defined in \eqref{eq:FM}-\eqref{eq:MAR}, and without loss of generality assume that the MAR is of order $P=1$. For a $p_1\times p_2$ matrix-valued covariance stationary process $\left\{\mathbf{Y}_t\right\}$ our data generating process is then given by:
	\begin{align}
		\label{eq:dmfm1}
		\mathbf{Y}_t &= \mathbf{R} \mathbf{F}_t \mathbf{C}' + \mathbf{E}_t,  \\
		\label{eq:dmfm2}
		\mathbf{F}_t &= \mathbf{A} \mathbf{F}_{t-1}\mathbf{B}' + \mathbf{U}_t,
	\end{align}
	where $\mathbf{R}$ is a $p_1\times k_1$ matrix of row loadings, $\mathbf{C}$ is a $p_2\times k_2$ matrix of column, $\mathbf{F}_t$ is a $k_1\times k_2$ matrix of latent factor, $\mathbf{E}_t$ is a $p_1\times p_2$ matrix of idiosyncratic components with covariances $\mathbf{H}$ and $\mathbf{K}$, $\mathbf A$ and $\mathbf B$ are both $k_1\times k_2$ matrices of MAR coefficients, and $\textbf{U}_t$ is a $k_1\times k_2$ matrix of innovations with covariances $\mathbf{P}$ and $\mathbf{Q}$. 
	As usual in factor models, for simplicity and without loss of generality, we assume $\mathbb{E}\left[\mathbf{F}_t\right]=\mathbf 0_{k_1,k_2}$ and $\mathbb{E}\left[\mathbf{E}_t\right]=\mathbf 0_{p_1,p_2}$. Therefore, model \eqref{eq:dmfm1}-\eqref{eq:dmfm2} implies that $\mathbb{E}\left[\mathbf{Y}_t\right]=\mathbf 0_{p_1,p_2}$ in other words, we implicitly assume for simplicity to be working with centered data. 
	
	
	Denote as $\mathsf{y}_t=\vect{\mathbf{Y}_t}$, $\mathsf{f}_t=\vect{\mathbf{F}_t}$, $\mathsf{e}_t=\vect{\mathbf{E}_t}$ and $\mathsf{u}_t=\vect{\mathbf{U}_t}$, the vectorized versions of the matrices $\mathbf{Y}_t$, $\mathbf{F}_t$, $\mathbf{E}_{t}$ and $\mathbf{U}_t$, respectively. 	Then, consider a sample of $T$ observations, and let $\mathsf{Y}_T=\left(\mathsf{y}'_1 \cdots \mathsf{y}'_T\right)'$ and $\mathsf{E}_T=\left(\mathsf{e}'_1 \cdots \mathsf{e}'_T\right)'$ be $(p_1p_2T)$-dimensional vectors containing all the observations and idiosyncratic components, respectively, and let $\mathsf{F}_T=\left(\mathsf{f}'_1 \cdots \mathsf{f}'_T\right)'$ be the $(k_1k_2T)$-dimensional vector of factors. 
	Let $\boldsymbol{\Omega}^{\mathsf{Y}}_T=\mathbb{E}\left[\mathsf{Y}_T\mathsf{Y}'_T\right]$, $\boldsymbol{\Omega}^{\mathsf{E}}_T=\mathbb{E}\left[\mathsf{E}_T\mathsf{E}'_T\right]$, $\boldsymbol{\Omega}^{\mathsf{F}}_T=\mathbb{E}\left[\mathsf{F}_T\mathsf{F}'_T\right]$ be covariance matrices containing all the cross-sectional row and column covariances and all the autocovariances up to lag $(T-1)$. Notice that $\boldsymbol{\Omega}^{\mathsf{F}}_T$ is fully characterized by the matrices of MAR parameter  $\mathbf{A}$, $\mathbf{B}$, $\mathbf{P}$ and $\mathbf{Q}$, thus, hereafter, we denote it as $\boldsymbol{\Omega}^{\mathsf{F}}_{T(\mathbf A,\mathbf B,\mathbf P,\mathbf Q)}$.
	
	It follows that the DMFM is fully characterized by the covariance matrix of $\mathsf{Y}_T$, which must be such that
$
	\boldsymbol{\Omega}^{\mathsf{Y}}_T = (\mathbb{I}_T\otimes \mathbf{C} \otimes \mathbf{R})\boldsymbol{\Omega}^{\mathsf{F}}_{T(\mathbf A,\mathbf B,\mathbf P,\mathbf Q)} (\mathbb{I}_T\otimes \mathbf{C} \otimes \mathbf{R})' + \boldsymbol{\Omega}^{\mathsf{E}}_T.
	$
%
%
	In an approximate DMFM as the one we consider, $\boldsymbol{\Omega}^{\mathsf{E}}_T$ is allowed to be a full-matrix, but this implies that it has $\frac{p_1p_2T(p_1p_2T+1)}2$ entries to be estimated 
while we have only $p_1p_2T $ observations. This  makes Maximum Likelihood estimation unfeasible.
	 
	 A solution consists in considering a misspecified  likelihood where the idiosyncratic components $\mathbf{E}_t$ are treated as if they were serially and cross-sectionally uncorrelated, i.e., when we replace $\boldsymbol{\Omega}^{\mathsf{E}}_T$ with $\mathbb{I}_T\otimes \dg{\mathbf{K}} \otimes \dg{\mathbf{H}}$. This is the approach followed in the vector factor model case, by, e.g., \citet{bai2016maximum} and \citet{doz2012quasi}.
	 Under this misspecification, the vector of parameters to be estimated reduces to
$
\boldsymbol{\theta}=\left(\vect{\mathbf{R}}',\vect{\mathbf{C}}',\vect{\dg{\mathbf{H}}}',\vect{\dg{\mathbf{K}}}',\vect{\mathbf{A}}',\vect{\mathbf{B}}',\vect{\mathbf{P}}',\vect{\mathbf{Q}}'\right)',
$
which has dimension now growing as $p_1+p_2$, thus it can be estimated using $p_1p_2T $ observations. Consequently, we consider a misspecified, or quasi, log-likelihood   given by:
	\begin{align}
		\ell\left(\mathsf{Y}_T; \boldsymbol{\theta}\right) =&\, \frac{p_1p_2T}{2} \log(2\pi) - \log\left(|(\mathbb{I}_T\otimes \mathbf{C} \otimes \mathbf{R})\boldsymbol{\Omega}^{\mathsf{F}}_{T(\mathbf A,\mathbf B,\mathbf P,\mathbf Q)} (\mathbb{I}_T\otimes \mathbf{C} \otimes \mathbf{R})' + \mathbb{I}_T\otimes \dg{\mathbf{K}} \otimes \dg{\mathbf{H}} |\right)\nonumber\\
		& - \frac{1}{2}\left[\mathsf{Y}'_T \left(
		(\mathbb{I}_T\otimes \mathbf{C} \otimes \mathbf{R})\boldsymbol{\Omega}^{\mathsf{F}}_{T(\mathbf A,\mathbf B,\mathbf P,\mathbf Q)} (\mathbb{I}_T\otimes \mathbf{C} \otimes \mathbf{R})' 
		+ 
		\mathbb{I}_T\otimes \dg{\mathbf{K}} \otimes \dg{\mathbf{H}}
		\right)^{-1}\mathsf{Y}_T \right].
		\label{eq:llk}
	\end{align}
	Due to the introduced misspecifications, we say that the maximizer of \eqref{eq:llk} is a QML estimator.

		In principle, QML estimation of $\boldsymbol{\theta}$ can be performed by writing the model in vectorized form and maximizing the prediction error decomposition of the Gaussian likelihood obtained from the Kalman filter (see e.g. section 7.2 in \citealp{durbin2012time}). This approach is applicable when $p_1$ and $p_2$ are relatively small, but it becomes quickly unfeasible in larger settings due to a lack of closed form solution. Furthermore, by vectorizing the data we lose the bilinear structure of the model.
		We resort to the EM algorithm instead.\vspace{-.5cm}
	
	\paragraph{EM algorithm.}
	\label{subsec:em}
	The EM algorithm is an iterative procedure proposed by \cite{dempster1977maximum} to maximize the log-likelihood in problems where missing or latent observations make the likelihood intractable. This procedure works in two steps: given a set of parameter values, the E-step computes the expectation of the log-likelihood conditional on the observed data, thus ``filling'' the missing observations; the M-step maximizes the expected log-likelihood with respect to the model parameters. These two steps are iterated until a convergence criterion is satisfied. 	As the factor process $\{\mathbf{F}_t\}$ is unobserved in our setting, the EM algorithm is a suitable option to perform QML estimation. 
	

	In general, \citet{wu1983convergence}  proves that, when considering a Gaussian quasi-likelihood the EM algorithm converges to one of its maxima. As such, \citet{doz2012quasi} consider their EM approach as 
	Quasi Maximum Likelihood (QML) estimation of  a dynamic vector factor model, a conjecture then proved by \citet{barigozzi2019quasi}. For this reasons in this section we refer to our estimator as a QML estimator.  However, to formally prove that the estimator defined below is effectively achieving QML estimation we would need more assumptions on the distribution of the data and the identification of the loadings space, which, in this paper, we refrain to make. For this reason, here we do not prove such equivalence but we limit to notice that, by construction,  the considered log-likelihood is effectively increasing at each iteration (see the numerical results in Section \ref{sec:Sim}). \vspace{-.5cm}

	
	\paragraph{Kalman smoother.}
For any iteration $n\geq 0$ of the EM algorithm, and given an estimator of the parameters $\widehat{\boldsymbol{\theta}}^{(n)}$, we run the Kalman smoother on a vectorized version of the DMFM \eqref{eq:dmfm1}-\eqref{eq:dmfm2}. This gives as an estimator of $\mathsf{f}_t=\vect{\mathbf{F}_t}$ the linear projection $\mathsf{f}^{(n)}_{t|T} = \mathrm{Proj}_{\widehat{\boldsymbol{\theta}}^{(n)}}\left[\mathsf{f}_t|\mathsf{Y}_{T}\right]$ and the associated MSE, denoted as $\mathbf{\Pi}^{(n)}_{t|T}$. Moreover, by considering the Kalman smoother for the augmented state vector $(\mathsf{f}_t^\prime \;\; \mathsf{f}_{t-1}^\prime)^\prime$, we denote the top-left $k_1k_2\times k_1k_2$ block of the associated $2k_1k_2\times 2k_1k_2$ MSE as $\mathbf{\Delta}^{(n)}_{t|T}$ (see, e.g., Section 4.4 in \citealp{durbin2012time}, for the explicit expressions of these quantities). In particular, to run the Kalman smoother we need first to run the Kalman filter, which, in turn, requires an estimate of the inverse idiosyncratic covariance matrix. This is a hard taks in high-dimensions, but here, consistently with the misspecified log-likelihood \eqref{eq:llk}, we always consider an estimator of the misspecified diagonal covariance matrix $\text{dg}(\mathbf K)\otimes \text{dg}(\mathbf H)$, which is always invertible. 

Note that the engineering literature proposes matrix versions of the Kalman filter for matrix state-space models like the one in \eqref{eq:dmfm1}-\eqref{eq:dmfm2}  (e.g. \citealp{choukroun2006kalman}). However, these approaches heavily rely on the $\vect{\cdot}$ operator and offer only minor computational advantages, primarily due to algebraic simplifications. Similarly to our approach, also \cite{yu2024dynamic} utilize the vectorized Kalman filter.

Under joint Gaussianity of $\mathsf F_T$ and $\mathsf Y_T$, it is known that $\mathsf{f}^{(n)}_{t|T}$, $\mathbf{\Pi}^{(n)}_{t|T}$, and $\mathbf{\Delta}^{(n)}_{t|T}$ are estimators of the first and second conditional moments of $\mathsf{f}_t$ given $\mathsf Y_T$, obtained when computing expectations using the estimated parameters $\widehat{\boldsymbol{\theta}}^{(n)}$. As mentioned above, here we do not make any Gaussianity assumption. Nevertheless, we show that in the present high-dimensional setting the Kalman smoother delivers consistent estimates of the factors, thus providing a good approximation (see the results in Section \ref{sec:asymp}).\vspace{-.5cm}
\paragraph{E-step.}
	In the E-step, we use the output of the Kalman smoother to compute the expected quasi log-likelihood of the approximate DMFM. Spcifically,
	given $\mathsf{Y}_T$ and $\widehat{\boldsymbol{\theta}}^{(n)}$,  by Bayes' rule, we have\footnote{Notice that conditioning on $\mathsf{Y}_T$, which is a vector, is equivalent to conditioning on the sequence of matrices $\{\mathbf Y_1,\ldots,\mathbf Y_T\}$, hence, we can  write the argument of the log-likelihood in both ways. The same applies for $\mathsf{F}_T$ and the sequence of matrices $\{\mathbf F_1,\ldots,\mathbf F_T\}$. Therefore, in order to avoid introducing further notation, hereafter, we use only the vector notation $\mathsf{Y}_T$ and $\mathsf{F}_T$ to indicate the conditioning random variables and the arguments of the log-likelihoods, even when the latter are expressed explicitly as function of matrix valued time series.}
	\[
	\ell\left(\mathsf{Y}_T; \boldsymbol{\theta}\right) = \underbrace{\mathbb{E}_{\widehat{\boldsymbol{\theta}}^{(n)}}\left[\ell\left(\mathsf{Y}_T |\mathsf{F}_T; \boldsymbol{\theta}\right) |\mathsf{Y}_T \right] +  \mathbb{E}_{\widehat{\boldsymbol{\theta}}^{(n)}}\left[\ell\left(\mathsf{F}_T; \boldsymbol{\theta}\right) |\mathsf{Y}_T \right]}_{\mathcal{Q}(\boldsymbol{\theta}, \widehat{\boldsymbol{\theta}}^{(n)})} - \mathbb{E}_{\widehat{\boldsymbol{\theta}}^{(n)}}\left[\ell\left(\mathsf{F}_T |\mathsf{Y}_T; \boldsymbol{\theta}\right)|\mathsf{Y}_T\right].
	\]
	
%
%
	As proved in \cite{dempster1977maximum}, maximizing $\ell\left(\mathsf{Y}_T; \boldsymbol{\theta}\right)$ is equivalent to maximizing $\mathcal{Q}(\boldsymbol{\theta}, \widehat{\boldsymbol{\theta}}^{(n)})$, and we thus need to compute only the latter. 
	Specifically, we have  (see Appendix \ref{ap:ellk} for the derivation)
	\begin{align}
	\begin{array}{lll}
	\mathbb{E}_{\widehat{\boldsymbol{\theta}}} \left[\ell(\mathsf{Y}_T|\mathsf{F}_T;\boldsymbol{\theta}) | \mathsf{Y}_T \right]  &=& -\frac{T}{2}\left(p_1\log\left(|\mathbf{K}|\right)+p_2\log\left(|\mathbf{H}|\right)\right) \\
	&& - \frac{1}{2}\sum^T_{t=1} \mathbb{E}_{\widehat{\boldsymbol{\theta}}^{(n)}} \left[ \textbf{tr}\left(\mathbf{H}^{-1}\left(\mathbf{Y}_t-\mathbf{R}\mathbf{F}_{t}\mathbf{C}'\right)\mathbf{K}^{-1}\left(\mathbf{Y}_t-\mathbf{R}\mathbf{F}_{t}\mathbf{C}'\right)'\right) | \mathsf{Y}_T \right],
	\end{array}
	\label{eq:llmat1}		
	\\
	\begin{array}{lll}
	\mathbb{E}_{\widehat{\boldsymbol{\theta}}} \left[\ell(\mathsf{F}_T;\boldsymbol{\theta}) | \mathsf{Y}_T \right] &=& - \frac{T-1}{2}\left(k_1\log\left(|\mathbf{Q}|\right)+ k_2\log\left(|\mathbf{P}|\right) \right) \\
	&& 
	-  \frac{1}{2}\sum^T_{t=1} \mathbb{E}_{\widehat{\boldsymbol{\theta}}^{(n)}} \left[ \textbf{tr} \left(\mathbf{P}^{-1}\left(\mathbf{F}_{t}-\mathbf{A}\mathbf{F}_{t-1}\mathbf{B}'\right)\mathbf{Q}^{-1}\left(\mathbf{F}_{t}-\mathbf{A}\mathbf{F}_{t-1}\mathbf{B}' \right)' \right) | \mathsf{Y}_T \right].
	\end{array}	
		\label{eq:llmat2}			
	\end{align}
	Notice that these log-likelihoods depend directly on the data in its matrix form. \vspace{-.5cm}

\paragraph{M-step.}\label{subsec:Mstep}
	In the M-step, we maximize \eqref{eq:llmat1} and \eqref{eq:llmat2} to obtain a new estimate of the parameters $\widehat{\boldsymbol{\theta}}^{(n+1)}$. In particular, at any $n\geq 0$ iteration, the row and column loadings estimators are given by (see Appendix \ref{ap:em} for the derivation):
	\begin{align}
	\widehat{\mathbf{R}}^{(n+1)} &= \left(\sum\limits_{t=1}^{T} \mathbf{Y}_t\widehat{\mathbf{K}}^{(n)-1}\widehat{\mathbf{C}}^{(n)}\mathbf{F}^{(n)\prime}_{t|T}\right)\left(\sum\limits_{t=1}^{T}  \left(\widehat{\mathbf{C}}^{(n)\prime}\widehat{\mathbf{K}}^{(n)-1}\widehat{\mathbf{C}}^{(n)}\right)\star \left(\mathsf{f}^{(n)}_{t|T}\mathsf{f}^{(n)\prime}_{t|T} + \mathbf{\Pi}^{(n)}_{t|T}  \right)\right)^{-1},  \label{eq:load.est}\\
	\widehat{\mathbf{C}}^{(n+1)} &= \left(\sum\limits_{t=1}^{T} \mathbf{Y}'_t\widehat{\mathbf{H}}^{(n)-1}\widehat{\mathbf{R}}^{(n+1)}\mathbf{F}^{(n)}_{t|T} \right) \left(\sum\limits_{t=1}^{T}   \left(\widehat{\mathbf{R}}^{(n+1)\prime}\widehat{\mathbf{H}}^{(n)-1}\widehat{\mathbf{R}}^{(n+1)}\right) \star \left(\mathbb{K}_{k_1k_2} \left( \mathsf{f}^{(n)}_{t|T}\mathsf{f}^{(n)\prime}_{t|T} + \mathbf{\Pi}^{(n)}_{t|T}  \right) \mathbb{K}'_{k_1k_2}\right) \right)^{-1}.\nonumber
	\end{align}
	Clearly, the estimators of $\mathbf R$ and $\mathbf C$ depend on each other. Here, we choose to first estimate $\mathbf R$ conditional on the previous iteration estimator of $\mathbf C$. Since, as shown in the next section, all these estimators are consistent at any iteration $n\ge 0$, provided we correctly initialize the EM algorithm, we ensure, in this way, that the bilinear structure of the DMFM is preserved. We can of course equivalently choose to first estimate $\mathbf C$ conditional on the previous iteration estimator of $\mathbf R$. 
	
	By using $\widehat{\mathbf{R}}^{(n+1)}$ and  $\widehat{\mathbf{C}}^{(n+1)}$ we can compute the estimators of ${\mathbf{H}}$ and ${\mathbf{K}}$ which are enforced to be diagonal matrices in agreement with the considered misspecified log-likelihood given in \eqref{eq:llk}. Thus, for $i=1,\ldots, p_1$, we have:
	\begin{align}
	[\widehat{\mathbf H}^{(n+1)}]_{ii} &= \frac{1}{Tp_2 }\sum\limits_{t=1}^{T}  \left[ \mathbf{Y}_t \widehat{\mathbf{K}}^{(n)-1}\mathbf{Y}'_t - \mathbf{Y}_t \widehat{\mathbf{K}}^{(n)-1} \widehat{\mathbf{C}}^{(n+1)} \mathbf{F}^{(n)\prime}_{t|T} \widehat{\mathbf{R}}^{(n+1)\prime} -  \widehat{\mathbf{R}}^{(n+1)} \mathbf{F}^{(n)}_{t|T} \widehat{\mathbf{C}}^{(n+1)\prime}\widehat{\mathbf{K}}^{(n)-1}\mathbf{Y}'_t \right. \nonumber\\
	& \qquad \qquad \left. + \left(\widehat{\mathbf{C}}^{(n+1)\prime}\widehat{\mathbf{K}}^{(n)-1} \widehat{\mathbf{C}}^{(n+1)}\right) \star \left(\left(\mathbb{I}_{k_2}\otimes\widehat{\mathbf{R}}^{(n+1)}\right) \left( \mathsf{f}^{(n)}_{t|T}\mathsf{f}^{(n)\prime}_{t|T} + \mathbf{\Pi}^{(n)}_{t|T}\right) \left( \mathbb{I}_{k_2}\otimes\widehat{\mathbf{R}}^{(n+1)}\right)'\right) \right]_{ii}, \nonumber
	\end{align}
	and $[\widehat{\mathbf H}^{(n+1)}]_{ij}=0$ if $i\ne j$. 
	Likewise, for $i=1,\ldots, p_2$, we have:
		\begin{align}
	[\widehat{\mathbf K}^{(n+1)}]_{ii} &= \frac{1}{Tp_1} \sum\limits_{t=1}^{T}  \left[ \mathbf{Y}'_t \widehat{\mathbf{H}}^{(n+1)-1} \mathbf{Y}_t - \mathbf{Y}_t'\widehat{\mathbf{H}}^{(n+1)-1} \widehat{\mathbf{R}}^{(n+1)} \mathbf{F}^{(n)}_{t|T} \widehat{\mathbf{C}}^{(n+1)\prime} - \widehat{\mathbf{C}}^{(n+1)}\mathbf{F}^{(n)\prime}_{t|T} \widehat{\mathbf{R}}^{(n+1)\prime}\widehat{\mathbf{H}}^{(n+1)-1}\mathbf{Y}_t \right. \nonumber\\
	&\left. \hspace*{-0.25in}+ \left(\widehat{\mathbf{R}}^{(n+1)\prime}\widehat{\mathbf{H}}^{(n+1)-1} \widehat{\mathbf{R}}^{(n+1)}\right) \star \left( \left(\mathbb{I}_{k_1}\otimes\widehat{\mathbf{C}}^{(n+1)} \right) \left(\mathbb{K}_{k_1k_2} \left(\mathsf{f}^{(n)}_{t|T}\mathsf{f}^{(n)\prime}_{t|T} + \mathbf{\Pi}^{(n)}_{t|T}  \right) \mathbb{K}'_{k_1k_2}\right) \left(\mathbb{I}_{k_1}\otimes\widehat{\mathbf{C}}^{(n+1)} \right)' \right) \right]_{ii}.\nonumber
	\end{align}	
	and $[\widehat{\mathbf K}^{(n+1)}]_{ij}=0$ if $i\ne j$. 
	As for the loadings,  the estimators of $\mathbf H$ and $\mathbf K$ depend on each other, and, in order to preserve the bilinear structure of the model, we first estimate $\mathbf H$ conditional on the previous iteration estimator of $\mathbf K$. 
	
Finally, while the individual MAR matrices are part of the model's structure, we opt for a more streamlined approach by estimating their Kronecker product directly. This choice simplifies implementation, particularly since the Kalman smoother in our algorithm is applied to vectorized data. Thus, at each iteration we compute the estimators $\widehat{\mathbf{B}\otimes\mathbf{A}}^{(n+1)}$ and $\widehat{\mathbf{Q}\otimes \mathbf{P}}^{(n+1)}$ (see Appendix \ref{ap:em} for their expressions and the expressions of the alternative estimators $\widehat{\mathbf{A}}^{(n+1)}$, $\widehat{\mathbf{B}}^{(n+1)}$, $\widehat{\mathbf{P}}^{(n+1)}$, and $\widehat{\mathbf{Q}}^{(n+1)}$). 
Clearly,  such  estimators do not satisfy the constraints imposed by the bilinear structure of the MAR. 
	Nevertheless, the asymptotic properties of the estimated loadings and factor matrices are unaffected by this choice. This is also confirmed by our simulations in Appendix \ref{app:sim}.\vspace{-.5cm}


	\paragraph{Initialization.}\label{subsec:EMinit}
	We use the projected estimator (PE) of \cite{yu2021projected} to obtain initial estimates $\widehat{\mathbf{R}}^{(0)} $, $\widehat{\mathbf{C}}^{(0)}$, and $\widetilde{\mathbf {F}}_t$ of the row and column loadings and the factor matrices. Initial estimates of the idiosyncratic variances, i.e., the diagonals of $\widehat{\mathbf{K}}^{(0)}$, $\widehat{\mathbf{H}}^{(0)}$ can be obtained by computing the sample variances of the  PE residual idiosyncratic components. 
	 Last, in agreement with the M-step, pre-estimators of the MAR parameters can be computed without imposing the bilinear structure, i.e., by fitting a VAR on $\widetilde{\mathsf{f}}_t\equiv \text{vec}(\widetilde{\mathbf {F}}_t)$, thus giving $\widehat{\mathbf{B\otimes\mathbf{A}}}^{(0)} $ and $\widehat{\mathbf{Q}\otimes\mathbf{P}}^{(0)}$. Expressions of all these pre-estimators are in Appendix \ref{app:yu}.
		
	Finally, we initialize the Kalman filter by setting $\mathsf{f}^{(0)}_{0|0} = \mathbf{0}_{k_1k_2}$ at $n=0$ and $\mathsf{f}^{(n)}_{0|0} = \mathsf{f}^{(n-1)}_{0|T}$ at $n\geq1$, and $\mathbf{\Pi}^{(n)}_{0|0} = \mathbb{I}_{k_1k_2}$ at $n\geq0$.  \vspace{-.5cm}

	\paragraph{Convergence.}
	\label{subsec:EMconv}
	As in \cite{doz2012quasi}, we run the EM algorithm for a finite pre-specified number of iterations $n_{\textrm{\tiny max}}$. For a given tolerance level $\epsilon$, the algorithm is stopped at the first iteration $n^*<n_{\textrm{\tiny max}}$ such that 
	$
	\Delta\mathcal{L}_{n^*}={|\mathcal{L}(\mathsf{Y}_t,\widehat{\boldsymbol{\theta}}^{(n^*+1)})-\mathcal{L}(\mathsf{Y}_t;\widehat{\boldsymbol{\theta}}^{(n^*)})|}/{\frac{1}{2}|\mathcal{L}(\mathsf{Y}_T;\widehat{\boldsymbol{\theta}}^{(n^*+1)}) +\mathcal{L}(\mathsf{Y}_t;\widehat{\boldsymbol{\theta}}^{(n^*)})|} < \epsilon,
	$
	where $	\mathcal{L}(\mathsf{Y}_T;\boldsymbol{\theta})$ is the one-step-ahead prediction error log-likelihood, computed via the Kalman filter. \vspace{-.5cm}
	
	\paragraph{Final estimators.}	Once the EM algorithm reaches convergence, we define the EM estimator of the model parameters as $\widehat{\boldsymbol{\theta}}\equiv \widehat{\boldsymbol{\theta}}^{(n^*+1)}$. In particular, the estimated factor loadings are given by $\widehat{\mathbf{R}}\equiv\widehat{\mathbf{R}}^{(n^*+1)}$ and $\widehat{\mathbf{C}}\equiv\widehat{\mathbf{C}}^{(n^*+1)}$. Finally, we obtain a final estimate of the factor matrices by running the Kalman smoother one last time, that is $\widehat{\mathbf{F}}_t \equiv \text{unvec}({\mathsf{f}^{(n^*+1)}_{t|T}})$ for any $t=1,\dots,T$.


	\section{Asymptotic results}\label{sec:asymp}

We make the following assumptions on the loadings and the factors.
	\begin{assumption}\textsc{(Common component).}
		\label{ass:common} 
		\begin{itemize}
			\subass \label{subass:loading} 
			$\lVert \mathbf{R} \rVert_{\textnormal{\tiny max}}\leq \bar{r}$ and $\lVert\mathbf{C}\rVert_{\textnormal{\tiny max}}\leq \bar{c}$, for finite positive reals $\bar{r}$ and $\bar{c}$, and, as $\min\{p_1,p_2\}\rightarrow\infty$, $\lVert p^{-1}_1\mathbf{R}'\mathbf{R} - \mathbb{I}_{k_1}\rVert \rightarrow 0$ and $\lVert p^{-1}_2\mathbf{C}'\mathbf{C} - \mathbb{I}_{k_2}\rVert \rightarrow 0$.
			\subass \label{subass:factor} 
			For all $t \in \mathbb Z$,
			$\mathbb{E}\left[\mathbf{F}_t\right]=\mathbf{0}_{k_1,k_2}$, $\mathbb{E}\lVert\mathbf{F}_t\rVert^4<\infty$, and, as $T\to\infty$, 
			$
			T^{-1}\sum_{t=1}^{T} \mathbf{F}_t\mathbf{F}'_t \xrightarrow{p} \boldsymbol{\Sigma}_1$ and  $T^{-1}\sum_{t=1}^{T} \mathbf{F}'_t\mathbf{F}_t \xrightarrow{p} \boldsymbol{\Sigma}_2 
			$
			where $\boldsymbol{\Sigma}_i$ is a $k_i\times k_i$ matrix with distinct eigenvalues and spectral decomposition $\boldsymbol{\Sigma}_i=\boldsymbol{\Gamma}^F_i\boldsymbol{\Lambda}^F_i\boldsymbol{\Gamma}^{F\prime}_i$, for $i=1,2$. The factor numbers $k_1$ and $k_2$ are finite and independent of $T$, $p_1$, and $p_2$.
			\subass \label{subass:armatrix} 
			$\lVert\mathbf{A}\otimes \mathbf{B}\rVert<1$.
			\subass \label{subass:arshock} For all $t \in \mathbb Z$, $\mathbb{E}\left[\mathbf{U}_t\right]=\mathbf{0}_{k_1,k_2}$, and $\mathbb{E}\left[\mathbf{U}_t\mathbf{U}'_t\right]=\mathbf{P}\textbf{\upshape tr}({\mathbf{Q}})$ and $\mathbb{E}\left[\mathbf{U}'_t\mathbf{U}_t\right]=\mathbf{Q}\textbf{\upshape tr}({\mathbf{P}})$, with $\mathbf{P}$ and $\mathbf{Q}$ $k_1\times k_1$ and $k_2\times k_2$ positive definite matrices such that $C^{-1}_P\leq [\mathbf P]_{ii}\leq  C_P$ and $C^{-1}_Q\leq [\mathbf Q]_{ii} \leq C_Q$, for some finite positive reals $C_P$ and $C_Q$ independent of $i$, and spectral decomposition $\mathbf{P}=\boldsymbol{\Gamma}^P\boldsymbol{\Lambda}^P\boldsymbol{\Gamma}^{P\prime}$ and $\mathbf{Q}=\boldsymbol{\Gamma}^Q\boldsymbol{\Lambda}^Q\boldsymbol{\Gamma}^{Q\prime}$. For all $t,k\in \mathbb{Z}$ with $k\not=0$, $\mathbf{U}_t$ and $\mathbf{U}_{t-k}$ are independent.
		\end{itemize}
	\end{assumption} 
	
	Assumptions \ref{subass:loading}-\ref{subass:factor} matches Assumptions B and C in \cite{yu2021projected}. Assumption \ref{subass:armatrix} guarantees the stationarity of the MAR(1) model. Assumption \ref{subass:arshock} requires the factor innovations $\{\mathbf{U}_t\}$ to have positive definite covariance matrix.
	
	We characterize the idiosyncratic component through the following assumption.
	
	\begin{assumption}\textsc{(Idiosyncratic component).}
		\label{ass:idio} 
		\begin{itemize}
			\subass \label{subass:idmix} The process $\{\mathsf{e}_t\}$ is $\alpha$-mixing, i.e. there exists $\gamma>2$ such that $\sum^{\infty}\limits_{h=1} \alpha(h)^{1-2/\gamma} \leq \infty$, with $\alpha({h})=\sup_{t\in\mathbb Z}\sup_{A\in\mathcal{F}^t_{-\infty},B\in\mathcal{F}^{\infty}_{t+h}} |\Pr(A\cap B)-\Pr(A)\cap\Pr(B)|$ and $\mathcal{F}^s_{\tau}$ the $\sigma$-field generated by $\{\mathsf{e}_t: \tau \leq t\leq s\}$.
			
			\end{itemize}
			There exists a finite positive real $c$, independent of $T$, $p_1$, and $p_2$, such that:
			\begin{itemize}
			\subass \label{subass:idshocks} for all $t=1,\dots,T$, $i=1,\dots,p_1$, $j=1,\dots,p_2$, $\mathbb{E}[e_{tij}] = 0$, $\mathbb{E}[|e_{tij}|^4] \leq c$,
			$
			\mathbb{E}\left[\mathbf{E}_t\mathbf{E}'_t\right]=\mathbf{H}\textbf{\upshape tr}({\mathbf{K}})$ and $\mathbb{E}\left[\mathbf{E}'_t\mathbf{E}_t\right]=\mathbf{K}\textbf{\upshape tr}({\mathbf{H}})$,
			with $\mathbf{H}$ and $\mathbf{K}$ $p_1\times p_1$ and $p_2\times p_2$ positive definite matrices such that $\textbf{\upshape tr}({\mathbf{H}})=p_1$, and  $C^{-1}_H\leq [\mathbf H]_{ii}\leq  C_H$ and $C^{-1}_K\leq [\mathbf K]_{ii} \leq C_K$, for some finite positive reals $C_H$ and $C_K$ independent of $i$;\vspace{-.8cm}
			\subass \label{subass:ecorr}  for all $T,p_1,p_2\in\mathbb N$,
				$
				(Tp_1p_2)^{-1} \sum_{t,s=1}^{T}\sum_{i_1,i_2=1}^{p_1} \sum_{j_1,j_2=1}^{p_2} |\mathbb{E}\left[e_{ti_2j_1}e_{si_1j_2} \right] | \leq c;
				$
			\subass \label{subass:ecorr2} 
			for all $t=1,\dots,T$, $i ,l_1=1,\dots,p_1$, $j,h_1=1,\dots,p_2$, and all $T,p_1,p_2\in\mathbb N$,\\
$
			\sum_{s=1}^{T} \sum_{l_2=1}^{p_1} \sum_{h=1}^{p_2} |\mathbb{C}\text{\upshape ov}\left[e_{tij}e_{tl_1j},e_{sih}e_{sl_2h}\right]| \leq c$;
			$\sum_{s=1}^{T}\sum_{l=1}^{p_1} \sum_{h_2=1}^{p_2} |\mathbb{C}\text{\upshape ov}\left[e_{tij}e_{tih_1},e_{sl_2j}e_{sih_2}\right]| \leq c
			$;\\
			$\sum_{s=1}^{T} \sum_{i,l_2=1}^{p_1} \sum_{j,h_2=1}^{p_2} \left(|\mathbb{C}\text{\upshape ov}\left[e_{tij}e_{tl_1h_1},e_{sij}e_{sl_2h_2}\right]| + |\mathbb{C}\text{\upshape ov}\left[e_{tl_1j}e_{tih_1},e_{sl_2j}e_{sih_2}\right]|\right) \leq c$.
		\end{itemize}
	\end{assumption} 
	
	Assumptions \ref{ass:idio} closely matches Assumptions A and D in \cite{yu2021projected}. In particular, Assumption
	\ref{subass:idmix} controls serial idiosyncratic dependence by requiring them to be $\alpha$-mixing (see also \citealp[Assumption 1]{chen2021statistical}). Assumption \ref{subass:idshocks} 
	 imposes finite absolute fourth moments and requires $\mathbf{H}$ and $\mathbf{K}$ to be positive definite matrices, and, since $\mathbf{H}$ and $\mathbf{K}$ are only determined up to a positive constant, and only their Kronecker product, $\mathbf{K} \otimes \mathbf{H}$, is uniquely defined, we impose the constraint $\textbf{tr}(\mathbf{H}) = p_1$ (see, e.g., \citealp{viroli2012matrix}). Assumption \ref{subass:ecorr} controls cross-sectional idiosyncratic dependence across both rows and columns (see also \citealp[Assumption 2]{chen2021statistical}), while Assumption \ref{subass:ecorr2} bounds fourth order cumulants to allow for consistent estimation of the second order moments.

	 Finally, the dependence between common and idiosyncratic components is controlled through the following assumption, which matches Assumption E in \cite{yu2021projected}.
	
	\begin{assumption}\textsc{(Components dependence).} There exists a finite positive real $c$, independent of $T$, $p_1$, and $p_2$, such that:
		\label{ass:dep} 
		\begin{itemize}
			\subass \label{subass:fecorr} 
			$
			\mathbb{E}[\lVert T^{-1/2}\sum_{t=1}^{T} (\mathbf{F}_t\mathbf{v}'\mathbf{E}_t\mathbf{w})\rVert^2_F] \leq c
			$
			for any deterministic vector $\mathbf{v}$ and $\mathbf{w}$ with $\lVert \mathbf{v}\rVert=1$ and $\lVert \mathbf{w} \rVert =1$;
			\subass \label{subass:fecorr2} for all $T,p_1,p_2\in\mathbb N$ and all $i= 1,\dots,p_1$, $j=1,\dots,p_2$, \\
			$\left\lVert\sum_{h=1}^{p_2}\mathbb{E}\left[\boldsymbol{\zeta}_{ij}\otimes\boldsymbol{\zeta}_{ih}\right] \right\rVert_{\textnormal{\tiny max}} \leq c$;
			$\left\lVert\sum_{l=1}^{p_1}\mathbb{E}\left[\boldsymbol{\zeta}_{ij}\otimes\boldsymbol{\zeta}_{lj}\right] \right\rVert_{\textnormal{\tiny max}} \leq c$,
			
			and for all $T,p_1,p_2\in\mathbb N$ and all $i_1,l_1=1, \dots, p_1$, $j_1,h_1=1,\dots,p_2$, letting  $\boldsymbol{\zeta}_{ij}=\text{\upshape vec}(T^{-1/2}\sum_{t=1}^{T}\mathbf{F}_te_{tij})$,\\
 $\left\lVert\sum_{j_2,h_2=1}^{p_2}\mathbb{C}\text{\upshape ov}\left[\boldsymbol{\zeta}_{i_1j_1}\otimes\boldsymbol{\zeta}_{l_1h_1},\boldsymbol{\zeta}_{i_1j_2}\otimes\boldsymbol{\zeta}_{l_1h_2}\right] \right\rVert_{\textnormal{\tiny max}} \leq c$; 
$\left\lVert\sum_{i_2=1}^{p_1}\sum_{h_2=1}^{p_2}\mathbb{C}\text{\upshape ov}\left[\boldsymbol{\zeta}_{i_1j_1}\otimes\boldsymbol{\zeta}_{l_1h_1},\boldsymbol{\zeta}_{i_2j_1}\otimes\boldsymbol{\zeta}_{l_1h_2}\right] \right\rVert_{\textnormal{\tiny max}} \leq c$;\\
$\left\lVert\sum_{i_2,l_2=1}^{p_1}\mathbb{C}\text{\upshape ov}\left[\boldsymbol{\zeta}_{i_1j_1}\otimes\boldsymbol{\zeta}_{l_1h_1},\boldsymbol{\zeta}_{i_2j_1}\otimes\boldsymbol{\zeta}_{l_2h_1}\right] \right\rVert_{\textnormal{\tiny max}} \leq c$; $\left\lVert\sum_{l_2=1}^{p_1}\sum_{j_2=1}^{p_2}\mathbb{C}\text{\upshape ov}\left[\boldsymbol{\zeta}_{i_1j_1}\otimes\boldsymbol{\zeta}_{l_1h_1},\boldsymbol{\zeta}_{i_1j_2}\otimes\boldsymbol{\zeta}_{l_2h_1}\right] \right\rVert_{\textnormal{\tiny max}} \leq c$.

		\end{itemize} 
	\end{assumption} 


	Under the above assumptions we can then derive theoretical results on the convergence rates of the EM estimators for the loading and factor matrices $\widehat{\mathbf{R}}$, $\widehat{\mathbf{C}}$, and $\widehat{\mathbf{F}}_t$, defined in Section \ref{subsec:em}. 

	\begin{proposition}	\label{prop:EMloadCons}
		Recall the definitions $\widehat{\mathbf{R}}\equiv\widehat{\mathbf{R}}^{(n^*+1)}$ and $\widehat{\mathbf{C}}\equiv\widehat{\mathbf{C}}^{(n^*+1)}$ of the EM estimators of the loadings, with $n^*\geq 0$. Under Assumptions \ref{ass:common} through \ref{ass:dep}, there exist matrices $\widehat{\mathbf{J}}_1$ of size $k_1\times k_1$ and $\widehat{\mathbf{J}}_2$ of size $k_2\times k_2$ satisfying $\widehat{\mathbf{J}}_1\widehat{\mathbf{J}}'_1 \xrightarrow{p} \mathbb{I}_{k_1}$ and $\widehat{\mathbf{J}}_2\widehat{\mathbf{J}}'_2 \xrightarrow{p} \mathbb{I}_{k_2} $, such that, as $\min\left\{p_1,p_2,T\right\} \rightarrow \infty$,
		\begin{align}
		\min\left(\sqrt{Tp_1},\sqrt{Tp_2},{p_1p_2}\right) \left\lVert \frac{\widehat{\mathbf{R}} - \mathbf{R}\widehat{\mathbf{J}}_1}{\sqrt {p_1}} \right\rVert &= O_p \left( 1\right), \quad
		\min\left(\sqrt{Tp_2},\sqrt{Tp_1},{p_1p_2}\right) \left\lVert\frac{ \widehat{\mathbf{C}}- \mathbf{C}\widehat{\mathbf{J}}_2}{{\sqrt {p_2}}} \right\rVert = O_p \left( 1\right) ,\nonumber
		\end{align}
		and, for any given $i =1,\dots,p_1$ and $ j=1,\dots,p_2$, as $\min\left\{p_1,p_2,T\right\} \rightarrow \infty$,
		\begin{align}
		\min\left(\sqrt{Tp_1},\sqrt{Tp_2},{p_1p_2}\right) \left\lVert \widehat{\mathbf{r}}_{i} - \mathbf{r}_{i}\widehat{\mathbf{J}}_1 \right\rVert &= O_p \left(1 \right), \quad
		\min\left(\sqrt{Tp_2},\sqrt{Tp_1},{p_1p_2}\right) \left\lVert \widehat{\mathbf{c}}_{j} - \mathbf{c}_{j}\widehat{\mathbf{J}}_2 \right\rVert = O_p \left(1\right). \nonumber
		 \end{align}
	\end{proposition}

		These rates are comparable with those of the PE, which we use to initialize the EM algorithm. In particular, from \citet[Theorem 3.1]{yu2021projected}, the PE are such that, as $\min\left\{p_1,p_2,T\right\} \rightarrow \infty$,
		\begin{align}
		&\min\left(\sqrt{Tp_1},{Tp_2},{p_1p_2}\right) \left\lVert \frac{\widehat{\mathbf{R}}^{(0)} - \mathbf{R}\widehat{\mathbf{J}}_1}{\sqrt {p_1}} \right\rVert = O_p \left( 1\right), \label{eq:yu1}\\
		&\min\left(\sqrt{Tp_2},{Tp_1},{p_1p_2}\right) \left\lVert\frac{ \widehat{\mathbf{C}}^{(0)}- \mathbf{C}\widehat{\mathbf{J}}_2}{{\sqrt {p_2}}} \right\rVert = O_p \left( 1\right). \label{eq:yu2}
		\end{align}	
		Consider, for example, the error we have for the initial estimator ${\widehat{\mathbf{R}}^{(0)}}$. While the first term $\sqrt {T p_1}$ is the same as the one we find in Proposition \ref{prop:EMloadCons}, we also have a slower comparable term $\sqrt{Tp_2}$ coming from the initial estimator ${\widehat{\mathbf{C}}^{(0)}}$ of the columns. 
		This is due to the fact that at each iteration we need also to estimate the idiosyncratic variances, which require estimating both the row and the column loadings first. 
		The price to be paid is however negligible or even null when $p_1$ and $p_2$ are of the same order of magnitude.

		

	\begin{proposition}		\label{prop:EMfactorCons}
		Recall the definition $\widehat{\mathbf{F}}_{t}\equiv\text{\upshape unvec}(\widehat{\mathsf{f}}^{(n^*+1)}_{t|T})$ of the Kalman smoother estimator of the factors computed using the estimated parameters $\widehat{\boldsymbol\theta}^{(n^*+1)}$, with $n^*\geq 0$. Under Assumptions \ref{ass:common} through \ref{ass:dep}, and given $\widehat{\mathbf{J}}_1$ and $\widehat{\mathbf{J}}_2$ as defined in Proposition \ref{prop:EMloadCons},
		for any given $t=1,\dots,T$, as $\min\left\{p_1,p_2,T\right\} \rightarrow \infty$,
		\[
		\min\left(\sqrt{Tp_1},\sqrt{Tp_2},\sqrt{p_1p_2}\right)\left\lVert \widehat{\mathbf{F}}_t-\widehat{\mathbf{J}}^{-1}_1\mathbf{F}_t\widehat{\mathbf{J}}^{-\prime}_2 \right\rVert = O_p \left( 1 \right).
		\]
	\end{proposition}
	
	The rate in \citet[Theorem 3.5]{yu2021projected} for the PE is instead $\min(\sqrt T p_1,\sqrt T p_2, \sqrt{p_1p_2})$. The first term $\sqrt{p_1p_2}$ is the same and corresponds to the case of known parameters, while the other rates, due to the estimation of the parameters, are slower because they inherit the slower rates of the EM estimator of the loadings loadings. Nevertheless, if ${p_2}/{T}\to 0$ and ${p_1}/{T}\to 0$ as $\min(p_1,p_2,T)\to\infty$, then the Kalman smoother and the PE have the same $\sqrt{p_1p_2}$ rate, which would be obtained by vectorizing the data and estimating the factors via projection onto the true loadings.

%

We conclude by noting that many applications in economics and finance may involve settings in which one dimension of the matrix $\mathbf Y_t$, say $p_1$, together with the sample size $T$, diverges ($p_1, T \to \infty$), while the other dimension, say $p_2$, remains fixed ($p_2 < \infty$). Alternatively, both dimensions may diverge ($p_1, p_2 \to \infty$) while the sample size is fixed. An example of the former case arises in asset pricing applications, where the number of traded assets can be large, whereas the number of liquidity or volatility proxies is typically finite. An example of the latter case is firm-level balance sheet data, where observations are available for many firms and variables but only at low frequency (e.g., annually). Our estimator remains consistent under both scenarios.

	


	\section{The case of missing data}\label{subsec:banbura}
	If the data contains missing values, the estimation of the factors and their second moments with the Kalman smoother is still possible (see, e.g., \citealp[Section 6.4]{durbin2012time} for details). Now, since $\mathbb{E}_{\widehat{\boldsymbol{\theta}}_{(n)}} \left[\ell(\mathsf{F}_T;\boldsymbol{\theta}) | \mathsf{Y}_T \right]$ depends only on the factors but not on the data, its expression remains unchanged. Thus, the estimators of $\mathbf A\otimes \mathbf B$ and $\mathbf P\otimes \mathbf Q$ are also unchanged. However, $\mathbb{E}_{\widehat{\boldsymbol{\theta}}_{(n)}} \left[\ell(\mathsf{Y}_T|\mathsf{F}_T;\boldsymbol{\theta}) | \mathsf{Y}_T \right]$  depends on the data and therefore its expression is affected by missing values. It follows that we need to adjust the estimators of $\mathbf{R}$, $\mathbf{C}$, $\mathbf{H}$, and $\mathbf{K}$ in the M-step accordingly. To this end, we extend the procedure of \cite{banbura2014maximum} to the matrix setting. 
	
	Let $\mathbf{W}_t$ be a $p_1 \times p_2$ matrix with $(i,j)$ entry equal to zero if $y_{tij}$ is missing and equal to one otherwise. For any iteration $n\geq 0$, the estimators of the row and column loadings  are then modified to (see Appendix \ref{ap:EMmiss} for the derivation of these expressions):
	\[
	\setlength{\thinmuskip}{0mu}
	\setlength{\medmuskip}{0mu}
	\setlength{\thickmuskip}{0mu}
	\begin{split}
	\vect{\widehat{\mathbf{R}}^{(n+1)}} &= \left(\sum\limits_{t=1}^{T} \sum\limits_{s=1}^{p_1} \sum\limits_{q=1}^{p_1} \left( \left(\widehat{\mathbf{C}}^{(n)\prime}\mathbb{D}^{[s,q]}_{\mathbf{W}_t} \widehat{\mathbf{K}}^{(n)-1} \widehat{\mathbf{C}}^{(n)}\right) \star \left(\mathsf{f}^{(n)}_{t|T} \mathsf{f}^{(n)\prime}_{t|T} + \mathbf{\Pi}^{(n)}_{t|T} \right) \right) \otimes \left(\mathbb{E}^{(s,q)}_{p_1,p_1}\widehat{\mathbf{H}}^{(n)-1}\right) \right)^{-1} \\
	& \hspace*{1in} \times \left(\sum\limits_{t=1}^{T} \vect{\left(\mathbf{W}_t \circ \widehat{\mathbf{H}}^{(n)-1}\mathbf{Y}_t\widehat{\mathbf{K}}^{(n)-1}\right) \widehat{\mathbf{C}}^{(n)} \mathbf{F}^{(n)\prime}_{t|T}} \right), \\
	\hspace*{-0.25in}\vect{\widehat{\mathbf{C}}^{(n+1)}} &= \left(\sum\limits_{t=1}^{T} \sum\limits_{k=1}^{p_2} \sum\limits_{q=1}^{p_2}  \left( \left(\widehat{\mathbf{R}}^{(n+1)\prime}\mathbb{D}^{[s,q]}_{\mathbf{W}'_t}\widehat{\mathbf{H}}^{(n)-1} \widehat{\mathbf{R}}^{(n+1)} \right)\star \left(\mathbb{K}_{k_1k_2} \left( \mathsf{f}^{(n)}_{t|T}\mathsf{f}^{(n)\prime}_{t|T} + \mathbf{\Pi}^{(n)}_{t|T}  \right) \mathbb{K}'_{k_1k_2}\right) \right) \otimes \left(\mathbb{E}^{(s,q)}_{p_2,p_2} \widehat{\mathbf{K}}^{(n)-1} \right) \right)^{-1}  \\
	& \hspace*{1in} \times \left(\sum\limits_{t=1}^{T} \vect{\left(\mathbf{W}_t \circ \widehat{\mathbf{H}}^{(n)-1}\mathbf{Y}_t\widehat{\mathbf{K}}^{(n)-1}\right)' \widehat{\mathbf{R}}^{(n)} \mathbf{F}^{(n)}_{t}} \right).
	\end{split}
	\]
		
	Since $\mathbf{Y}_t$ contains missing observations, the initialization procedure described in Section~\ref{subsec:EMinit} cannot be applied directly. To address this, we introduce an additional preliminary step in which the missing entries are imputed using a suitable imputation method. To this end, an obvious choice consists in applying the imputation procedure for tensor factor models proposed by \cite{cen2024tensor}.
	
	Let $w_{t,i,j}$ be the entry $(i,j)$ of $\mathbf W_t$ and let $\eta$ be the fraction of missing data, i.e., such that 
	$(1-\eta)\le \min_{i=1,\ldots, p_1}\min_{i=j,\ldots, p_2}T^{-1}{\sum_{t=1}^T w_{t,i,j}}$. Then, from \citet[Corollary 1.1]{cen2024tensor}, we see that, under our assumptions, we have initial estimators such that,  as $\min\left\{p_1,p_2,T\right\} \rightarrow \infty$,
		\begin{align}
			&\min\left(\sqrt{Tp_2},p_1,\max\left(\sqrt T, \frac{1-\eta}{\eta}\right)\right) \left\lVert \frac{\widehat{\mathbf{R}}^{(0)} - \mathbf{R}\widehat{\mathbf{J}}_1}{\sqrt {p_1}} \right\rVert = O_p \left( 1\right), \label{eq:clifford1}\\
			&\min\left(\sqrt{Tp_1},p_2,\max\left(\sqrt T, \frac{1-\eta}{\eta}\right)\right) \left\lVert\frac{ \widehat{\mathbf{C}}^{(0)}- \mathbf{C}\widehat{\mathbf{J}}_2}{{\sqrt {p_2}}} \right\rVert = O_p \left( 1\right). \label{eq:clifford2}
			\end{align}
	These results extend the findings of \citet{xiong2023large} from the vector  to the matrix setting. They are comparable to the results for the initial PC-based estimator studied by \citet[Theorem 3.3]{yu2021projected} and \citet[Theorem 1]{chen2021statistical}. However, because no projection step is involved in the method of \citet{cen2024tensor}, the rates in \eqref{eq:clifford1}-\eqref{eq:clifford2} are not directly comparable to those for the PE analyzed in \citet[Theorem 3.1]{yu2021projected}.

	
	From the discussion after Proposition \ref{prop:EMloadCons} it is clear that, when dealing with missing data and initializing the EM algorithm with the estimator by \cite{cen2024tensor}, the consistency rates for our estimated loadings will be the minimum of the rates in \eqref{eq:clifford1} and \eqref{eq:clifford2}. Specifically, the loadings estimated via the EM algorithm are such that, as $\min\left\{p_1,p_2,T\right\} \rightarrow \infty$,
	\begin{align}
			&\min\left(\sqrt{Tp_1},\sqrt{Tp_2}, p_1,p_2, \max\left(\sqrt T,\frac{1-\eta}{\eta}\right)\right) \left\lVert \frac{\widehat{\mathbf{R}} - \mathbf{R}\widehat{\mathbf{J}}_1}{\sqrt {p_1}} \right\rVert = O_p \left( 1\right),\label{eq:conject1}\\ 
			& \min\left(\sqrt{Tp_1},\sqrt{Tp_2}, p_1,p_2,\max\left(\sqrt T, \frac{1-\eta}{\eta}\right)\right) \left\lVert\frac{ \widehat{\mathbf{C}}- \mathbf{C}\widehat{\mathbf{J}}_2}{{\sqrt {p_2}}} \right\rVert = O_p \left( 1\right). \label{eq:conject2}
	\end{align}
	The same rates hold for the single row or column estimators $\widehat{\mathbf r}_{i}$, $i=1,\ldots, p_1$, and $\widehat {\mathbf c}_{j}$, $j=1,\ldots, p_2$.
			
	By the same arguments, we expect the Kalman smoother computed using the EM estimator of the parameters to be a consistent estimator of the factors with a rate given by the minimum of the rates in \eqref{eq:clifford1} and \eqref{eq:clifford2} and $\sqrt{p_1p_2}$ corresponding to the rate for known parameters. Hence, as $\min\left\{p_1,p_2,T\right\} \rightarrow \infty$,
	\begin{align}
	&\min\left(\sqrt{Tp_1},\sqrt{Tp_2}, p_1,p_2,\max\left(\sqrt T, \frac{1-\eta}{\eta}\right),\sqrt{p_1p_2}\right) \left\lVert \widehat{\mathbf{F}}_t-\widehat{\mathbf{J}}^{-1}_1\mathbf{F}_t\widehat{\mathbf{J}}^{-\prime}_2 \right\rVert = O_p \left( 1 \right),\label{eq:conject3}
	\end{align}
	which matches the rate in \citet[Corollary 1.2]{cen2024tensor}.

	A formal proof of the statements \eqref{eq:conject1}, \eqref{eq:conject2}, and \eqref{eq:conject3} would follow \textit{verbatim} the same steps of the proofs of Propositions \ref{prop:EMloadCons} and \ref{prop:EMfactorCons}, respectively, but when using as initial estimators those satisfying \eqref{eq:clifford1}-\eqref{eq:clifford2} instead of the PE which satisfy \eqref{eq:yu1}-\eqref{eq:yu2}. Hence, such proof is omitted.
		

	\section{The case of non-stationary data}\label{subsec:unitroot}

	In this section, we consider the case in which $\{\text{vec}(\mathbf Y_t)\}$ is no more covariance stationary but is instead an $I(1)$ process, following a matrix factor model as   in \eqref{eq:dmfm1}, but under the assumption that $\{\text{vec}(\mathbf F_t)\}$ is $I(1)$ and $\{\text{vec}(\mathbf E_t)\}$ is $I(0)$.
	
	In a macroeconomic context, it is reasonable to assume that the elements of $\mathbf F_t$ are driven both by common trends, which are $I(1)$, and stationary components, which we could consider as common cycles (see, e.g., the applications in Section \ref{sec:EA} and in \citealp{barigozzi2023measuring}). In such a case, there exist a $k_1\times k_1$ invertible matrix $\bm{\mathcal R}$ and a $k_2\times k_2$ invertible matrix $\bm{\mathcal C}$ such that our model can be rewritten as (see Appendix \ref{sec:gianni} for the explicit expressions)
	\begin{equation}
	\mathbf Y_t =\mathbf R\mathbf F_t\mathbf C' +\mathbf E_t=\mathbf R \bm{\mathcal R}^{-1}  \bm{\mathcal R}\mathbf F_t \bm{\mathcal C}' \bm{\mathcal C}^{-1'}\mathbf C'+\mathbf E_t
	= \mathbf R_1 \mathbf G_{1t} \mathbf C_{1}' +  \mathbf R_0 \mathbf G_{0t} \mathbf C_{0}' +  \mathbf E_{t} ,\label{eq:gianni}
	\end{equation}
	where $\{\text{vec}(\mathbf G_{1t})\}$ is an $I(1)$ process with $\mathbf G_{1t}$ being $r_1\times r_2$, $\mathbf R_1$ being $p_1\times r_1$ and $\mathbf C_1$ being $p_2\times r_2$, while $\{\text{vec}(\mathbf G_{0t})\}$ and $\{\text{vec}(\mathbf E_t)\}$ are $I(0)$, with $\mathbf G_{0t}$ being $q_1\times q_2$, $\mathbf R_0$ being $p_1\times q_1$ and $\mathbf C_0$ being $p_2\times q_2$, so that $k_1=r_1+q_1$ and $k_2=r_2+q_2$.
	 
The model on the rightmost side of \eqref{eq:gianni} is introduced by \citet{chen2025inference} who propose PC-type estimators of both $\mathbf G_{1t}$ and $\mathbf G_{0t}$.
Here, instead we focus on the estimation of the common component, i.e., $\mathbf S_t=\mathbf R\mathbf F_t\mathbf C'$, which does not require identification of the common trends. If, according to \eqref{eq:gianni}, we make the assumption that $\{\text{vec}(\mathbf F_t)\}$ is a cointegrated process driven by $r_1r_2$ common trends, then the correct specification for its dynamics is either via a VECM or a VAR in levels. Hence, the DMFM must be estimated by applying the EM algorithm and the Kalman smoother on the levels of the data, i.e., without differencing them in order to achieve stationarity.

Under the assumption that $\{\text{vec}(\mathbf E_t)\}$ is stationary, we can still adopt the same initialization as in the stationary case, i.e., we can still use the PE as described in Section~\ref{subsec:EMinit}.  From \citet[Theorem 4]{chen2025inference} we see that, under our assumptions plus the assumption of cointegrated factors, we have initial estimators such that,  as $\min\left\{p_1,p_2,T\right\} \rightarrow \infty$,
		\begin{align}
		&\min\left(T\sqrt{p_2},T^2,T^{3/2}\sqrt {p_1}\right) \left\lVert \frac{\widehat{\mathbf{R}}^{(0)} - \mathbf{R}\widehat{\mathbf{J}}_1}{\sqrt {p_1}} \right\rVert = O_p \left( 1\right), \label{eq:gianni1}\\
		&\min\left(T\sqrt{p_1},T^2,T^{3/2}\sqrt {p_2}\right) \left\lVert\frac{ \widehat{\mathbf{C}}^{(0)}- \mathbf{C}\widehat{\mathbf{J}}_2}{{\sqrt {p_2}}} \right\rVert = O_p \left( 1\right). \label{eq:gianni2}
		\end{align}
		These results generalize to the matrix case the results by \citet{bai2004estimating} for the vector case.
		
	Once again, our results can then be directly adapted to this setting. From the discussion after Proposition \ref{prop:EMloadCons} it is clear that the consistency rates for our estimated loadings will be the minimum of the rates in \eqref{eq:gianni1} and \eqref{eq:gianni2}, i.e., the loadings estimated via the EM algorithm are such that, as $\min\left\{p_1,p_2,T\right\} \rightarrow \infty$,
	\begin{align}
		&\min\left(T\sqrt{p_2},T^2,T\sqrt {p_1}\right) \left\lVert \frac{\widehat{\mathbf{R}} - \mathbf{R}\widehat{\mathbf{J}}_1}{\sqrt {p_1}} \right\rVert = O_p \left( 1\right),\label{eq:conject1ns}\\ 
		&\min\left(T\sqrt{p_1},T^2,T\sqrt {p_2}\right) \left\lVert\frac{ \widehat{\mathbf{C}}- \mathbf{C}\widehat{\mathbf{J}}_2}{{\sqrt {p_2}}} \right\rVert = O_p \left( 1\right). \label{eq:conject2ns}
	\end{align}
	The same rates hold for the single row or column estimators $\widehat{\mathbf r}_{i}$, $i=1,\ldots, p_1$, and $\widehat {\mathbf c}_{j}$, $j=1,\ldots, p_2$. A formal proof of the statements \eqref{eq:conject1ns} and \eqref{eq:conject2ns} would follow \textit{verbatim} the same steps of the proofs of Propositions \ref{prop:EMloadCons} and \ref{prop:EMfactorCons}, respectively, but when using as initial estimators those satisfying \eqref{eq:gianni1}-\eqref{eq:gianni2} instead of the PE which satisfy \eqref{eq:yu1}-\eqref{eq:yu2}. Hence, such proof is omitted.

	By the same arguments, we expect the Kalman smoother computed using the EM estimator of the parameters to be a consistent estimator of the factors with rate the minimum between the rates in \eqref{eq:gianni1} and \eqref{eq:gianni2}, divided by $\sqrt T$ due to non-stationarity, and $\sqrt{p_1p_2}$ corresponding to the rate for known parameters.  
	Hence, as $\min\left\{p_1,p_2,T\right\} \rightarrow \infty$,
	\begin{align}
	&\min\left(\sqrt{T p_1},\sqrt {T p_2}, T^{3/2}, \sqrt{p_1p_2}\right) \left\lVert \widehat{\mathbf{F}}_t-\widehat{\mathbf{J}}^{-1}_1\mathbf{F}_t\widehat{\mathbf{J}}^{-\prime}_2 \right\rVert = O_p \left( 1 \right),\label{eq:conject3I1}
	\end{align}
	which matches the rate in \citet[Theorem 5]{chen2025inference}. A formal proof of this statement is, however, less straightforward and, thus, it  should  be regarded  just as an informed conjecture. 

	We conclude with three remarks. First, in the presence of missing observations, the EM algorithm can still be applied using the update modifications discussed in Section~\ref{subsec:banbura}. Since no iputation method exists for the non-stationary case, we propose to initialize the algorithm by running the EM procedure on a fully observed subset of the original matrix $\mathbf{Y}_t$. Simulation results in Appendix~\ref{app:sim} confirm the effectiveness of this approach.
	
	Second, the number of common trends can be determined by following the same approach proposed in \citet[Theorem 7]{chen2025inference} and based on eigenvalue ratios of suitable second moment matrices.
	
	Third, if all or some of the idiosyncratic components were non-stationary due to the presence of stochastic trends, then, the above approach would not be consistent. Indeed, in that case the loadings should be estimated from the differenced data as explained in \citet{bai2004panic} and \citet{barigozzi2021large} in the vector case. In this case, the estimated loadings would retain the same rates as the PE for stationary data given in \eqref{eq:yu1} and \eqref{eq:yu2}. Moreover, the Kalman smoother should be run by adding as additional latent states all those idiosyncratic components which are $I(1)$, in a way similar to the approach proposed by \citet{banbura2014maximum} for serially correlated, but stationary, idiosyncratic components. This case is left for further research.

%

	\section{Simulation study}
	\label{sec:Sim}

	We perform Monte Carlo simulations in order to assess the finite sample properties of the proposed EM estimator and the Kalman smoother. 
	For $t=1,\dots,T$, we generate observations according to the following DMFM:
	\[
	\begin{array}{ll}
	\qquad \mathbf{Y}_t = \mathbf{R} \mathbf{F}_t \mathbf{C}' + \mathbf{E}_{t}, \quad  \mathbf{F}_t = \mathbf{A} \mathbf{F}_{t-1}\mathbf{B}' + \mathbf{U}_t, & \textbf{U}_t\sim \mathfrak{D}_{k_1,k_2}(\mathbf{0}_{k_1,k_2}, \mathbb{I}_{k_1}, \mathbb{I}_{k_2}), \\
	\qquad\mathbf{E}_t = \mathbf{D} \mathbf{E}_{t-1}\mathbf{G}' + \mathbf{V}_t, & \mathbf{V}_t\sim \mathfrak{D}_{p_1,p_2}(\mathbf{0}_{p_1,p_2}, \mathbf{H}, \mathbf{K}),
	\end{array}	
	\]
	 where $\mathfrak{D}_{k_1,k_2}(\mathbf{0}_{k_1,k_2}, \mathbf{P}, \mathbf{Q}) $ and $\mathfrak{D}_{p_1,p_2}(\mathbf{0}_{p_1,p_2}, \mathbf{H}, \mathbf{K})$ denote general matrix distributions of dimensions $k_1\times k_2$ and $p_1\times p_2$, centered on zero, and with covariance matrices $\mathbf{P}, \mathbf{Q}$ and  $\mathbf{H}, \mathbf{K}$, respectively. We consider $\mathfrak{D}$ either to be a matrix normal (N) or a matrix skew-t (St) distribution with 4 degrees of freedom.
	The loading matrices are such that $\left[\mathbf{R}\right]_{ij},\left[\mathbf{C}\right]_{ij} \sim \mathcal{U}(-1,1)$. 
	The matrix of latent factors follows a MAR(1) process with $\mathbf{B}=\mu \frac{\mathbf{B}^*}{\lvert \nu^{(1)}\left(\mathbf{B}^{*} \otimes \mathbf{A}\right) \rvert}$ where 
	$[\textbf{B}^*]_{ii}, [\textbf{A}]_{ii} \sim\mathcal{U}(0.7,0.9)$ and $[\textbf{B}^*]_{ij}, [\textbf{A}]_{ij} \sim \mathcal{U}(0,0.5)$ for $i\ne j$.
%
	Note that $\mu$ defines the maximum eigenvalue of the matrix $\mathbf{B}\otimes\mathbf{A}$ allowing us to control whether the matrix factor process is stationary or not. In particular when $\mu=1$, the simulated factors are driven by one common $I(1)$ trend. Throughout, we set $k_1=2$ and $k_2=2$. 
	
	The idiosyncratic components follow a MAR(1) process with 
	\[
	\left[\textbf{D}\right]_{ij},\left[\textbf{G}\right]_{ij} =\begin{cases} \mathcal{U}(0,\delta), & i=j, \\0, & i \not=j, \end{cases}  \qquad \left[\textbf{H}\right]_{ij}, \left[\textbf{K}\right]_{ij} = \begin{cases}	\mathcal{U}(0.7,1.2), & i=j, \\ \tau^{|i-j|},& i \not=j, \end{cases}
	\]
	with $\tau$  and $\delta$ controlling the degree of cross-sectional and serial correlation, respectively. 
		
			For each performance measure considered, we report its average and standard deviation over 100 replications. We use the column space distance $\mathcal{D}(\mathbf{R}, \widehat{\mathbf{R}})$ and $\mathcal{D}(\mathbf{C}, \widehat{\mathbf{C}})$ to evaluate the loadings matrices estimators, which, for any $m\times n$ matrix $\mathbf{A}$, is defined as
$
	\mathcal{D}(\mathbf{A}, \widehat{\mathbf{A}}) = \big\lVert \widehat{\mathbf{A}} \left(\widehat{\mathbf{A}}'\widehat{\mathbf{A}}\right)^{-1}\widehat{\mathbf{A}}'-\mathbf{A}\left(\mathbf{A}'\mathbf{A}\right)^{-1}\mathbf{A}' \big\rVert.
$
	We also consider the mean squared error in recovering the signal $\mathbf{S}_t=\mathbf{R}\mathbf{F}_t\mathbf{C}'_t$, defined as
	$
	\textrm{MSE}_{\textbf{S}} = (Tp_1p_2)^{-1} \sum_{t=1}^{T} \lVert \widehat{\mathbf{S}}_t - \mathbf{S}_t  \rVert^2_{\textrm{F}},
	$
	where $\widehat{\mathbf{S}}_t= \widehat{\mathbf{R}}\widehat{\mathbf{F}}_t\widehat{\mathbf{C}}$  denotes the estimated signal, as described in Sections \ref{sec:QML} or  \ref{subsec:unitroot}.

	In Table \ref{tab:EMvsPCA} we compare the performance of the EM estimators for the loading and factor matrices with those of the PE both in the stationary and the $I(1)$ cases. The EM algorithm improves upon PE across all the different settings. Furthermore, in Figure \ref{fig:EMconv} we show for one replication the log-likelihood as function of the number of iterations of the EM algorithm. As expected the log-likelihood increases monotonically and the first few iterations seem to be the most important ones.		
	
	\begin{figure}[t!]
		\caption{Log-likelihood as a function of EM iterations.}
		\centering
		\vskip .2cm
		\begin{tabular}{cc}
				Stationary case & $I(1)$ case\\
		\includegraphics[width=0.3\textwidth]{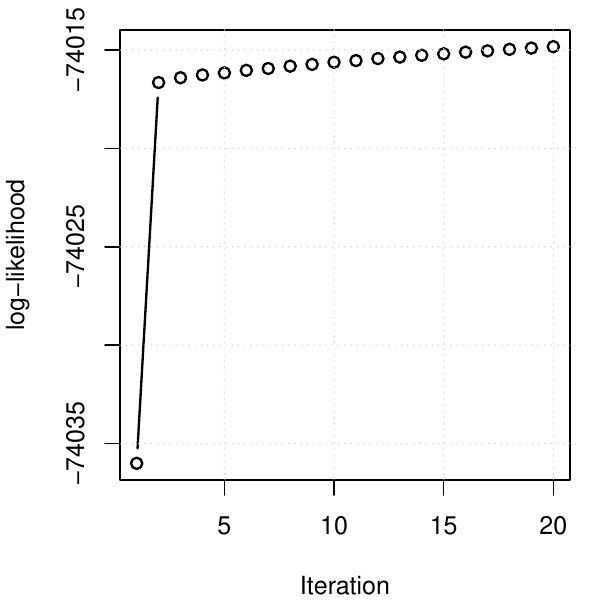}&
		\includegraphics[width=0.3\textwidth]{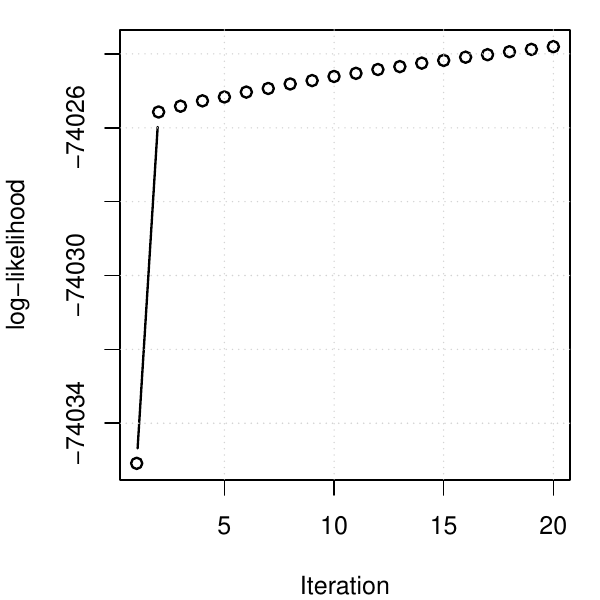}
		\end{tabular}  
		\label{fig:EMconv}
	\end{figure}
	
	\begin{table}[h!]
		\caption{Average and standard deviation (in parenthesis) of the ratio between the performance  of the EM estimator and PE over 100 replications, for each of $\mathcal{D}(\mathbf{R}, \widehat{\mathbf{R}})$, $\mathcal{D}(\mathbf{C}, \widehat{\mathbf{C}})$, and $\textrm{MSE}_{\textbf{S}}$. }
		\centering
		\begin{tabular}{cccccccccccc}
			\toprule
			& & & & & & \multicolumn{3}{c}{$T=100$} & \multicolumn{3}{c}{$T=400$}   \\
			\midrule
			$\mu$ & $\delta$ & $\tau$ & $\mathfrak{D}$ & $p_1$ & $p_2$ &  $\mathcal{D}(\mathbf{R}, \widehat{\mathbf{R}})$ & $\mathcal{D}(\mathbf{C}, \widehat{\mathbf{C}})$ & $\textrm{MSE}_{\textbf{S}}$ & $\mathcal{D}(\mathbf{R}, \widehat{\mathbf{R}})$ & $\mathcal{D}(\mathbf{C}, \widehat{\mathbf{C}})$ & $\textrm{MSE}_{\textbf{S}}$   \\ 
			\midrule
			\multirow{2}{*}{$0.7$} & \multirow{2}{*}{$0$} & \multirow{2}{*}{$0$} & \multirow{2}{*}{N}& 20 & 20 & 0.98 & 0.97 & 0.92& 0.98 & 0.96 & 0.91 \\
			&&&& &   & (0.05) & (0.05) & (0.03) & (0.05) & (0.05) & (0.01)  \\
			&&&& 10 & 30 & 0.98 & 0.96 & 0.90 & 0.96 & 0.96 & 0.90 \\
			&&&& &   & (0.09) & (0.05) & (0.03) & (0.10) & (0.05) & (0.01) \\
			\midrule
			 \multirow{2}{*}{$0.7$} & \multirow{2}{*}{$0.7$} & \multirow{2}{*}{$0.5$} & \multirow{2}{*}{N} & 20 & 20 & 0.80 & 0.71 & 0.73 & 0.74 & 0.65 & 0.75 \\
			&&&& &   & (0.07) & (0.08) & (0.04) & (0.06) & (0.06) & (0.02)  \\
			&&&& 10 & 30 & 0.87 & 0.68 & 0.70 & 0.82 & 0.63 & 0.75  \\
			&&&& &   & (0.10) & (0.08) & (0.05) & (0.10) & (0.06) & (0.03) \\
			\midrule
			\multirow{2}{*}{$0.7$} &\multirow{2}{*}{$0$} & \multirow{2}{*}{$0$} & \multirow{2}{*}{St} & 20 & 20 & 0.97 & 0.97 & 0.91 & 0.96 & 0.98 & 0.91  \\
			&&&& &   & (0.07) & (0.05) & (0.07)  & (0.06) & (0.06) & (0.03)\\
			&&&& 10 & 30 & 0.97 & 0.97 & 0.9 & 0.96 & 0.95 & 0.89   \\
			&&&& &   & (0.11) & (0.05) & (0.05) & (0.12) & (0.06) & (0.03) \\
			\midrule
			\multirow{2}{*}{$0.7$} &\multirow{2}{*}{$0.7$} & \multirow{2}{*}{$0.5$} & \multirow{2}{*}{St} &  20 & 20 & 0.9 & 0.86 & 0.95 & 0.92 & 0.9 & 1.01  \\
			&&&& &   &  (0.11) & (0.15) & (0.14) & (0.08) & (0.1) & (0.05)  \\
			&&&& 10 & 30 & 0.99 & 0.63 & 0.81 & 1.06 & 0.58 & 0.89   \\
			&&&& &   & (0.15) & (0.19) & (0.18) & (0.13) & (0.18) & (0.14)  \\
			\midrule
			\multirow{2}{*}{$1$} &\multirow{2}{*}{$0$} & \multirow{2}{*}{$0$} & \multirow{2}{*}{N}&  20 & 20 & 0.99 & 0.98 & 0.92 & 0.99 & 0.98 & 0.91 \\
			&&&&&& (0.06) & (0.06) & (0.02) &  (0.06) & (0.06) & (0.01) \\
			&&&& 10 & 30 &  0.97 & 0.97 & 0.9  & 0.96 & 0.97 & 0.9 \\
			&&&& &   & (0.1) & (0.04) & (0.03) & (0.14) & (0.05) & (0.02) \\
			\midrule
			\multirow{2}{*}{$1$} &\multirow{2}{*}{$0.7$} & \multirow{2}{*}{$0.5$} & \multirow{2}{*}{N} &  20 & 20 & 0.84 & 0.8 & 0.8 & 0.79 & 0.76 & 0.8 \\
			&&&& &   &  (0.07) & (0.07) & (0.03) & (0.07) & (0.07) & (0.02) \\
			&&&& 10 & 30 & 0.96 & 0.8 & 0.8 & 0.94 & 0.77 & 0.82   \\
			&&&& &   & (0.1) & (0.08) & (0.04) & (0.11) & (0.07) & (0.02) \\
			\bottomrule  
		\end{tabular}
		\label{tab:EMvsPCA}
	\end{table}
	
	We then introduce missing observations in the data generating process. 
	After simulating the matrix $\mathbf Y_t$ with no missing values as described above we introduce two patterns of missing observations widely seen in empirical application: (i) randomly missing, i.e., removing at each point in time observations of $\mathbf Y_t$ at random with a constant probability $\pi=\{25\%,50\%\}$; (ii) block missing, i.e., when a fixed portion $\pi=\{25\%,50\%\}$ of $\mathbf{Y}_t$ is removed for a given period of time. For the case of block missing we remove the bottom-right quarter and the right-half of $\mathbf{Y}_t$ for the first half of the time series when $\pi=25\%$ and $\pi=50\%$, respectively. 
	
	In this case, besides $\mathcal{D}(\mathbf{R}, \widehat{\mathbf{R}})$, $\mathcal{D}(\mathbf{C}, \widehat{\mathbf{C}})$, and $\textrm{MSE}_{\textbf{S}}$, we also investigate the goods of our imputation method by computing $
	\textrm{MSE}_{\textbf{Y}^{(0)}} =(Tp_1p_2)^{-1} \sum_{t=1}^{T} \lVert (\widehat{\mathbf{S}}_t - \mathbf{Y}_t)\circ (\mathbf{1}_{p_1,p_2}-\mathbf{W}_t)  \rVert^2_{\textrm{F}}$,
	where $\widehat{\mathbf{S}}_t= \widehat{\mathbf{R}}\widehat{\mathbf{F}}_t\widehat{\mathbf{C}}$ denotes the estimated signal, as described in Section \ref{subsec:banbura}, and $\mathbf{W}_t$ is the binary matrix indicating  observed entries.
	
	Following the discussion in Section \ref{subsec:banbura}, we adopt the imputation method proposed by \cite{cen2024tensor} to fill in missing values prior to initializing the EM algorithm. Because this method requires stationarity, we restrict the analysis to stationary settings. Table \ref{tab:EMvsPCAmiss} reports summary statistics for the relative performance of the EM estimator compared to the PE estimator applied to the imputed data. The results indicate that the EM algorithm yields systematically improved estimates over PE. Additional simulation results based on initialization using a balanced subpanel are in Appendix \ref{app:sim}.

		\begin{table}[h!]
		\caption{Average and standard deviation (in parenthesis) of ratio between the performance  of the EM estimator and PE  over 100 replications, for each of $\mathcal{D}(\mathbf{R}, \widehat{\mathbf{R}})$, $\mathcal{D}(\mathbf{C}, \widehat{\mathbf{C}})$, $\textrm{MSE}_{\textbf{S}}$, and $\textrm{MSE}_{\textbf{Y}^{(0)}}$. }
		\centering
		\begin{tabular}{cccccccccccc}
			\toprule
			& & & & \multicolumn{4}{c}{$T=100$} & \multicolumn{4}{c}{$T=400$}   \\
			\midrule
			$\mathfrak{D}$ & $\pi$ & $p_1$ & $p_2$ &  $\mathcal{D}(\mathbf{R}, \widehat{\mathbf{R}})$ & $\mathcal{D}(\mathbf{C}, \widehat{\mathbf{C}})$ & $\textrm{MSE}_{\textbf{S}}$ & $\textrm{MSE}_{\textbf{Y}^{(0)}}$& $\mathcal{D}(\mathbf{R}, \widehat{\mathbf{R}})$ & $\mathcal{D}(\mathbf{C}, \widehat{\mathbf{C}})$ & $\textrm{MSE}_{\textbf{S}}$ & $\textrm{MSE}_{\textbf{Y}^{(0)}}$   \\ 
			\midrule
			& &&& \multicolumn{8}{c}{\emph{Randomly missing}} \\
			\midrule
			\multirow{2}{*}{N} & \multirow{2}{*}{$25\%$} & 20 & 20 & 0.9 & 0.91 & 0.94 & 1.00 & 0.78 & 0.81 & 0.95 & 1.00   \\
			&& &  & (0.11) & (0.08) & (0.03) & (0.00) & (0.1) & (0.08) & (0.01) & (0.00)\\
			&& 10 & 30 & 0.72 & 0.92 & 0.89 & 1.00 & 0.48 & 0.87 & 0.89 & 1.00 \\
			&& &   & (0.15) & (0.06) & (0.03) & (0.00) & (0.08) & (0.06) & (0.02) & (0.00) \\
			\midrule
			\multirow{2}{*}{N} & \multirow{2}{*}{$50\%$} & 20 & 20 & 0.67 & 0.7 & 0.81 & 0.99 & 0.57 & 0.61 & 0.88 & 1.00 \\
			&& &  & (0.12) & (0.11) & (0.05) & (0) & (0.1) & (0.11) & (0.02) & (0.00) \\
			&& 10 & 30 &  0.53 & 0.72 & 0.77 & 0.99 & 0.32 & 0.67 & 0.82 & 0.99 \\
			&& &   &   (0.14) & (0.09) & (0.04) & (0) & (0.08) & (0.09) & (0.02) & (0.00) \\
			\midrule
			\multirow{2}{*}{St} &\multirow{2}{*}{$25\%$} & 20 & 20 & 0.77 & 0.89 & 0.89 & 0.99 & 0.65 & 1.01 & 0.93 & 0.99  \\
			&& &  & (0.14) & (0.16) & (0.11) & (0.02) & (0.12) & (0.17) & (0.06) & (0.00) \\
			&& 10 & 30 & 0.70 & 0.87 & 0.85 & 0.99 & 0.55 & 0.93 & 0.88 & 0.99 \\
			&& &   & (0.20) & (0.12) & (0.08) & (0.00)  & (0.14) & (0.06) & (0.05) & (0.00)  \\
			\midrule
			\multirow{2}{*}{St} & \multirow{2}{*}{$50\%$} & 20 & 20 & 0.43 & 0.55 & 0.62 & 0.97 & 0.37 & 0.63 & 0.77 & 0.99 \\
			&& &  & (0.12) & (0.14) & (0.14) & (0.05) & (0.09) & (0.17) & (0.12) & (0.02)  \\
			&& 10 & 30 & 0.41 & 0.52 & 0.61 & 0.96 &  0.32 & 0.69 & 0.76 & 0.98 \\
			&& &   & (0.16) & (0.14) & (0.13) & (0.04) &  (0.13) & (0.14) & (0.10) & (0.01) \\
			\midrule
			&&&& \multicolumn{8}{c}{\emph{Block missing}} \\
			\midrule
			\multirow{2}{*}{N} &\multirow{2}{*}{$25\%$} & 20 & 20 & 0.82 & 0.92 & 0.92 & 0.99 & 0.65 & 0.87 & 0.94 & 1.00  \\
			&& &  & (0.17) & (0.06) & (0.06) & (0.01) & (0.14) & (0.07) & (0.02) & (0.00) \\
			&& 10 & 30 &  0.86 & 0.97 & 0.91 & 1.00 & 0.67 & 0.95 & 0.9 & 1.00 \\
			&& &   &  (0.15) & (0.05) & (0.03) & (0.00) & (0.14) & (0.05) & (0.01) & (0.00)\\
			\midrule
			\multirow{2}{*}{N} &\multirow{2}{*}{$50\%$} & 20 & 20 & 0.82 & 0.92 & 0.92 & 0.99 & 0.73 & 0.7 & 0.87 & 0.99  \\
			&& &  & (0.17) & (0.06) & (0.06) & (0.01) & (0.1) & (0.12) & (0.04) & (0.00)  \\
			&& 10 & 30 & 0.76 & 0.88 & 0.82 & 0.98  & 0.54 & 0.84 & 0.85 & 0.99\\
			&& & &  (0.14) & (0.12) & (0.07) & (0.02) & (0.09) & (0.1) & (0.02) & (0.00) \\
			\midrule
			\multirow{2}{*}{St} &\multirow{2}{*}{$25\%$} & 20 & 20 & 0.87 & 0.98 & 0.91 & 0.98 & 0.74 & 1.09 & 0.96 & 0.99  \\
			&& &  & (0.16) & (0.12) & (0.14) & (0.04) & (0.13) & (0.08) & (0.03) & (0.00) \\
			&& 10 & 30 & 0.87 & 0.92 & 0.88 &  0.98 & 0.68 & 0.95& 0.90 & 0.99 \\
			&& &   &  (0.20) & (0.09) & (0.08) & (0.02) & (0.16) & (0.02) & (0.03) & (0.00) \\
			\midrule
			\multirow{2}{*}{St} &\multirow{2}{*}{$50\%$} & 20 & 20 & 0.86 & 0.69 & 0.69 & 0.93 & 0.70 & 0.80& 0.83 & 0.98 \\
			&& &  & (0.08) & (0.23) & (0.22) & (0.13) & (0.08) & (0.19) & (0.11) & (0.01) \\
			&& 10 & 30 & 0.81 & 0.77 & 0.74 & 0.95  & 0.59 & 0.85 & 0.78 & 0.9 8\\
			&&& & (0.21) & (0.21) & (0.2) & (0.08)& (0.14) & (0.12) & (0.11) & (0.02) \\
			\bottomrule 
		\end{tabular}
		\label{tab:EMvsPCAmiss}
	\end{table}

	\section{Empirical applications}	\label{sec:App}
	\paragraph{Forecasting volatilities.}
	
	Despite the abundant use of high-frequency data in the financial econometrics literature, their availability is often limited to major equity indices or large U.S. stocks  \citep{bollerslev2018risk}, limiting the chance of building high-frequency-based estimates of volatility for a large number of traded companies. Given that volatility measures tend to covary across assets \citep{barigozzi2016generalized}, a natural question is whether high-frequency-based volatility measures on a set of assets can be used to improve volatility estimates for a set of assets for which only daily observations are available. 
	
	We collect daily returns and realized measures for $30$ assets listed in the S\&P500 under the Financial GICS sector. The data covers the period that goes from the beginning of 2006 to the end of 2010, covering the Great Financial Crisis.  We consider $10$ realized measures of the daily integrated volatility. In particular, we have 7 high-frequency measures based on intra-daily data\footnote{These are: 5-min and 15-min realized variance, autocorrelation-corrected 5-min realized variance \citep{hansen2006realized}, realized range \citep{christensen2007realized}, realized kernel \citep{barndorff2008designing}, pre-averaged realized variance \citep{jacod2009microstructure}, maximum likelihood realized variance \citep{xiu2010quasi}.} and three low-frequency proxies based on the opening (O), highest (H), lowest (L), and closing (C) daily prices (OHLC hereafter).\footnote{These are: the daily range $(H - L)^2/(4\log 2)$, the O/C adjusted daily range $0.5(H - L)^2 - (2\log 2 - 1)(C - O)^2$, and the O/C adjusted daily range $(H -C)(H -O)+(L-C)(L-O)$.} These measures are available only for half of the stocks in the sample, as we have access to high-frequency data solely for those assets. For the remaining stocks we only have daily data, and can therefore compute just the three OHLC variance  proxies.  We thus obtain a matrix time series of $p_2=10$ daily variance proxies on $p_1=30$ assets for $T=1259$ days, with a block of missing observations corresponding to 35\% of the total number of possible observations which is $p_1p_2T$.
	
	Our data can be modeled as a 2-layers hierarchical factor model which in turn is equivalent to a matrix factor model. First, let $\sigma^2_{i,t}$ be the $t$th day latent variance of the $i$th asset  and define $\boldsymbol{\widetilde{\sigma}}^2_{t}=(\widetilde{\sigma}^2_{1,t},\dots,\widetilde{\sigma}^2_{p_1,t})'$, with $\widetilde{\sigma}^2_{i,t}=\sigma^2_{i,t}/\bar{\sigma}^2_{i,t}$ and $\overline{\sigma}^2_{i}=(\prod^T_{t=1}\sigma^2_{i,t})^{{1}/{T}}$, for all $i=1,\dots,p_1$. We assume that the vector of centered latent log-variances for all assets, $\log(\boldsymbol{\widetilde{\sigma}}^2_{t})$, follows a factor model with $\mathbf{f}_{t}$ being a vector of $k_1$ common factors, e.g., representing the stock market, that is	
\begin{align}
\log\left(\boldsymbol{\widetilde{\sigma}}^2_{t}\right) = \mathbf{R} \textbf{f}_{t} + \boldsymbol{\varepsilon}_{t},\label{eq:strato1}
\end{align}
where $\mathbf{R}$ is a $p_1 \times k_1$ loading matrix and $ \boldsymbol{\varepsilon}_{t}$ contains the idiosyncratic component for each asset.

 Second, let $\mathbf{s}_{i,t}$ be the vector of $p_2$ variance proxies for the $i$th asset on the $t$th day and define $\widetilde{\mathbf{s}}_{i,t}=(\widetilde{s}_{i,1,t},\dots,\widetilde{s}_{1,p_2,t})'$, with $\widetilde{s}_{i,j,t}=s_{i,j,t}/\bar{s}_i$ and $\overline{s}_i=(\prod^T_{t=1}\prod_{j=1}^{p_2} s_{i,j,t})^{{1}/{(Tp_2)}}$. It is reasonable to assume that the centered vector of log-variance proxies of asset $i$ follows a one factor model, where the common factor is the latent volatility $\log (\widetilde{\sigma}^2_{i,t})$ of asset $i$, that is
 \begin{align}
 \log \left(\widetilde{\mathbf{s}}_{i,t}\right) = \textbf{c} \log \left(\widetilde{\sigma}^2_{i,t}\right) + \boldsymbol{\epsilon}_{i,t},\label{eq:strato2}
 \end{align}
 where $\textbf{c}$ is a $p_2$-dimensional loading vector and $\boldsymbol{\epsilon}_{i,t}$ is a $p_2$-dimensional vector contaning the measurement errors of all variance proxies of asset $i$. 

It follows that the $p_2$ vector of observed centered log-transformed variance proxies for the $i$th asset follows a 2-layer factor model. Indeed, by substituting the transposed of
 \eqref{eq:strato1}
into \eqref{eq:strato2}, we have
 \begin{align}
 \log \left(\widetilde{\mathbf{s}}_{i,t}\right)' = \mathbf{r}_i' \mathbf f_t  \textbf{c}'+  {\varepsilon}_{i,t} \textbf{c} '+ \boldsymbol{\epsilon}_{i,t}',\label{eq:2layer}
 \end{align}
 where $\mathbf r_i'$ is the $i$th row of $\mathbf R$ and ${\varepsilon}_{i,t}$ is the $i$th element of $\boldsymbol{\varepsilon}_{t}$. By letting $\mathbf{Y}_t = \left(\log \left(\widetilde{\mathbf{s}}_{1,t}\right)', \dots, \log \left(\widetilde{\mathbf{s}}_{p_1,t}\right)' \right)'$, we see that \eqref{eq:2layer} is 
 equivalent to the matrix factor model in \eqref{eq:dmfm1} with $\mathbf{E}_t= \boldsymbol {\varepsilon}_{t}\textbf{c}' + 
 (\boldsymbol{\epsilon}_{1,t}'\cdots\boldsymbol{\epsilon}_{p_1,t}')'$.
	For economic reasons we fix the number of columns factors  to $k_2=1$, indeed, this corresponds to the number of latent variance factor underlying all proxies. As for the number of row factors the eigenvalue-ratio criterion by \citet{cen2024tensor} suggests to set $k_1=1$.
	
	
	We then conduct a forecasting exercise. 
	We define an in-sample window of $750$ observations for the models estimation and leave $509$ observations for the out-of-sample forecast evaluation. We estimate a DMFM on the in-sample window using our proposed EM algorithm modeling $f_t$, which, since $k_1=k_2=1$ is now a scalar, as an AR(1), and obtain one-step-ahead forecasts of $\widetilde{\sigma}^2_{i,t}$ as 
	$\widehat{\widetilde{\sigma}}^2_{i,t|t-1} 
	= \exp(\widehat{r}_{i} \, \widehat{A} \, \widehat{f}_{t-1|t-1})$, where $\widehat{r}_i$ is the estimated row loading for the $i$th asset and $\widehat{A}$ is the estimated autoregressive coefficient.
	For comparison, we also estimate an analogous DMFM on the in-sample window using the proposed EM algorithm, but restricted to the $15 \times 3$ sub-matrix of assets for which only low-frequency volatility measures are available, i.e., to a balanced subpanel of the considered dataset.


Table \ref{tab:MSEvol} reports the out-of-sample MSE ratios comparing the model estimated on the reduced matrix to that estimated on the full dataset, along with the p-values from the \cite{diebold1995comparing} test for each financial asset at the daily frequency. The out-of-sample MSE for the model estimated on the reduced matrix is higher for twelve out of fifteen assets, reaching up to 7\% in some cases. According to the Diebold-Mariano test of equal predictive accuracy, these differences are statistically significant for nine assets. This finding underscores the advantage of incorporating high-frequency volatility proxies from assets that covary with those for which we only have access to low-frequency measures, and thus shows the importance of having a method which allows us to deal with panels with missing observations.

\begin{table}[htbp]
	\centering
	\caption{Out-of-sample MSE ratios and $p$-values from the \cite{diebold1995comparing} test, comparing model performance on the reduced matrix versus the full dataset for each financial asset observed at the daily frequency; $^{*}$ indicates $p$-values below 0.10.}
	\begin{tabular}{lcl}
		\toprule
		Ticker & MSE Ratio & DM \\
		\hline
		AMP & $1.042$ & $0.094^{*}$ \\
		BEN & $1.042$ & $0.020^{*}$ \\
		CMA & $0.999$ & $0.530$ \\
		CME & $1.055$ & $0.073^{*}$ \\
		FITB & $1.031$ & $0.092^{*}$ \\
		HBAN & $1.008$ & $0.311$ \\
		ICE & $1.002$ & $0.441$ \\
		MCO & $1.029$ & $0.166$ \\
		MTB & $1.065$ & $0.069^{*}$ \\
		NDAQ & $1.072$ & $0.033^{*}$ \\
		NTRS & $0.995$ & $0.765$ \\
		SCHW & $0.981$ & $0.806$ \\
		TROW & $1.058$ & $0.046^{*}$ \\
		USB & $1.048$ & $0.041^{*}$ \\
		ZION & $1.034$ & $0.076^{*}$ \\
		\toprule
	\end{tabular}
	\label{tab:MSEvol}
\end{table}

\paragraph{Macroeconomic trends in the Euro Area.}\label{sec:EA}
We analyze a collection of macroeconomic indicators from EA countries.\footnote{The data is available at \href{https://zenodo.org/doi/10.5281/zenodo.10514667}{https://zenodo.org/doi/10.5281/zenodo.10514667}.} Specifically, we consider 39 real macroeconomic indicators across three categories: National Accounts, Labor Market Indicators, and Industrial Production and Turnover. These indicators are collected at either monthly or quarterly frequency for eight countries: Austria, Belgium, Germany, Spain, France, Italy, the Netherlands, and Portugal, resulting in a matrix-valued time series of dimensions $(p_1, p_2) = (8, 39)$. The dataset spans the period from January 2000 to November 2024 ($T = 299$). 

By applying the eigenvalue ratio criterion by \citet{yu2021projected} on the differenced data we find evidence of one row factor and three column factors, but we cannot say whether any of these is $I(1)$ or stationary. To this end we can instead apply the eigenvalue ratio approach proposed by \citet{chen2025inference} on the non-differenced data,  showing evidence of just one $I(1)$ common factor, i.e., a common trend.  Since the factor matrix is actually a 3-dimensional vector this implies that the process of latent factors is indeed cointegrated with two cointegrating relations. As explained before, and differently from \citet{chen2025inference}, here we are not interested in identifying the trend or the other factors separately, but we are interested in recovering the whole common component of the data, i.e., $\widehat{\mathbf{S}}_t= \widehat{\mathbf{R}}\widehat{\mathbf{F}}_t\widehat{\mathbf{C}}$. 
Hence, we can apply the methodology described in Section \ref{subsec:unitroot}. Moreover, since the considered dataset contains both monthly and quarterly varaibles, we apply our method when also imputing missing values as described in Section \ref{subsec:banbura}. Figure \ref{fig:MacroGDP} reports the GDP of Germany, France, Spain, and Italy (in black), which are quarterly, together with their estimated common components  $\widehat{\mathbf{S}}_t$ (in red), which are monthly time series. While the GDPs of Germany and France are strongly related to the common EA factors, Spain and Italy display more idiosyncratic behavior, hinting at a different level of commonality among EA countries.




\begin{figure}
	\caption{Estimated GDP for selected countries}
	\centering
	\includegraphics[width=.7\textwidth]{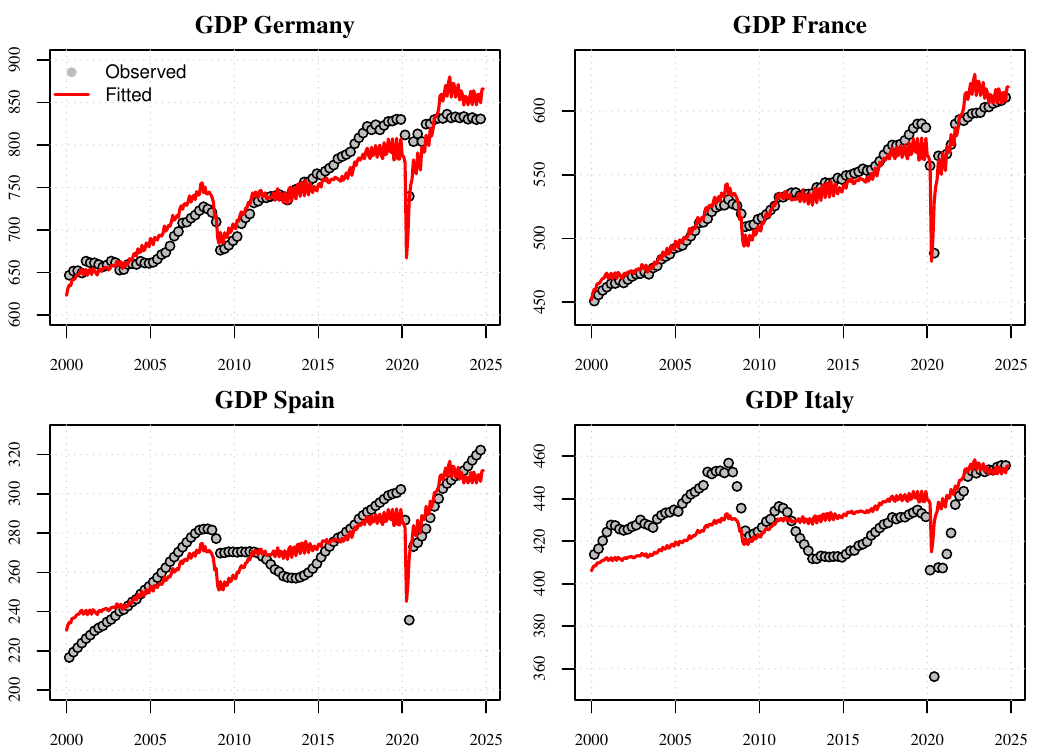}
	\label{fig:MacroGDP}
\end{figure}


\section{Conclusions}
This paper introduces a methodology for estimating a large approximate DMFM using the EM algorithm combined with Kalman filtering. 
We establish the consistency of the spaces spanned by the estimated loadings and factors as $\min\left(p_1, p_2, T\right) \to \infty$. Our estimation framework accommodates missing observations and unit root data.

Our approach can be readily adapted to include additional constraints on the model parameters \citep{chen2020constrained} or to explicitly model the dynamics of the idiosyncratic components which can be modeled as additional latent states \citep{banbura2014maximum}. Moreover, the proposed approach can be further and straightforwardly extended to tensor data of  higher order, enhancing its applicability to more complex data structures.

	\bibliographystyle{myplainnat}
	\bibliography{EMmatrix}
	
	\newpage
	\appendix
	\begin{center}
	{\LARGE \bf Supplementary Material}
	\end{center}
	\section{Notation and results on matrix operations}\label{ap:notation}
	We adopt the following notation.
	\begin{itemize}
	\item The Hadamard and Kronecker product are denoted with $\circ$ and $\otimes$, respectively. 
	\item We use $\vect{\cdot}$ and $\unvect{\cdot}$ to denote the vectorization operation and its inverse. 
	\item $\mathbb{I}_d$ denotes the $d\times d$ identity matrix. 
	\item  We use $\mathbf{1}_{m}$ and $\mathbf{0}_{m}$ to denote an $m$ dimensional vector filled with ones and zeros, and use  $\mathbf{1}_{m,n}$ and $\mathbf{0}_{m,n}$ to denote $m\times n$ matrices filled with ones and zeros, respectively.
	\item   Let $\mathbf{A}$ be an $n\times n$ matrix, the matrix $\dg{\mathbf{A}}$ denotes the matrix having the diagonal entries of $\mathbf{A}$ in the diagonal and zeros in the off-diagonal entries. 
	\item Let $\mathbf{A}$ be an $m\times n$ matrix, we denote with $\mathbb{D}_{\mathbf{A}}$ a $mn\times mn$ matrix stacking the columns of $\mathbf{A}$ on its diagonal, i.e. $\mathbb{D}_{\mathbf{A}} = \mathbb{I}_{mn} \vect{\mathbf{A}}$.
	\item  $\mathbb{E}^{(i,j)}_{m,n}$ denotes a standard basis $(m\times n)$ matrix with a one in the $(i,j)$ entry.
	\item $\mathbb{K}_{nm}$ denotes an $nm\times nm$ commutation matrix,
	$
	\mathbb{K}_{nm}=\sum_{i=1}^{n} \sum_{j=1}^{m} \mathbb{E}^{(i,j)}_{n,m} \otimes \mathbb{E}^{(i,j)\prime}_{n,m}.
	$
	\item 	The generic $(i,j)$ entry of a matrix $\mathbf{A}$ is denoted as $a_{ij}\equiv[\mathbf{A}]_{ij}$, while $\mathbf{a}_{i\cdot}\equiv[\mathbf{A}]_{i\cdot}$ and $\mathbf{a}_{.j}\equiv[\mathbf{A}]_{.j}$ denote the generic $i$th row and $j$th column of a matrix $\mathbf{A}$, respectively.
		\item 	Let $\mathbf{A}$ be a $mp\times nq$ matrix, we denote with $\textbf{A}^{[i,j]}$ the special partition of dimension $m\times n$, 
	$
	\textbf{A}^{[i,j]} = \sum^{m}_{r=1} \sum^{n}_{s=1} a_{(rp-p+i)(sq-q+j)} \mathbb{E}^{(r,s)}_{m,n},
	$
	for $i=1,\dots,p$ and $j=1,\dots,q$.
	\item Let $\mathbf{A}$ and $\mathbf{B}$ be $m\times n$ and $mp\times nq$ matrices such that $\textbf{B}$ can be partitioned into $mn$ sub-matrices of dimension $p\times q$. The star product between $\mathbf{A}$ and $\mathbf{B}$ is defined as
	$
	\mathbf{A}\star \mathbf{B} = \sum^m_{i=1} \sum_{j=1}^{n} a_{ij} \mathbf{B}^{(p,q)}_{ij},
	$
	where $\mathbf{B}^{(p,q)}_{ij}$ denotes the $ij$th block of dimension$p\times q$ of the matrix $\mathbf{B}$.
	\item  We denote by $\mathcal{U}(a,b)$ the uniform distribution on the interval $(a,b)$, $\mathcal{N}(\mathbf{A},\boldsymbol{\Sigma})$ the normal distribution with mean $\mathbf{A}$ and covariance matrix $\boldsymbol{\Sigma}$, and by  $\mathcal{MN}_{m,n}(\mathbf{A},\boldsymbol{\Sigma},\boldsymbol{\Omega})$ the $(m\times n)$ matrix normal distribution with mean $\mathbf{A}$, row covariance $\boldsymbol{\Sigma}$, and column covariance $\boldsymbol{\Omega}$. 
	\item 	We denote as $\nu^{(k)}(\mathbf{A})$ the $k$th largest eigenvalue of a generic squared matrix $\mathbf{A}$.
	 The matrix norm induced by a vector $p$-norm is denoted as $\lVert \mathbf{A} \rVert_p$, with $\lVert \mathbf{A} \rVert$ the spectral norm. The Frobenious norm is denoted $\lVert \mathbf{A} \rVert_F = \sqrt{\sum\limits^n_{i=1}\sum\limits^n_{j=1}a_{ij}}=\textbf{tr}(\mathbf{A}\mathbf{A}')$. The max norm is denoted as $\lVert \mathbf{A}\rVert_{\textrm{\tiny max}}=\max\limits_{ij} a_{ij}$. 
	 \item 	 The $o_p$ is for convergence to zero in probability and $O_p$ is for stochastic boundedness. For two random series, $X_n$ and $Y_n$, $X_n \lesssim Y_n$ means that $X_n = O_p(Y_n)$, and $X_n \gtrsim Y_n$ means that $Y_n = O_p(X_n)$. The notation $X_n \asymp Y_n$ means that $X_n \lesssim Y_n$ and $X_n \gtrsim Y_n$.
\item We use $\mathbb{E}[\cdot]$ to denote the expectation with respect to the true unknown distribution.
	\end{itemize}

We make also use of the following matrix results. Recall that $\mathbb{K}_{nm}$ is the $(nm \times nm)$ commutation matrix. Let $\textbf{X}$, $\textbf{Y}$, and $\textbf{Z}$ be $m\times n$, $n\times p$ and $p\times q$ matrices,  
	\begin{eqnarray}
	\vect{\textbf{X}\textbf{Z}\textbf{Y}} &=& \left(\textbf{Y}' \otimes \textbf{X}\right) \vect{\textbf{Z}}, \label{eq:vectrick} \\
	\vect{\textbf{X}'}  &=& \mathbb{K}_{mn}\vect{\textbf{X}}, \label{eq:vectrasp}\\
	\textbf{X}\textbf{Z}\textbf{Y} &=& \textbf{Z} \star \left(\vect{\textbf{X}}\vect{\textbf{Y}'}'\right), \label{eq:startrick} \\
	\textbf{tr}(\textbf{X}(\textbf{Y}\circ \textbf{Z})) &=&\textbf{tr}((\textbf{X}'\circ \textbf{Y})' \textbf{Z}). \label{eq:hadamtrick}
	\end{eqnarray}	
	Let $\mathbf{A}$ be a $mp\times nq$ matrix. Recall that $\textbf{A}^{[i,j]} = \sum^{m}_{r=1} \sum^{n}_{s=1} a_{(rp-p+i)(sq-q+j)} \mathbb{E}^{(r,s)}_{m,n}$, we can decompose $\mathbf{A}$ as follows
	\begin{equation}
	\mathbf{A} = \sum^p_{i=1} \sum^q_{j=1} \textbf{A}^{[i,j]} \otimes \mathbb{E}^{(ij)}_{p,q} .	 
	\label{eq:matdec}
	\end{equation}
		
	\section{Details on the EM algorithm}\label{app:EM}

	\subsection{Derivation of the expected likelihoods}
	\label{ap:ellk}
	The expressions of the expected log-likelihoods stated in terms of the data in its original matrix form, which are given in \eqref{eq:llmat1} and \eqref{eq:llmat2} (up to constant terms and initial conditions), 
	are obtained from the usual expressions for vectorized data as follows (here we consider expectations computed using a generic estimator of the parameters $\widehat{\boldsymbol{\theta}}$):
	\[
	\setlength{\thinmuskip}{0mu}
	\setlength{\medmuskip}{0mu}
	\setlength{\thickmuskip}{0mu}
	\begin{array}{rcl}
	\mathbb{E}_{\widehat{\boldsymbol{\theta}}}\left[\ell\left(\mathsf{Y}_T |\mathsf{F}_T; \boldsymbol{\theta}\right) |\mathsf{Y}_T \right] &=& -\frac{T}{2} \log\left(|\mathbf{K}\otimes\mathbf{H}|\right)- \frac{1}{2}\sum^T\limits_{t=1} \mathbb{E}_{\widehat{\boldsymbol{\theta}}} \left[\left(\mathsf{y}_t - \left(\mathbf{C}\otimes \mathbf{R} \right) \mathsf{f}_t \right)' \left(\mathbf{K}\otimes\mathbf{H}\right)^{-1} \left(\mathsf{y}_t - \left(\mathbf{C}\otimes \mathbf{R} \right)\mathsf{f}_t \right) |\mathsf{Y}_T\right]\\[0.1in]
	&=&-\frac{T}{2}\left(\log\left(|\mathbf{K}|^{p_1}|\mathbf{H}|^{p_2}\right)\right) - \frac{1}{2}\sum^T\limits_{t=1} \mathbb{E}_{\widehat{\boldsymbol{\theta}}} \left[\vect{\mathbf{Y}_t - \mathbf{R} \mathbf{F}_t\mathbf{C}' }' \left(\mathbf{K}^{-1}\otimes\mathbf{H}^{-1}\right) \vect{\mathbf{Y}_t - \mathbf{R} \mathbf{F}_t\mathbf{C}' } |\mathsf{Y}_T\right] \\[0.1in]
	&=&-\frac{T}{2}\left(\log\left(|\mathbf{K}|^{p_1}|\mathbf{H}|^{p_2}\right)\right) - \frac{1}{2}\sum^T\limits_{t=1} \mathbb{E}_{\widehat{\boldsymbol{\theta}}} \left[\vect{\mathbf{Y}_t - \mathbf{R} \mathbf{F}_t\mathbf{C}' }' \vect{ \mathbf{H}^{-1}\left(\mathbf{Y}_t - \mathbf{R} \mathbf{F}_t\mathbf{C}' \right)\mathbf{K}^{-1}} |\mathsf{Y}_T\right] \\[0.1in]
	&=&-\frac{T}{2}\left(\log\left(|\mathbf{K}|^{p_1}|\mathbf{H}|^{p_2}\right)\right) - \frac{1}{2}\sum^T\limits_{t=1} \mathbb{E}_{\widehat{\boldsymbol{\theta}}} \left[ \textbf{tr}\left(\mathbf{H}^{-1}\left(\mathbf{Y}_t-\mathbf{R}\mathbf{F}_{t}\mathbf{C}'\right)\mathbf{K}^{-1}\left(\mathbf{Y}_t-\mathbf{R}\mathbf{F}_{t}\mathbf{C}'\right)'\right) | \mathsf{Y}_T \right].
	\end{array}
	\]
	\[
	\setlength{\thinmuskip}{0mu}
	\setlength{\medmuskip}{0mu}
	\setlength{\thickmuskip}{0mu}
	\begin{array}{rcl}
	\mathbb{E}_{\widehat{\boldsymbol{\theta}}}\left[\ell\left(\mathsf{F}_T; \boldsymbol{\theta}\right) |\mathsf{Y}_T \right] &=& -\frac{T-1}{2} \log\left(|\mathbf{Q}\otimes\mathbf{P}|\right)- \frac{1}{2}\sum^T\limits_{t=2} \mathbb{E}_{\widehat{\boldsymbol{\theta}}} \left[\left(\mathsf{f}_t - \left(\mathbf{B}\otimes \mathbf{A} \right) \mathsf{f}_{t-1} \right)' \left(\mathbf{Q}\otimes\mathbf{P}\right)^{-1} \left(\mathsf{f}_t - \left(\mathbf{B}\otimes \mathbf{A} \right)\mathsf{f}_{t-1} \right) |\mathsf{Y}_T\right]\\[0.1in]
	&=&-\frac{T-1}{2}\left(\log\left(|\mathbf{Q}|^{k_1}|\mathbf{P}|^{k_2}\right)\right) - \frac{1}{2}\sum^T\limits_{t=2} \mathbb{E}_{\widehat{\boldsymbol{\theta}}} \left[\vect{\mathbf{F}_t - \mathbf{A} \mathbf{F}_{t-1}\mathbf{B}' }' \left(\mathbf{Q}\otimes\mathbf{P}\right)^{-1} \vect{\mathbf{F}_t - \mathbf{A} \mathbf{F}_{t-1}\mathbf{B}' } |\mathsf{Y}_T\right] \\[0.1in]
	&=&-\frac{T-1}{2}\left(\log\left(|\mathbf{Q}|^{k_1}|\mathbf{P}|^{k_2}\right)\right) - \frac{1}{2}\sum^T\limits_{t=2} \mathbb{E}_{\widehat{\boldsymbol{\theta}}} \left[\vect{\mathbf{f}_t - \mathbf{A} \mathbf{F}_{t-1}\mathbf{B}' }' \vect{ \mathbf{P}^{-1}\left(\mathbf{F}_t - \mathbf{A} \mathbf{F}_{t-1}\mathbf{B}' \right)\mathbf{Q}^{-1}} |\mathsf{Y}_T\right] \\[0.1in]
	&=&-\frac{T-1}{2}\left(\log\left(|\mathbf{Q}|^{k_1}|\mathbf{P}|^{k_2}\right)\right) - \frac{1}{2}\sum^T\limits_{t=2} \mathbb{E}_{\widehat{\boldsymbol{\theta}}} \left[ \textbf{tr}\left(\mathbf{P}^{-1}\left(\mathbf{F}_t-\mathbf{A}\mathbf{F}_{t-1}\mathbf{B}'\right)\mathbf{Q}^{-1}\left(\mathbf{F}_t-\mathbf{A}\mathbf{F}_{t-1}\mathbf{B}'\right)'\right) | \mathsf{Y}_T \right].
	\end{array}
	\]

	\subsection{EM updates}
	\label{ap:em}
	To obtain the EM updates we first compute the derivatives of $\mathcal{Q}(\boldsymbol{\theta}, \widehat{\boldsymbol{\theta}})$ with respect to each parameter in $\boldsymbol{\theta}$, and obtain the following
	\[
	\begin{array}{lll}
	\frac{\partial \mathcal{Q}(\boldsymbol{\theta}, \widehat{\boldsymbol{\theta}})}{\partial \textbf{R}} &=& \frac{\partial \mathbb{E}_{\widehat{\boldsymbol{\theta}}} \left[\ell(\mathsf{Y}_T|\mathbf{F}_T;\boldsymbol{\theta}) | \mathsf{Y}_T \right]}{\partial \textbf{R}} \\
	&=& \frac{1}{2} \sum\limits_{t=1}^{T} \left(\frac{\partial  \textbf{tr}\left(\mathbf{H}^{-1}\mathbf{Y}_t\mathbf{K}^{-1}\mathbf{C}\mathbb{E}_{\widehat{\boldsymbol{\theta}}} \left[\mathbf{F}'_{t}| \mathsf{Y}_T \right]\mathbf{R}'\right)}{\partial \textbf{R}}
	+ \frac{\partial  \textbf{tr}\left(\mathbf{H}^{-1}\mathbf{R}\mathbb{E}_{\widehat{\boldsymbol{\theta}}} \left[\mathbf{F}_{t}| \mathsf{Y}_T \right]\mathbf{C}'\mathbf{K}^{-1}\mathbf{Y}'_t\right)}{\partial \textbf{R}}
	-\frac{\partial  \textbf{tr}\left( \mathbb{E}_{\widehat{\boldsymbol{\theta}}} \left[\mathbf{H}^{-1}\mathbf{R}\mathbf{F}_{t}\mathbf{C}'\mathbf{K}^{-1}\mathbf{C}\mathbf{F}'_{t}\mathbf{R}'| \mathsf{Y}_T \right]\right)}{\partial \textbf{R}} \right) \\
	&=& \frac{1}{2} \sum\limits_{t=1}^{T} \left( 2 \mathbf{H}^{-1}\mathbf{Y}_t\mathbf{K}^{-1}\mathbf{C}\mathbb{E}_{\widehat{\boldsymbol{\theta}}} \left[\mathbf{F}'_{t}| \mathsf{Y}_T \right] 
	- 2 \mathbb{E}_{\widehat{\boldsymbol{\theta}}} \left[\mathbf{H}^{-1}\mathbf{R}\mathbf{F}_t\mathbf{C'}\mathbf{K}^{-1}\mathbf{C}\mathbf{F}'_{t}| \mathsf{Y}_T \right]\right),
	\end{array}	 
	\]
	\[
	\begin{array}{lll}
	\frac{\partial \mathcal{Q}(\boldsymbol{\theta}, \widehat{\boldsymbol{\theta}})}{\partial \textbf{C}} &=& \frac{\partial \mathbb{E}_{\widehat{\boldsymbol{\theta}}} \left[\ell(\mathsf{Y}_T|\mathbf{F}_T;\boldsymbol{\theta}) | \mathsf{Y}_T \right]}{\partial \textbf{C}} \\
	&=& \frac{1}{2} \sum\limits_{t=1}^{T} \left(\frac{\partial  \textbf{tr}\left(\mathbf{H}^{-1}\mathbf{Y}_t\mathbf{K}^{-1}\mathbf{C}\mathbb{E}_{\widehat{\boldsymbol{\theta}}} \left[\mathbf{F}'_{t}| \mathsf{Y}_T \right]\mathbf{R}'\right)}{\partial \textbf{C}}
	+ \frac{\partial  \textbf{tr}\left(\mathbf{H}^{-1}\mathbf{R}\mathbb{E}_{\widehat{\boldsymbol{\theta}}} \left[\mathbf{F}_{t}| \mathsf{Y}_T \right]\mathbf{C}'\mathbf{K}^{-1}\mathbf{Y}'_t\right)}{\partial \textbf{C}}
	-\frac{\partial  \textbf{tr}\left( \mathbb{E}_{\widehat{\boldsymbol{\theta}}} \left[\mathbf{H}^{-1}\mathbf{R}\mathbf{F}_{t}\mathbf{C}'\mathbf{K}^{-1}\mathbf{C}\mathbf{F}'_{t}\mathbf{R}'| \mathsf{Y}_T \right]\right)}{\partial \textbf{C}} \right) \\
	&=& \frac{1}{2} \sum\limits_{t=1}^{T} \left( 2 \mathbf{K}^{-1}\mathbf{Y}'_t\mathbf{H}^{-1}\mathbf{R}\mathbb{E}_{\widehat{\boldsymbol{\theta}}} \left[\mathbf{F}_{t}| \mathsf{Y}_T \right] 
	- 2 \mathbb{E}_{\widehat{\boldsymbol{\theta}}} \left[\mathbf{K}^{-1}\mathbf{C}\mathbf{F}'_t\mathbf{R'}\mathbf{H}^{-1}\mathbf{R}\mathbf{F}_{t}| \mathsf{Y}_T \right]\right),
	\end{array}	 
	\]
	\[
	\begin{array}{lll}
	\frac{\partial \mathcal{Q}(\boldsymbol{\theta}, \widehat{\boldsymbol{\theta}})}{\partial \textbf{H}} &=& \frac{\partial \mathbb{E}_{\widehat{\boldsymbol{\theta}}} \left[\ell(\mathsf{Y}_T|\mathbf{F}_T;\boldsymbol{\theta}) | \mathsf{Y}_T \right]}{\partial \textbf{H}} \\
	&=& -\frac{Tp_2}{2} \frac{\partial\log\left(|\mathbf{H}|\right)}{\partial\mathbf{H}} - \frac{1}{2} \sum\limits_{t=1}^{T} \frac{\partial \textbf{tr}\left( \mathbf{H}^{-1} \mathbb{E}_{\widehat{\boldsymbol{\theta}}} \left[\left(\mathbf{Y}_t-\mathbf{R}\mathbf{F}_{t}\mathbf{C}'\right)\mathbf{K}^{-1}\left(\mathbf{Y}_t-\mathbf{R}\mathbf{F}_{t}\mathbf{C}'\right)'| \mathsf{Y}_T \right] \right) }{\partial \mathbf{H}} \\
	&=& -\frac{Tp_2}{2} \mathbf{H}^{-1} + \frac{1}{2} \sum\limits_{t=1}^{T} \mathbf{H}^{-1} \mathbb{E}_{\widehat{\boldsymbol{\theta}}} \left[\left(\mathbf{Y}_t-\mathbf{R}\mathbf{F}_{t}\mathbf{C}'\right)\mathbf{K}^{-1}\left(\mathbf{Y}_t-\mathbf{R}\mathbf{F}_{t}\mathbf{C}'\right)'| \mathsf{Y}_T \right]' \mathbf{H}^{-1}, \qquad \qquad \qquad
	\end{array}	 
	\]
	\[
	\begin{array}{lll}
	\frac{\partial \mathcal{Q}(\boldsymbol{\theta}, \widehat{\boldsymbol{\theta}})}{\partial \textbf{K}} &=& \frac{\partial \mathbb{E}_{\widehat{\boldsymbol{\theta}}} \left[\ell(\mathsf{Y}_T|\mathbf{F}_T;\boldsymbol{\theta}) | \mathsf{Y}_T \right]}{\partial \textbf{K}} \\
	&=& -\frac{Tp_1}{2} \frac{\partial\log\left(|\mathbf{K}|\right)}{\partial\mathbf{K}} - \frac{1}{2} \sum\limits_{t=1}^{T} \frac{\partial \textbf{tr}\left( \mathbf{K}^{-1} \mathbb{E}_{\widehat{\boldsymbol{\theta}}} \left[\left(\mathbf{Y}_t-\mathbf{R}\mathbf{F}_{t}\mathbf{C}'\right)'\mathbf{H}^{-1}\left(\mathbf{Y}_t-\mathbf{R}\mathbf{F}_{t}\mathbf{C}'\right)| \mathsf{Y}_T \right] \right) }{\partial \mathbf{K}} \\
	&=& -\frac{Tp_1}{2} \mathbf{K}^{-1} + \frac{1}{2} \sum\limits_{t=1}^{T} \mathbf{K}^{-1} \mathbb{E}_{\widehat{\boldsymbol{\theta}}} \left[\left(\mathbf{Y}_t-\mathbf{R}\mathbf{F}_{t}\mathbf{C}'\right)'\mathbf{H}^{-1}\left(\mathbf{Y}_t-\mathbf{R}\mathbf{F}_{t}\mathbf{C}'\right)| \mathsf{Y}_T \right]' \mathbf{K}^{-1}, \qquad \qquad \qquad
	\end{array}	 
	\]
	\[
	\setlength{\thinmuskip}{0mu}
	\setlength{\medmuskip}{0mu}
	\setlength{\thickmuskip}{0mu}
	\begin{array}{lll}
	\frac{\partial \mathcal{Q}(\boldsymbol{\theta}, \widehat{\boldsymbol{\theta}})}{\partial \textbf{A}} &=& \frac{\partial \mathbb{E}_{\widehat{\boldsymbol{\theta}}} \left[\ell(\mathbf{F}_T;\boldsymbol{\theta}) | \mathsf{Y}_T \right]}{\partial \textbf{A}} \\
	&=& \frac{1}{2} \sum\limits_{t=2}^{T} \left(\frac{\partial  \textbf{tr}\left(\mathbb{E}_{\widehat{\boldsymbol{\theta}}} \left[\mathbf{P}^{-1}\mathbf{F}_{t}\mathbf{Q}^{-1}\mathbf{B}\mathbf{F}'_{t-1}\mathbf{A}'| \mathsf{Y}_T \right]\right)}{\partial \textbf{A}} + \frac{\partial  \textbf{tr}\left(\mathbb{E}_{\widehat{\boldsymbol{\theta}}} \left[\mathbf{P}^{-1}\mathbf{A}\mathbf{F}_{t-1}\mathbf{B}'\mathbf{Q}^{-1}\mathbf{F}'_{t}| \mathsf{Y}_T \right]\right)}{\partial \textbf{A}} \right. \\
	&& \left. \hspace*{3in} - \frac{\partial  \textbf{tr}\left(\mathbb{E}_{\widehat{\boldsymbol{\theta}}} \left[\mathbf{P}^{-1}\mathbf{A}\mathbf{F}_{t-1}\mathbf{B}'\mathbf{Q}^{-1}\mathbf{B}\mathbf{F}'_{t-1}\mathbf{A}'| \mathsf{Y}_T \right]\right)}{\partial \textbf{A}}\right)  \\
	&=& \frac{1}{2} \sum\limits_{t=2}^{T} 2 \left(\mathbb{E}_{\widehat{\boldsymbol{\theta}}} \left[\mathbf{P}^{-1}\mathbf{F}_{t}\mathbf{Q}^{-1}\mathbf{B}\mathbf{F}'_{t-1}| \mathsf{Y}_T \right] - 2 \mathbb{E}_{\widehat{\boldsymbol{\theta}}} \left[\mathbf{P}^{-1}\mathbf{A}\mathbf{F}_{t-1}\mathbf{B}'\mathbf{Q}^{-1}\mathbf{B}\mathbf{F}'_{t-1}| \mathsf{Y}_T \right] \right),
	\end{array}	 
	\]
	\[
	\setlength{\thinmuskip}{0mu}
	\setlength{\medmuskip}{0mu}
	\setlength{\thickmuskip}{0mu}
	\begin{array}{lll}
	\frac{\partial \mathcal{Q}(\boldsymbol{\theta}, \widehat{\boldsymbol{\theta}})}{\partial \textbf{B}} &=& \frac{\partial \mathbb{E}_{\widehat{\boldsymbol{\theta}}} \left[\ell(\mathbf{F}_T;\boldsymbol{\theta}) | \mathsf{Y}_T \right]}{\partial \textbf{B}} \\
	&=& \frac{1}{2} \sum\limits_{t=2}^{T} \left(\frac{\partial  \textbf{tr}\left(\mathbb{E}_{\widehat{\boldsymbol{\theta}}} \left[\mathbf{P}^{-1}\mathbf{F}_{t}\mathbf{Q}^{-1}\mathbf{B}\mathbf{F}'_{t-1}\mathbf{A}'| \mathsf{Y}_T \right]\right)}{\partial \textbf{B}} + \frac{\partial  \textbf{tr}\left(\mathbb{E}_{\widehat{\boldsymbol{\theta}}} \left[\mathbf{P}^{-1}\mathbf{A}\mathbf{F}_{t-1}\mathbf{B}'\mathbf{Q}^{-1}\mathbf{F}'_{t}| \mathsf{Y}_T \right]\right)}{\partial \textbf{B}} \right. \\
	&& \hspace*{3in} \left. - \frac{\partial  \textbf{tr}\left(\mathbb{E}_{\widehat{\boldsymbol{\theta}}} \left[\mathbf{P}^{-1}\mathbf{A}\mathbf{F}_{t-1}\mathbf{B}'\mathbf{Q}^{-1}\mathbf{B}\mathbf{F}'_{t-1}\mathbf{A}'| \mathsf{Y}_T \right]\right)}{\partial \textbf{B}} \right)\\
	&=& \frac{1}{2} \sum\limits_{t=2}^{T} \left(2 \mathbb{E}_{\widehat{\boldsymbol{\theta}}} \left[\mathbf{Q}^{-1}\mathbf{F}'_{t}\mathbf{P}^{-1}\mathbf{A}\mathbf{F}'_{t-1}| \mathsf{Y}_T \right] - 2 \mathbb{E}_{\widehat{\boldsymbol{\theta}}} \left[\mathbf{Q}^{-1}\mathbf{B}\mathbf{F}'_{t-1}\mathbf{A}'\mathbf{P}^{-1}\mathbf{A}\mathbf{F}_{t-1}| \mathsf{Y}_T \right]\right),
	\end{array}	 
	\]
	\[
	\begin{array}{lll}
	\frac{\partial \mathcal{Q}(\boldsymbol{\theta}, \widehat{\boldsymbol{\theta}})}{\partial \textbf{P}} &=& \frac{\partial \mathbb{E}_{\widehat{\boldsymbol{\theta}}} \left[\ell(\mathbf{F}_T;\boldsymbol{\theta}) | \mathsf{Y}_T \right]}{\partial \textbf{P}} \\
	&=& -\frac{(T-1)k_2}{2} \frac{\partial\log\left(|\mathbf{P}|\right)}{\partial\mathbf{P}} + \frac{1}{2} \sum\limits_{t=2}^{T} \frac{\partial \textbf{tr}\left( \mathbf{P}^{-1} \mathbb{E}_{\widehat{\boldsymbol{\theta}}} \left[\left(\mathbf{F}_{t}-\mathbf{A}\mathbf{F}_{t-1}\mathbf{B}'\right)\mathbf{Q}^{-1}\left(\mathbf{F}_{t}-\mathbf{A}\mathbf{F}_{t-1}\mathbf{B}' \right)' | \mathsf{Y}_T \right] \right) }{\partial \mathbf{P}} \\
	&=& -\frac{(T-1)k_2}{2} \mathbf{P}^{-1} + \frac{1}{2} \sum\limits_{t=2}^{T} \mathbf{P}^{-1} \mathbb{E}_{\widehat{\boldsymbol{\theta}}} \left[\left(\mathbf{F}_{t}-\mathbf{A}\mathbf{F}_{t-1}\mathbf{B}'\right)\mathbf{Q}^{-1}\left(\mathbf{F}_{t}-\mathbf{A}\mathbf{F}_{t-1}\mathbf{B}' \right)' | \mathsf{Y}_T \right]' \mathbf{P}^{-1}, \qquad \qquad \qquad
	\end{array}	 
	\]
	\[
	\begin{array}{lll}
	\frac{\partial \mathcal{Q}(\boldsymbol{\theta}, \widehat{\boldsymbol{\theta}})}{\partial \textbf{Q}} &=& \frac{\partial \mathbb{E}_{\widehat{\boldsymbol{\theta}}} \left[\ell(\mathbf{F}_T;\boldsymbol{\theta}) | \mathsf{Y}_T \right]}{\partial \textbf{Q}} \\
	&=& -\frac{(T-1)k_1}{2} \frac{\partial\log\left(|\mathbf{Q}|\right)}{\partial\mathbf{Q}} + \frac{1}{2} \sum\limits_{t=2}^{T} \frac{\partial \textbf{tr}\left( \mathbf{Q}^{-1} \mathbb{E}_{\widehat{\boldsymbol{\theta}}} \left[ \left(\mathbf{F}_{t}-\mathbf{A}\mathbf{F}_{t-1}\mathbf{B}'\right)'\mathbf{P}^{-1}\left(\mathbf{F}_{t}-\mathbf{A}\mathbf{F}_{t-1}\mathbf{B}' \right)| \mathsf{Y}_T \right] \right) }{\partial \mathbf{Q}} \\
	&=& -\frac{(T-1)k_1}{2} \mathbf{Q}^{-1} + \frac{1}{2} \sum\limits_{t=2}^{T} \mathbf{Q}^{-1} \mathbb{E}_{\widehat{\boldsymbol{\theta}}} \left[ \left(\mathbf{F}_{t}-\mathbf{A}\mathbf{F}_{t-1}\mathbf{B}'\right)'\mathbf{P}^{-1}\left(\mathbf{F}_{t}-\mathbf{A}\mathbf{F}_{t-1}\mathbf{B}' \right)| \mathsf{Y}_T \right]' \mathbf{Q}^{-1}. \qquad \qquad \qquad
	\end{array}	 
	\]
	First order conditions (FOC) then yield
	\[
	\mathbf{R} = \left(\sum_{t=1}^{T} \mathbf{Y}_t\mathbf{K}^{-1}\mathbf{C}\mathbb{E}_{\widehat{\boldsymbol{\theta}}} \left[\mathbf{F}'_{t}| \mathsf{Y}_T \right]\right)\left( \sum_{t=1}^{T} \mathbb{E}_{\widehat{\boldsymbol{\theta}}} \left[\mathbf{F}_t\mathbf{C'}\mathbf{K}^{-1}\mathbf{C}\mathbf{F}'_{t}| \mathsf{Y}_T \right] \right)^{-1},
	\]
	\[
	\mathbf{C} = \left(\sum_{t=1}^{T} \mathbf{Y}'_t\mathbf{H}^{-1}\mathbf{R}\mathbb{E}_{\widehat{\boldsymbol{\theta}}} \left[\mathbf{F}_{t}| \mathsf{Y}_T \right] \right)\left(\sum_{t=1}^{T}  \mathbb{E}_{\widehat{\boldsymbol{\theta}}} \left[\mathbf{F}'_t\mathbf{R'}\mathbf{H}^{-1}\mathbf{R}\mathbf{F}_{t}| \mathsf{Y}_T \right]\right)^{-1},
	\]
	\[
	\mathbf{H} = \frac{1}{Tp_2 }\sum_{t=1}^{T}  \mathbb{E}_{\widehat{\boldsymbol{\theta}}} \left[\left(\mathbf{Y}_t-\mathbf{R}\mathbf{F}_{t}\mathbf{C}'\right)\mathbf{K}^{-1}\left(\mathbf{Y}_t-\mathbf{R}\mathbf{F}_{t}\mathbf{C}'\right)'| \mathsf{Y}_T \right],
	\]
	\[
	\mathbf{K} = \frac{1}{Tp_1} \sum_{t=1}^{T}  \mathbb{E}_{\widehat{\boldsymbol{\theta}}} \left[\left(\mathbf{Y}_t-\mathbf{R}\mathbf{F}_{t}\mathbf{C}'\right)'\mathbf{H}^{-1}\left(\mathbf{Y}_t-\mathbf{R}\mathbf{F}_{t}\mathbf{C}'\right)| \mathsf{Y}_T \right],
	\]
	\[
	\mathbf{A} = \left(\sum_{t=1}^{T} \mathbb{E}_{\widehat{\boldsymbol{\theta}}} \left[\mathbf{F}_{t}\mathbf{Q}^{-1}\mathbf{B}\mathbf{F}'_{t-1}| \mathsf{Y}_T \right]\right)\left(\sum_{t=1}^{T}  \mathbb{E}_{\widehat{\boldsymbol{\theta}}} \left[\mathbf{F}_{t-1}\mathbf{B}'\mathbf{Q}^{-1}\mathbf{B}\mathbf{F}'_{t-1}| \mathsf{Y}_T \right] \right)^{-1},
	\]
	\[
	\mathbf{B} = \left(\sum_{t=1}^{T}  \mathbb{E}_{\widehat{\boldsymbol{\theta}}} \left[\mathbf{F}'_{t}\mathbf{P}^{-1}\mathbf{A}\mathbf{F}_{t-1}| \mathsf{Y}_T \right] \right)\left(\sum_{t=1}^{T}  \mathbb{E}_{\widehat{\boldsymbol{\theta}}} \left[\mathbf{F}'_{t-1}\mathbf{A}'\mathbf{P}^{-1}\mathbf{A}\mathbf{F}_{t-1}| \mathsf{Y}_T \right] \right)^{-1},
	\]
	\[
	\mathbf{P}  = \frac{1}{(T-1)k_2} \sum_{t=2}^{T} \mathbb{E}_{\widehat{\boldsymbol{\theta}}} \left[\left(\mathbf{F}_{t}-\mathbf{A}\mathbf{F}_{t-1}\mathbf{B}'\right)\mathbf{Q}^{-1}\left(\mathbf{F}_{t}-\mathbf{A}\mathbf{F}_{t-1}\mathbf{B}' \right)' | \mathsf{Y}_T \right],
	\]
	\[
	\mathbf{Q} = \frac{1}{(T-1)k_1} \sum_{t=2}^{T} \mathbb{E}_{\widehat{\boldsymbol{\theta}}} \left[ \left(\mathbf{F}_{t}-\mathbf{A}\mathbf{F}_{t-1}\mathbf{B}'\right)'\mathbf{P}^{-1}\left(\mathbf{F}_{t}-\mathbf{A}\mathbf{F}_{t-1}\mathbf{B}' \right)| \mathsf{Y}_T \right].
	\]
	Using the conditional moments of the Kalman smoother recursions and \eqref{eq:vectrick}-\eqref{eq:startrick}, we obtain 
	\[
	\mathbb{E}_{\widehat{\boldsymbol{\theta}}} \left[\mathbf{F}_{t}| \mathsf{Y}_T \right] = \unvect{\mathbb{E}_{\widehat{\boldsymbol{\theta}}} \left[\mathsf{f}_{t}| \mathsf{Y}_T \right]} = \unvect{\mathsf{f}_{t|T}} = \mathbf{F}_{t|T},
	\]
	\[
	\mathbb{E}_{\widehat{\boldsymbol{\theta}}} \left[\mathbf{F}'_{t}| \mathsf{Y}_T \right] = \unvect{\mathbb{E}_{\widehat{\boldsymbol{\theta}}} \left[\mathsf{f}_{t}| \mathsf{Y}_T \right]}' = \unvect{\mathsf{f}_{t|T}}' = \mathbf{F}'_{t|T},
	\]
	\[
	\begin{array}{rcl}
	\mathbb{E}_{\widehat{\boldsymbol{\theta}}} \left[\mathbf{F}_t\mathbf{C'}\mathbf{K}^{-1}\mathbf{C}\mathbf{F}'_{t}| \mathsf{Y}_T \right] &=& \left(\mathbf{C'}\mathbf{K}^{-1}\mathbf{C}\right)\star \mathbb{E}_{\widehat{\boldsymbol{\theta}}} \left[\vect{\mathbf{F}_t}\vect{\mathbf{F}_{t}}'| \mathsf{Y}_T \right] \\
	&=&	\left(\mathbf{C'}\mathbf{K}^{-1}\mathbf{C}\right)\star \mathbb{E}_{\widehat{\boldsymbol{\theta}}} \left[\mathsf{f}_t\mathsf{f}'_t| \mathsf{Y}_T \right] \\
	&=&	\left(\mathbf{C'}\mathbf{K}^{-1}\mathbf{C}\right)\star \left( \mathsf{f}_{t|T}\mathsf{f}'_{t|T} + \boldsymbol{\Pi}_{t|T}  \right),
	\end{array}
	\]
	\[
	\begin{array}{rcl}
	\mathbb{E}_{\widehat{\boldsymbol{\theta}}} \left[\mathbf{F}'_t\mathbf{R'}\mathbf{H}^{-1}\mathbf{R}\mathbf{F}_{t}| \mathsf{Y}_T \right] &=& \left(\mathbf{R'}\mathbf{H}^{-1}\mathbf{R}\right) \star \mathbb{E}_{\widehat{\boldsymbol{\theta}}} \left[\vect{\mathbf{F}'_t}\vect{\mathbf{F}_{t}'}'| \mathsf{Y}_T \right] \\
	&=& \left(\mathbf{R'}\mathbf{H}^{-1}\mathbf{R}\right) \star \left(\mathbb{K}_{k_1k_2} \mathbb{E}_{\widehat{\boldsymbol{\theta}}} \left[\mathsf{f}_t \mathsf{f}'_t| \mathsf{Y}_T \right] \mathbb{K}'_{k_1k_2}\right) \\
	&=& \left(\mathbf{R'}\mathbf{H}^{-1}\mathbf{R}\right) \star \left(\mathbb{K}_{k_1k_2} \left( \mathsf{f}_{t|T}\mathsf{f}'_{t|T} + \boldsymbol{\Pi}_{t|T}  \right) \mathbb{K}'_{k_1k_2}\right),
	\end{array}
	\]
	\[
	\begin{array}{rcl}
	\mathbb{E}_{\widehat{\boldsymbol{\theta}}} \left[\mathbf{R}\mathbf{F}_{t}\mathbf{C}'\mathbf{K}^{-1} \mathbf{C} \mathbf{F}'_{t}\mathbf{R}'| \mathsf{Y}_T \right] &=& \left(\mathbf{C}'\mathbf{K}^{-1} \mathbf{C}\right) \star \mathbb{E}_{\widehat{\boldsymbol{\theta}}} \left[\vect{\mathbf{R}\mathbf{F}_{t}} \vect{\mathbf{R}\mathbf{F}_{t}}' | \mathsf{Y}_T \right] \\
	&=& \left(\mathbf{C}'\mathbf{K}^{-1} \mathbf{C}\right) \star \left(\left( \mathbb{I}_{k_2}\otimes\mathbf{R}\right) \mathbb{E}_{\widehat{\boldsymbol{\theta}}} \left[\mathsf{f}_t \mathsf{f}'_t | \mathsf{Y}_T \right] \left( \mathbb{I}_{k_2}\otimes\mathbf{R}\right)'\right) \\
	&=& \left(\mathbf{C}'\mathbf{K}^{-1} \mathbf{C}\right) \star \left(\left( \mathbb{I}_{k_2}\otimes\mathbf{R}\right) \left( \mathsf{f}_{t|T}\mathsf{f}'_{t|T} + \boldsymbol{\Pi}_{t|T}\right) \left( \mathbb{I}_{k_2}\otimes\mathbf{R}\right)'\right),
	\end{array}
	\]
	\[
	\begin{array}{rcl}
	\mathbb{E}_{\widehat{\boldsymbol{\theta}}} \left[ \mathbf{C}\mathbf{F}'_{t} \mathbf{R}'\mathbf{H}^{-1} \mathbf{R}\mathbf{F}_{t}\mathbf{C}' | \mathsf{Y}_T \right] &=&  \left(\mathbf{R}'\mathbf{H}^{-1} \mathbf{R}\right) \star\mathbb{E}_{\widehat{\boldsymbol{\theta}}} \left[ \vect{\mathbf{C}\mathbf{F}'_{t}} \vect{\mathbf{C}\mathbf{F}'_{t}}' | \mathsf{Y}_T \right] \\ 
	&=& \left(\mathbf{R}'\mathbf{H}^{-1} \mathbf{R}\right) \star \left( \left(\mathbb{I}_{k_1}\otimes\mathbf{C} \right) \left(\mathbb{K}_{k_1k_2} \mathbb{E}_{\widehat{\boldsymbol{\theta}}} \left[\mathsf{f}_t \mathsf{f}'_t| \mathsf{Y}_T \right] \mathbb{K}'_{k_1k_2}\right) \left(\mathbb{I}_{k_1}\otimes\mathbf{C} \right)' \right) \\
	&=& \left(\mathbf{R}'\mathbf{H}^{-1} \mathbf{R}\right) \star \left( \left(\mathbb{I}_{k_1}\otimes\mathbf{C} \right) \left(\mathbb{K}_{k_1k_2} \left( \mathsf{f}_{t|T}\mathsf{f}'_{t|T} + \boldsymbol{\Pi}_{t|T}  \right) \mathbb{K}'_{k_1k_2}\right) \left(\mathbb{I}_{k_1}\otimes\mathbf{C} \right)' \right),
	\end{array}
	\]
	\[
	\begin{array}{rcl}
	\mathbb{E}_{\widehat{\boldsymbol{\theta}}} \left[\left(\mathbf{Y}_t-\mathbf{R}\mathbf{F}_{t}\mathbf{C}'\right)\mathbf{K}^{-1}\left(\mathbf{Y}_t-\mathbf{R}\mathbf{F}_{t}\mathbf{C}'\right)'| \mathsf{Y}_T \right] &=&  \mathbf{Y}_t \mathbf{K}^{-1}\mathbf{Y}'_t - \mathbf{Y}_t \mathbf{K}^{-1} \mathbf{C} \mathbb{E}_{\widehat{\boldsymbol{\theta}}} \left[\mathbf{F}'_{t}| \mathsf{Y}_T \right]\mathbf{R}' \\ && -  \mathbf{R}\mathbb{E}_{\widehat{\boldsymbol{\theta}}} \left[\mathbf{F}_{t}| \mathsf{Y}_T \right]\mathbf{C}'\mathbf{K}^{-1}\mathbf{Y}'_t + \mathbb{E}_{\widehat{\boldsymbol{\theta}}} \left[\mathbf{R}\mathbf{F}_{t}\mathbf{C}'\mathbf{K}^{-1} \mathbf{C} \mathbf{F}'_{t}\mathbf{R}'| \mathsf{Y}_T \right],
	\end{array}
	\]
	\[
	\begin{array}{rcl}
	\mathbb{E}_{\widehat{\boldsymbol{\theta}}} \left[\left(\mathbf{Y}_t-\mathbf{R}\mathbf{F}_{t}\mathbf{C}'\right)'\mathbf{H}^{-1}\left(\mathbf{Y}_t-\mathbf{R}\mathbf{F}_{t}\mathbf{C}'\right)| \mathsf{Y}_T \right] &=& \mathbf{Y}'_t \mathbf{H}^{-1} \mathbf{Y}_t - \mathbf{Y}_t'\mathbf{H}^{-1} \mathbf{R}\mathbb{E}_{\widehat{\boldsymbol{\theta}}} \left[\mathbf{F}_{t} | \mathsf{Y}_T \right]\mathbf{C}' \\
	&& - \mathbf{C}\mathbb{E}_{\widehat{\boldsymbol{\theta}}} \left[\mathbf{F}'_{t}| \mathsf{Y}_T \right]\mathbf{R}'\mathbf{H}^{-1}\mathbf{Y}_t  + \mathbb{E}_{\widehat{\boldsymbol{\theta}}} \left[ \mathbf{C}\mathbf{F}'_{t}\mathbf{R}'\mathbf{H}^{-1} \mathbf{R}\mathbf{F}_{t}\mathbf{C}' | \mathsf{Y}_T \right],
	\end{array}
	\]
	\[
	\begin{array}{rcl}
	\mathbb{E}_{\widehat{\boldsymbol{\theta}}} \left[\mathbf{F}_{t}\mathbf{Q}^{-1}\mathbf{F}'_{t} | \mathsf{Y}_T \right] &=& \mathbf{Q}^{-1} \star \mathbb{E}_{\widehat{\boldsymbol{\theta}}} \left[\vect{\mathbf{F}_{t}}\vect{\mathbf{F}_{t}}'| \mathsf{Y}_T \right] \\
	&=& \mathbf{Q}^{-1} \star \mathbb{E}_{\widehat{\boldsymbol{\theta}}} \left[\mathsf{f}_{t} \mathsf{f}_{t}'| \mathsf{Y}_T \right] \\
	&=& \mathbf{Q}^{-1} \star \left( \mathsf{f}_{t|T}\mathsf{f}'_{t|T} + \boldsymbol{\Pi}_{t|T}  \right),
	\end{array}
	\]
	\[
	\begin{array}{rcl}
	\mathbb{E}_{\widehat{\boldsymbol{\theta}}} \left[ \mathbf{F}'_{t}\mathbf{P}^{-1}\mathbf{F}_{t}| \mathsf{Y}_T \right]  &=& \mathbf{P}^{-1} \star \mathbb{E}_{\widehat{\boldsymbol{\theta}}} \left[\vect{\mathbf{F}'_{t}} \vect{\mathbf{F}'_{t}}'| \mathsf{Y}_T \right] \\
	&=& \mathbf{P}^{-1} \star \left(\mathbb{K}_{k_1k_2} \mathbb{E}_{\widehat{\boldsymbol{\theta}}} \left[\mathsf{f}_t \mathsf{f}'_{t}| \mathsf{Y}_T \right] \mathbb{K}'_{k_1k_2}\right) \\
	&=& \mathbf{P}^{-1} \star \left(\mathbb{K}_{k_1k_2} \left( \mathsf{f}_{t|T}\mathsf{f}'_{t|T} + \boldsymbol{\Pi}_{t|T}  \right) \mathbb{K}'_{k_1k_2}\right),
	\end{array}	
	\]
	\[
	\begin{array}{rcl}
	\mathbb{E}_{\widehat{\boldsymbol{\theta}}} \left[\mathbf{F}_{t}\mathbf{Q}^{-1}\mathbf{B}\mathbf{F}'_{t-1}| \mathsf{Y}_T \right] &=& \left(\mathbf{Q}^{-1}\mathbf{B}\right) \star \mathbb{E}_{\widehat{\boldsymbol{\theta}}} \left[\vect{\mathbf{F}_{t}}\vect{\mathbf{F}_{t-1}}'| \mathsf{Y}_T \right] \\
	&=& \left(\mathbf{Q}^{-1}\mathbf{B}\right) \star \mathbb{E}_{\widehat{\boldsymbol{\theta}}} \left[\mathsf{f}_{t} \mathsf{f}_{t-1}'| \mathsf{Y}_T \right] \\
	&=& \left(\mathbf{Q}^{-1}\mathbf{B}\right) \star \left( \mathsf{f}_{t|T}\mathsf{f}'_{t-1|T} +  \boldsymbol{\Delta}_{t|T}  \right),
	\end{array}
	\]
	\[
	\begin{array}{rcl}
	\mathbb{E}_{\widehat{\boldsymbol{\theta}}} \left[\mathbf{F}_{t-1}\mathbf{B}'\mathbf{Q}^{-1}\mathbf{B}\mathbf{F}'_{t-1}| \mathsf{Y}_T \right] &=& \left(\mathbf{B}'\mathbf{Q}^{-1}\mathbf{B}\right) \star \mathbb{E}_{\widehat{\boldsymbol{\theta}}} \left[\vect{\mathbf{F}_{t-1}}\vect{\mathbf{F}_{t-1}}'| \mathsf{Y}_T \right] \\
	&=& \left(\mathbf{B}'\mathbf{Q}^{-1}\mathbf{B}\right) \star \mathbb{E}_{\widehat{\boldsymbol{\theta}}} \left[\mathsf{f}_{t-1} \mathsf{f}_{t-1}'| \mathsf{Y}_T \right] \\
	&=& \left(\mathbf{B}'\mathbf{Q}^{-1}\mathbf{B}\right) \star \left( \mathsf{f}_{t-1|T}\mathsf{f}'_{t-1|T} + \boldsymbol{\Pi}_{t-1|T}  \right),
	\end{array}
	\]
	\[
	\begin{array}{rcl}
	\mathbb{E}_{\widehat{\boldsymbol{\theta}}} \left[\mathbf{F}'_{t}\mathbf{P}^{-1}\mathbf{A}\mathbf{F}_{t-1}| \mathsf{Y}_T \right] &=& \left(\mathbf{P}^{-1}\mathbf{A}\right) \star \mathbb{E}_{\widehat{\boldsymbol{\theta}}} \left[\vect{\mathbf{F}'_{t}} \vect{\mathbf{F}'_{t-1}}'| \mathsf{Y}_T \right] \\
	&=& \left(\mathbf{P}^{-1}\mathbf{A}\right) \star \left(\mathbb{K}_{k_1k_2} \mathbb{E}_{\widehat{\boldsymbol{\theta}}} \left[\mathsf{f}_t \mathsf{f}'_{t-1}| \mathsf{Y}_T \right] \mathbb{K}'_{k_1k_2}\right) \\
	&=& \left(\mathbf{P}^{-1}\mathbf{A}\right) \star \left(\mathbb{K}_{k_1k_2} \left( \mathsf{f}_{t|T}\mathsf{f}'_{t-1|T} +  \boldsymbol{\Delta}_{t|T}  \right) \mathbb{K}'_{k_1k_2}\right),
	\end{array}
	\]
	\[
	\begin{array}{rcl}
	\mathbb{E}_{\widehat{\boldsymbol{\theta}}} \left[\mathbf{F}'_{t-1}\mathbf{A}'\mathbf{P}^{-1}\mathbf{A}\mathbf{F}_{t-1}| \mathsf{Y}_T \right] &=& \left(\mathbf{A}'\mathbf{P}^{-1}\mathbf{A}\right) \star \mathbb{E}_{\widehat{\boldsymbol{\theta}}} \left[\vect{\mathbf{F}'_{t-1}} \vect{\mathbf{F}'_{t-1}}'| \mathsf{Y}_T \right] \\
	&=& \left(\mathbf{A}'\mathbf{P}^{-1}\mathbf{A}\right) \star \left(\mathbb{K}_{k_1k_2} \mathbb{E}_{\widehat{\boldsymbol{\theta}}} \left[\mathsf{f}_{t-1} \mathsf{f}'_{t-1}| \mathsf{Y}_T \right] \mathbb{K}'_{k_1k_2}\right)  \\
	&=& \left(\mathbf{A}'\mathbf{P}^{-1}\mathbf{A}\right) \star \left(\mathbb{K}_{k_1k_2} \left( \mathsf{f}_{t-1|T}\mathsf{f}'_{t-1|T} + \boldsymbol{\Pi}_{t-1|T}  \right) \mathbb{K}'_{k_1k_2}\right), \\
	\end{array}
	\]
	\[
	\begin{array}{rcl}
	\mathbb{E}_{\widehat{\boldsymbol{\theta}}} \left[\left(\mathbf{F}_{t}-\mathbf{A}\mathbf{F}_{t-1}\mathbf{B}'\right)\mathbf{Q}^{-1}\left(\mathbf{F}_{t}-\mathbf{A}\mathbf{F}_{t-1}\mathbf{B}' \right)' | \mathsf{Y}_T \right] &=& \mathbb{E}_{\widehat{\boldsymbol{\theta}}} \left[\mathbf{F}_{t}\mathbf{Q}^{-1}\mathbf{F}'_{t} | \mathsf{Y}_T \right] - \mathbb{E}_{\widehat{\boldsymbol{\theta}}} \left[\mathbf{F}_{t}\mathbf{Q}^{-1} \mathbf{B}\mathbf{F}'_{t-1}| \mathsf{Y}_T \right] \mathbf{A}' \\
	&& - \mathbf{A}\mathbb{E}_{\widehat{\boldsymbol{\theta}}} \left[\mathbf{F}_{t-1}\mathbf{B}'\mathbf{Q}^{-1}\mathbf{F}'_{t} | \mathsf{Y}_T \right] \\
	&& + \mathbf{A} \mathbb{E}_{\widehat{\boldsymbol{\theta}}} \left[\mathbf{F}_{t-1}\mathbf{B}'\mathbf{Q}^{-1}\mathbf{B}\mathbf{F}'_{t-1}| \mathsf{Y}_T \right]\mathbf{A}',
	\end{array}
	\]
	\[
	\begin{array}{rcl}
	\mathbb{E}_{\widehat{\boldsymbol{\theta}}} \left[ \left(\mathbf{F}_{t}-\mathbf{A}\mathbf{F}_{t-1}\mathbf{B}'\right)'\mathbf{P}^{-1}\left(\mathbf{F}_{t}-\mathbf{A}\mathbf{F}_{t-1}\mathbf{B}' \right)| \mathsf{Y}_T \right] &=& \mathbb{E}_{\widehat{\boldsymbol{\theta}}} \left[ \mathbf{F}'_{t}\mathbf{P}^{-1}\mathbf{F}_{t}| \mathsf{Y}_T \right] - \mathbb{E}_{\widehat{\boldsymbol{\theta}}} \left[ \mathbf{F}'_{t}\mathbf{P}^{-1} \mathbf{A}\mathbf{F}_{t-1}| \mathsf{Y}_T \right] \mathbf{B}' \\ 
	&& - \mathbf{B}\mathbb{E}_{\widehat{\boldsymbol{\theta}}} \left[ \mathbf{F}'_{t-1}\mathbf{A}'\mathbf{P}^{-1} \mathbf{F}_{t}| \mathsf{Y}_T \right] \\
	&& + \mathbf{B}\mathbb{E}_{\widehat{\boldsymbol{\theta}}} \left[ \mathbf{F}'_{t-1}\mathbf{A}'\mathbf{P}^{-1} \mathbf{A}\mathbf{F}_{t-1}| \mathsf{Y}_T \right]\mathbf{B}'.
	\end{array}
	\]
	Combining these results together with the FOC we obtain
	\[
	\mathbf{R} = \left(\sum_{t=1}^{T} \mathbf{Y}_t\mathbf{K}^{-1}\mathbf{C}\mathbf{F}'_{t|T}\right)\left(\sum_{t=1}^{T}  \left(\mathbf{C'}\mathbf{K}^{-1}\mathbf{C}_t\right)\star \left( \mathsf{f}_{t|T}\mathsf{f}'_{t|T} + \boldsymbol{\Pi}_{t|T}  \right)\right)^{-1},
	\]
	\[
	\mathbf{C} = \left(\sum_{t=1}^{T} \mathbf{Y}'_t\mathbf{H}^{-1}\mathbf{R}\mathbf{F}_{t|T} \right) \left(\sum_{t=1}^{T}   \left(\mathbf{R'}\mathbf{H}^{-1}\mathbf{R}\right) \star \left(\mathbb{K}_{k_1k_2} \left( \mathsf{f}_{t|T}\mathsf{f}'_{t|T} + \boldsymbol{\Pi}_{t|T}  \right) \mathbb{K}'_{k_1k_2}\right) \right)^{-1},
	\]
	\[
	\begin{split}
	\mathbf{H} = \frac{1}{Tp_2 }\sum_{t=1}^{T}  & \left[ \mathbf{Y}_t \mathbf{K}^{-1}\mathbf{Y}'_t - \mathbf{Y}_t \mathbf{K}^{-1} \mathbf{C} \mathbf{F}'_{t|T} \mathbf{R}' -  \mathbf{R} \mathbf{F}_{t|T} \mathbf{C}'\mathbf{K}^{-1}\mathbf{Y}'_t \right. \\ &  \left. + \left(\mathbf{C}'\mathbf{K}^{-1} \mathbf{C}\right) \star \left(\left( \mathbb{I}_{k_2}\otimes\mathbf{R}\right) \left( \mathsf{f}_{t|T}\mathsf{f}'_{t|T} + \boldsymbol{\Pi}_{t|T}\right) \left( \mathbb{I}_{k_2}\otimes\mathbf{R}\right)'\right) \right],
	\end{split}	
	\]
	\[
	\begin{split}
	\mathbf{K} = \frac{1}{Tp_1} \sum_{t=1}^{T} & \left[ \mathbf{Y}'_t \mathbf{H}^{-1} \mathbf{Y}_t - \mathbf{Y}_t'\mathbf{H}^{-1} \mathbf{R} \mathbf{F}_{t|T} \mathbf{C}' - \mathbf{C}\mathbf{F}'_{t|T} \mathbf{R}'\mathbf{H}^{-1}\mathbf{Y}_t \right. \\ &
	\left. + \left(\mathbf{R}'\mathbf{H}^{-1} \mathbf{R}\right) \star \left( \left(\mathbb{I}_{k_1}\otimes\mathbf{C} \right) \left(\mathbb{K}_{k_1k_2} \left( \mathsf{f}_{t|T}\mathsf{f}'_{t|T} + \boldsymbol{\Pi}_{t|T}  \right) \mathbb{K}'_{k_1k_2}\right) \left(\mathbb{I}_{k_1}\otimes\mathbf{C} \right)' \right) \right],
	\end{split}	
	\]
	\[
	\mathbf{A} = \left(\sum_{t=1}^{T} \left(\mathbf{Q}^{-1}\mathbf{B}\right) \star \left( \mathsf{f}_{t|T}\mathsf{f}'_{t-1|T} + \boldsymbol{\Delta}_{t|T}  \right) \right)\left(\sum_{t=1}^{T} \left(\mathbf{B}'\mathbf{Q}^{-1}\mathbf{B}\right) \star \left( \mathsf{f}_{t-1|T}\mathsf{f}'_{t-1|T} + \boldsymbol{\Pi}_{t-1|T}  \right) \right)^{-1},
	\]
	\[
	\setlength{\thinmuskip}{0mu}
	\setlength{\medmuskip}{0mu}
	\setlength{\thickmuskip}{0mu}
	\mathbf{B} = \left(\sum_{t=1}^{T}  \left(\mathbf{P}^{-1}\mathbf{A}\right) \star \left(\mathbb{K}_{k_1k_2} \left( \mathsf{f}_{t|T}\mathsf{f}'_{t-1|T} +  \boldsymbol{\Delta}_{t|T}  \right) \mathbb{K}'_{k_1k_2}\right) \right)\left(\sum_{t=1}^{T}  \left(\mathbf{A}'\mathbf{P}^{-1}\mathbf{A}\right) \star \left(\mathbb{K}_{k_1k_2} \left( \mathsf{f}_{t-1|T}\mathsf{f}'_{t-1|T} + \boldsymbol{\Pi}_{t-1|T}  \right) \mathbb{K}'_{k_1k_2}\right) \right)^{-1},
	\]
	\[
	\begin{split}
	\mathbf{P}  = \frac{1}{(T-1)k_2} \sum_{t=2}^{T} & \left\{ \mathbf{Q}^{-1} \star \left( \mathsf{f}_{t|T}\mathsf{f}'_{t|T} + \boldsymbol{\Pi}_{t|T}  \right) - \left[\left(\mathbf{Q}^{-1}\mathbf{B}\right) \star \left( \mathsf{f}_{t|T}\mathsf{f}'_{t-1|T} +  \boldsymbol{\Delta}_{t|T}  \right)\right] \mathbf{A}' \right. - \\ 
	& \left. \mathbf{A} \left[\left(\mathbf{Q}^{-1}\mathbf{B}\right) \star \left( \mathsf{f}_{t|T}\mathsf{f}'_{t-1|T} +  \boldsymbol{\Delta}_{t|T}  \right)\right]' + \mathbf{A} \left[\left(\mathbf{B}'\mathbf{Q}^{-1}\mathbf{B}\right) \star \left( \mathsf{f}_{t-1|T}\mathsf{f}'_{t-1|T} + \boldsymbol{\Pi}_{t-1|T}  \right)\right] \mathbf{A}' \right\},
	\end{split} 
	\]
	\[
	\begin{split}
	\mathbf{Q} = \frac{1}{(T-1)k_1} \sum_{t=2}^{T} & \left\{\mathbf{P}^{-1} \star \left(\mathbb{K}_{k_1k_2} \left( \mathsf{f}_{t|T}\mathsf{f}'_{t|T} + \boldsymbol{\Pi}_{t|T}  \right) \mathbb{K}'_{k_1k_2}\right)  - \left[\left(\mathbf{P}^{-1}\mathbf{A}\right) \star \left(\mathbb{K}_{k_1k_2} \left( \mathsf{f}_{t|T}\mathsf{f}'_{t-1|T} +  \boldsymbol{\Delta}_{t|T}  \right) \mathbb{K}'_{k_1k_2}\right)\right] \mathbf{B}' \right. \\ & - \mathbf{B} \left[\left(\mathbf{P}^{-1}\mathbf{A}\right) \star \left(\mathbb{K}_{k_1k_2} \left( \mathsf{f}_{t|T}\mathsf{f}'_{t-1|T} +  \boldsymbol{\Delta}_{t|T}  \right) \mathbb{K}'_{k_1k_2}\right)\right]' \\ & + \left. \mathbf{B} \left[\left(\mathbf{A}'\mathbf{P}^{-1}\mathbf{A}\right) \star \left(\mathbb{K}_{k_1k_2} \left( \mathsf{f}_{t-1|T}\mathsf{f}'_{t-1|T} +\boldsymbol{\Pi}_{t-1|T}  \right) \mathbb{K}'_{k_1k_2}\right) \right] \mathbf{B}' \right\}.
	\end{split}	
	\]
	For any iteration $n\ge 0$, given an estimator of the parameters $\widehat{\boldsymbol{\theta}}^{(n)}$, the explicit solutions for $\mathbf R$, $\mathbf C$, $\mathbf H$, and $\mathbf K$ obtained from the previous FOCs are given in Section \ref{subsec:Mstep}. While the solutions for $\mathbf A$, $\mathbf B$, $\mathbf P$, and $\mathbf Q$ are the following:
	\[
	\setlength{\thinmuskip}{0mu}
	\setlength{\medmuskip}{0mu}
	\setlength{\thickmuskip}{0mu}
	\begin{array}{rcl}
	\widehat{\mathbf{A}}^{(n+1)} &=& \left(\sum\limits_{t=2}^{T} \left(\widehat{\mathbf{Q}}^{(n)-1}\widehat{\mathbf{B}}^{(n)}\right) \star \left( \mathsf{f}^{(n)}_{t|T}\mathsf{f}^{(n)\prime}_{t-1|T} + \mathbf{\Delta}^{(n)}_{t|T}  \right) \right)\left(\sum\limits_{t=1}^{T} \left(\widehat{\mathbf{B}}^{(n)\prime}\widehat{\mathbf{Q}}^{(n)-1}\widehat{\mathbf{B}}^{(n)}\right) \star \left(\mathsf{f}^{(n)}_{t-1|T}\mathsf{f}^{(n)\prime}_{t-1|T} + \mathbf{\Pi}^{(n)}_{t-1|T}  \right) \right)^{-1}, \\[0.25in]
	\widehat{\mathbf{B}}^{(n+1)} &=&  \left(\sum\limits_{t=2}^{T}  \left(\widehat{\mathbf{P}}^{(n)-1}\widehat{\mathbf{A}}^{(n+1)}\right) \star \left(\mathbb{K}_{k_1k_2} \left( \mathsf{f}^{(n)}_{t|T}\mathsf{f}^{(n)\prime}_{t-1|T} + \mathbf{\Delta}^{(n)}_{t|T}  \right) \mathbb{K}'_{k_1k_2}\right) \right) \\[0.1in]
	&&	 \hspace{1in} \times \left(\sum\limits_{t=1}^{T}  \left(\widehat{\mathbf{A}}^{(n+1)\prime}\widehat{\mathbf{P}}^{(n)-1}\widehat{\mathbf{A}}^{(n+1)}\right) \star \left(\mathbb{K}_{k_1k_2} \left( \mathsf{f}^{(n)}_{t-1|T}\mathsf{f}^{(n)\prime}_{t-1|T} + \mathbf{\Pi}^{(n)}_{t-1|T} \right) \mathbb{K}'_{k_1k_2}\right) \right)^{-1},
	\end{array}
	\]
	\[
	\setlength{\thinmuskip}{0mu}
	\setlength{\medmuskip}{0mu}
	\setlength{\thickmuskip}{0mu}
	\begin{array}{rcl}
	\widehat{\mathbf{P}}^{(n+1)} &=& \frac{1}{(T-1)k_2} \sum\limits_{t=2}^{T}  \left[ \widehat{\mathbf{Q}}^{(n)-1} \star \left( \mathsf{f}^{(n)}_{t|T}\mathsf{f}^{(n)\prime}_{t|T} + \mathbf{\Pi}^{(n)}_{t|T}  \right)  \right. \\[0.1in] 
	&& \hspace*{0.8in} - \left(\left(\widehat{\mathbf{Q}}^{(n)-1}\widehat{\mathbf{B}}^{(n+1)}\right) \star \left( \mathsf{f}^{(n)}_{t|T}\mathsf{f}^{(n)\prime}_{t-1|T} + \mathbf{\Delta}^{(n)}_{t|T}  \right)\right)  \widehat{\mathbf{A}}^{(n+1)\prime}   \\[0.1in]
	&& \hspace*{0.8in} -  \widehat{\mathbf{A}}^{(n+1)} \left( \left(\widehat{\mathbf{Q}}^{(n)-1}\widehat{\mathbf{B}}^{(n+1)}\right) \star \left( \mathsf{f}^{(n)}_{t|T}\mathsf{f}^{(n)\prime}_{t-1|T} + \mathbf{\Delta}^{(n)}_{t|T}  \right)\right)' \\[0.1in]
	&& \hspace*{0.8in} \left.+ \widehat{\mathbf{A}}^{(n+1)} \left(\left(\widehat{\mathbf{B}}^{(n+1)\prime}\widehat{\mathbf{Q}}^{(n)-1}\widehat{\mathbf{B}}^{(n+1)}\right) \star \left( \mathsf{f}^{(n)}_{t-1|T}\mathsf{f}^{(n)\prime}_{t-1|T} + \mathbf{\Pi}^{(n)}_{t-1|T}  \right)\right) \widehat{\mathbf{A}}^{(n+1)\prime} \right], \\[0.25in]
	\widehat{\mathbf{Q}}^{(n+1)} &=& \frac{1}{(T-1)k_1} \sum\limits_{t=2}^{T}  \left[ \widehat{\mathbf{P}}^{(n)-1} \star \left(\mathbb{K}_{k_1k_2} \left( \mathsf{f}^{(n)}_{t|T}\mathsf{f}^{(n)\prime}_{t|T} + \mathbf{\Pi}^{(n)}_{t|T} \right) \mathbb{K}'_{k_1k_2}\right)\right. \\[0.1in]
	&& \hspace*{0.8in} - \left( \left(\mathbf{P}^{(n)-1}\mathbf{A}^{(n+1)}\right) \star \left(\mathbb{K}_{k_1k_2} \left( \mathsf{f}^{(n)}_{t|T}\mathsf{f}^{(n)\prime}_{t-1|T} + \mathbf{\Delta}^{(n)}_{t|T}  \right) \mathbb{K}'_{k_1k_2}\right) \right) \widehat{\mathbf{B}}'_{(n+1)}  \\[0.1in]
	&& \hspace*{0.8in} - \widehat{\mathbf{B}}^{(n+1)} \left( \left(\widehat{\mathbf{P}}^{(n)-1}\widehat{\mathbf{A}}^{(n+1)}\right) \star \left(\mathbb{K}_{k_1k_2} \left( \mathsf{f}^{(n)}_{t|T}\mathsf{f}^{(n)\prime}_{t-1|T} + \mathbf{\Delta}^{(n)}_{t|T}  \right) \mathbb{K}'_{k_1k_2}\right)\right)' \\[0.1in]
	&& \hspace*{0.8in} + \left. \widehat{\mathbf{B}}^{(n+1)} \left( \left(\widehat{\mathbf{A}}^{(n+1)\prime}\widehat{\mathbf{P}}^{(n)-1}\widehat{\mathbf{A}}^{(n+1)}\right) \star \left(\mathbb{K}_{k_1k_2} \left( \mathsf{f}^{(n)}_{t-1|T}\mathsf{f}^{(n)\prime}_{t-1|T} + \mathbf{\Pi}^{(n)}_{t-1|T}  \right) \mathbb{K}'_{k_1k_2}\right) \right) \widehat{\mathbf{B}}^{(n+1)\prime} \right]. 
	\end{array}	
	\]
	In practice, estimating the factor matrices using 
	$\widehat{\mathbf{B}}^{(n+1)}\otimes \widehat{\mathbf{A}}^{(n+1)}$ and $\widehat{\mathbf{P}}^{(n+1)}\otimes \widehat{\mathbf{Q}}^{(n+1)}$
	or using  those computed directly for the vectorized MAR, i.e., 
			\begin{align}
	\widehat{\mathbf{B}\otimes\mathbf{A}}^{(n+1)} = \left(\sum^T_{t=2}\mathsf{f}^{(n+1)}_{t|T}  \mathsf{f}^{(n)\prime}_{t-1|T}  + \boldsymbol{\Delta}^{(n)}_{t|T}\right) \left(\sum_{t=2}^{T} \mathsf{f}^{(n)}_{t-1|T}  \mathsf{f}^{(n)\prime}_{t-1|T}  + \boldsymbol{\Pi}^{(n)}_{t-1|T}\right)^{-1},\nonumber\\
	\widehat{\mathbf{Q}\otimes \mathbf{P}}^{(n+1)} = \frac{1}{T} \sum_{t=2}^{T} \mathsf{f}^{(k)}_{t|T} \mathsf{f}^{(n)\prime}_{t|T} + \boldsymbol{\Pi}^{(n)}_{t|T} - \left( \mathsf{f}^{(n)}_{t|T} \mathsf{f}^{(n)\prime}_{t-1|T} + \boldsymbol{\Delta}^{(n)}_{t|T} \right)	\widehat{\mathbf{B}\otimes\mathbf{A}}^{(n+1)\prime},\nonumber
	\end{align}
	does not make any appreciable difference, so for ease of computation we suggest to use the latter.

\subsection{Initial estimators}\label{app:yu}
	Let $\mathbf{M}_1=(p_1p_2T)^{-1}\sum_{t=1}^{T}\mathbf{Y}_t\mathbf{Y}'_t$ and $\mathbf{M}_2=(p_1p_2T)^{-1}\sum_{t=1}^{T}\mathbf{Y}'_t\mathbf{Y}_t$, and define $\overline{\mathbf{X}}_t=p^{-1}_2\mathbf{Y}_t\overline{\mathbf{C}}$ and $\overline{\mathbf{Z}}_t=p^{-1}_1\mathbf{Y}'_t\overline{\mathbf{R}}$, where $\overline{\mathbf{R}}=\sqrt{p_1}\ \boldsymbol{\Gamma}^{M_1}$ and $\overline{\mathbf{C}}=\sqrt{p_2}\ \boldsymbol{\Gamma}^{M_2}$, with $\boldsymbol{\Gamma}^{M_i}$ containing the $k_i$ leading eigenvectors of $\mathbf{M}_i$, for $i=1,2$. The estimators $\overline{\mathbf{R}}$ and $\overline{\mathbf{C}}$ are called initial estimators are equivalent to those introdcued by \citet{chen2021statistical}. However, a better estimator of the row (column) loadings can be obtained by PC of the data projected onto the space spanned by the column (row) loadings. Specifically, let $\overline{\mathbf{M}}_1=(p_1p_2T)^{-1}\sum_{t=1}^{T}\overline{\mathbf{X}}_t\overline{\mathbf{X}}'_t$ and $\overline{\mathbf{M}}_2=(p_1p_2T)^{-1}\sum_{t=1}^{T}\overline{\mathbf{Z}}'_t\overline{\mathbf{Z}}_t$. Pre-estimators of $\mathbf{R}$ and $\mathbf{C}$ are given by:	
	\begin{equation}
	\widehat{\mathbf{R}}^{(0)} = \sqrt{p_1}\ \boldsymbol{\Gamma}^{\overline{M}_1}, \qquad \widehat{\mathbf{C}}^{(0)}=\sqrt{p_2}\  \boldsymbol{\Gamma}^{\overline{M}_2},\nonumber
	\end{equation}
	with $\boldsymbol{\Gamma}^{\overline M_i}$ containing the $k_i$ leading eigenvectors of $\overline{\mathbf{M}}_i$, for $i=1,2$. 
	
	The pre-estimator of the factor matrix is then obtained by linear projection as:
	\[
	\widetilde{\mathbf {F}}_t=\frac {
		\widehat{\mathbf{R}}^{(0)'}
	\mathbf{Y}_t
		\widehat{\mathbf{C}}^{(0)}}{p_1p_2}.
		\] 
		Then, letting $\widehat{\mathbf{E}}^{(0)} = \mathbf{Y}_t - \widehat{\mathbf{R}}^{(0)}\widetilde{\mathbf {F}}_t\widehat{\mathbf{C}}^{(0)\prime}$, the pre-estimators of $\mathbf{H}$ and $\mathbf{K}$ are given by:
	
	\begin{align}
	&[\widehat{\mathbf{H}}^{(0)}]_{ii}=\frac{1}{Tp_2} \sum_{t=1}^{T}\left[{\widehat{\mathbf{E}}^{(0)}_t\widehat{\mathbf{E}}^{(0)\prime}_t}\right]_{ii}, &&\quad[\widehat{\mathbf{H}}^{(0)}]_{ij}=0,\qquad i,j=1,\ldots, p_1, \quad i\ne j.	\nonumber\\
	&[\widehat{\mathbf{K}}^{(0)}]_{ii}=\frac{1}{Tp_1} \sum_{t=1}^{T}\left[{\widehat{\mathbf{E}}^{(0)\prime}_t\widehat{\mathbf{H}}^{(0)-1}\widehat{\mathbf{E}}^{(0)}_t}\right]_{ii},&& \quad [\widehat{\mathbf{K}}^{(0)}]_{ij}=0,\qquad i,j=1,\ldots, p_2, \quad i\ne j,	\nonumber
	\end{align}
	Notice that only the pre-estimators of the diagonal terms are needed for running the EM algorithm.

	Then, denoting the pre-estimator of the vectorized factors as $\widetilde{\mathsf{f}}_t=\frac{(\widehat{\mathbf{C}}^{(0)}\otimes\widehat{\mathbf{R}}^{(0)})'\mathsf{y}_t}{p_1p_2}$, the pre-estimators for the MAR parameters  are given by:
	\begin{align}
	\widehat{\mathbf{B\otimes\mathbf{A}}}^{(0)} &= \left(\sum^T_{t=2} \widetilde{\mathsf{f}}_t \widetilde{\mathsf{f}}'_{t-1} \right)\left(\sum^T_{t=2} \widetilde{\mathsf{f}}_{t-1} \widetilde{\mathsf{f}}'_{t-1}\right)^{-1}, \nonumber\\
	\widehat{\mathbf{Q}\otimes\mathbf{P}}^{(0)}& = \left(\sum^T_{t=2} \widetilde{\mathsf{f}}_t - \widehat{\mathbf{B\otimes\mathbf{A}}}^{(0)} \widetilde{\mathsf{f}}_{t-1} \right)\left(\sum^T_{t=2} \widetilde{\mathsf{f}}_t - \widehat{\mathbf{B\otimes\mathbf{A}}}^{(0)} \widetilde{\mathsf{f}}_{t-1} \right)',\nonumber
	\end{align}
	Alternatively, we can obtain the pre-estimators $\widehat{\mathbf{A}}^{(0)}$, $\widehat{\mathbf{B}}^{(0)}$, $\widehat{\mathbf{P}}^{(0)}$ and $\widehat{\mathbf{Q}}^{(0)}$ with the projection method of \cite{chen2021autoregressive} computed when using the estimated factors $\widetilde{\mathsf{f}}_t$.

	\subsection{EM updates with missing observations}
	\label{ap:EMmiss}
	Let $\mathbf{W}_t$ be a $p_1\times p_2$ matrix with ones corresponding to the non-missing entries in $\mathbf{Y}_t$ and zeros otherwise. Decomposing $\mathbf{Y}_t$ as follows
	\[
	\mathbf{Y}_t = \mathbf{W}_t \circ \mathbf{Y}_t + \left(\mathbf{1}_{p_1,p_2}-\mathbf{W}_t\right) \circ \mathbf{Y}_t,
	\]
	we can write 
	\begin{equation}
	\begin{array}{l}
	\textbf{tr}\left(\mathbf{H}^{-1}\left(\mathbf{Y}_t-\mathbf{R}\mathbf{F}_{t}\mathbf{C}'\right)\mathbf{K}^{-1}\left(\mathbf{Y}_t-\mathbf{R}\mathbf{F}_{t}\mathbf{C}'\right)'\right) \\
	= \quad \textbf{tr}\left(\mathbf{H}^{-1}\left(\mathbf{W}_t \circ (\mathbf{Y}_t-\mathbf{R}\mathbf{F}_{t}\mathbf{C}') \right)\mathbf{K}^{-1}\left(\mathbf{W}_t \circ (\mathbf{Y}_t-\mathbf{R}\mathbf{F}_{t}\mathbf{C}') \right)'\right) \\
	\quad + \; \textbf{tr}\left(\mathbf{H}^{-1}\left(\mathbf{W}_t \circ (\mathbf{Y}_t-\mathbf{R}\mathbf{F}_{t}\mathbf{C}') \right)\mathbf{K}^{-1}\left( \left(\mathbf{1}_{p_1,p_2}-\mathbf{W}_t\right) \circ (\mathbf{Y}_t-\mathbf{R}\mathbf{F}_{t}\mathbf{C}')\right)'\right) \\
	\quad+ \; \textbf{tr}\left(\mathbf{H}^{-1}\left(\left(\mathbf{1}_{p_1,p_2}-\mathbf{W}_t\right) \circ (\mathbf{Y}_t-\mathbf{R}\mathbf{F}_{t}\mathbf{C}')\right)\mathbf{K}^{-1}\left(\mathbf{W}_t \circ (\mathbf{Y}_t-\mathbf{R}\mathbf{F}_{t}\mathbf{C}')\right)'\right) \\
	\quad + \; \textbf{tr}\left(\mathbf{H}^{-1}\left(\left(\mathbf{1}_{p_1,p_2}-\mathbf{W}_t\right) \circ (\mathbf{Y}_t-\mathbf{R}\mathbf{F}_{t}\mathbf{C}')\right)\mathbf{K}^{-1}\left(\left(\mathbf{1}_{p_1,p_2}-\mathbf{W}_t\right) \circ (\mathbf{Y}_t-\mathbf{R}\mathbf{F}_{t}\mathbf{C}')\right)'\right).
	\end{array}
	\label{eq:llkdeco}
	\end{equation}
	Moreover, by the law of iterated expectations, we have that 
	\begin{equation}
	\begin{split}
	\mathbb{E}_{\widehat{\boldsymbol{\theta}}}& \left[  \textbf{tr}\left(\mathbf{H}^{-1}\left(\mathbf{Y}_t-\mathbf{R}\mathbf{F}_{t}\mathbf{C}'\right)\mathbf{K}^{-1}\left(\mathbf{Y}_t-\mathbf{R}\mathbf{F}_{t}\mathbf{C}'\right)'\right) | \mathsf{Y}_T \right] =  \\ & \mathbb{E}_{\widehat{\boldsymbol{\theta}}}\left[\mathbb{E}_{\widehat{\boldsymbol{\theta}}}\left[ \textbf{tr}\left(\mathbf{H}^{-1}\left(\mathbf{Y}_t-\mathbf{R}\mathbf{F}_{t}\mathbf{C}'\right)\mathbf{K}^{-1}\left(\mathbf{Y}_t-\mathbf{R}\mathbf{F}_{t}\mathbf{C}'\right)'\right)| \mathbf{F}_t, \mathsf{Y}_T  \right] \right].
	\end{split}
	\label{eq:LIEllk}
	\end{equation}
	Using the properties of Hadamard product and \eqref{eq:hadamtrick}, we obtain
	\begin{equation}
\begin{array}{l}
	\mathbb{E}_{\widehat{\boldsymbol{\theta}}}\left[\textbf{tr}\left(\mathbf{H}^{-1}\left(\mathbf{W}_t \circ (\mathbf{Y}_t-\mathbf{R}\mathbf{F}_{t}\mathbf{C}') \right)\mathbf{K}^{-1}\left( \left(\mathbf{1}_{p_1,p_2}-\mathbf{W}_t\right) \circ (\mathbf{Y}_t-\mathbf{R}\mathbf{F}_{t}\mathbf{C}')\right)'\right)| \mathbf{F}_t, \mathsf{Y}_T  \right] \\
	= \mathbb{E}_{\widehat{\boldsymbol{\theta}}}\left[\textbf{tr}\left(\left(\mathbf{W}_t \circ (\mathbf{H}^{-1}(\mathbf{Y}_t-\mathbf{R}\mathbf{F}_{t}\mathbf{C}'))\mathbf{K}^{-1} \right)\left( \left(\mathbf{1}_{p_1,p_2}-\mathbf{W}_t\right)' \circ (\mathbf{Y}_t-\mathbf{R}\mathbf{F}_{t}\mathbf{C}')' \right)\right)| \mathbf{F}_t, \mathsf{Y}_T \right] \\
	= \mathbb{E}_{\widehat{\boldsymbol{\theta}}}\left[\textbf{tr}\left(\left(\mathbf{W}'_t \circ (\mathbf{K}^{-1}(\mathbf{Y}_t-\mathbf{R}\mathbf{F}_{t}\mathbf{C}')'\mathbf{H}^{-1}) \circ \left(\mathbf{1}_{p_1,p_2}-\mathbf{W}_t\right)' \right) (\mathbf{Y}_t-\mathbf{R}\mathbf{F}_{t}\mathbf{C}')'\right) | \mathbf{F}_t, \mathsf{Y}_T  \right] \\
	= \mathbb{E}_{\widehat{\boldsymbol{\theta}}}\left[\textbf{tr}\left(\left( \mathbf{0}_{p_1,p_2} \circ (\mathbf{K}^{-1}(\mathbf{Y}_t-\mathbf{R}\mathbf{F}_{t}\mathbf{C}')'\mathbf{H}^{-1} ) \right) (\mathbf{Y}_t-\mathbf{R}\mathbf{F}_{t}\mathbf{C}')'\right) | \mathbf{F}_t, \mathsf{Y}_T  \right] = 0,
	\end{array} 
	\label{eq:LIEllkdecoI}
	\end{equation}
	and
	\begin{equation}
	\begin{array}{l}
	\mathbb{E}_{\widehat{\boldsymbol{\theta}}}\left[\textbf{tr}\left(\mathbf{H}^{-1}\left(\left(\mathbf{1}_{p_1,p_2}-\mathbf{W}_t\right) \circ (\mathbf{Y}_t-\mathbf{R}\mathbf{F}_{t}\mathbf{C}') \right)\mathbf{K}^{-1}\left( \left(\mathbf{1}_{p_1,p_2}-\mathbf{W}_t\right) \circ (\mathbf{Y}_t-\mathbf{R}\mathbf{F}_{t}\mathbf{C}')\right)'\right) | \mathbf{F}_t, \mathsf{Y}_T  \right] \\
	= \mathbb{E}_{\widehat{\boldsymbol{\theta}}}\left[\textbf{tr}\left(\left((\mathbf{H}^{-1}\left(\mathbf{1}_{p_1,p_2}-\mathbf{W}_t\right) )\circ (\mathbf{Y}_t-\mathbf{R}\mathbf{F}_{t}\mathbf{C}') \right)\left((\mathbf{K}^{-1}\left(\mathbf{1}_{p_1,p_2}-\mathbf{W}_t\right)') \circ (\mathbf{Y}_t-\mathbf{R}\mathbf{F}_{t}\mathbf{C}')'\right)\right) | \mathbf{F}_t, \mathsf{Y}_T  \right] \\
	= \mathbb{E}_{\widehat{\boldsymbol{\theta}}}\left[\textbf{tr}\left(\left((\left(\mathbf{1}_{p_1,p_2}-\mathbf{W}_t\right)'\mathbf{H}^{-1}) \circ (\mathbf{Y}_t-\mathbf{R}\mathbf{F}_{t}\mathbf{C}')' \circ (\mathbf{K}^{-1}\left(\mathbf{1}_{p_1,p_2}-\mathbf{W}_t\right)') \right)' (\mathbf{Y}_t-\mathbf{R}\mathbf{F}_{t}\mathbf{C}')'\right) | \mathbf{F}_t, \mathsf{Y}_T  \right] \\
	= \mathbb{E}_{\widehat{\boldsymbol{\theta}}}\left[\textbf{tr}\left( \left(( \mathbf{H}^{-1}\left(\mathbf{1}_{p_1,p_2}-\mathbf{W}_t\right)\mathbf{K}^{-1}) \circ (\mathbf{Y}_t-\mathbf{R}\mathbf{F}_{t}\mathbf{C}') \right) (\mathbf{Y}_t-\mathbf{R}\mathbf{F}_{t}\mathbf{C}')'\right) | \mathbf{F}_t, \mathsf{Y}_T  \right] \\
	= \mathbb{E}_{\widehat{\boldsymbol{\theta}}}\left[\textbf{tr}\left( \left((\mathbf{Y}_t-\mathbf{R}\mathbf{F}_{t}\mathbf{C}') \circ (\mathbf{Y}_t-\mathbf{R}\mathbf{F}_{t}\mathbf{C}')\right)' \left( \mathbf{H}^{-1}\left(\mathbf{1}_{p_1,p_2}-\mathbf{W}_t\right)\mathbf{K}^{-1}  \right) \right) | \mathbf{F}_t, \mathsf{Y}_T  \right] \\
\end{array} 
	\label{eq:LIEllkdecoII}
	\end{equation}
	Combining \eqref{eq:llkdeco}, \eqref{eq:LIEllk}, \eqref{eq:LIEllkdecoI} and \eqref{eq:LIEllkdecoII} we get, up to constant terms, that
	\[
	\begin{array}{lll}
	\mathbb{E}_{\widehat{\boldsymbol{\theta}}} \left[\ell(\mathsf{Y}_T|\mathbf{F}_T;\boldsymbol{\theta}) | \mathsf{Y}_T \right]  &=& -\frac{T}{2}\left(p_1\log\left(|\mathbf{K}|\right)+p_2\log\left(|\mathbf{H}|\right)\right) \\[0.1in]
	&& - \frac{1}{2}\sum^T_{t=1} \mathbb{E}_{\widehat{\boldsymbol{\theta}}} \left[ \textbf{tr}\left(\mathbf{H}^{-1}\left(\mathbf{W}_t\circ\left(\mathbf{Y}_t-\mathbf{R}\mathbf{F}_{t}\mathbf{C}'\right)\right)\mathbf{K}^{-1}\left(\mathbf{W}_t\circ\left(\mathbf{Y}_t-\mathbf{R}\mathbf{F}_{t}\mathbf{C}'\right)\right)'\right) | \mathsf{Y}_T \right]  \\[0.1in]
	&& - \frac{1}{2}\sum^T_{t=1} \textbf{tr}\left( \mathbf{H}^{-1}\left(\mathbf{1}_{p_1,p_2}-\mathbf{W}_t\right)\mathbf{K}^{-1}   \left(\widehat{\textbf{H}} \mathbf{1}_{p_1,p_2} \widehat{\textbf{K}} \right)' \right).
	\end{array}		
	\]
	Therefore, the EM updates must be modified accordingly. We first compute derivatives of $\mathbb{E}_{\widehat{\boldsymbol{\theta}}} \left[\ell(\mathsf{Y}_T|\mathbf{F}_T;\boldsymbol{\theta}) | \mathsf{Y}_T \right] $ with respect to $\mathbf{R}$, $\mathbf{C}$, $\mathbf{H}$, and $\mathbf{K}$, obtaining
	\[
	\begin{array}{lll}
	\frac{\partial \mathbb{E}_{\widehat{\boldsymbol{\theta}}} \left[\ell(\mathsf{Y}_T|\mathbf{F}_T;\boldsymbol{\theta}) | \mathsf{Y}_T \right]}{\partial \textbf{R}} &=&\frac{1}{2} \sum\limits_{t=1}^{T} \left(2\frac{\partial \textbf{tr}\left(\mathbf{H}^{-1}\mathbf{W}_t\mathbf{K}^{-1}\left(\mathbf{Y}'_t \circ \mathbf{C}\mathbb{E}_{\widehat{\boldsymbol{\theta}}} \left[\mathbf{F}'_{t}| \mathsf{Y}_T \right]\mathbf{R}'\right)\right)}{\partial \textbf{R}}
	-\frac{\partial \mathbb{E}_{\widehat{\boldsymbol{\theta}}} \left[ \textbf{tr}\left(\mathbf{H}^{-1}\mathbf{W}_t\mathbf{K}^{-1}\left(\mathbf{C}\mathbf{F}'_{t}\mathbf{R}'\circ \mathbf{C}\mathbf{F}'_{t} \mathbf{R}'\right) \right) | \mathsf{Y}_T \right]}{\partial \textbf{R}} \right) \\
	&=&\frac{1}{2} \sum\limits_{t=1}^{T} \left(2\frac{\partial \textbf{tr}\left( \left(\mathbf{H}^{-1}\mathbf{W}_t\mathbf{K}^{-1} \circ \mathbf{Y}_t\right) \mathbf{C}\mathbb{E}_{\widehat{\boldsymbol{\theta}}} \left[\mathbf{F}'_{t}| \mathsf{Y}_T \right]\mathbf{R}'\right)}{\partial \textbf{R}}
	-\frac{\partial   \mathbb{E}_{\widehat{\boldsymbol{\theta}}} \left[ \textbf{tr}\left( \left(\mathbf{H}^{-1}\mathbf{W}_t\mathbf{K}^{-1} \circ \mathbf{R} \mathbf{F}_{t}\mathbf{C}'\right)  \mathbf{C} \mathbf{F}'_{t} \mathbf{R}'\right) | \mathsf{Y}_T \right]}{\partial \textbf{R}} \right) \\
	&=&\frac{1}{2} \sum\limits_{t=1}^{T} \left( 2 \left(\mathbf{H}^{-1}\mathbf{W}_t\mathbf{K}^{-1} \circ \mathbf{Y}_t\right) \mathbf{C}\mathbb{E}_{\widehat{\boldsymbol{\theta}}} \left[\mathbf{F}'_{t}| \mathsf{Y}_T \right] 
	- 2 \mathbb{E}_{\widehat{\boldsymbol{\theta}}} \left[ \left(\mathbf{H}^{-1}\mathbf{W}_t\mathbf{K}^{-1} \circ \mathbf{R}\mathbf{F}_{t}\mathbf{C}'\right)  \mathbf{C} \mathbf{F}'_{t} | \mathsf{Y}_T \right] \right),
	\end{array}	 
	\]
	\[
	\begin{array}{lll}
	\frac{\partial \mathbb{E}_{\widehat{\boldsymbol{\theta}}} \left[\ell(\mathsf{Y}_T|\mathbf{F}_T;\boldsymbol{\theta}) | \mathsf{Y}_T \right]}{\partial \textbf{C}} &=&\frac{1}{2} \sum\limits_{t=1}^{T} \left(2\frac{\partial \textbf{tr}\left(\mathbf{H}^{-1}\mathbf{W}_t\mathbf{K}^{-1}\left(\mathbf{Y}'_t \circ \mathbf{C}\mathbb{E}_{\widehat{\boldsymbol{\theta}}} \left[\mathbf{F}'_{t}| \mathsf{Y}_T \right]\mathbf{R}'\right)\right)}{\partial \textbf{R}}
	-\frac{\partial \mathbb{E}_{\widehat{\boldsymbol{\theta}}} \left[ \textbf{tr}\left(\mathbf{H}^{-1}\mathbf{W}_t\mathbf{K}^{-1}\left(\mathbf{C}\mathbf{F}'_{t}\mathbf{R}'\circ \mathbf{C}\mathbf{F}'_{t} \mathbf{R}'\right) \right) | \mathsf{Y}_T \right]}{\partial \textbf{C}} \right) \\
	&=&\frac{1}{2} \sum\limits_{t=1}^{T} \left(2\frac{\partial \textbf{tr}\left( \left(\mathbf{H}^{-1}\mathbf{W}_t\mathbf{K}^{-1} \circ \mathbf{Y}_t\right) \mathbf{C}\mathbb{E}_{\widehat{\boldsymbol{\theta}}} \left[\mathbf{F}'_{t}| \mathsf{Y}_T \right]\mathbf{R}'\right)}{\partial \textbf{C}}
	-\frac{\partial   \mathbb{E}_{\widehat{\boldsymbol{\theta}}} \left[ \textbf{tr}\left( \left(\mathbf{H}^{-1}\mathbf{W}_t\mathbf{K}^{-1} \circ \mathbf{R} \mathbf{F}_{t}\mathbf{C}'\right)  \mathbf{C} \mathbf{F}'_{t} \mathbf{R}'\right) | \mathsf{Y}_T \right]}{\partial \textbf{C}} \right) \\
	&=&\frac{1}{2} \sum\limits_{t=1}^{T} \left( 2 \left(\mathbf{H}^{-1}\mathbf{W}_t\mathbf{K}^{-1} \circ \mathbf{Y}_t\right)' \mathbf{R}\mathbb{E}_{\widehat{\boldsymbol{\theta}}} \left[\mathbf{F}_{t}| \mathsf{Y}_T \right] 
	- 2 \mathbb{E}_{\widehat{\boldsymbol{\theta}}} \left[ \left(\mathbf{H}^{-1}\mathbf{W}_t\mathbf{K}^{-1} \circ \mathbf{R}\mathbf{F}_{t}\mathbf{C}'\right)'  \mathbf{R} \mathbf{F}_{t} | \mathsf{Y}_T \right] \right),
	\end{array}	 
	\]
	\[
	\setlength{\thinmuskip}{0mu}
	\setlength{\medmuskip}{0mu}
	\setlength{\thickmuskip}{0mu}
	\begin{array}{lll}
	\frac{\partial \mathbb{E}_{\widehat{\boldsymbol{\theta}}} \left[\ell(\mathsf{Y}_T|\mathbf{F}_T;\boldsymbol{\theta}) | \mathsf{Y}_T \right]}{\partial \textbf{H}} \\
	=-\frac{Tp_2}{2} \frac{\partial \log(|\mathbf{H}|)}{\partial \mathbf{H}} - \frac{1}{2} \sum\limits_{t=1}^{T} \frac{\partial   \textbf{tr}\left(\mathbf{H}^{-1} \left(\mathbb{E}_{\widehat{\boldsymbol{\theta}}} \left[\left(\mathbf{W}_t\circ\left(\mathbf{Y}_t-\mathbf{R}\mathbf{F}_{t}\mathbf{C}'\right)\right)\mathbf{K}^{-1}\left(\mathbf{W}_t\circ\left(\mathbf{Y}_t-\mathbf{R}\mathbf{F}_{t}\mathbf{C}'\right)\right)'| \mathsf{Y}_T \right] + \left(\mathbf{1}_{p_1,p_2}-\mathbf{W}_t\right)\mathbf{K}^{-1}   \left(\widehat{\textbf{H}} \mathbf{1}_{p_1,p_2} \widehat{\textbf{K}} \right)'\right) \right) }{\partial \mathbf{H}} \\
	= -\frac{Tp_2}{2} \mathbf{H}^{-1} + \frac{1}{2} \sum\limits_{t=1}^{T} \mathbf{H}^{-1} \left(\mathbb{E}_{\widehat{\boldsymbol{\theta}}} \left[\left(\mathbf{W}_t\circ\left(\mathbf{Y}_t-\mathbf{R}\mathbf{F}_{t}\mathbf{C}'\right)\right)\mathbf{K}^{-1}\left(\mathbf{W}_t\circ\left(\mathbf{Y}_t-\mathbf{R}\mathbf{F}_{t}\mathbf{C}'\right)\right)'| \mathsf{Y}_T \right]'  \right. \\
	\hspace*{4in}\left.	+ \left(\widehat{\textbf{H}} \mathbf{1}_{p_1,p_2} \widehat{\textbf{K}} \right) \mathbf{K}^{-1}  \left(\mathbf{1}_{p_1,p_2}-\mathbf{W}_t\right)'  \right)  \mathbf{H}^{-1}, \\
	\end{array}	 
	\]
	\[
	\setlength{\thinmuskip}{0mu}
	\setlength{\medmuskip}{0mu}
	\setlength{\thickmuskip}{0mu}
	\begin{array}{lll}
	\frac{\partial \mathbb{E}_{\widehat{\boldsymbol{\theta}}} \left[\ell(\mathsf{Y}_T|\mathbf{F}_T;\boldsymbol{\theta}) | \mathsf{Y}_T \right]}{\partial \textbf{K}} \\
	=-\frac{Tp_1}{2} \frac{\partial \log(|\mathbf{K}|)}{\partial \mathbf{K}} - \frac{1}{2} \sum\limits_{t=1}^{T} \frac{\partial   \textbf{tr}\left(\mathbf{K}^{-1} \left(\mathbb{E}_{\widehat{\boldsymbol{\theta}}} \left[\left(\mathbf{W}_t\circ\left(\mathbf{Y}_t-\mathbf{R}\mathbf{F}_{t}\mathbf{C}'\right)\right)'\mathbf{H}^{-1}\left(\mathbf{W}_t\circ\left(\mathbf{Y}_t-\mathbf{R}\mathbf{F}_{t}\mathbf{C}'\right)\right)| \mathsf{Y}_T \right] + \left(\widehat{\textbf{H}} \mathbf{1}_{p_1,p_2} \widehat{\textbf{K}} \right)'\mathbf{H}^{-1} \left(\mathbf{1}_{p_1,p_2}-\mathbf{W}_t\right)   \right) \right) }{\partial \mathbf{H}} \\
	= -\frac{Tp_1}{2} \mathbf{K}^{-1} + \frac{1}{2} \sum\limits_{t=1}^{T} \mathbf{K}^{-1} \left(\mathbb{E}_{\widehat{\boldsymbol{\theta}}} \left[\left(\mathbf{W}_t\circ\left(\mathbf{Y}_t-\mathbf{R}\mathbf{F}_{t}\mathbf{C}'\right)\right)'\mathbf{H}^{-1}\left(\mathbf{W}_t\circ\left(\mathbf{Y}_t-\mathbf{R}\mathbf{F}_{t}\mathbf{C}'\right)\right)| \mathsf{Y}_T \right]'  \right. \\
	\hspace*{4in}\left.	+  \left(\mathbf{1}_{p_1,p_2}-\mathbf{W}_t\right)'\mathbf{H}^{-1} \left(\widehat{\textbf{H}} \mathbf{1}_{p_1,p_2} \widehat{\textbf{K}} \right) \right)  \mathbf{K}^{-1}. \\
	\end{array}	 
	\]
	The FOC for $\mathbf{R}$ then yields
	\[
		\sum_{t=1}^{T} \left(\mathbf{H}^{-1}\mathbf{W}_t\mathbf{K}^{-1} \circ \mathbf{Y}_t\right) \mathbf{C}\mathbb{E}_{\widehat{\boldsymbol{\theta}}} \left[\mathbf{F}'_{t}| \mathsf{Y}_T \right] 
		= \sum_{t=1}^{T}  \mathbb{E}_{\widehat{\boldsymbol{\theta}}} \left[ \left(\mathbf{H}^{-1}\mathbf{W}_t\mathbf{K}^{-1} \circ \mathbf{R}\mathbf{F}_{t}\mathbf{C}'\right)  \mathbf{C} \mathbf{F}'_{t} | \mathsf{Y}_T \right],
	\]
	from which we obtain
	\[
	\begin{array}{rll}
	\sum\limits_{t=1}^{T} \vect{\left(\mathbf{W}_t \circ \mathbf{H}^{-1}\mathbf{Y}_t\mathbf{K}^{-1}\right) \mathbf{C}\mathbb{E}_{\widehat{\boldsymbol{\theta}}} \left[\mathbf{F}'_{t}| \mathsf{Y}_T \right]}
	&=& \sum\limits_{t=1}^{T}  \mathbb{E}_{\widehat{\boldsymbol{\theta}}} \left[\vect{ \left(\mathbf{W}_t \circ \mathbf{H}^{-1}\mathbf{R}\mathbf{F}_{t}\mathbf{C}'\mathbf{K}^{-1} \right)  \mathbf{C} \mathbf{F}'_{t}} | \mathsf{Y}_T \right]  \\
	&=& \sum\limits_{t=1}^{T} \mathbb{E}_{\widehat{\boldsymbol{\theta}}} \left[ \left( \mathbf{F}_{t}\mathbf{C}' \otimes \mathbb{I}_{p_1} \right) \vect{ \left(\mathbf{W}_t \circ \mathbf{H}^{-1}\mathbf{R}\mathbf{F}_{t}\mathbf{C}'\mathbf{K}^{-1} \right)  \mathbf{C} \mathbf{F}'_{t}} | \mathsf{Y}_T \right] \\
	&=& \sum\limits_{t=1}^{T} \mathbb{E}_{\widehat{\boldsymbol{\theta}}} \left[ \left( \mathbf{F}_{t}\mathbf{C}' \otimes \mathbb{I}_{p_1} \right) \mathbb{D}_{\mathbf{W}_t} \vect{ \mathbf{H}^{-1}\mathbf{R}\mathbf{F}_{t}\mathbf{C}'\mathbf{K}^{-1}  } | \mathsf{Y}_T \right] \\
	&=& \sum\limits_{t=1}^{T} \mathbb{E}_{\widehat{\boldsymbol{\theta}}} \left[ \left( \mathbf{F}_{t}\mathbf{C}' \otimes \mathbb{I}_{p_1} \right) \mathbb{D}_{\mathbf{W}_t} \left(\mathbf{K}^{-1} \mathbf{C} \mathbf{F}'_{t} \otimes \mathbf{H}^{-1}\right)  | \mathsf{Y}_T \right]\vect{\mathbf{R}},
	\end{array}
	\]
	implying that
	\[
	\begin{array}{l}
	 \vect{\mathbf{R}} = \left(\sum\limits_{t=1}^{T} \mathbb{E}_{\widehat{\boldsymbol{\theta}}} \left[ \left( \mathbf{F}_{t}\mathbf{C}' \otimes \mathbb{I}_{p_1} \right) \mathbb{D}_{\mathbf{W}_t} \left(\mathbf{K}^{-1} \mathbf{C} \mathbf{F}'_{t} \otimes \mathbf{H}^{-1}\right)  | \mathsf{Y}_T \right] \right)^{-1} \left(\sum\limits_{t=1}^{T} \vect{\left(\mathbf{W}_t \circ \mathbf{H}^{-1}\mathbf{Y}_t\mathbf{K}^{-1}\right) \mathbf{C}\mathbb{E}_{\widehat{\boldsymbol{\theta}}} \left[\mathbf{F}'_{t}| \mathsf{Y}_T \right]}\right).
	 \end{array}
	\]
	The FOC for $\mathbf{C}$  yields
	\[
	\sum\limits_{t=1}^{T}  \left(\mathbf{H}^{-1}\mathbf{W}_t\mathbf{K}^{-1} \circ \mathbf{Y}_t\right)' \mathbf{R}\mathbb{E}_{\widehat{\boldsymbol{\theta}}} \left[\mathbf{F}_{t}| \mathsf{Y}_T \right] 
	= \sum\limits_{t=1}^{T} \mathbb{E}_{\widehat{\boldsymbol{\theta}}} \left[ \left(\mathbf{H}^{-1}\mathbf{W}_t\mathbf{K}^{-1} \circ \mathbf{R}\mathbf{F}_{t}\mathbf{C}'\right)'  \mathbf{R} \mathbf{F}_{t} | \mathsf{Y}_T \right], \\
	\]
	from which we obtain
	\[
	\begin{array}{rll}
	\sum\limits_{t=1}^{T} \vect{\left(\mathbf{W}_t \circ \mathbf{H}^{-1}\mathbf{Y}_t\mathbf{K}^{-1}\right)' \mathbf{R}\mathbb{E}_{\widehat{\boldsymbol{\theta}}} \left[\mathbf{F}_{t}| \mathsf{Y}_T \right]}
	&=& \sum\limits_{t=1}^{T}  \mathbb{E}_{\widehat{\boldsymbol{\theta}}} \left[\vect{ \left(\mathbf{W}_t \circ \mathbf{H}^{-1}\mathbf{R}\mathbf{F}_{t}\mathbf{C}'\mathbf{K}^{-1} \right)'  \mathbf{R} \mathbf{F}_{t}} | \mathsf{Y}_T \right]  \\
	&=& \sum\limits_{t=1}^{T} \mathbb{E}_{\widehat{\boldsymbol{\theta}}} \left[ \left( \mathbf{F}'_{t}\mathbf{R}' \otimes \mathbb{I}_{p_2} \right) \vect{ \left(\mathbf{W}_t \circ \mathbf{H}^{-1}\mathbf{R}\mathbf{F}_{t}\mathbf{C}'\mathbf{K}^{-1} \right)'  \mathbf{R} \mathbf{F}_{t}} | \mathsf{Y}_T \right] \\
	&=& \sum\limits_{t=1}^{T} \mathbb{E}_{\widehat{\boldsymbol{\theta}}} \left[ \left( \mathbf{F}'_{t}\mathbf{R}' \otimes \mathbb{I}_{p_2} \right) \mathbb{D}_{\mathbf{W}'_t} \vect{ \mathbf{K}^{-1}\mathbf{C}\mathbf{F}'_{t}\mathbf{R}'\mathbf{H}^{-1}  } | \mathsf{Y}_T \right] \\
	&=& \sum\limits_{t=1}^{T} \mathbb{E}_{\widehat{\boldsymbol{\theta}}} \left[ \left( \mathbf{F}'_{t}\mathbf{R}' \otimes \mathbb{I}_{p_2} \right) \mathbb{D}_{\mathbf{W}'_t} \left(\mathbf{H}^{-1} \mathbf{R} \mathbf{F}_{t} \otimes \mathbf{K}^{-1}\right)  | \mathsf{Y}_T \right] \vect{\mathbf{C}},
	\end{array}
	\]
	implying that
	\[
	\begin{array}{l}
	\vect{\mathbf{C}} = \left(\sum\limits_{t=1}^{T} \mathbb{E}_{\widehat{\boldsymbol{\theta}}} \left[ \left( \mathbf{F}'_{t}\mathbf{R}' \otimes \mathbb{I}_{p_2} \right) \mathbb{D}_{\mathbf{W}'_t} \left(\mathbf{H}^{-1} \mathbf{R} \mathbf{F}_{t} \otimes \mathbf{K}^{-1}\right)  | \mathsf{Y}_T \right] \right)^{-1} \left(\sum\limits_{t=1}^{T} \vect{\left(\mathbf{W}_t \circ \mathbf{H}^{-1}\mathbf{Y}_t\mathbf{K}^{-1}\right)' \mathbf{R}\mathbb{E}_{\widehat{\boldsymbol{\theta}}} \left[\mathbf{F}_{t}| \mathsf{Y}_T \right]}\right).
	\end{array}
	\]
	The FOC for $\mathbf{H}$  yields
	\[
	\begin{array}{l}
	\frac{Tp_2}{2} \mathbf{H}^{-1} = \frac{1}{2} \sum_{t=1}^{T} \mathbf{H}^{-1} \left(\mathbb{E}_{\widehat{\boldsymbol{\theta}}} \left[\left(\mathbf{W}_t\circ\left(\mathbf{Y}_t-\mathbf{R}\mathbf{F}_{t}\mathbf{C}'\right)\right)\mathbf{K}^{-1}\left(\mathbf{W}_t\circ\left(\mathbf{Y}_t-\mathbf{R}\mathbf{F}_{t}\mathbf{C}'\right)\right)'| \mathsf{Y}_T \right]'  \right. \\
	\hspace*{4in}\left.	+ \left(\widehat{\textbf{H}} \mathbf{1}_{p_1,p_2} \widehat{\textbf{K}} \right) \mathbf{K}^{-1}  \left(\mathbf{1}_{p_1,p_2}-\mathbf{W}_t\right)'  \right)  \mathbf{H}^{-1},  \\
	\end{array}
	\]
	implying that, for $i=1,\ldots, p_1$,
	\[
	\begin{array}{l}
	 [\mathbf{H}]_{ii} = \left[ \frac{\sum_{t=1}^{T} \left(\mathbb{E}_{\widehat{\boldsymbol{\theta}}} \left[\left(\mathbf{W}_t\circ\left(\mathbf{Y}_t-\mathbf{R}\mathbf{F}_{t}\mathbf{C}'\right)\right)\mathbf{K}^{-1}\left(\mathbf{W}_t\circ\left(\mathbf{Y}_t-\mathbf{R}\mathbf{F}_{t}\mathbf{C}'\right)\right)'| \mathsf{Y}_T \right]'  + \left(\widehat{\textbf{H}} \mathbf{1}_{p_1,p_2} \widehat{\textbf{K}} \right) \mathbf{K}^{-1}  \left(\mathbf{1}_{p_1,p_2}-\mathbf{W}_t\right)'  \right)}{Tp_2} \right]_{ii}.
	 \end{array}
	\]
	The FOC for $\mathbf{K}$  yields
	\[
	\begin{array}{l}
	\frac{Tp_1}{2} \mathbf{K}^{-1} = \frac{1}{2} \sum_{t=1}^{T} \mathbf{K}^{-1} \left(\mathbb{E}_{\widehat{\boldsymbol{\theta}}} \left[\left(\mathbf{W}_t\circ\left(\mathbf{Y}_t-\mathbf{R}\mathbf{F}_{t}\mathbf{C}'\right)\right)'\mathbf{H}^{-1}\left(\mathbf{W}_t\circ\left(\mathbf{Y}_t-\mathbf{R}\mathbf{F}_{t}\mathbf{C}'\right)\right)| \mathsf{Y}_T \right]'  \right. \\
	\hspace*{4in}\left.	+  \left(\mathbf{1}_{p_1,p_2}-\mathbf{W}_t\right)'\mathbf{H}^{-1} \left(\widehat{\textbf{H}} \mathbf{1}_{p_1,p_2} \widehat{\textbf{K}} \right) \right)  \mathbf{K}^{-1}, \\
	\end{array}
	\]
	implying that, for $i=1,\ldots, p_2$,
	\[
	\begin{array}{l}
		 [\mathbf{K}]_{ii} = \left[ \frac{\sum_{t=1}^{T} \left( \mathbb{E}_{\widehat{\boldsymbol{\theta}}} \left[\left(\mathbf{W}_t\circ\left(\mathbf{Y}_t-\mathbf{R}\mathbf{F}_{t}\mathbf{C}'\right)\right)'\mathbf{H}^{-1}\left(\mathbf{W}_t\circ\left(\mathbf{Y}_t-\mathbf{R}\mathbf{F}_{t}\mathbf{C}'\right)\right)| \mathsf{Y}_T \right]' +  \left(\mathbf{1}_{p_1,p_2}-\mathbf{W}_t\right)'\mathbf{H}^{-1} \left(\widehat{\textbf{H}} \mathbf{1}_{p_1,p_2} \widehat{\textbf{K}} \right) \right)}{Tp_1} \right]_{ii}.
	\end{array}
	\]
	Using the conditional moments computed with the Kalman recursions and \eqref{eq:matdec}, we obtain
	\[
	\begin{array}{l}
	\mathbb{E}_{\widehat{\boldsymbol{\theta}}} \left[ \left( \mathbf{F}_{t}\mathbf{C}' \otimes \mathbb{I}_{p_1} \right) \mathbb{D}_{\mathbf{W}_t} \left(\mathbf{K}^{-1} \mathbf{C} \mathbf{F}'_{t} \otimes \mathbf{H}^{-1}\right)  | \mathsf{Y}_T \right] \\
	\hspace*{1in}= \mathbb{E}_{\widehat{\boldsymbol{\theta}}} \left[ \left( \mathbf{F}_{t}\mathbf{C}' \otimes \mathbb{I}_{p_1} \right) \left(\sum\limits_{s=1}^{p_1} \sum\limits_{q=1}^{p_1} \mathbb{D}^{[s,q]}_{\mathbf{W}_t} \otimes \mathbb{E}^{(s,q)}_{p_1,p_1} \right) \left(\mathbf{K}^{-1} \mathbf{C} \mathbf{F}'_{t} \otimes \mathbf{H}^{-1}\right)  | \mathsf{Y}_T \right]  \\
	\hspace*{1in}= \sum\limits_{s=1}^{p_1} \sum\limits_{q=1}^{p_1} \mathbb{E}_{\widehat{\boldsymbol{\theta}}} \left[ \left( \mathbf{F}_{t}\mathbf{C}' \otimes \mathbb{I}_{p_1} \right) \left( \mathbb{D}^{[s,q]}_{\mathbf{W}_t} \otimes \mathbb{E}^{(s,q)}_{p_1,p_1} \right) \left(\mathbf{K}^{-1} \mathbf{C} \mathbf{F}'_{t} \otimes \mathbf{H}^{-1}\right)  | \mathsf{Y}_T \right] \\
	\hspace*{1in}= \sum\limits_{s=1}^{p_1} \sum\limits_{q=1}^{p_1} \mathbb{E}_{\widehat{\boldsymbol{\theta}}} \left[ \left( \mathbf{F}_{t}\mathbf{C}'\mathbb{D}^{[s,q]}_{\mathbf{W}_t}   \otimes \mathbb{I}_{p_1} \mathbb{E}^{(s,q)}_{p_1,p_1} \right) \left(\mathbf{K}^{-1} \mathbf{C} \mathbf{F}'_{t} \otimes \mathbf{H}^{-1}\right)  | \mathsf{Y}_T \right]  \\
	\hspace*{1in}= \sum\limits_{s=1}^{p_1} \sum\limits_{q=1}^{p_1} \mathbb{E}_{\widehat{\boldsymbol{\theta}}} \left[ \left( \mathbf{F}_{t}\mathbf{C}'\mathbb{D}^{[s,q]}_{\mathbf{W}_t} \mathbf{K}^{-1} \mathbf{C} \mathbf{F}'_{t}  \otimes \mathbb{I}_{p_1} \mathbb{E}^{(s,q)}_{p_1,p_1}\mathbf{H}^{-1}\right)  | \mathsf{Y}_T \right] 	\\
\end{array}	
	\]
	\[
	\begin{array}{l}
	\mathbb{E}_{\widehat{\boldsymbol{\theta}}} \left[ \left( \mathbf{F}'_{t}\mathbf{R}' \otimes \mathbb{I}_{p_2} \right) \mathbb{D}_{\mathbf{W}'_t} \left(\mathbf{H}^{-1} \mathbf{R} \mathbf{F}_{t} \otimes \mathbf{K}^{-1}\right)  | \mathsf{Y}_T \right] \\
	\hspace*{1in}= \mathbb{E}_{\widehat{\boldsymbol{\theta}}} \left[ \left( \mathbf{F}'_{t}\mathbf{R}' \otimes \mathbb{I}_{p_2} \right) \left(\sum\limits_{s=1}^{p_2} \sum\limits_{q=1}^{p_2} \mathbb{D}^{[s,q]}_{\mathbf{W}'_t} \otimes \mathbb{E}^{(s,q)}_{p_2,p_2}\right) \left(\mathbf{H}^{-1} \mathbf{R} \mathbf{F}_{t} \otimes \mathbf{K}^{-1}\right)  | \mathsf{Y}_T \right] \\
	\hspace*{1in}= \sum\limits_{s=1}^{p_2} \sum\limits_{q=1}^{p_2} \mathbb{E}_{\widehat{\boldsymbol{\theta}}} \left[ \left( \mathbf{F}'_{t}\mathbf{R}' \otimes \mathbb{I}_{p_2} \right) \left( \mathbb{D}^{[s,q]}_{\mathbf{W}'_t} \otimes \mathbb{E}^{(s,q)}_{p_2,p_2} \right) \left(\mathbf{H}^{-1} \mathbf{R} \mathbf{F}_{t} \otimes \mathbf{K}^{-1}\right)  | \mathsf{Y}_T \right] \\
	\hspace*{1in}= \sum\limits_{s=1}^{p_2} \sum\limits_{q=1}^{p_2} \mathbb{E}_{\widehat{\boldsymbol{\theta}}} \left[ \left( \mathbf{F}'_{t}\mathbf{R}'\mathbb{D}^{[s,q]}_{\mathbf{W}'_t}  \otimes \mathbb{I}_{p_2}\mathbb{E}^{(s,q)}_{p_2,p_2} \right) \left(\mathbf{H}^{-1} \mathbf{R} \mathbf{F}_{t} \otimes \mathbf{K}^{-1}\right)  | \mathsf{Y}_T \right] \\
	\hspace*{1in}= \sum\limits_{s=1}^{p_2} \sum\limits_{q=1}^{p_2} \mathbb{E}_{\widehat{\boldsymbol{\theta}}} \left[ \left( \mathbf{F}'_{t}\mathbf{R}'\mathbb{D}^{[s,q]}_{\mathbf{W}'_t}\mathbf{H}^{-1} \mathbf{R} \mathbf{F}_{t}  \otimes \mathbb{I}_{p_2}\mathbb{E}^{(s,q)}_{p_2,p_2} \mathbf{K}^{-1}\right)  | \mathsf{Y}_T \right] \\
\end{array}	
	\]
	Moreover,
	\[
	\begin{array}{l}
	\mathbb{E}_{\widehat{\boldsymbol{\theta}}} \left[ \left(\mathbf{W}_t\circ\mathbf{R}\mathbf{F}_{t}\mathbf{C}'\right) \mathbf{K}^{-1}_t \left(\mathbf{W}_t\circ\mathbf{R}\mathbf{F}_{t}\mathbf{C}'\right)'| \mathsf{Y}_T \right] \\
	\hspace*{1in}= \mathbf{K}^{-1}_t \star \mathbb{E}_{\widehat{\boldsymbol{\theta}}} \left[ \vect{\mathbf{W}_t\circ\mathbf{R}\mathbf{F}_{t}\mathbf{C}'} \vect{\mathbf{W}_t\circ\mathbf{R}\mathbf{F}_{t}\mathbf{C}'}'| \mathsf{Y}_T \right] \\
	\hspace*{1in}= \mathbf{K}^{-1}_t \star \mathbb{E}_{\widehat{\boldsymbol{\theta}}} \left[ \mathbb{D}_{\mathbf{W}_t} \vect{\mathbf{R}\mathbf{F}_{t}\mathbf{C}'} \vect{\mathbf{R}\mathbf{F}_{t}\mathbf{C}'}' \mathbb{D}'_{\mathbf{W}_t}| \mathsf{Y}_T \right] \\
	\hspace*{1in}= \mathbf{K}^{-1}_t \star \left(\mathbb{D}_{\mathbf{W}_t} \left(\mathbf{C}\otimes\mathbf{R}\right) \mathbb{E}_{\widehat{\boldsymbol{\theta}}} \left[ \mathsf{f}_t \mathsf{f}'_t| \mathsf{Y}_T \right] \left(\mathbf{C}\otimes\mathbf{R}\right)' \mathbb{D}'_{\mathbf{W}_t}\right) \\
	\hspace*{1in}= \mathbf{K}^{-1}_t \star \left(\mathbb{D}_{\mathbf{W}_t} \left(\mathbf{C}\otimes\mathbf{R}\right) \left(\mathsf{f}_{t|T} \mathsf{f}'_{t|T} + \boldsymbol{\Pi}_{t|T}\right) \left(\mathbf{C}\otimes\mathbf{R}\right)' \mathbb{D}'_{\mathbf{W}_t}\right), \\
	\end{array}
	\]
	\[
	\begin{array}{l}
	\mathbb{E}_{\widehat{\boldsymbol{\theta}}} \left[ \left(\mathbf{W}_t\circ\mathbf{R}\mathbf{F}_{t}\mathbf{C}'\right)' \mathbf{H}^{-1} \left(\mathbf{W}_t\circ\mathbf{R}\mathbf{F}_{t}\mathbf{C}'\right)| \mathsf{Y}_T \right] \\
	\hspace*{1in}= \mathbf{H}^{-1} \star \mathbb{E}_{\widehat{\boldsymbol{\theta}}} \left[ \vect{\left(\mathbf{W}_t\circ\mathbf{R}\mathbf{F}_{t}\mathbf{C}'\right)'} \vect{\left(\mathbf{W}_t\circ\mathbf{R}\mathbf{F}_{t}\mathbf{C}'\right)'}' | \mathsf{Y}_T \right] \\
	\hspace*{1in}= \mathbf{H}^{-1} \star \mathbb{E}_{\widehat{\boldsymbol{\theta}}} \left[ \mathbb{D}_{\mathbf{W}'_t} \vect{\mathbf{C}\mathbf{F}'_{t}\mathbf{R}'}  \vect{\mathbf{C}\mathbf{F}'_{t}\mathbf{R}'}' \mathbb{D}'_{\mathbf{W}'_t}| \mathsf{Y}_T \right] \\
	\hspace*{1in}= \mathbf{H}^{-1} \star  \left(\mathbb{D}_{\mathbf{W}'_t} \left(\mathbf{R}\otimes\mathbf{C}\right) \mathbb{E}_{\widehat{\boldsymbol{\theta}}} \left[\vect{\mathbf{F}'_{t}}  \vect{\mathbf{F}'_{t}}'| \mathsf{Y}_T \right] \left(\mathbf{R}\otimes\mathbf{C}\right)' \mathbb{D}'_{\mathbf{W}'_t}\right) \\
	\hspace*{1in}= \mathbf{H}^{-1} \star  \left(\mathbb{D}_{\mathbf{W}'_t} \left(\mathbf{R}\otimes\mathbf{C}\right) \left(\mathbb{K}_{k_1k_2} \left( \mathsf{f}_{t|T}\mathsf{f}'_{t|T} + \boldsymbol{\Pi}_{t|T}  \right) \mathbb{K}'_{k_1k_2}\right) \left(\mathbf{R}\otimes\mathbf{C}\right)' \mathbb{D}'_{\mathbf{W}'_t}\right),
	\end{array}
	\]
	\[
	\begin{array}{l}
	\mathbb{E}_{\widehat{\boldsymbol{\theta}}} \left[ \left(\mathbf{W}_t\circ\left(\mathbf{Y}_t-\mathbf{R}\mathbf{F}_{t}\mathbf{C}'\right)\right)\mathbf{K}^{-1}\left(\mathbf{W}_t\circ\left(\mathbf{Y}_t-\mathbf{R}\mathbf{F}_{t}\mathbf{C}'\right)\right)'| \mathsf{Y}_T \right] \\
	\hspace*{1in}= \left(\mathbf{W}_t\circ\mathbf{Y}_t\right)\mathbf{K}^{-1}\left(\mathbf{W}_t\circ\mathbf{Y}_t\right)' - \left(\mathbf{W}_t\circ\mathbf{Y}_t\right)\mathbf{K}^{-1}\left(\mathbf{W}_t\circ\mathbf{R}\mathbb{E}_{\widehat{\boldsymbol{\theta}}} \left[\mathbf{F}_{t}|\mathsf{Y}_T \right] \mathbf{C}'\right)'   \\
	\hspace*{1.2in}- \left(\mathbf{W}_t\circ\mathbf{R}\mathbb{E}_{\widehat{\boldsymbol{\theta}}} \left[\mathbf{F}_{t}|\mathsf{Y}_T \right]\mathbf{C}'\right) \mathbf{K}^{-1} \left(\mathbf{W}_t\circ \mathbf{Y}_t\right)' \\
	\hspace*{1.2in}+	\mathbb{E}_{\widehat{\boldsymbol{\theta}}} \left[ \left(\mathbf{W}_t\circ\mathbf{R}\mathbf{F}_{t}\mathbf{C}'\right) \mathbf{K}^{-1} \left(\mathbf{W}_t\circ\mathbf{R}\mathbf{F}_{t}\mathbf{C}'\right)'| \mathsf{Y}_T \right] \\
	\hspace*{1in}= \left(\mathbf{W}_t\circ\mathbf{Y}_t\right)\mathbf{K}^{-1}\left(\mathbf{W}_t\circ\mathbf{Y}_t\right)' - \left(\mathbf{W}_t\circ\mathbf{Y}_t\right)\mathbf{K}^{-1}\left(\mathbf{W}_t\circ\mathbf{R}\mathbf{F}_{t|T} \mathbf{C}'\right)'   \\
	\hspace*{1.2in}- \left(\mathbf{W}_t\circ\mathbf{R} \mathbf{F}_{t|T} \mathbf{C}'\right) \mathbf{K}^{-1} \left(\mathbf{W}_t\circ \mathbf{Y}_t\right)' \\
	\hspace*{1.2in}+	\mathbf{K}^{-1}_t \star \left(\mathbb{D}_{\mathbf{W}_t} \left(\mathbf{C}\otimes\mathbf{R}\right) \left(\mathsf{f}_{t|T} \mathsf{f}'_{t|T} + \boldsymbol{\Pi}_{t|T}\right) \left(\mathbf{C}\otimes\mathbf{R}\right)' \mathbb{D}'_{\mathbf{W}_t}\right),
	\end{array}
	\]
	\[
	\begin{array}{l}
	\mathbb{E}_{\widehat{\boldsymbol{\theta}}} \left[\left(\mathbf{W}_t\circ\left(\mathbf{Y}_t-\mathbf{R}\mathbf{F}_{t}\mathbf{C}'\right)\right)'\mathbf{H}^{-1}\left(\mathbf{W}_t\circ\left(\mathbf{Y}_t-\mathbf{R}\mathbf{F}_{t}\mathbf{C}'\right)\right)| \mathsf{Y}_T \right] \\
	\hspace*{1in}= \left(\mathbf{W}_t\circ\mathbf{Y}_t\right)'\mathbf{H}^{-1}\left(\mathbf{W}_t\circ\mathbf{Y}_t\right) - \left(\mathbf{W}_t\circ\mathbf{Y}_t\right)'\mathbf{H}^{-1}\left(\mathbf{W}_t\circ\mathbf{R}\mathbb{E}_{\widehat{\boldsymbol{\theta}}} \left[\mathbf{F}_{t}|\mathsf{Y}_T \right] \mathbf{C}'\right)   \\
	\hspace*{1.2in}- \left(\mathbf{W}_t\circ\mathbf{R}\mathbb{E}_{\widehat{\boldsymbol{\theta}}} \left[\mathbf{F}_{t}|\mathsf{Y}_T \right]\mathbf{C}'\right)' \mathbf{H}^{-1} \left(\mathbf{W}_t\circ \mathbf{Y}_t\right)\\
	\hspace*{1.2in}+	\mathbb{E}_{\widehat{\boldsymbol{\theta}}} \left[ \left(\mathbf{W}_t\circ\mathbf{R}\mathbf{F}_{t}\mathbf{C}'\right)' \mathbf{H}^{-1} \left(\mathbf{W}_t\circ\mathbf{R}\mathbf{F}_{t}\mathbf{C}'\right)| \mathsf{Y}_T \right] \\
	\hspace*{1in}= \left(\mathbf{W}_t\circ\mathbf{Y}_t\right)'\mathbf{H}^{-1}\left(\mathbf{W}_t\circ\mathbf{Y}_t\right) - \left(\mathbf{W}_t\circ\mathbf{Y}_t\right)'\mathbf{H}^{-1}\left(\mathbf{W}_t\circ\mathbf{R}\mathbf{F}_{t|T} \mathbf{C}'\right) \\
	\hspace*{1.2in} - \left(\mathbf{W}_t\circ\mathbf{R} \mathbf{F}_{t|T} \mathbf{C}'\right)' \mathbf{H}^{-1} \left(\mathbf{W}_t\circ \mathbf{Y}_t\right) \\
 	\hspace*{1.2in} +	\mathbf{H}^{-1} \star  \left(\mathbb{D}_{\mathbf{W}'_t} \left(\mathbf{R}\otimes\mathbf{C}\right) \left(\mathbb{K}_{k_1k_2} \left( \mathsf{f}_{t|T}\mathsf{f}'_{t|T} + \boldsymbol{\Pi}_{t|T}  \right) \mathbb{K}'_{k_1k_2}\right) \left(\mathbf{R}\otimes\mathbf{C}\right)' \mathbb{D}'_{\mathbf{W}'_t}\right).
	\end{array}
	\]
	Combining these results with the FOC we obtain
	\[
	\begin{array}{lcr}
	\vect{\mathbf{R}} &=& \left(\sum\limits_{t=1}^{T} \sum\limits_{s=1}^{p_1} \sum\limits_{q=1}^{p_1} \left( \left(\mathbf{C}'\mathbb{D}^{[s,q]}_{\mathbf{W}_t} \mathbf{K}^{-1} \mathbf{C}\right) \star \left(\mathsf{f}_{t|T} \mathsf{f}'_{t|T} + \boldsymbol{\Pi}_{t|T} \right) \right) \otimes \left(\mathbb{I}_{p_1} \mathbb{E}^{(s,q)}_{p_1,p_1}\mathbf{H}^{-1}\right) \right)^{-1} \\ && \left(\sum\limits_{t=1}^{T} \vect{\left(\mathbf{W}_t \circ \mathbf{H}^{-1}\mathbf{Y}_t\mathbf{K}^{-1}\right) \mathbf{C} \mathbf{F}'_{t|T}} \right),
	\end{array}	
	\]
	\[
	\begin{array}{lcr}
	\vect{\mathbf{C}} &=& \left(\sum\limits_{t=1}^{T} \sum\limits_{s=1}^{p_2} \sum\limits_{q=1}^{p_2}  \left( \mathbf{R}'\mathbb{D}^{[k,q]}_{\mathbf{W}'_t}\mathbf{H}^{-1} \mathbf{R} \star \left(\mathbb{K}_{k_1k_2} \left( \mathsf{f}_{t|T}\mathsf{f}'_{t|T} + \boldsymbol{\Pi}_{t|T}  \right) \mathbb{K}'_{k_1k_2}\right) \right) \otimes \left(\mathbb{I}_{p_2}\mathbb{E}^{(s,q)}_{p_2,p_2} \mathbf{K}^{-1} \right) \right)^{-1}  \\ 
	&& \left(\sum\limits_{t=1}^{T} \vect{\left(\mathbf{W}_t \circ \mathbf{H}^{-1}\mathbf{Y}_t\mathbf{K}^{-1}\right)' \mathbf{R} \mathbf{F}_{t}} \right),
	\end{array}
	\]
	\[
	\begin{array}{lcl}
	\left[\mathbf{H}\right]_{ii} &=& \frac{1}{Tp_2}\sum\limits_{t=1}^{T}  \left[  \left( \left(\mathbf{W}_t\circ\mathbf{Y}_t\right)\mathbf{K}^{-1}\left(\mathbf{W}_t\circ\mathbf{Y}_t\right)' - \left(\mathbf{W}_t\circ\mathbf{Y}_t\right)\mathbf{K}^{-1}\left(\mathbf{W}_t\circ\mathbf{R}\mathbf{F}_{t|T} \mathbf{C}'\right)' \right. \right. \\
	&& \hspace*{1in} - \left(\mathbf{W}_t\circ\mathbf{R} \mathbf{F}_{t|T} \mathbf{C}'\right) \mathbf{K}^{-1} \left(\mathbf{W}_t\circ \mathbf{Y}_t\right)' \\ 
	&&  \hspace*{1in}\left. +	\mathbf{K}^{-1}_t \star \left(\mathbb{D}_{\mathbf{W}_t} \left(\mathbf{C}\otimes\mathbf{R}\right) \left(\mathsf{f}_{t|T} \mathsf{f}'_{t|T} + \boldsymbol{\Pi}_{t|T}\right) \left(\mathbf{C}\otimes\mathbf{R}\right)' \mathbb{D}'_{\mathbf{W}_t}\right) \right)' \\
	&&   \hspace*{1in}\left.  + \left(\widehat{\textbf{H}} \mathbf{1}_{p_1,p_2} \widehat{\textbf{K}} \right) \mathbf{K}^{-1}  \left(\mathbf{1}_{p_1,p_2}-\mathbf{W}_t\right)' \right]_{ii},
	\end{array}
	\]
	\[
	\begin{array}{lcl}
	\left[\mathbf{K}\right]_{ii} &=& \frac{1}{Tp_1}\sum\limits_{t=1}^{T}  \left[  \left(\left(\mathbf{W}_t\circ\mathbf{Y}_t\right)'\mathbf{H}^{-1}\left(\mathbf{W}_t\circ\mathbf{Y}_t\right) - \left(\mathbf{W}_t\circ\mathbf{Y}_t\right)'\mathbf{H}^{-1}\left(\mathbf{W}_t\circ\mathbf{R}\mathbf{F}_{t|T} \mathbf{C}'\right) \right. \right. \\
	 &&  \hspace*{1in}- \left(\mathbf{W}_t\circ\mathbf{R} \mathbf{F}_{t|T} \mathbf{C}'\right)' \mathbf{H}^{-1} \left(\mathbf{W}_t\circ \mathbf{Y}_t\right) \\
	 &&  \left.\hspace*{1in} + \mathbf{H}^{-1} \star  \left(\mathbb{D}_{\mathbf{W}'_t} \left(\mathbf{R}\otimes\mathbf{C}\right) \left(\mathbb{K}_{k_1k_2} \left( \mathsf{f}_{t|T}\mathsf{f}'_{t|T} + \boldsymbol{\Pi}_{t|T}  \right) \mathbb{K}'_{k_1k_2}\right) \left(\mathbf{R}\otimes\mathbf{C}\right)' \mathbb{D}'_{\mathbf{W}'_t}\right) \right)'  \\
	 && \left. \hspace*{1in} + \left(\mathbf{1}_{p_1,p_2}-\mathbf{W}_t\right)'\mathbf{H}^{-1} \left(\widehat{\textbf{H}} \mathbf{1}_{p_1,p_2} \widehat{\textbf{K}} \right)  \right]_{ii}.
	\end{array}
	\]
	For any iteration $n\ge 0$, given an estimator of the parameters $\widehat{\boldsymbol{\theta}}^{(n)}$, the explicit solutions for $\mathbf R$ and $\mathbf C$ obtained from the previous FOCs are given in Section \ref{subsec:banbura}. While the solutions for the diagonal elements of $\mathbf H$ and $\mathbf K$ are the following:
	\[
	\setlength{\thinmuskip}{0mu}
	\setlength{\medmuskip}{0mu}
	\setlength{\thickmuskip}{0mu}
	\begin{split}
	[\widehat{\mathbf H}^{(n+1)}]_{ii} &= \frac{1}{Tp_2}\sum\limits_{t=1}^{T} \left[  \left( \left(\mathbf{W}_t\circ\mathbf{Y}_t\right)\widehat{\mathbf{K}}^{(n)-1}\left(\mathbf{W}_t\circ\mathbf{Y}_t\right)'  \right. \right. \\[0.1in]
	& \hspace*{0.75in} - \left(\mathbf{W}_t\circ\mathbf{Y}_t\right)\widehat{\mathbf{K}}^{(n)-1}\left(\mathbf{W}_t\circ\left(\widehat{\mathbf{R}}^{(n+1)}\mathbf{F}^{(n)}_{t|T} \widehat{\mathbf{C}}^{(k+1)\prime}\right)\right)'  \\[0.1in]
	& \hspace*{0.75in} - \left(\mathbf{W}_t\circ\left(\widehat{\mathbf{R}}^{(n+1)} \mathbf{F}^{(n)}_{t|T} \widehat{\mathbf{C}}^{(k+1)\prime}\right)\right) \widehat{\mathbf{K}}^{(n)-1} \left(\mathbf{W}_t\circ \mathbf{Y}_t\right)' \\[0.1in]
	& \hspace*{0.75in} \left. +	\widehat{\mathbf{K}}^{(n)-1} \star \left(\mathbb{D}_{\mathbf{W}_t} \left(\widehat{\mathbf{C}}^{(n+1)}\otimes\widehat{\mathbf{R}}^{(n+1)}\right) \left(\mathsf{f}^{(n)}_{t|T} \mathsf{f}^{(n)\prime}_{t|T} + \mathbf{\Pi}^{(n)}_{t|T}\right) \left(\widehat{\mathbf{C}}^{(n+1)}\otimes\widehat{\mathbf{R}}^{(n+1)}\right)' \mathbb{D}'_{\mathbf{W}_t}\right) \right)' \\[0.1in]
	& \hspace*{0.75in} \left. + \left(\widehat{\textbf{H}}^{(n)} \mathbf{1}_{p_1,p_2} \widehat{\textbf{K}}^{(n)} \right) \widehat{\mathbf{K}}^{(n)-1} \left(\mathbf{1}_{p_1,p_2}-\mathbf{W}_t\right)' \right]_{ii}, \\[0.25in]
	[\widehat{\mathbf K}^{(n+1)}]_{ii} &= \frac{1}{Tp_1}\sum\limits_{t=1}^{T}  \left[  \left(\left(\mathbf{W}_t\circ\mathbf{Y}_t\right)'\widehat{\mathbf{H}}^{(n+1)-1}\left(\mathbf{W}_t\circ\mathbf{Y}_t\right) \right. \right. \\[0.1in]
	& \hspace*{0.75in} -\left(\mathbf{W}_t\circ\mathbf{Y}_t\right)'\widehat{\mathbf{H}}^{(n+1)-1}\left(\mathbf{W}_t\circ\left(\widehat{\mathbf{R}}^{(n+1)}\mathbf{F}^{(n)}_{t|T} \widehat{\mathbf{C}}^{(k+1)\prime}\right)\right) \\[0.1in]
	& \hspace*{0.75in} - \left(\mathbf{W}_t\circ\left(\widehat{\mathbf{R}}^{(n+1)} \mathbf{F}^{(n)}_{t|T} \widehat{\mathbf{C}}^{(k+1)\prime}\right)\right)' \widehat{\mathbf{H}}^{(n+1)-1} \left(\mathbf{W}_t\circ \mathbf{Y}_t\right) \\[0.1in]
	&  \hspace*{0.4in} \left.  + \widehat{\mathbf{H}}^{(n+1)-1} \star  \left(\mathbb{D}_{\mathbf{W}'_t} \left(\widehat{\mathbf{R}}^{(n+1)}\otimes\widehat{\mathbf{C}}^{(n+1)}\right) \left(\mathbb{K}_{k_1k_2} \left( \mathsf{f}^{(n)}_{t|T}\mathsf{f}^{(n)\prime}_{t|T} + \mathbf{\Pi}^{(n)}_{t|T}  \right) \mathbb{K}'_{k_1k_2}\right) \left(\widehat{\mathbf{R}}^{(n+1)}\otimes\widehat{\mathbf{C}}^{(n+1)}\right)' \mathbb{D}'_{\mathbf{W}'_t}\right) \right)'  \\[0.1in]
	&\hspace*{0.75in} \left. + \left(\mathbf{1}_{p_1,p_2}-\mathbf{W}_t\right)'\widehat{\mathbf{H}}^{(n+1)-1} \left(\widehat{\textbf{H}}^{(n+1)} \mathbf{1}_{p_1,p_2} \widehat{\textbf{K}}^{(n)} \right)  \right]_{ii}.
	\end{split}
	\]

	\section{Asymptotic results}\label{app:proofs}
	As $k_1$, $k_2$ are both fixed constants, without loss of generality, we assume $k_1 = k_2 = 1$ in some parts of the proofs as long as it simplifies the notations.
	
	\subsection{Proof of main results}
		\begin{proof}[Proof of Proposition \ref{prop:EMloadCons}]
		Recall that 
			\[
\begin{array}{rll}
				\widehat{\mathbf{R}}^{(n+1)}&=&\left( \sum\limits_{t=1}^{T} \mathbf{Y}_t\widehat{\mathbf{K}}^{(n)-1}\widehat{\mathbf{C}}^{(n)} \widehat{\mathbf{F}}^{(n)\prime}_{t|T} \right) \left(\left(\widehat{\mathbf{C}}^{(n)\prime}\widehat{\mathbf{K}}^{(n)-1}\widehat{\mathbf{C}}^{(n)}\right)\star \left(\sum\limits_{t=1}^{T} \mathsf{f}^{(n)}_{t|T} \mathsf{f}^{(n)\prime}_{t|T} + \mathbf{\Pi}^{(n)}_{t|T} \right) \right)^{-1} \\[0.1in]
				&=& \mathbf{R}\left( \sum\limits_{t=1}^{T} \mathbf{F}_t\mathbf{C}\widehat{\mathbf{K}}^{(n)-1}\widehat{\mathbf{C}}^{(n)} \widehat{\mathbf{F}}^{(n)\prime}_{t|T} \right) \left(\left(\widehat{\mathbf{C}}^{(n)\prime}\widehat{\mathbf{K}}^{(0)-1}\widehat{\mathbf{C}}^{(n)}\right)\star \left(\sum\limits_{t=1}^{T} \mathsf{f}^{(n)}_{t|T} \mathsf{f}^{(n)\prime}_{t|T} + \mathbf{\Pi}^{(n)}_{t|T} \right) \right)^{-1}  \\[0.1in]
				&& + \left( \sum\limits_{t=1}^{T} \mathbf{E}_t\widehat{\mathbf{K}}^{(n)-1}\widehat{\mathbf{C}}^{(n)} \widehat{\mathbf{F}}^{(n)\prime}_{t|T} \right) \left(\left(\widehat{\mathbf{C}}^{(n)\prime}\widehat{\mathbf{K}}^{(n)-1}\widehat{\mathbf{C}}^{(n)}\right)\star \left(\sum\limits_{t=1}^{T} \mathsf{f}^{(n)}_{t|T} \mathsf{f}^{(n)\prime}_{t|T} + \mathbf{\Pi}^{(n)}_{t|T} \right) \right)^{-1}  \\[0.1in]
			\end{array}
			\]
			we can thus write
			\begin{equation}
				\begin{array}{rll}
					\widehat{\mathbf{R}}^{(n+1)} - \mathbf{R}\widehat{\mathbf{J}}_1 
					&=&  \left\{ \mathbf{R}\widehat{\mathbf{J}}_1\left(\frac{1}{p_2}\mathbf{\widehat{\mathbf{J}}'_2C'} \widehat{\mathbf{K}}^{(n)-1} \widehat{\mathbf{C}}^{(n)}\right) \star \left(\frac{1}{T}\sum\limits_{t=1}^{T}\left(\mathsf{f}^{(n)}_{t|T} - \widehat{\mathbf{J}}^{-1} \mathsf{f}_t\right)  \mathsf{f}^{(n)\prime}_{t|T}\right) \right.\\[0.1in]
					&& \left.- \mathbf{R}\widehat{\mathbf{J}}_1 \left(\frac{1}{p_2}\left(\widehat{\mathbf{C}}^{(n)}-\mathbf{C\widehat{J}_2}\right)'\widehat{\mathbf{K}}^{(n)-1}\widehat{\mathbf{C}}^{(n)}\right) \star \left( \frac{1}{T}\sum\limits_{t=1}^{T}\mathsf{f}^{(n)}_{t|T}\mathsf{f}^{(n)\prime}_{t|T} \right)  \right. \\[0.1in]
					&& \left.- \mathbf{R}\widehat{\mathbf{J}}_1 \left(\frac{1}{p_2}\widehat{\mathbf{C}}^{(n)\prime}\widehat{\mathbf{K}}^{(n)-1}\widehat{\mathbf{C}}^{(n)}\right) \star \left( \frac{1}{T}\sum\limits_{t=1}^{T}\mathbf{ \Pi}^{(n)}_{t|T} \right) \right.  \\[0.1in]
					&& +  \left. \frac{1}{Tp_2}\sum\limits_{t=1}^{T} \mathbf{E}_t \widehat{\mathbf{K}}^{(n)-1} \widehat{\mathbf{C}}^{(n)} \widehat{\mathbf{F}}^{(n)\prime}_{t|T}\right\}  \\[0.1in]
					&& \times \left(  \left(\frac{1}{p_2}\widehat{\mathbf{C}}^{(n)\prime}\widehat{\mathbf{K}}^{(n)-1}\widehat{\mathbf{C}}^{(n)}\right) \star \left( \frac{1}{T} \sum\limits_{t=1}^{T} \mathsf{f}^{(n)}_{t|T}\mathsf{f}^{(n)\prime}_{t|T} + \mathbf{ \Pi}^{(n)}_{t|T} \right)  \right)^{-1}.			
				\end{array}
				\label{eq:rh.1}
			\end{equation}
			Now, let $n=0$, we have
			\begin{equation}
				\begin{array}{l}
					\frac{1}{\sqrt{p_1}}\left\lVert \mathbf{R}\widehat{\mathbf{J}}_1 \left(\frac{\mathbf{\widehat{\mathbf{J}}'_2C'} \widehat{\mathbf{K}}^{(0)-1} \widehat{\mathbf{C}}^{(0)}}{p_2}\right) \star \left(\frac{1}{T}\sum\limits_{t=1}^{T}\left(\mathsf{f}^{(0)}_{t|T} - \widehat{\mathbf{J}}^{-1} \mathsf{f}_t\right)  \mathsf{f}^{(0)\prime}_{t|T}\right)  \right\rVert \\[0.1in] 
					\hspace*{2in} \leq k^2_2 \frac{\left\lVert 	\mathbf{R}\widehat{\mathbf{J}}_1\right\rVert}{\sqrt{p_1}} \left\lVert \frac{\mathbf{\widehat{\mathbf{J}}'_2C'} \widehat{\mathbf{K}}^{(0)-1} \widehat{\mathbf{C}}^{(0)}}{p_2} \right\rVert \left\lVert \frac{1}{T}\sum\limits_{t=1}^{T}\left(\mathsf{f}^{(0)}_{t|T} - \widehat{\mathbf{J}}^{-1} \mathsf{f}_t\right)  \mathsf{f}^{(0)\prime}_{t|T} \right\rVert \\[0.1in]
					\hspace*{2in} \lesssim \frac{\left\lVert 	\mathbf{R}\widehat{\mathbf{J}}_1\right\rVert}{\sqrt{p_1}} \frac{\left\lVert \mathbf{C\widehat{\mathbf{J}}_2}\right\rVert}{\sqrt{p_2}} 	\left\lVert \frac{ \widehat{\mathbf{C}}^{(0)\prime}\widehat{\mathbf{K}}^{(0)-1} - (\mathbf{C}\widehat{\mathbf{J}}_2)'\dg{\mathbf{K}}^{-1}}{\sqrt{p_2}} \right \rVert  \left\lVert \frac{1}{T}\sum\limits_{t=1}^{T}\left(\mathsf{f}^{(0)}_{t|T} - \widehat{\mathbf{J}}^{-1} \mathsf{f}_t\right)  \mathsf{f}^{(0)\prime}_{t|T} \right\rVert \\[0.1in]
					\hspace*{2.5in} + \frac{\left\lVert	\mathbf{R}\widehat{\mathbf{J}}_1\right\rVert}{\sqrt{p_1}} \frac{\left\lVert \mathbf{C\widehat{\mathbf{J}}_2}\right\rVert}{\sqrt{p_2}}  \left\lVert \frac{ \mathbf{C}'\dg{\mathbf{K}}^{-1}}{\sqrt{p_2}} \right \rVert  \left\lVert \frac{1}{T}\sum\limits_{t=1}^{T}\left(\mathsf{f}^{(0)}_{t|T} - \widehat{\mathbf{J}}^{-1} \mathsf{f}_t\right)  \mathsf{f}^{(0)\prime}_{t|T} \right\rVert \\[0.1in]
					\hspace*{2in} = O_p\left(\max\left\{\frac{1}{\sqrt{Tp_1}},\frac{1}{\sqrt{Tp_2}},\frac{1}{p_1p_2}\right\}\right)
				\end{array}
				\label{eq:rh.1.i}
			\end{equation}
			by Lemmas \ref{lemma:starnorm}, \ref{lemma:tildeBound}(iii), \ref{lemma:tildeBound}(v), \ref{lemma:est.loadscale}(iii), and \ref{lemma:crossf}(i),
			\begin{equation}
				\begin{array}{l}
					\frac{1}{\sqrt{p_1}}\left\lVert  \mathbf{R}\widehat{\mathbf{J}}_1 \left(\frac{1}{p_2}\left(\widehat{\mathbf{C}}^{(0)}-\mathbf{C\widehat{J}_2}\right)'\widehat{\mathbf{K}}^{(0)-1}\widehat{\mathbf{C}}^{(0)}\right)  \star \left( \frac{1}{T}\sum\limits_{t=1}^{T}\mathsf{f}^{(0)}_{t|T}\mathsf{f}^{(0)\prime}_{t|T} \right) \right\rVert \\[0.1in]
					 \hspace*{2in} \leq  k^2_2\frac{\left\lVert 	\mathbf{R}\widehat{\mathbf{J}}_1\right\rVert}{\sqrt{p_1}} \left\lVert \frac{\left(\widehat{\mathbf{C}}^{(0)}-\mathbf{C\widehat{J}_2}\right)'\widehat{\mathbf{K}}^{(0)-1}\widehat{\mathbf{C}}^{(0)}}{p_2} \right\rVert  \left\lVert\frac{1}{T}\sum\limits_{t=1}^{T}\mathsf{f}^{(0)}_{t|T}\mathsf{f}^{(0)\prime}_{t|T} \right\rVert \\[0.1in]
					 \hspace*{2in}\lesssim \frac{\left\lVert 	\mathbf{R}\widehat{\mathbf{J}}_1\right\rVert}{\sqrt{p_1}} \frac{\left\lVert\widehat{\mathbf{C}}^{(0)}-\mathbf{C\widehat{J}_2}\right\rVert}{\sqrt{p_2}}\left\lVert\widehat{\mathbf{K}}^{(0)-1}\right\rVert \frac{\left\lVert\widehat{\mathbf{C}}^{(0)}-\mathbf{C}\widehat{\mathbf{J}}_2\right\rVert}{\sqrt{p_2}}  \left\lVert\frac{1}{T}\sum\limits_{t=1}^{T}\mathsf{f}^{(0)}_{t|T}\mathsf{f}^{(0)\prime}_{t|T} \right\rVert \\[0.1in]
					 \hspace*{2.5in}+ \frac{\left\lVert 	\mathbf{R}\widehat{\mathbf{J}}_1\right\rVert}{\sqrt{p_1}} \frac{\left\lVert\widehat{\mathbf{C}}^{(0)}-\mathbf{C\widehat{J}_2}\right\rVert}{\sqrt{p_2}}\left\lVert\widehat{\mathbf{K}}^{(0)-1}\right\rVert \frac{\left\lVert\mathbf{C}\widehat{\mathbf{J}}_2\right\rVert}{\sqrt{p_2}}  \left\lVert\frac{1}{T}\sum\limits_{t=1}^{T}\mathsf{f}^{(0)}_{t|T}\mathsf{f}^{(0)\prime}_{t|T} \right\rVert \\[0.1in]
					\hspace*{2in}= O_p\left(\max\left\{ \frac{1}{\sqrt{Tp_1}}, \frac{1}{p_1p_2}, \frac{1}{Tp_2} \right\}\right)
				\end{array}
				\label{eq:rh.1.ii}
			\end{equation}
			by Lemmas \ref{lemma:starnorm}, \ref{lemma:tildeBound}(iii), \ref{lemma:boundRC}(i), \ref{lemma:est.loadscale}(i), and since
			\[
			\begin{array}{rll}
				\left\lVert \frac{1}{T}\sum\limits^T_{t=1} \mathsf{f}^{(0)}_{t|T} \mathsf{f}^{(0)\prime}_{t|T}  \right\rVert &\leq& \left\lVert \frac{1}{T}\sum\limits^T_{t=1} \left(\mathsf{f}^{(0)}_{t|T} - \widehat{\mathbf{J}}^{-1}\mathsf{f}_t\right) \left(\mathsf{f}^{(0)\prime}_{t|T} - \widehat{\mathbf{J}}^{-1}\mathsf{f}_t\right) \right\rVert + 2 \left\lVert \frac{1}{T}\sum\limits^T_{t=1} \left(\mathsf{f}^{(0)}_{t|T} - \widehat{\mathbf{J}}^{-1}\mathsf{f}_t\right) \widehat{\mathbf{J}}^{-1}\mathsf{f}_t \right\rVert + \left\lVert \frac{1}{T}\sum\limits^T_{t=1} \widehat{\mathbf{J}}^{-1}\mathsf{f}_t \mathsf{f}'_t \widehat{\mathbf{J}}^{-\prime} \right\rVert  \\[0.1in]
				&\lesssim& \left\lVert \frac{1}{T}\sum\limits^T_{t=1} \left(\mathsf{f}^{(0)}_{t|T} - \widehat{\mathbf{J}}^{-1}\mathsf{f}_t\right) \widehat{\mathbf{J}}^{-1}\mathsf{f}_t \right\rVert + \left\lVert \frac{1}{T}\sum\limits^T_{t=1} \mathsf{f}_t \mathsf{f}'_t \right\rVert  \\
				&=&O_p(1)
			\end{array}
			\]
			by Assumption \ref{subass:factor} and Lemma \ref{lemma:crossf}(i),
			\begin{equation}
				\begin{array}{rll}
					\frac{1}{\sqrt{p_1}}\left\lVert\mathbf{R}\widehat{\mathbf{J}}_1 \left(\frac{\widehat{\mathbf{C}}^{(0)\prime}\widehat{\mathbf{K}}^{(0)-1}\widehat{\mathbf{C}}^{(0)}}{p_2}\right) \star \left( \frac{1}{T}\sum\limits_{t=1}^{T}\mathbf{ \Pi}^{(0)}_{t|T} \right)\right\rVert &\leq& k^2_2 \frac{\left\lVert 	\mathbf{R}\widehat{\mathbf{J}}_1\right\rVert}{\sqrt{p_1}} \left\lVert \frac{\widehat{\mathbf{C}}^{(0)\prime}\widehat{\mathbf{K}}^{(0)-1}\widehat{\mathbf{C}}^{(0)}}{p_2} \right\rVert  \max\limits_{1\leq t\leq T} \left\lVert \mathbf{ \Pi}^{(0)}_{t|T}\right\rVert  \\[0.1in]
					&\lesssim& \frac{\left\lVert 	\mathbf{R}\widehat{\mathbf{J}}_1\right\rVert}{\sqrt{p_1}} \frac{\left\lVert \widehat{\mathbf{C}}^{(0)} - \mathbf{C}\widehat{\mathbf{J}}_2  \right\rVert^2}{p_2} \left\lVert \widehat{\mathbf{K}}^{(0)-1} \right\rVert  \max\limits_{1\leq t\leq T} \left\lVert \mathbf{ \Pi}^{(0)}_{t|T}\right\rVert \\[0.1in]
					&& + \frac{\left\lVert 	\mathbf{R}\widehat{\mathbf{J}}_1\right\rVert}{\sqrt{p_1}} \frac{\left\lVert \mathbf{C}\widehat{\mathbf{J}}_2  \right\rVert^2}{p_2} \left\lVert \widehat{\mathbf{K}}^{(0)-1} \right\rVert  \max\limits_{1\leq t\leq T} \left\lVert\mathbf{ \Pi}^{(0)}_{t|T}\right\rVert \\[0.1in]
					&=& O_p\left(\frac{1}{p_1p_2}\right)
				\end{array}
				\label{eq:rh.1.iii}
			\end{equation}
			by Lemmas \ref{lemma:starnorm}, \ref{lemma:tildeBound}(iii),  \ref{lemma:boundRC}(ii), \ref{lemma:est.loadscale}(i), and Lemma D.12 in \cite{barigozzi2019quasi}. Moreover,
			\begin{equation}
				\begin{array}{rll}
					\frac{1}{\sqrt{p_1}}\left\lVert \frac{1}{Tp_2}\sum\limits_{t=1}^{T} \mathbf{E}_t \widehat{\mathbf{K}}^{(0)-1} \widehat{\mathbf{C}}^{(0)} \widehat{\mathbf{F}}^{(0)\prime}_{t|T}  \right\rVert &\lesssim& \frac{1}{\sqrt{p_1}}\left\lVert \frac{1}{Tp_2}\sum\limits_{t=1}^{T} \mathbf{E}_t \dg{\mathbf{K}}^{-1} \mathbf{C} \left(\widehat{\mathbf{F}}^{(0)}_{t|T} - \widehat{\mathbf{J}}^{-1}_1\mathbf{F}_t\widehat{\mathbf{J}}^{-\prime}_2 \right)' \right\rVert \\[0.1in]
					&& \frac{1}{\sqrt{p_1}}\left\lVert \frac{1}{Tp_2}\sum\limits_{t=1}^{T} \mathbf{E}_t \left( \widehat{\mathbf{K}}^{(0)-1}\widehat{\mathbf{C}}^{(0)} - \dg{\mathbf{K}}^{-1}\mathbf{C}\widehat{\mathbf{J}}_2 \right) \mathbf{F}'_t  \right\rVert \\[0.1in]
					&& \frac{1}{\sqrt{p_1}}\left\lVert \frac{1}{Tp_2}\sum\limits_{t=1}^{T} \mathbf{E}_t \dg{\mathbf{K}}^{-1} \mathbf{C} \mathbf{F}'_t  \right\rVert \\[0.1in]
					&=& O_p\left(\max\left\{ \frac{1}{\sqrt{Tp_1}}, \frac{1}{\sqrt{Tp_2}}, \frac{1}{p_1p_2}  \right\}\right)
				\end{array}
				\label{eq:rh.1.iv}
			\end{equation}
			since
			\[
			\begin{array}{rll}
				\frac{1}{T\sqrt{p_1}p_2}\left\lVert \sum\limits_{t=1}^{T} \mathbf{E}_t \dg{\mathbf{K}}^{-1} \mathbf{C} \left(\widehat{\mathbf{F}}^{(0)}_{t|T} - \widehat{\mathbf{J}}^{-1}_1\mathbf{F}_t\widehat{\mathbf{J}}^{-\prime}_2 \right)' \right\rVert &=& \frac{1}{T\sqrt{p_1}p_2}\left\lVert \left(\dg{\mathbf{K}}^{-1} \mathbf{C}\right)\star\left(\sum\limits_{t=1}^{T} \mathsf{e}_t \left(\mathsf{f}^{(0)}_{t|T} - \widehat{\mathbf{J}}^{-1}\mathsf{f}_t \right)' \right) \right\rVert \\[0.1in]
				&=& \frac{1}{T\sqrt{p_1}p_2}\left\lVert \sum\limits_{i=1}^{p_2} \sum\limits_{j=1}^{k_2} k^{-1}_{ii} c_{ij}\left[\sum\limits_{t=1}^{T} \mathsf{e}_t \left(\mathsf{f}^{(0)}_{t|T} - \widehat{\mathbf{J}}^{-1}\mathsf{f}_t \right)' \right]^{(p_1,k_1)}_{ij} \right\rVert  \\[0.1in]
				&\leq& \frac{\bar{c}C_K}{T\sqrt{p_1}p_2}\left\lVert \sum\limits_{i=1}^{p_2} \sum\limits_{j=1}^{k_2} \left[\sum\limits_{t=1}^{T} \mathsf{e}_t \left(\mathsf{f}^{(0)}_{t|T} - \widehat{\mathbf{J}}^{-1}\mathsf{f}_t \right)' \right]^{(p_1,k_1)}_{ij} \right\rVert  \\[0.1in]
				&\leq& \bar{c}C_K \left\lVert \frac{1}{T\sqrt{p_1}p_2} \sum\limits_{i=1}^{p_2} \sum\limits_{t=1}^{T} \mathbf{e}_{t\cdot i} \left(\mathsf{f}^{(0)}_{t|T} - \widehat{\mathbf{J}}^{-1}\mathsf{f}_t \right)'  \right\rVert  \\[0.1in]
				&=& O_p\left(\max\left\{ \frac{1}{\sqrt{Tp_1}}, \frac{1}{\sqrt{Tp_2}}, \frac{1}{p_1p_2}  \right\}\right)
			\end{array}
			\]
			by Lemma \ref{lemma:crossf}(iv), and
			\[
			\begin{array}{rll}
				\frac{1}{\sqrt{p_1}}\left\lVert \frac{1}{Tp_2}\sum\limits_{t=1}^{T} \mathbf{E}_t \left( \widehat{\mathbf{K}}^{(0)-1}\widehat{\mathbf{C}} - \dg{\mathbf{K}}^{-1}\mathbf{C} \right) \mathbf{F}'_t  \right\rVert &\leq& \frac{1}{\sqrt{Tp_1p_2}} \left\lVert \frac{1}{\sqrt{T}} \sum\limits_{t=1}^{T} \mathbf{E}_t \otimes \mathbf{F}'_t  \right\rVert \left\lVert \frac{ \widehat{\mathbf{K}}^{(0)-1}\widehat{\mathbf{C}} - \dg{\mathbf{K}}^{-1}\mathbf{C}}{\sqrt{p_2}} \right\rVert \\[0.1in]
				&=& \frac{1}{\sqrt{T}} O_p\left( \max\left\{ \frac{1}{\sqrt{Tp_1}}, \frac{1}{p_1p_2}, \frac{1}{Tp_2} \right\}  \right)
			\end{array}
			\]
			by Assumption \ref{subass:fecorr} and Lemma \ref{lemma:est.loadscale}(iii),
			\[
			\frac{1}{\sqrt{p_1}}\left\lVert \frac{1}{Tp_2}\sum\limits_{t=1}^{T} \mathbf{E}_t \dg{\mathbf{K}}^{-1} \mathbf{C} \mathbf{F}'_t  \right\rVert = O_p\left(\frac{1}{\sqrt{Tp_2}}\right)
			\]
			by Lemma A.1 in \cite{yu2021projected}.	Finally, we have that
			\begin{equation}
				\begin{array}{l}
					\left\lVert\left(\frac{\widehat{\mathbf{C}}^{(0)\prime}\widehat{\mathbf{K}}^{(0)-1}\widehat{\mathbf{C}}^{(0)}}{p_2}\right) \star \left( \frac{1}{T} \sum\limits_{t=1}^{T} \mathsf{f}^{(0)}_{t|T}\mathsf{f}^{(0)\prime}_{t|T} + \mathbf{ \Pi}^{(0)}_{t|T} \right) - \left(\frac{\mathbf{C}'\dg{\mathbf{K}}^{-1}\mathbf{C}}{p_2}\right)\star \left(\frac{1}{T}\sum^T_{t=1}\widehat{\mathbf{J}}^{-1}\mathsf{f}_t \mathsf{f}'_t \widehat{\mathbf{J}}^{\prime} \right)\right\rVert \\[0.1in]
					\hspace*{2in}\leq \left\lVert \left(\frac{\widehat{\mathbf{C}}^{(0)\prime}\widehat{\mathbf{K}}^{(0)-1}\widehat{\mathbf{C}}^{(0)}-\mathbf{C}'\dg{\mathbf{K}}^{-1}\mathbf{C}}{p_2}\right) \star \left( \frac{1}{T} \sum\limits_{t=1}^{T} \mathsf{f}^{(0)}_{t|T}\mathsf{f}^{(0)\prime}_{t|T} + \mathbf{ \Pi}^{(0)}_{t|T} - \widehat{\mathbf{J}}^{-1}\mathsf{f}_t \mathsf{f}'_t \widehat{\mathbf{J}}^{\prime}  \right) \right\rVert \\[0.1in]
					\hspace*{2.5in} + \left\lVert \left(\frac{\mathbf{C}'\dg{\mathbf{K}}^{-1}\mathbf{C}}{p_2}\right) \star \left( \frac{1}{T} \sum\limits_{t=1}^{T} \mathsf{f}^{(0)}_{t|T}\mathsf{f}^{(0)\prime}_{t|T} + \mathbf{ \Pi}^{(0)}_{t|T} - \widehat{\mathbf{J}}^{-1}\mathsf{f}_t \mathsf{f}'_t \widehat{\mathbf{J}}^{\prime}  \right) \right\rVert \\[0.1in]
					\hspace*{2.5in} + \left\lVert \left(\frac{\widehat{\mathbf{C}}^{(0)\prime}\widehat{\mathbf{K}}^{(0)-1}\widehat{\mathbf{C}}^{(0)}-\mathbf{C}'\dg{\mathbf{K}}^{-1}\mathbf{C}}{p_2}\right) \star \left( \frac{1}{T} \sum\limits_{t=1}^{T} \widehat{\mathbf{J}}^{-1}\mathsf{f}_t \mathsf{f}'_t \widehat{\mathbf{J}}^{\prime}  \right) \right\rVert \\[0.1in]
					\hspace*{2in} = O_p\left(\max\left\{\frac{1}{\sqrt{Tp_1}}, \frac{1}{\sqrt{Tp_2}},\frac{1}{p_1p_2}\right\}\right) 
				\end{array}
				\label{eq:rh.1.v}
			\end{equation}
			since the first addendum is dominated by the other two terms and
			\[
			%
				\label{eq:rh1.vi}
			\end{equation}
			by Lemma \ref{lemma:tildeBound}(iv) and Assumption \ref{subass:factor}.	Combining \eqref{eq:rh.1} with \eqref{eq:rh.1.i}-\eqref{eq:rh.1.iv} and \eqref{eq:rh.1.v}-\eqref{eq:rh1.vi}, we obtain 
			\[
			\frac{1}{\sqrt{p_1}}\left\lVert\widehat{\mathbf{R}}^{(1)}-\mathbf{R}\widehat{\mathbf{J}}_1 \right\rVert = O_p\left(\max\left\{\frac{1}{\sqrt{Tp_1}},\frac{1}{\sqrt{Tp_2}},\frac{1}{p_1p_2}\right\}\right).	
			\]
			Consider now $n>0$, the consistency result for $\widehat{\mathbf{R}}^{(n+1)}$ follows iterating the same steps but using Lemma \ref{lemma:est.loadscale.1} and this proposition in place of Lemmas \ref{lemma:boundRC} and \ref{lemma:est.loadscale}.  The proof for $\widehat{\mathbf{C}}$ follows analogously and it is omitted.
			
			For the row-wise consistency note that we can use the decomposition in \eqref{eq:rh.1} using $\mathbf{r}_{i\cdot}$ in place of $\mathbf{R}$, then \eqref{eq:rh.1.i}-\eqref{eq:rh.1.iii} follows analogously as $\lVert\mathbf{r}_{i\cdot}\rVert=O_p(1)$ by Assumption \ref{subass:loading}. Then
			\begin{equation}
				\begin{array}{rll}
					\left\lVert \frac{1}{Tp_2}\sum\limits_{t=1}^{T} \mathbf{e}'_{ti\cdot} \widehat{\mathbf{K}}^{(0)-1} \widehat{\mathbf{C}}^{(0)} \widehat{\mathbf{F}}^{(0)\prime}_{t|T}  \right\rVert &\lesssim& \left\lVert \frac{1}{Tp_2}\sum\limits_{t=1}^{T} \mathbf{e}'_{ti\cdot} \dg{\mathbf{K}}^{-1} \mathbf{C} \left(\widehat{\mathbf{F}}^{(0)}_{t|T} - \widehat{\mathbf{J}}^{-1}_1\mathbf{F}_t\widehat{\mathbf{J}}^{-\prime}_2 \right)' \right\rVert \\[0.1in]
					&& \left\lVert \frac{1}{Tp_2}\sum\limits_{t=1}^{T} \mathbf{e}'_{ti\cdot} \left( \widehat{\mathbf{K}}^{(0)-1}\widehat{\mathbf{C}}^{(0)} - \dg{\mathbf{K}}^{-1}\mathbf{C}\widehat{\mathbf{J}}_2 \right) \mathbf{F}'_t  \right\rVert \\[0.1in]
					&& \left\lVert \frac{1}{Tp_2}\sum\limits_{t=1}^{T} \mathbf{e}'_{ti\cdot} \dg{\mathbf{K}}^{-1} \mathbf{C} \mathbf{F}'_t  \right\rVert \\[0.1in]
					&=& O_p\left(\max\left\{ \frac{1}{\sqrt{Tp_1}}, \frac{1}{\sqrt{Tp_2}}, \frac{1}{p_1p_2}  \right\}\right)
				\end{array}
			\end{equation}
			since
			\[
			\begin{array}{rll}
				\frac{1}{Tp_2}\left\lVert \sum\limits_{t=1}^{T} \mathbf{e}'_{ti\cdot} \dg{\mathbf{K}}^{-1} \mathbf{C} \left(\widehat{\mathbf{F}}^{(0)}_{t|T} - \widehat{\mathbf{J}}^{-1}_1\mathbf{F}_t\widehat{\mathbf{J}}^{-\prime}_2 \right)' \right\rVert 
				&=& \frac{1}{Tp_2}\left\lVert \sum\limits_{t=1}^{T}  \sum\limits_{j=1}^{p_2} \sum\limits_{q=1}^{k_2}  e_{tij} k^{-1}_{jj} c_{jq}\left[\widehat{\mathbf{F}}^{(0)}_{t|T} - \widehat{\mathbf{J}}^{-1}_1\mathbf{F}_t\widehat{\mathbf{J}}^{-\prime}_2 \right]_{\cdot q} \right\rVert  \\[0.1in]
				&\leq& \bar{c}C_K  \left\lVert \frac{1}{Tp_2}\sum\limits_{j=1}^{p_2}\sum\limits_{t=1}^{T} e_{tij} \left(\mathsf{f}^{(0)}_{t|T} - \widehat{\mathbf{J}}^{-1}\mathsf{f}_t \right)  \right\rVert  \\[0.1in]
				&=& O_p\left(\max\left\{ \frac{1}{\sqrt{Tp_1}}, \frac{1}{\sqrt{Tp_2}}, \frac{1}{p_1p_2}  \right\}\right)
			\end{array}
			\]
			by Lemma \ref{lemma:crossf}(ii), 
			\[
			\begin{array}{rll}
				\left\lVert \frac{1}{Tp_2}\sum\limits_{t=1}^{T}  \mathbf{e}'_{ti\cdot} \left( \widehat{\mathbf{K}}^{(0)-1}\widehat{\mathbf{C}}^{(0)} - \dg{\mathbf{K}}^{-1}\mathbf{C} \right) \mathbf{F}'_t  \right\rVert &\lesssim& \frac{1}{\sqrt{T}} \frac{\left\lVert  \widehat{\mathbf{K}}^{(0)-1}\widehat{\mathbf{C}}^{(0)} - \dg{\mathbf{K}}^{-1}\mathbf{C} \right\rVert }{\sqrt{p_2}}\left\lVert \frac{1}{\sqrt{Tp_2}}\sum\limits_{t=1}^{T} \sum\limits_{j=1}^{p_2} e_{tij} \mathbf{F}'_t  \right\rVert \\[0.1in]
				&=& \frac{1}{\sqrt{T}} O_p\left( \max\left\{ \frac{1}{\sqrt{Tp_1}}, \frac{1}{p_1p_2}, \frac{1}{Tp_2} \right\}  \right)
			\end{array}
			\]
			by Assumption \ref{subass:fecorr} and Lemma \ref{lemma:est.loadscale}(iii),
			\[
			\left\lVert \frac{1}{Tp_2}\sum\limits_{t=1}^{T} \mathbf{e}_{ti\cdot} \dg{\mathbf{K}}^{-1} \mathbf{C} \mathbf{F}'_t  \right\rVert \leq \frac{C_K}{\sqrt{Tp_2}}\left\lVert \frac{1}{\sqrt{T}}\sum\limits_{t=1}^{T}\sum\limits_{j=1}^{p_2}  e_{tij} \frac{\mathbf{c}_{j\cdot}}{\sqrt{p_2}}\mathbf{F}'_t  \right\rVert =  O_p\left(\frac{1}{\sqrt{Tp_2}}\right)
			\]
			follows directly from Assumption \ref{subass:fecorr}. Iterating the same steps using Lemma \ref{lemma:est.loadscale.1} in place of Lemma \ref{lemma:est.loadscale} yields the result for $n>0$. Row-wise consistency for $\widehat{\mathbf{c}}_{j\cdot}$ can be established analogously. 
	\end{proof}
	
	\begin{proof}[Proof of Proposition \ref{prop:EMfactorCons}]
		Since $\vect{\widehat{\mathbf{F}}^{(n)}_t} =  \mathsf{f}^{(n)}_{t|T}$, we have
		\[
		\begin{array}{rll}
			\left\lVert \widehat{\mathbf{F}}^{(n)}_t-\widehat{\mathbf{J}}^{-1}_t\mathbf{F}_t\widehat{\mathbf{J}}^{-\prime}_2 \right\rVert &\leq& \left\lVert \widehat{\mathbf{F}}^{(n)}_t-\widehat{\mathbf{J}}^{-1}_t\mathbf{F}_t\widehat{\mathbf{J}}^{-\prime}_2  \right\rVert_F \\[0.1in]
			&=& 	\left\lVert \mathsf{f}^{(n)}_{t|T}-\widehat{\mathbf{J}}^{-1}\mathsf{f}_t \right\rVert \\[0.1in]
			&\leq& \left\lVert\mathsf{f}^{(n)}_{t|T}-\mathsf{f}^{(n)}_{t|t}\right\rVert + \left\lVert\mathsf{f}^{(n)}_{t|t}-\mathsf{f}^{LS(n)}_t\right\rVert + \left\lVert\mathsf{f}^{LS(n)}_{t}-\widehat{\mathbf{J}}^{-1}\mathsf{f}_t\right\rVert \\[0.1in]
		\end{array}
		\]
		with $\widehat{\mathbf{J}}=\widehat{\mathbf{J}}_2\otimes\widehat{\mathbf{J}}_1$ and
		\[
		\mathsf{f}^{LS(n)}_{t} = \left( \left(\left(\widehat{\mathbf{C}}^{(n)\prime}\widehat{\mathbf{K}}^{(n)-1}\widehat{\mathbf{C}}^{(n)}\right)^{-1}\widehat{\mathbf{C}}^{(n)\prime}\widehat{\mathbf{K}}^{(n)-1}\right)\otimes\left(\left(\widehat{\mathbf{R}}^{(n)\prime}\widehat{\mathbf{H}}^{(n)-1}\widehat{\mathbf{R}}^{(n)}\right)^{-1}\widehat{\mathbf{R}}^{(n)\prime}\widehat{\mathbf{H}}^{(n)-1}\right)\right)\mathsf{y}_t.
		\]
		Consider the case $n=0$. Combining Lemmas D.15 and D.16 in \cite{barigozzi2019quasi} with Lemmas \ref{lemma:boundRC}, \ref{lemma:boundAB},  \ref{lemma:est.loadscale}, we have that
		\[
			\left\lVert\mathsf{f}^{(0)}_{t|T}-\mathsf{f}_t\right\rVert \leq  \left\lVert\mathsf{f}^{LS(0)}_{t}-\widehat{\mathbf{J}}^{-1}\mathsf{f}_t\right\rVert + O_p\left(\frac{1}{p_1p_2}\right)
		\]
		Then,
		\[
		\begin{array}{rll}
			\left\lVert\mathsf{f}^{LS(0)}_{t}-\widehat{\mathbf{J}}^{-1}\mathsf{f}_t\right\rVert &\leq& \left\lVert \left( \left(\left(\widehat{\mathbf{C}}^{(0)\prime}\widehat{\mathbf{K}}^{(0)-1}\widehat{\mathbf{C}}^{(0)}\right)^{-1}\widehat{\mathbf{C}}^{(0)\prime}\widehat{\mathbf{K}}^{(0)-1}\right)\otimes\left(\left(\widehat{\mathbf{R}}^{(0)\prime}\widehat{\mathbf{H}}^{(0)-1}\widehat{\mathbf{R}}^{(0)}\right)^{-1}\widehat{\mathbf{R}}^{(0)\prime}\widehat{\mathbf{H}}^{(0)-1}\right)\right) \right. \\[0.1in]
			&& \hspace*{1.5in}\left. \times \left( \mathbf{C}\widehat{\mathbf{J}}_2 \otimes \mathbf{R}\widehat{\mathbf{J}}_1 - \widehat{\mathbf{C}}^{(0)} \otimes \widehat{\mathbf{R}}^{(0)} \right) \right\rVert \left\lVert \widehat{\mathbf{J}}^{-1}\mathsf{f}_t \right\rVert \\[0.1in]
			&& + \left\lVert \left( \left(\left(\widehat{\mathbf{C}}^{(0)\prime}\widehat{\mathbf{K}}^{(0)-1}\widehat{\mathbf{C}}^{(0)}\right)^{-1}\widehat{\mathbf{C}}^{(0)\prime}\widehat{\mathbf{K}}^{(0)-1}\right)\otimes\left(\left(\widehat{\mathbf{R}}^{(0)\prime}\widehat{\mathbf{H}}^{(0)-1}\widehat{\mathbf{R}}^{(0)}\right)^{-1}\widehat{\mathbf{R}}^{(0)\prime}\widehat{\mathbf{H}}^{(0)-1}\right)\right)\mathsf{e}_t \right\rVert \\[0.1in]
			&\leq&  \left\lVert \left(\left(\left(\mathbf{C}'\mathbf{K}^{-1}\mathbf{C}\right)^{-1}\mathbf{C}'\mathbf{K}^{-1}\right) \otimes \left(\left(\mathbf{R}^{\prime}\mathbf{H}^{-1}\mathbf{R}\right)^{-1}\mathbf{R}^{\prime}\mathbf{H}^{-1} \right)\right) \left( \mathbf{C}\widehat{\mathbf{J}}_2 \otimes \mathbf{R}\widehat{\mathbf{J}}_1 - \widehat{\mathbf{C}}^{(0)} \otimes \widehat{\mathbf{R}}^{(0)} \right) \right\rVert \left\lVert \widehat{\mathbf{J}}^{-1}\mathsf{f}_t \right\rVert \\[0.1in]
			&& + \left\lVert \left( \left(\left(\widehat{\mathbf{C}}^{(0)\prime}\widehat{\mathbf{K}}^{(0)-1}\widehat{\mathbf{C}}^{(0)}\right)^{-1}\widehat{\mathbf{C}}^{(0)\prime}\widehat{\mathbf{K}}^{(0)-1}\right)\otimes\left(\left(\widehat{\mathbf{R}}^{(0)\prime}\widehat{\mathbf{H}}^{(0)-1}\widehat{\mathbf{R}}^{(0)}\right)^{-1}\widehat{\mathbf{R}}^{(0)\prime}\widehat{\mathbf{H}}^{(0)-1}\right)\right)  \right. \\[0.1in]
			&& \left. \hspace*{1in} - \left(\left(\mathbf{C}'\mathbf{K}^{-1}\mathbf{C}\right)^{-1}\mathbf{C}'\mathbf{K}^{-1}\right) \otimes \left(\left(\mathbf{R}^{\prime}\mathbf{H}^{-1}\mathbf{R}\right)^{-1}\mathbf{R}^{\prime}\mathbf{H}^{-1} \right) \right\rVert\\[0.1in]
			&& \hspace*{0.5in}\times \left\lVert \left( \mathbf{C}\widehat{\mathbf{J}}_2 \otimes \mathbf{R}\widehat{\mathbf{J}}_1 - \widehat{\mathbf{C}}^{(0)} \otimes \widehat{\mathbf{R}}^{(0)} \right) \right\rVert \left\lVert \mathsf{f}_t \right\rVert \\[0.1in]
			&& + \left\lVert \left(\left(\widehat{\mathbf{C}}^{(0)\prime}\widehat{\mathbf{K}}^{(0)-1}\widehat{\mathbf{C}}^{(0)}\right)^{-1}\otimes\left(\widehat{\mathbf{R}}^{(0)\prime}\widehat{\mathbf{H}}^{(0)-1}\widehat{\mathbf{R}}^{(0)}\right)^{-1}\right)\left(\left(\mathbf{C}'\dg{\mathbf{K}}^{-1}\right) \otimes \left(\mathbf{R}'\dg{\mathbf{H}}^{-1} \right)\right) \mathsf{e}_t \right\rVert \\[0.1in]
			&& + \left\lVert \left(\left(\widehat{\mathbf{C}}^{(0)\prime}\widehat{\mathbf{K}}^{(0)-1}\widehat{\mathbf{C}}^{(0)}\right)^{-1}\otimes\left(\widehat{\mathbf{R}}^{(0)\prime}\widehat{\mathbf{H}}^{(0)-1}\widehat{\mathbf{R}}^{(0)}\right)^{-1}\right) \right.  \\[0.1in]
			&& \left.  \hspace*{1in} \times \left( \left(\widehat{\mathbf{C}}^{(0)\prime}\widehat{\mathbf{K}}^{(0)-1}\right)\otimes\left(\widehat{\mathbf{R}}^{(0)\prime}\widehat{\mathbf{H}}^{(0)-1}\right) -\left(\mathbf{C}'\dg{\mathbf{K}}^{-1}\right) \otimes \left(\mathbf{R}'\dg{\mathbf{H}}^{-1} \right)\right)  \mathsf{e}_t \right\rVert \\[0.1in]
			&=& I + II + III + IV
		\end{array}
		\]
		Now,
		\[
		\begin{array}{rll}
			I &\leq& \left\lVert\left(\frac{\mathbf{C}'\dg{\mathbf{K}}^{-1}\mathbf{C}}{p_2}\right)^{-1} \right\rVert \left\lVert \frac{\mathbf{C}'\dg{\mathbf{K}}^{-1}}{\sqrt{p_2}} \right\rVert \left\lVert\left(\frac{\mathbf{R}'\dg{\mathbf{H}}^{-1}\mathbf{R}}{p_1}\right)^{-1} \right\rVert \left\lVert \frac{\mathbf{R}'\dg{\mathbf{H}}^{-1}}{\sqrt{p_2}} \right\rVert \left\lVert\frac{ \mathbf{C}\widehat{\mathbf{J}}_2 \otimes \mathbf{R}\widehat{\mathbf{J}}_1 - \widehat{\mathbf{C}}^{(0)} \otimes \widehat{\mathbf{R}}^{(0)} }{\sqrt{p_1p_2}} \right\rVert \left\lVert \widehat{\mathbf{J}}^{-1}\mathsf{f}_t \right\rVert \\[0.1in]
			&=& O_p\left(\max\left\{ \frac{1}{\sqrt{Tp_1}}, \frac{1}{\sqrt{Tp_2}}, \frac{1}{p_1p_2} \right\}\right)
		\end{array}
		\]
		by Lemmas \ref{lemma:tildeBound}(iv), \ref{lemma:tildeBound}(v), \ref{lemma:boundRC}(iii), and since $ \left\lVert \mathsf{f}_t \right\rVert=O_p(1)$ by Assumption \ref{subass:factor},
		\[
		\begin{array}{rll}
			II &\leq&  \sqrt{p_2}\left\lVert \left(\widehat{\mathbf{C}}^{(0)\prime}\widehat{\mathbf{K}}^{(0)-1}\widehat{\mathbf{C}}^{(0)}\right)^{-1}\widehat{\mathbf{C}}^{(0)\prime}\widehat{\mathbf{K}}^{(0)-1}-\left(\mathbf{C}'\mathbf{K}^{-1}\mathbf{C}\right)^{-1}\mathbf{C}'\mathbf{K}^{-1} \right\rVert  \left\lVert \left(\frac{\mathbf{R}^{\prime}\mathbf{H}^{-1}\mathbf{R}}{p_1}\right)^{-1}\right\rVert \frac{\left\lVert\mathbf{R}^{\prime}\mathbf{H}^{-1}\right\rVert}{\sqrt{p_1}} \\[0.1in]
			&& \hspace*{0.5in}\times \frac{1}{\sqrt{p_1p_2}} \left\lVert \left( \mathbf{C}\widehat{\mathbf{J}}_2 \otimes \mathbf{R}\widehat{\mathbf{J}}_1 - \widehat{\mathbf{C}}^{(0)} \otimes \widehat{\mathbf{R}}^{(0)} \right) \right\rVert \left\lVert \widehat{\mathbf{J}}^{-1}\mathsf{f}_t \right\rVert \\[0.1in]
			&&  \left\lVert \left(\frac{\mathbf{C}^{\prime}\mathbf{K}^{-1}\mathbf{C}}{p_2}\right)^{-1}\right\rVert \frac{\left\lVert\mathbf{C}^{\prime}\mathbf{K}^{-1}  \right\rVert}{\sqrt{p_2}} \sqrt{p_1} \left\lVert \left(\widehat{\mathbf{R}}^{(0)\prime}\widehat{\mathbf{H}}^{(0)-1}\widehat{\mathbf{R}}^{(0)}\right)^{-1}\widehat{\mathbf{R}}^{(0)\prime}\widehat{\mathbf{H}}^{(0)-1}-\left(\mathbf{R}'\mathbf{H}^{-1}\mathbf{R}\right)^{-1}\mathbf{R}'\mathbf{H}^{-1} \right\rVert   \\[0.1in]
			&& \hspace*{0.5in}\times \frac{1}{\sqrt{p_1p_2}} \left\lVert \left( \mathbf{C}\widehat{\mathbf{J}}_2 \otimes \mathbf{R}\widehat{\mathbf{J}}_1 - \widehat{\mathbf{C}}^{(0)} \otimes \widehat{\mathbf{R}}^{(0)} \right) \right\rVert \left\lVert \widehat{\mathbf{J}}^{-1}\mathsf{f}_t \right\rVert \\[0.1in]
			&=& o_p \left(\max\left\{ \frac{1}{\sqrt{Tp_1}}, \frac{1}{\sqrt{Tp_2}}, \frac{1}{p_1p_2} \right\}\right)
		\end{array}
		\]
		by Lemmas \ref{lemma:tildeBound}(iv), \ref{lemma:tildeBound}(v), \ref{lemma:boundRC}(iii), \ref{lemma:est.loadscale}(ix), \ref{lemma:est.loadscale}(x), and Assumption \ref{subass:factor},
		\[
		\begin{array}{rll}
			III &\leq& \left\lVert\left(\frac{\widehat{\mathbf{C}}^{(0)\prime}\widehat{\mathbf{K}}^{(0)-1}\widehat{\mathbf{C}}^{(0)}}{p_2}\right)^{-1}\right\rVert \left\lVert\left(\frac{\widehat{\mathbf{R}}^{(0)\prime}\widehat{\mathbf{H}}^{(0)-1}\widehat{\mathbf{R}}^{(0)}}{p_1}\right)^{-1}\right\rVert \frac{1}{p_1p_2} \left\lVert\mathbf{R}'\dg{\mathbf{H}}^{-1} \mathbf{E}_t\dg{\mathbf{K}}^{-1}\mathbf{C} \right\rVert_F \\[0.1in]
			&=&O_p\left(\frac{1}{\sqrt{p_1p_2}}\right)
		\end{array}
		\]
		by Lemma \ref{lemma:est.loadscale}(v), \ref{lemma:est.loadscale}(vi), and since
		\[
		\begin{array}{rll}
			\mathbb{E}\left[\left\lVert\mathbf{R}'\dg{\mathbf{H}}^{-1} \mathbf{E}_t\dg{\mathbf{K}}^{-1}\mathbf{C} \right\rVert^2_F \right] &=& \mathbb{E}\left[ \sum\limits_{i=1}^{k_1}\sum\limits_{j=1}^{k_2} \left( \mathbf{r}'_{\cdot i}\dg{\mathbf{H}}^{-1} \mathbf{E}_t\dg{\mathbf{K}}^{-1}\mathbf{c}_{\cdot j} \right)^2 \right] \\[0.1in]
			&=& \sum\limits_{i=1}^{k_1}\sum\limits_{j=1}^{k_2} \mathbb{E}\left[ \left( \sum\limits_{s=1}^{p_1}\sum\limits_{q=1}^{p_2} r_{s i}h^{-1}_{ss} e_{tsq} k^{-1}_{qq}c_{qj} \right)^2 \right] \\[0.1in]
			&\leq& k \bar{r}^2\bar{c}^2 C^2_K C^2_H  \mathbb{E}\left[ \left( \sum\limits_{s=1}^{p_1}\sum\limits_{q=1}^{p_2} e_{tsq} \right)^2 \right] \\[0.1in]
			&\lesssim&  \sum\limits_{s_1,s_2}^{p_1}\sum\limits_{q_1,q_2}^{p_2} \left\lvert \mathbb{E}\left[ e_{ts_1q_1} e_{ts_2q_2} \right] \right\rvert \\[0.1in]
			&=& O_p\left(p_1p_2\right)
		\end{array}
		\]
		by Assumptions \ref{subass:loading}, \ref{subass:idshocks}, and \ref{subass:ecorr},
		\[
		\begin{array}{rll}
			IV &\leq& \left\lVert\left(\frac{\widehat{\mathbf{C}}^{(0)\prime}\widehat{\mathbf{K}}^{(0)-1}\widehat{\mathbf{C}}^{(0)}}{p_2}\right)^{-1}\right\rVert \left\lVert\left(\frac{\widehat{\mathbf{R}}^{(0)\prime}\widehat{\mathbf{H}}^{(0)-1}\widehat{\mathbf{R}}^{(0)}}{p_1}\right)^{-1}\right\rVert \left\lVert\frac{\widehat{\mathbf{C}}^{(0)\prime}\widehat{\mathbf{K}}^{(0)-1}-\mathbf{C}'\dg{\mathbf{K}}^{-1}}{\sqrt{p_2}}\right\rVert \left\lVert \frac{\mathbf{R}'\dg{\mathbf{H}}^{-1}}{\sqrt{p_1}} \right\rVert  \left\lVert\frac{\mathsf{e}_t}{\sqrt{p_1p_2}}\right\rVert \\[0.1in]
			&& +\left\lVert\left(\frac{\widehat{\mathbf{C}}^{(0)\prime}\widehat{\mathbf{K}}^{(0)-1}\widehat{\mathbf{C}}^{(0)}}{p_2}\right)^{-1}\right\rVert \left\lVert\left(\frac{\widehat{\mathbf{R}}^{(0)\prime}\widehat{\mathbf{H}}^{(0)-1}\widehat{\mathbf{R}}^{(0)}}{p_1}\right)^{-1}\right\rVert \left\lVert\frac{\widehat{\mathbf{R}}^{(0)}\widehat{\mathbf{H}}^{(0)-1}-\mathbf{R}'\dg{\mathbf{H}}^{-1}}{\sqrt{p_1}}\right\rVert \left\lVert \frac{\mathbf{C}'\dg{\mathbf{K}}^{-1}}{\sqrt{p_2}} \right\rVert \\[0.1in]
			&=& O_p\left(\max\left\{ \frac{1}{\sqrt{Tp_1}}, \frac{1}{\sqrt{Tp_2}}, \frac{1}{p_1p_2} \right\}\right)
		\end{array}		
		\]
		by Lemma \ref{lemma:tildeBound}(v), \ref{lemma:est.loadscale}(iii)-\ref{lemma:est.loadscale}(vi) and since
		\[
		\mathbb{E}\left[\left\lVert\mathsf{e}_t\right\rVert^2\right]=\mathbb{E}\left[\left\lVert\mathbf{E}_t\right\rVert^2_F\right]= \sum_{i=1}^{p_1}\sum_{j=1}^{p_2} \mathbb{E}\left[|e_{tij}|^2 \right]\leq p_1p_2 C_kC_h=O(p_1p_2)
		\]
		by Assumption \ref{subass:idshocks}. Iterating the same steps, using Proposition \ref{prop:EMloadCons}, Proposition \ (a.4)-(a.5) in \cite{barigozzi2019quasi}, and Lemma \ref{lemma:est.loadscale.1} in place of Lemmas \ref{lemma:boundRC}, \ref{lemma:boundAB},  \ref{lemma:est.loadscale}, we can obtain the result for $n>0$.
	\end{proof}

	\subsection{Auxiliary lemmata}
	
	\subsubsection{Preliminary results}
	
	\begin{lemma}
		Let $\mathbf{A}$ and $\mathbf{B}$ be $m\times n$ and $mp\times nq$ matrices, respectively. We have that
		\[
		\left\lVert \mathbf{A}\star \mathbf{B}  \right\rVert \leq mn\left\lVert\mathbf{A}\right\rVert_{\max} \left\lVert\mathbf{B} \right\rVert
		\]
		\label{lemma:starnorm}
	\end{lemma}
	\begin{proof}
		\[
		\begin{array}{rll}
			\lVert \mathbf{A}\star \mathbf{B}  \lVert &=& \left\lVert \sum\limits_{i=1}^m \sum\limits_{j=1}^{n} a_{ij} \mathbf{B}^{(p,q)}_{ij}  \right\rVert \\
			&\leq & \left\lVert\mathbf{A}\right\rVert_{\max} \sum\limits_{i=1}^m \sum\limits_{j=1}^{n}  \left\lVert\mathbf{B}^{(p,q)}_{ij}  \right\rVert \\
			&\leq& m n \left\lVert\mathbf{A}\right\rVert_{\max}\max\limits_{i,j}  \left\lVert\mathbf{B}^{(p,q)}_{ij}  \right\rVert \\
			&\leq & m n \left\lVert\mathbf{A}\right\rVert_{\max} \left\lVert \mathbf{B} \right\rVert \\
		\end{array}
		\]
	\end{proof}

	\begin{lemma}
		For any $k_1\times k_1$ and $k_2\times k_2$ orthogonal matrices $\mathbf{J}_1$ and $\mathbf{J}_2$, the DMFM in \eqref{eq:dmfm1}-\eqref{eq:dmfm2} is equivalent to
		\begin{align}
			\label{eq:dmfmt1}
			\mathbf{Y}_t &= \widetilde{\mathbf{R}} \widetilde{\mathbf{F}}_t \widetilde{\mathbf{C}}' + \mathbf{E}_t  \\
			\label{eq:dmfmt2}
			\widetilde{\mathbf{F}}_t &= \widetilde{\mathbf{A}} \widetilde{\mathbf{F}}_{t-1}\widetilde{\mathbf{B}}' + \widetilde{\mathbf{U}}_t
		\end{align}
		and its vectorized form is
		\begin{align}
			\label{eq:vdmfmt1}
			\mathsf{y}_t &= \left(\widetilde{\mathbf{C}} \otimes \widetilde{\mathbf{R}}\right) \widetilde{\mathsf{f}}_t + \mathsf{e}_t \\
			\label{eq:vdmfmt2}
			\widetilde{\mathsf{f}}_t &= \left(\widetilde{\mathbf{B}} \otimes \widetilde{\mathbf{A}}\right) \widetilde{\mathsf{f}}_{t-1} + \widetilde{ \mathsf{u}}_t 
		\end{align}
		with $\mathbb{E}\left[\widetilde{\mathbf{U}}_t\widetilde{\mathbf{U}}'_t\right]=\widetilde{\mathbf{P}}\tr{\widetilde{\mathbf{Q}}}$ and $\mathbb{E}\left[\widetilde{\mathbf{U}}'_t\widetilde{\mathbf{U}}_t\right]=\widetilde{\mathbf{Q}}\tr{\widetilde{\mathbf{P}}}$ such that $\widetilde{\mathbf{P}}=\widetilde{\boldsymbol{\Gamma}}^P\widetilde{\boldsymbol{\Lambda}}^P\widetilde{\boldsymbol{\Gamma}}^{P\prime}$ and $\widetilde{\mathbf{Q}}=\widetilde{\boldsymbol{\Gamma}}^Q\widetilde{\boldsymbol{\Lambda}}^Q\widetilde{\boldsymbol{\Gamma}}^{Q\prime}$, where
		\[
		\begin{array}{lll}
			\widetilde{\mathbf{R}}=\mathbf{R}\mathbf{J}_1, & \qquad \widetilde{\mathbf{C}}=\mathbf{C}\mathbf{J}_2, & \qquad \widetilde{\mathbf{F}}_t = \mathbf{J}^{-1}_1 \mathbf{F}_t (\mathbf{J}^{-1}_2)', \\
			\widetilde{\mathbf{A}} = \mathbf{J}^{-1}_1 \mathbf{A}\mathbf{J}_1, & \qquad \widetilde{\mathbf{B}}= \mathbf{J}^{-1}_2\mathbf{B}\mathbf{J}_2, & \qquad  \widetilde{\mathbf{U}}_t= \mathbf{J}^{-1}_1 \mathbf{U}_t (\mathbf{J}^{-1}_2)' \\
			\widetilde{\mathbf{P}}=\mathbf{J}^{-1}_1\mathbf{P}\mathbf{J}_1, &  \qquad \widetilde{\boldsymbol{\Gamma}}^P = \mathbf{J}^{-1}_1 \boldsymbol{\Gamma}^P & \qquad \widetilde{\boldsymbol{\Lambda}}^{P}=\boldsymbol{\Lambda}^P\\
			\widetilde{\mathbf{Q}}=\mathbf{J}^{-1}_2\mathbf{Q}\mathbf{J}_2 & \qquad  \widetilde{\boldsymbol{\Gamma}}^Q = \mathbf{J}^{-1}_2 \boldsymbol{\Gamma}^Q & \qquad \widetilde{\boldsymbol{\Lambda}}^{Q}=\boldsymbol{\Lambda}^Q\\
		\end{array}
		\]
		\label{lemma:dmfmt}
	\end{lemma}
	\begin{proof} Plug-in all the rotated (``tilde'') matrices to obtain,
		\[
		\begin{array}{rcl}
			\mathbf{Y}_t &=& \widetilde{\mathbf{R}} \widetilde{\mathbf{F}}_t \widetilde{\mathbf{C}}' + \mathbf{E}_t  \\
			&=& \mathbf{R}\mathbf{J}_1\mathbf{J}^{-1}_1 \mathbf{F}_t \mathbf{J}^{-1}_2\mathbf{J}'_2 \mathbf{C}'  + \mathbf{E}_t \\
			&=& \mathbf{R} \mathbf{F}_t \mathbf{C}' + \mathbf{E}_t \\[0.25in]
			\widetilde{\mathbf{F}}_t &=& \widetilde{\mathbf{A}} \widetilde{\mathbf{F}}_{t-1}\widetilde{\mathbf{B}}' + \widetilde{\mathbf{U}}_t \\
			\mathbf{J}_1\widetilde{\mathbf{F}}_t \mathbf{J}'_2&=& \mathbf{J}_1\widetilde{\mathbf{A}} \widetilde{\mathbf{F}}_{t-1}\widetilde{\mathbf{B}}'\mathbf{J}'_2 + \mathbf{J}_1\widetilde{\mathbf{U}}_t \mathbf{J}'_2 \\
			\mathbf{J}_1 \mathbf{J}^{-1}_1 \mathbf{F}_t (\mathbf{J}^{-1}_2)' \mathbf{J}'_2&=& \mathbf{J}_1\mathbf{J}^{-1}_1 \mathbf{A}\mathbf{J}_1 \mathbf{J}^{-1}_1 \mathbf{F}_t (\mathbf{J}^{-1}_2)'  \mathbf{J}'_2\mathbf{B}'(\mathbf{J}^{-1}_2)' \mathbf{J}'_2 + \mathbf{J}_1\mathbf{J}^{-1}_1 \mathbf{U}_t (\mathbf{J}^{-1}_2)'\mathbf{J}'_2\\
			\mathbf{F}_t &=& \mathbf{A} \mathbf{F}_{t-1}\mathbf{B}' + \mathbf{U}_t
		\end{array}
		\]
		Moreover, using Assumption \ref{subass:arshock}, we have
		\[
		\begin{array}{rll}
			\mathbb{E}\left[\widetilde{\mathbf{U}}_t\widetilde{\mathbf{U}}'_t\right] &=& \mathbb{E}\left[\mathbf{J}^{-1} _1\mathbf{U}_t(\mathbf{J}^{-1}_2)'\mathbf{J}^{-1}_2\mathbf{U}_t(\mathbf{J}^{-1}_1)'\right] \\
			&=& \mathbf{J}^{-1} _1 \mathbf{P}\tr{\mathbf{Q}}(\mathbf{J}^{-1}_1)' \\
			&=& \mathbf{J}^{-1} _1 \mathbf{P}(\mathbf{J}^{-1}_1)' \tr{\mathbf{Q}(\mathbf{J}^{-1}_2)'\mathbf{J}^{-1}_2}\\
			&=& \widetilde{\mathbf{P}}\tr{\widetilde{\mathbf{Q}}}
		\end{array}
		\]
		and analogous derivation can be obtained for $\mathbb{E}\left[\widetilde{\mathbf{U}}'_t\widetilde{\mathbf{U}}_t\right]$. The derivation of the vectorized form follows naturally.
	\end{proof}

	\begin{lemma}
		Consider the rotated system \eqref{eq:dmfmt1}-\eqref{eq:dmfmt2} defined in Lemma \ref{lemma:dmfmt}. For any $k_1\times k_1$ and $k_2\times k_2$ orthogonal matrices $\mathbf{J}_1$ and $\mathbf{J}_2$,  under Assumptions \ref{ass:common}-\ref{ass:idio}, we have that
		\begin{enumerate}
			\item[(i)] $\lVert \widetilde{\mathbf{A}} \rVert \leq1$, $\lVert \widetilde{\mathbf{B}} \rVert<1$, and $\lVert \widetilde{\mathbf{B}}\otimes\widetilde{\mathbf{A}} \rVert < 1$
			\item[(ii)] $\lVert {\mathbf{H}} \rVert=O(1)$, $\lVert {\mathbf{K}} \rVert=O(1)$, and $\lVert {\mathbf{K}}\otimes{\mathbf{H}} \rVert = O(1)$
			\item[(iii)] $p^{-1/2}_1\left\lVert {\widetilde{\mathbf{R}}} \right\rVert=O(1)$, and $p^{-1/2}_2\left\lVert {\widetilde{\mathbf{C}}} \right\rVert=O(1)$,
			\item[(iv)] $p_2 \left\lVert\left(\widetilde{\mathbf{C}}'\dg{\mathbf{K}}^{-1}\widetilde{\mathbf{C}}\right)^{-1} \right\rVert = O(1)$ and $p_1 \left\lVert\left(\widetilde{\mathbf{R}}'\dg{\mathbf{H}}^{-1}\widetilde{\mathbf{R}}\right)^{-1} \right\rVert = O(1)$
			\item[(v)] $\frac{1}{\sqrt{p_2}} \left\lVert\widetilde{\mathbf{C}}'\dg{\mathbf{K}}^{-1} \right\rVert = O(1)$ and $\frac{1}{\sqrt{p_1}} \left\lVert\widetilde{\mathbf{R}}'\dg{\mathbf{H}}^{-1}\right\rVert = O(1)$
		\end{enumerate}
		\label{lemma:tildeBound}
	\end{lemma}
	\begin{proof}
		To show (i), note that by Assumption \ref{subass:armatrix},
		\[
		\lVert \widetilde{\mathbf{A}} \rVert\leq\lVert \mathbf{J}^{-1}_1 \mathbf{A}\mathbf{J}_1 \rVert_F= \sqrt{\tr{\mathbf{J}^{-1}_1 \mathbf{A}\mathbf{J}_1\mathbf{J}'_1 \mathbf{A}'(\mathbf{J}^{-1}_1)'}}=\sqrt{\tr{\ \mathbf{A} \mathbf{A}'}}=\lVert\mathbf{A}\rVert<1
		\]
		\[
		\lVert \widetilde{\mathbf{B}} \rVert = \lVert \mathbf{J}^{-1}_2 \mathbf{B}\mathbf{J}_2 \rVert <  \lVert \mathbf{J}^{-1}_2 \rVert \lVert\mathbf{B}\rVert \lVert \mathbf{J}_2 \rVert = \lVert\mathbf{B}\rVert < 1
		\]
		and $\lVert \widetilde{\mathbf{B}}\otimes\widetilde{\mathbf{A}} \rVert \leq \lVert \widetilde{\mathbf{B}} \rVert \lVert \widetilde{\mathbf{A}} \rVert<1$.
		To show (ii), start noticing that since $\mathbf{H}$ and $\mathbf{K}$ are symmetric matrices then also $\mathbf{K}\otimes\mathbf{H}$ is symmetric, implying that    $\lVert \mathbf{K}\otimes\mathbf{H} \rVert=\rho(\mathbf{K}\otimes\mathbf{H})<\lVert \mathbf{K}\otimes\mathbf{H} \rVert_1$. Note that each column of $\mathbf{K}\otimes\mathbf{H} $ can be written as $\mathbf{k}_{\cdot j} \otimes \mathbf{h}_{\cdot i}$ for $j=1,\dots,p_2$ and $i=1,\dots,p_1$. By the symmetry of $\mathbf{K}$ and $\mathbf{H}$, we have that $\mathbf{k}_{\cdot j} \otimes \mathbf{h}_{\cdot i}=\mathbf{k}_{j\cdot} \otimes \mathbf{h}_{i\cdot}$ and $\mathbf{k}_{j\cdot} \otimes \mathbf{h}_{i\cdot}=\vect{\mathbb{E}\left[\mathbf{e}_{\cdot j}\mathbf{e}'_{i\cdot}\right]}$. We conclude $\lVert \mathbf{K}\otimes\mathbf{H} \rVert_1=O(1)$ by Assumption \ref{subass:ecorr}. The same conclusion follows for $\lVert {\mathbf{H}} \rVert$ and $\lVert {\mathbf{K}} \rVert$ noticing that $\lVert \mathbf{K}\otimes\mathbf{H} \rVert= \lVert \mathbf{K}\rVert \lVert \mathbf{H} \rVert$.
		To show (iii), note that by Assumption \ref{subass:loading}
		\[
		\frac{1}{p_1}\left\lVert\widetilde{\mathbf{R}}\right\rVert^2 \asymp \frac{1}{p_1} \left\lVert\mathbf{R}\right\rVert^2_F = \frac{1}{p_1} \sum\limits_{i=1}^{p_1}\sum\limits_{j=1}^{k_1} r^2_{ij} \leq k_1 \bar{r} = O(1),
		\]
		and
		\[
		\frac{1}{p_2}\left\lVert\widetilde{\mathbf{C}}\right\rVert^2 \asymp \frac{1}{p_2} \left\lVert\mathbf{C}\right\rVert^2_F = \frac{1}{p_2} \sum\limits_{i=1}^{p_2}\sum\limits_{j=1}^{k_2} c^2_{ij} \leq k_2 \bar{c} = O(1).
		\]
		To show (iv), note that by Theorem 1 in \cite{merikoski2004inequalities}, we have
		\[
		\begin{array}{rll}
			p_2 \left\lVert\left(\widetilde{\mathbf{C}}'\dg{\mathbf{K}}^{-1}\widetilde{\mathbf{C}}\right)^{-1} \right\rVert &=& \frac{p_2}{\nu^{(k_2)}\left(\widetilde{\mathbf{C}}'\dg{\mathbf{K}}^{-1}\widetilde{\mathbf{C}} \right)} \\[0.1in]
			&\leq& \frac{p_2}{\nu^{(k_2)}\left(\widetilde{\mathbf{C}}'\widetilde{\mathbf{C}} \right)\nu^{(p_2)}\left(\dg{\mathbf{K}}^{-1}\right)}  \\[0.1in]
			&\leq& \frac{p_2}{\nu^{(k_2)}\left(\widetilde{\mathbf{C}}'\widetilde{\mathbf{C}} \right) \left(\nu^{(1)}\left(\dg{\mathbf{K}}\right) \right)^{-1}}  \\[0.1in]
			&\lesssim& \frac{1}{C^{-1}_K}  \\[0.1in]
		\end{array}
		\]
		by Assumptions \ref{subass:loading} and \ref{subass:idshocks}. The proof for $p_1 \left\lVert\left(\widetilde{\mathbf{R}}'\dg{\mathbf{H}}^{-1}\widetilde{\mathbf{R}}\right)^{-1} \right\rVert $ follows the same steps.
		To show (v), note that
		\[
		\begin{array}{rll}
			\frac{1}{\sqrt{p_2}} \left\lVert\widetilde{\mathbf{C}}'\dg{\mathbf{K}}^{-1} \right\rVert &\leq& \frac{1}{\sqrt{p_2}} \left\lVert\widetilde{\mathbf{C}} \right\rVert \left\lVert\dg{\mathbf{K}}^{-1} \right\rVert \\[0.1in]
			&\leq& \frac{1}{\sqrt{p_2}} \left\lVert\widetilde{\mathbf{C}} \right\rVert \max\limits_{j=1,\dots,p_2}  k^{-1}_{ii} \\[0.1in]
			&=& O(1) 
		\end{array}
		\]
		by (iii) and Assumption \ref{subass:idshocks}. The proof for $\frac{1}{\sqrt{p_1}} \left\lVert\widetilde{\mathbf{R}}'\dg{\mathbf{H}}^{-1}\right\rVert$ follows the same steps.
	\end{proof}

\subsubsection{Results on pre-estimators}
\begin{lemma}
	Under Assumptions \ref{ass:common} through \ref{ass:dep}, there exist matrices $\widehat{\mathbf{J}}_1$ and $\widehat{\mathbf{J}}_2$ satisfying $\widehat{\mathbf{J}}_1\widehat{\mathbf{J}}'_1 \xrightarrow{p} \mathbb{I}_{k_1k_1}$ and $\widehat{\mathbf{J}}_2\widehat{\mathbf{J}}'_2 \xrightarrow{p} \mathbb{I}_{k_2k_2} $, such that as $\min\left\{p_1,p_2,T\right\} \rightarrow \infty$,
	\begin{itemize}
		\item[(i)] $\frac{1}{\sqrt{p_1}} \left\lVert \widehat{\mathbf{R}}^{(0)} - \mathbf{R}\widehat{\mathbf{J}}_1 \right\rVert=O_p\left( \max\left\{\frac{1}{\sqrt{Tp_2}},\frac{1}{p_1p_2},\frac{1}{Tp_1}\right\} \right)$ 
		\item[(ii)] $\frac{1}{\sqrt{p_2}} \left\lVert \widehat{\mathbf{C}}^{(0)} - \mathbf{C}\widehat{\mathbf{J}}_2 \right\rVert=O_p\left( \max\left\{\frac{1}{\sqrt{Tp_1}},\frac{1}{p_1p_2},\frac{1}{Tp_2}\right\} \right)$
		\item[(iii)] $\frac{1}{\sqrt{p_1p_2}} \left\lVert \widehat{\mathbf{C}}^{(0)}\otimes\widehat{\mathbf{R}}^{(0)} - \mathbf{C}\widehat{\mathbf{J}}_2\otimes\mathbf{R}\widehat{\mathbf{J}}_1 \right\rVert=O_p\left(\max\left\{\frac{1}{\sqrt{Tp_1}}, \frac{1}{\sqrt{Tp_2}},\frac{1}{p_1p_2}\right\}\right)$ 
	\end{itemize}
	\label{lemma:boundRC}
\end{lemma}
\begin{proof}
	Note that (i) and (ii) follows immediately from Theorem 3.1 in \cite{yu2021projected}. Moreover, consider (iii) and note that
	\[
	\begin{array}{rll}
		\frac{1}{\sqrt{p_1p_2}} \left\lVert \widehat{\mathbf{C}}^{(0)}\otimes\widehat{\mathbf{R}}^{(0)} - \mathbf{C}\widehat{\mathbf{J}}_2\otimes\mathbf{R}\widehat{\mathbf{J}}_1 \right\rVert & \leq &  \frac{1}{\sqrt{p_1p_2}} \left\lVert \left(\widehat{\mathbf{C}}^{(0)}-\mathbf{C}\widehat{\mathbf{J}}_2\right)\otimes\left(\widehat{\mathbf{R}}^{(0)}-\mathbf{R}\widehat{\mathbf{J}}_1\right) \right\rVert \\[0.1in]
		 && +\frac{1}{\sqrt{p_1p_2}} \left\lVert\left(\mathbf{C}\widehat{\mathbf{J}}_2\right)\otimes\left(\widehat{\mathbf{R}}^{(0)} -\mathbf{R}\widehat{\mathbf{J}}_1\right)\right\rVert  \\[0.1in]
		 && +  \frac{1}{\sqrt{p_1p_2}} \left\lVert\left(\widehat{\mathbf{C}}^{(0)} -\mathbf{C}\widehat{\mathbf{J}_2}\right)\otimes\left(\mathbf{R}\widehat{\mathbf{J}}_1\right)\right\rVert  \\[0.1in]
		 &\lesssim&  \left\lVert\frac{\mathbf{C}\widehat{\mathbf{J}}_2}{\sqrt{p_2}}\right\rVert\left\lVert\frac{\widehat{\mathbf{R}}^{(0)} -\mathbf{R}\widehat{\mathbf{J}}_1}{\sqrt{p_1}}\right\rVert + \left\lVert\frac{\widehat{\mathbf{C}}^{(0)} -\mathbf{C}\widehat{\mathbf{J}}_2}{\sqrt{p_2}}\right\rVert   \left\lVert\frac{\mathbf{R}\widehat{\mathbf{J}}_1}{\sqrt{p_1}}\right\rVert \\[0.1in]
		 &=& O_p\left( \max\left\{\frac{1}{\sqrt{Tp_2}},\frac{1}{p_1p_2},\frac{1}{Tp_1}\right\} \right) + O_p\left( \max\left\{\frac{1}{\sqrt{Tp_1}},\frac{1}{p_1p_2},\frac{1}{Tp_2}\right\} \right)
	\end{array}
	\]
	by (i), (ii) and Lemma \ref{lemma:tildeBound}(iii).
\end{proof}

\begin{lemma}
	Under Assumptions \ref{ass:common} through \ref{ass:dep}, there exist matrices $\widehat{\mathbf{J}}_1$ and $\widehat{\mathbf{J}}_2$ satisfying $\widehat{\mathbf{J}}_1\widehat{\mathbf{J}}'_1 \xrightarrow{p} \mathbb{I}_{k_1k_1}$ and $\widehat{\mathbf{J}}_2\widehat{\mathbf{J}}'_2 \xrightarrow{p} \mathbb{I}_{k_2k_2} $, such that for $s_1,s_2\in\{0,1\}$ as $\min\left\{p_1,p_2,T\right\} \rightarrow \infty$,
	\begin{itemize}
		\item[(i)] $\frac{1}{T}\sum_{t}^{T}\left(\left(\frac{\widehat{\mathbf{C}}^{(0)}}{p_2}\otimes\frac{\widehat{\mathbf{R}}^{(0)}}{p_1}\right)'\mathsf{y}_{t-s_1}-\widehat{\mathbf{J}}^{-1}\mathsf{f}_{t-s_1}\right) \mathsf{f}'_{t-s_2} \widehat{\mathbf{J}}^{-1\prime}=O_p\left(\max\left\{\frac{1}{\sqrt{T}},\frac{1}{p_1p_2}\right\}\right)$,
		\item[(ii)] $\frac{1}{T}\sum_{t}^{T}\left(\left(\frac{\widehat{\mathbf{C}}^{(0)}}{p_2}\otimes\frac{\widehat{\mathbf{R}}^{(0)}}{p_1}\right)'\mathsf{y}_{t-s_1}-\widehat{\mathbf{J}}^{-1}\mathsf{f}_{t-s_1}\right) \mathsf{u}'_{t}=O_p\left(\max\left\{\frac{1}{\sqrt{T}},\frac{1}{p_1p_2}\right\}\right),$
	\end{itemize}
	with $\widehat{\mathbf{J}}=\widehat{\mathbf{J}}_2\otimes\widehat{\mathbf{J}}_1$.
	\label{lemma:innerboundAB}
\end{lemma}
\begin{proof}
	Note that the left hand side of (i) can be written as
	\begin{equation}
		\frac{1}{T}\sum_{t=2}^{T} \left(\vect{ \frac{\widehat{\mathbf{R}}^{(0)\prime}}{p_1}\mathbf{R}\widehat{\mathbf{J}}_1\widehat{\mathbf{J}}^{-1}_1\mathbf{F}_{t-s_1}\widehat{\mathbf{J}}^{-1\prime}_2\widehat{\mathbf{J}}'_2\mathbf{C}'\frac{\widehat{\mathbf{C}}^{(0)}}{p_2} + \frac{\widehat{\mathbf{R}}^{(0)\prime}}{p_1}\mathbf{E}_{t}\frac{\widehat{\mathbf{C}}^{(0)}}{p_2}}- \widehat{\mathbf{J}}^{-1}\mathsf{f}_{t-s_1} \right) \mathsf{f}'_{t-s_2}\widehat{\mathbf{J}}^{-1\prime}
		\label{eq:AB1b}
	\end{equation}
	By Theorem 3.1 in \cite{yu2021projected}, we have that 
	\begin{equation}
		p^{-1}_1\widehat{\mathbf{R}}^{(0)\prime}\mathbf{R} =\widehat{\mathbf{J}}'_1 + o_p(1), \qquad p^{-1}_2\widehat{\mathbf{C}}^{(0)\prime}\mathbf{C}=\widehat{\mathbf{J}}'_2+ o_p(1),
		\label{eq:convRhR}
	\end{equation}
	therefore \eqref{eq:AB1b} is asymptotically equivalent to
	\begin{equation}
		\frac{1}{T} \sum^T_{t=2} \left(\frac{\widehat{\mathbf{C}}^{(0)} }{p_2}\otimes \frac{\widehat{\mathbf{R}}^{(0)}}{p_1}\right)' \mathsf{e}_{t-s_1} \mathsf{f}'_{t-s_2} \widehat{\mathbf{J}}^{-1\prime}.
		\label{eq:AB1b.as}
	\end{equation}
	We can bound \eqref{eq:AB1b.as} as follows
	\[
	\begin{array}{rll}
		\frac{1}{T} \sum\limits^T_{t=2} \left(\frac{\widehat{\mathbf{C}}^{(0)} }{p_2}\otimes \frac{\widehat{\mathbf{R}}^{(0)}}{p_1}\right)' \mathsf{e}_{t} \mathsf{f}'_{t} \widehat{\mathbf{J}}^{-1\prime} &\lesssim&  \frac{1}{T\sqrt{p_1p_2}} \sum\limits^T_{t=2} \left(\frac{\widehat{\mathbf{C}}^{(0)} }{\sqrt{p_2}}\otimes \frac{\widehat{\mathbf{R}}^{(0)}}{\sqrt{p_1}}-\frac{\mathbf{C}\widehat{\mathbf{J}}_2}{\sqrt{p_2}}\otimes \frac{\mathbf{R}\widehat{\mathbf{J}}_1}{\sqrt{p_1}}\right)' \mathsf{e}_{t} \mathsf{f}'_{t} \widehat{\mathbf{J}}^{-1\prime} \\
		&& + \frac{1}{T\sqrt{p_1p_2}} \sum\limits^T_{t=2} \left(\frac{\mathbf{C}\widehat{\mathbf{J}}_2}{\sqrt{p_2}}\otimes \frac{\mathbf{R}\widehat{\mathbf{J}}_1}{\sqrt{p_1}}\right)' \mathsf{e}_{t} \mathsf{f}'_{t} \widehat{\mathbf{J}}^{-1\prime} \\
		&\lesssim& O_p\left(\max\left\{\frac{1}{\sqrt{T}}, \frac{1}{p_1p_2} \right\}\right).
	\end{array}
	\]
	by Assumption \ref{subass:fecorr}, and Lemmas \ref{lemma:tildeBound} and \ref{lemma:boundRC}. Consider (ii) and note that by \eqref{eq:convRhR}, we have 
	\[
	\begin{array}{rll}
		\frac{1}{T}\sum\limits_{t}^{T}\left(\left(\frac{\widehat{\mathbf{C}}^{(0)}}{p_2}\otimes\frac{\widehat{\mathbf{R}}^{(0)}}{p_1}\right)'\mathsf{y}_{t-s_1}-\widehat{\mathbf{J}}^{-1}\mathsf{f}_{t-s_1}\right) \mathsf{u}'_{t} & \lesssim &
		\frac{1}{T\sqrt{p_1p_2}} \sum\limits^T_{t=2} \left(\frac{\widehat{\mathbf{C}}^{(0)} }{\sqrt{p_2}}\otimes \frac{\widehat{\mathbf{R}}^{(0)}}{\sqrt{p_1}}\right)' \mathsf{e}_{t} \mathsf{u}'_{t}. \\
		&\lesssim& O_p\left(\max\left\{\frac{1}{\sqrt{T}}, \frac{1}{p_1p_2} \right\}\right).
	\end{array}
	\]
	by  Lemmas \ref{lemma:tildeBound} and \ref{lemma:boundRC}, and since
	We can bound \eqref{eq:AB1b.as} as follows
	\[
	\begin{array}{rll}
		\mathbb{E}\left[\left\lVert\frac{1}{T\sqrt{p_1p_2}}  \sum\limits^T_{t=2} \mathsf{e}_{t} \mathsf{u}'_{t} \right\rVert^2\right]&\leq&  \frac{1}{T^2p_1p_2} \sum\limits_{i_1=1}^{p_1}\sum\limits_{i_2=1}^{p_2} \sum\limits_{j_1=1}^{k_1}\sum\limits_{j_2=1}^{k_2} \sum\limits^T_{t,s=1} \mathbb{E}\left[e_{ti_1i_2}e_{si_1i_2} u_{tj_1j_2}u_{sj_1j_2} \right]  \\
		&\leq&   \frac{kC_PC_Q}{T^2p_1p_2} \sum\limits_{i_1=1}^{p_1}\sum\limits_{i_2=1}^{p_2}  \sum\limits^T_{t,s=1} \left\lvert \mathbb{E}\left[e_{ti_1i_2}e_{si_1i_2} \right] \right\rvert \\
		&\lesssim& O_p\left(\frac{1}{T} \right).
	\end{array}
	\]
	by Assumptions \ref{subass:arshock} and \ref{subass:ecorr}.
\end{proof}

\begin{lemma}
	Under Assumptions \ref{ass:common} through \ref{ass:dep}, there exist matrices $\widehat{\mathbf{J}}_1$ and $\widehat{\mathbf{J}}_2$ satisfying $\widehat{\mathbf{J}}_1\widehat{\mathbf{J}}'_1 \xrightarrow{p} \mathbb{I}_{k_1}$ and $\widehat{\mathbf{J}}_2\widehat{\mathbf{J}}'_2 \xrightarrow{p} \mathbb{I}_{k_2} $, such that as $\min\left\{p_1,p_2,T\right\} \rightarrow \infty$,
	\begin{itemize} 
		\item[(i)] $\left\lVert \widehat{\mathbf{B}\otimes\mathbf{A}}^{(0)} - \widehat{\mathbf{J}}^{-1}(\mathbf{B}\otimes\mathbf{A})\widehat{\mathbf{J}}  \right\rVert = O_p\left( \max\left\{\frac{1}{\sqrt{T}}, \frac{1}{p_1p_2} \right\}\right)$
		\item[(ii)]  $\left\lVert \widehat{\mathbf{Q} \otimes \mathbf{P}}^{(0)} - \widehat{\mathbf{J}}^{-1}(\mathbf{Q}\otimes\mathbf{P})\widehat{\mathbf{J}} \right\rVert = O_p\left( \max\left\{\frac{1}{\sqrt{T}}, \frac{1}{p_1p_2} \right\}\right)$,
	\end{itemize}
	with $\widehat{\mathbf{J}}=\widehat{\mathbf{J}}_2\otimes\widehat{\mathbf{J}}_1$.
	\label{lemma:boundAB}
\end{lemma}

\begin{proof}
	Consider (i) and note that
	\[
	\begin{array}{rll}
		\left\lVert \widehat{\mathbf{B}\otimes\mathbf{A}}^{(0)} - \widehat{\mathbf{J}}^{-1}(\mathbf{B}\otimes\mathbf{A})\widehat{\mathbf{J}}\right\rVert &\leq& \left\lVert \left(\frac{1}{T}\sum\limits_{t=2}^{T} \widetilde{\mathsf{f}}_t \widetilde{\mathsf{f}}'_{t-1} - \frac{1}{T} \sum\limits_{t=1}^{T} \widehat{\mathbf{J}}^{-1}\mathsf{f}_{t}\mathsf{f}'_{t-1}\widehat{\mathbf{J}}^{-\prime} \right)\left(\frac{1}{T}\sum\limits_{t=2}^{T} \widetilde{\mathsf{f}}_{t-1} \widetilde{\mathsf{f}}'_{t-1} \right)^{-1}  \right\rVert \\[0.2in]
		&& + \left\lVert \widehat{\mathbf{J}}^{-1}(\mathbf{B}\otimes\mathbf{A})\widehat{\mathbf{J}}\left(\frac{1}{T}\sum\limits_{t=2}^{T} \widetilde{\mathsf{f}}_{t-1} \widetilde{\mathsf{f}}'_{t-1} - \frac{1}{T} \sum\limits_{t=1}^{T} \widehat{\mathbf{J}}^{-1}\mathsf{f}_{t-1}\mathsf{f}'_{t-1}\widehat{\mathbf{J}}^{-\prime} \right)\left(\frac{1}{T}\sum\limits_{t=2}^{T} \widetilde{\mathsf{f}}_{t-1} \widetilde{\mathsf{f}}'_{t-1} \right)^{-1}  \right\rVert \\[0.2in]
		&& + \left\lVert \left(\frac{1}{T}\sum\limits_{t=2}^{T} \mathsf{u}_t \widetilde{\mathsf{f}}'_{t-1} \right)\left(\frac{1}{T}\sum\limits_{t=2}^{T} \widetilde{\mathsf{f}}_{t-1} \widetilde{\mathsf{f}}'_{t-1} \right)^{-1}  \right\rVert.
	\end{array}
	\]
	where $\widetilde{\mathsf{f}}_t=\left(\frac{\widehat{\mathbf{C}}^{(0)}}{p_2}\otimes\frac{\widehat{\mathbf{R}}^{(0)}}{p_1}\right)'\mathsf{y}_t$.	Now, for $s=\{0,1\}$,
	\[
	\begin{array}{rll}
		\frac{1}{T}\sum\limits_{t=2}^{T} \widetilde{\mathsf{f}}_{t-s} \widetilde{\mathsf{f}}'_{t-1} - \frac{1}{T} \sum\limits_{t=2}^{T} \widehat{\mathbf{J}}^{-1}\mathsf{f}_{t-s}\mathsf{f}'_{t-1}\widehat{\mathbf{J}}^{-\prime} &=& \frac{1}{T}\sum\limits_{t=2}^{T} \left(\widetilde{\mathsf{f}}_{t-s} - \widehat{\mathbf{J}}^{-1}\mathsf{f}_{t-s}\right) \mathsf{f}'_{t-1}\widehat{\mathbf{J}}^{-\prime} + \frac{1}{T}\sum\limits_{t=2}^{T}  \widehat{\mathbf{J}}^{-1}\mathsf{f}_{t-s} \left(\widetilde{\mathsf{f}}_{t-1}-\widehat{\mathbf{J}}^{-1}\mathsf{f}_{t-1}\right)' \\[0.1in]
		&& + \frac{1}{T} \sum\limits_{t=2}^{T} \left(\widetilde{\mathsf{f}}_{t-s} -\widehat{\mathbf{J}}^{-1}\mathsf{f}_{t-s}\right)\left(\widetilde{\mathsf{f}}_{t-1}-\widehat{\mathbf{J}}^{-1}\mathsf{f}_{t-1}\right)'.
	\end{array}
	\]
	Note that the first two terms dominate the third one. By Lemma \ref{lemma:innerboundAB}(i) we have that
	\[
	\left\lVert \frac{1}{T}\sum\limits_{t=2}^{T} \widetilde{\mathsf{f}}_{t} \widetilde{\mathsf{f}}'_{t-1} - \frac{1}{T} \sum\limits_{t=2}^{T} \widehat{\mathbf{J}}^{-1}\mathsf{f}_{t}\mathsf{f}'_{t-1}\widehat{\mathbf{J}}^{-\prime}  \right\rVert = O_p\left(\max\left\{\frac{1}{\sqrt{T}},\frac{1}{p_1p_2}\right\}\right)
	\]
	\[
	\left\lVert \frac{1}{T}\sum\limits_{t=2}^{T} \widetilde{\mathsf{f}}_{t-1} \widetilde{\mathsf{f}}'_{t-1} - \frac{1}{T} \sum\limits_{t=2}^{T} \widehat{\mathbf{J}}^{-1}\mathsf{f}_{t-1}\mathsf{f}'_{t-1} \widehat{\mathbf{J}}^{-\prime} \right\rVert = O_p\left(\max\left\{\frac{1}{\sqrt{T}},\frac{1}{p_1p_2}\right\}\right).
	\]
	Moreover, combining the latter with Assumption \ref{subass:factor}, we have that,
	\[
	\left(\frac{1}{T}\sum\limits_{t=2}^{T} \widetilde{\mathsf{f}}_{t-1} \widetilde{\mathsf{f}}'_{t-1} \right)^{-1} = O_p(1).
	\]
	Finally, note that 
	\[
	\left\lVert \frac{1}{T}\sum\limits_{t=2}^{T} \mathsf{u}_t \widetilde{\mathsf{f}}'_{t-1} \right\rVert \leq	\left\lVert \frac{1}{T}\sum\limits_{t=2}^{T} \mathsf{u}_t \left(\widetilde{\mathsf{f}}_{t-1}-\mathsf{f}_{t-1}\right) \right\rVert + \left\lVert \frac{1}{T}\sum\limits_{t=2}^{T} \mathsf{u}_t \mathsf{f}_{t-1} \right\rVert = O_p\left(\min\left\{\frac{1}{\sqrt{T}},\frac{1}{p_1p_2}\right\}\right),
	\]
	by Lemma \ref{lemma:innerboundAB}(ii) and Assumption \ref{subass:arshock}, concluding the proof. The result for (ii) follows using the same steps of (i) noting that
	\[
	\begin{array}{rll}
	\widehat{\mathbf{Q} \otimes \mathbf{P}}^{(0)} &=& \frac{1}{T}\sum\limits_{t=2}^{T} \left\{\widetilde{\mathsf{f}}_t  - \widehat{\mathbf{B}\otimes\mathbf{A}} \widetilde{\mathsf{f}}_{t-1}\right\} \left\{\widetilde{\mathsf{f}}_t  - \widehat{\mathbf{B}\otimes \mathbf{A}} \widetilde{\mathsf{f}}_{t-1}\right\}' \\[0.2in]
	&=& \frac{1}{T}\sum\limits_{t=2}^{T} \left\{\left(\widetilde{\mathsf{f}}_t -\widehat{\mathbf{J}}^{-1} \mathsf{f}_t \right) + \left(\widehat{\mathbf{B}\otimes\mathbf{A}} -\mathbf{B}\otimes\mathbf{A} \right)\left( \widetilde{\mathsf{f}}_{t-1} - \widehat{\mathbf{J}}^{-1}\mathsf{f}_{t-1} \right)  \right. \\[0.2in]
	&& \hspace{1.25cm} \left. + \left(\widehat{\mathbf{B}\otimes\mathbf{A}} -\mathbf{B}\otimes\mathbf{A} \right) \widetilde{\mathsf{f}}_{t-1} + \mathbf{B}\otimes\mathbf{A} \left( \widetilde{\mathsf{f}}_{t-1} - \widehat{\mathbf{J}}^{-1}\mathsf{f}_{t-1} \right) + \widehat{\mathbf{J}}^{-1}\mathsf{u}_t  \right\} \\[0.2in]
	&&\hspace{1cm} \left\{\left(\widetilde{\mathsf{f}}_t -\widehat{\mathbf{J}}^{-1} \mathsf{f}_t \right) + \left(\widehat{\mathbf{B}\otimes\mathbf{A}} -\mathbf{B}\otimes\mathbf{A} \right)\left( \widetilde{\mathsf{f}}_{t-1} - \widehat{\mathbf{J}}^{-1}\mathsf{f}_{t-1} \right)  \right. \\[0.2in]
	&& \hspace{1.25cm} \left. + \left(\widehat{\mathbf{B}\otimes\mathbf{A}} -\mathbf{B}\otimes\mathbf{A} \right) \widetilde{\mathsf{f}}_{t-1} + \mathbf{B}\otimes\mathbf{A} \left( \widetilde{\mathsf{f}}_{t-1} - \widehat{\mathbf{J}}^{-1}\mathsf{f}_{t-1} \right) + \widehat{\mathbf{J}}^{-1}\mathsf{u}_t  \right\} '
	\end{array}
	\]
	and that
	\[
	\left\lVert\frac{1}{T}\sum\limits_{t=1}^{T} \mathsf{u}_t\mathsf{u}'_t - \mathbf{Q}\otimes\mathbf{P} \right\rVert = O_p\left(\frac{1}{\sqrt{T}}\right)
	\]
	by Assumption \ref{subass:arshock}.
\end{proof}

\begin{lemma}
	Under Assumption \eqref{ass:common} through \eqref{ass:dep},  as $\min\left\{p_1,p_2,T\right\} \rightarrow \infty$,
	\begin{itemize}
		\item[(i)] $\left\lvert\widehat{k}^{(0)}_{jj} - k_{jj} \right\rvert = O_p\left(\max\left\{\frac{1}{\sqrt{Tp_1}},\frac{1}{p_1p_2}, \frac{1}{Tp_2},\right\}\right)$ uniformly in $j$
		\item[(ii)] $\left\lvert\widehat{h}^{(0)}_{ii} - h_{ii} \right\rvert = O_p\left(\max\left\{\frac{1}{\sqrt{Tp_2}},\frac{1}{p_1p_2},\frac{1}{Tp_1}\right\}\right)$ uniformly in $i$
		\item[(iii)] $\frac{1}{p_2}\left\lvert \sum\limits_{j=1}^{p_2}\left(\widehat{k}^{(0)}_{jj} - k_{jj}\right) \right\rvert = O_p\left(\max\left\{\frac{1}{\sqrt{Tp_1}},\frac{1}{p_1p_2},\frac{1}{Tp_2}\right\}\right)$
		\item[(iv)] $\frac{1}{p_1}\left\lvert \sum\limits_{i=1}^{p_1}\left(\widehat{h}^{(0)}_{ii} - h_{ii}\right) \right\rvert =O_p\left(\max\left\{\frac{1}{\sqrt{Tp_2}},\frac{1}{p_1p_2},\frac{1}{Tp_1}\right\}\right)$
	\end{itemize}
	\label{lemma:boundkh}
\end{lemma}

\begin{proof}
	Start from (i), and recall by Assumption \ref{subass:idshocks} we have that $\tr{\mathbf{H}}=p_1$, thus $\mathbf{K}=\frac{1}{p_1}\mathbb{E}\left[\mathbf{E}'_t\mathbf{E}_t\right]$. We can then write 
	\[
\begin{array}{rll}
		\left\lvert\widehat{k}^{(0)}_{jj} - k_{jj} \right\rvert &=&  \left\lvert \frac{1}{Tp_1} \sum\limits_{t=1}^{T} \sum\limits_{i=1}^{p_1} (\widehat{e}^{(0)}_{tij})^2 - \frac{1}{p_1} \sum\limits_{i=1}^{p_1} \mathbb{E}\left[e^2_{tij}\right] \right\rvert \\[0.1in]
		&=&\left\lvert \frac{1}{Tp_1} \sum\limits_{t=1}^{T} \sum\limits_{i=1}^{p_1} \left((s_{tij}-\widehat{s}^{(0)}_{tij})^2 + 2(s_{tij}-\widehat{s}^{(0)}_{tij})e_{tij} + e^2_{tij} \right) - \frac{1}{p_1} \sum\limits_{i=1}^{p_1} \mathbb{E}\left[e^2_{tij}\right] \right\rvert \\[0.1in]
		&\leq&  \left\lvert\frac{1}{Tp_1} \sum\limits_{t=1}^{T} \sum\limits_{i=1}^{p_1} (s_{tij}-\widehat{s}^{(0)}_{tij})^2 \right\rvert + 2 \left\lvert\frac{1}{Tp_1} \sum\limits_{t=1}^{T} \sum\limits_{i=1}^{p_1} (s_{tij}-\widehat{s}^{(0)}_{tij})e_{tij} \right\rvert  \\[0.1in]
		&&  + \left\lvert\frac{1}{Tp_1} \sum\limits_{t=1}^{T} \sum\limits_{i=1}^{p_1} e^2_{tij} - \frac{1}{p_1} \sum\limits_{i=1}^{p_1} \mathbb{E}\left[e^2_{tij}\right]  \right\rvert 
	\end{array}
	\]	
	Consider the first addendum, we have
	\begin{equation}
	\begin{array}{rll}
		\left\lvert\frac{1}{Tp_1} \sum\limits_{t=1}^{T} \sum\limits_{i=1}^{p_1} (s_{tij}-\widehat{s}^{(0)}_{tij})^2 \right\rvert &=& 	\left\lvert\frac{1}{Tp_1} \sum\limits_{t=1}^{T} (\mathbf{s}_{t\cdot j}-\widehat{\mathbf{s}}^{(0)}_{t\cdot j})^2 \right\rvert \\[0.1in]
		&\lesssim & \left\lvert\frac{1}{Tp_1} \sum\limits_{t=1}^{T} \mathbf{c}_{j\cdot}\mathbf{F}'_t \left(\widehat{\mathbf{R}}^{(0)}-\mathbf{R}\widehat{\mathbf{J}}_1\right)'\left(\widehat{\mathbf{R}}^{(0)}-\mathbf{R}\widehat{\mathbf{J}}_1\right)\mathbf{F}_t\mathbf{c}_{j\cdot} \right\rvert \\[0.1in]
		&& + \left\lvert\frac{1}{Tp_1} \sum\limits_{t=1}^{T} \mathbf{c}'_{j\cdot}\mathbf{F}'_t \left(\widehat{\mathbf{R}}^{(0)}-\mathbf{R}\widehat{\mathbf{J}}_1\right)' \mathbf{R}\left(\widetilde{\mathbf{F}}_t-\widehat{\mathbf{J}}^{-1}_1\mathbf{F}_t\widehat{\mathbf{J}}_2 \right)\mathbf{c}_{j\cdot} \right\rvert \\[0.1in]
		&& + \left\lvert\frac{1}{Tp_1} \sum\limits_{t=1}^{T} \mathbf{c}'_{j\cdot}\mathbf{F}'_t \left(\widehat{\mathbf{R}}^{(0)}-\mathbf{R}\widehat{\mathbf{J}}_1\right)' \mathbf{R}\mathbf{F}_t \left(\widehat{\mathbf{c}}_{j\cdot} - \mathbf{c}'_{j\cdot}\widehat{\mathbf{J}}_2\right) \right\rvert \\[0.1in]
		&& + \left\lvert\frac{1}{Tp_1} \sum\limits_{t=1}^{T} \mathbf{c}_{j\cdot}\left(\widetilde{\mathbf{F}}_t-\widehat{\mathbf{J}}^{-1}_1\mathbf{F}_t\widehat{\mathbf{J}}_2 \right)'\mathbf{R}'\mathbf{R}\left(\widetilde{\mathbf{F}}_t-\widehat{\mathbf{J}}^{-1}_1\mathbf{F}_t\widehat{\mathbf{J}}_2 \right)\mathbf{c}_{j\cdot} \right\rvert \\[0.1in]
		&& + \left\lvert\frac{1}{Tp_1} \sum\limits_{t=1}^{T} \mathbf{c}_{j\cdot}\left(\widetilde{\mathbf{F}}_t-\widehat{\mathbf{J}}^{-1}_1\mathbf{F}_t\widehat{\mathbf{J}}_2 \right)'\mathbf{R}'\mathbf{R}\mathbf{F}_t \left(\widehat{\mathbf{c}}_{j\cdot} - \mathbf{c}'_{j\cdot}\widehat{\mathbf{J}}_2\right) \right\rvert \\
\end{array}
	\label{eq:k.h_I}
	\end{equation}
	since
	\[
	\begin{array}{rll}
		\left\lvert\frac{1}{Tp_1} \sum\limits_{t=1}^{T} \mathbf{c}_{j\cdot}\mathbf{F}'_t \left(\widehat{\mathbf{R}}^{(0)}-\mathbf{R}\widehat{\mathbf{J}}_1\right)'\left(\widehat{\mathbf{R}}^{(0)}-\mathbf{R}\widehat{\mathbf{J}}_1\right)\mathbf{F}_t\mathbf{c}_{j\cdot} \right\rvert &\lesssim& \frac{1}{p_1} \left\lVert \widehat{\mathbf{R}}^{(0)}-\mathbf{R}\widehat{\mathbf{J}}_1 \right\rVert^2 \left\lVert\frac{1}{T}\sum^T\limits_{t=1} {\mathbf{F}'_t}{\mathbf{F}_t} \right\rVert \\
		&=& O_p\left( \max\left\{\frac{1}{Tp_2},\frac{1}{p^2_1p^2_2},\frac{1}{T^2p^2_1} \right\} \right)
	\end{array}
	\]
	because of Lemma \ref{lemma:boundRC}(i) and Assumption \ref{subass:factor}
	\[
	%
	\]
	because of Lemma E.1 in \cite{yu2021projected} and since
	\begin{equation}
		%
		\label{eq:F_dec}
	\end{equation}
	we have that
	\begin{equation}
	\left\lVert \frac{1}{T}\sum\limits _{t=1}^{T} \mathbf{F}'_t \left( \widetilde{\mathbf{F}}_t-\widehat{\mathbf{J}}^{-1}_1\mathbf{F}_t\widehat{\mathbf{J}}_2\right) \right\rVert \lesssim \frac{1}{Tp_1p_2}\sum\limits _{t=1}^{T} \mathbf{F}'_t  \mathbf{R}'\mathbf{E}_t \mathbf{C} = O_p\left( \frac{1}{\sqrt{Tp_1p_2}} \right)
	\label{eq:F'Ft}
	\end{equation}
	by Lemma A.1 in \cite{yu2021projected} and Chebyshev inequality,
	\[
	%
	\]
	because of Theorem 3.1 and Lemma E.1 in \cite{yu2021projected} and Assumption \ref{subass:factor}, 
	\[
	%
	\]
	because of Assumption \ref{subass:factor} and Theorem 3.1 in \cite{yu2021projected}. Consider now the second addendum,
	\begin{equation}
	%
	\]
	by the same steps as in the proof of Lemma B.3 in \cite{yu2021projected},
	\[
	%
	\]
	by Theorem 3.1 and assumption \ref{subass:fecorr}. Finally, note that 
	\begin{equation}
		%
	\label{eq:k.h_III_exp}
	\end{equation}
	by Assumption \ref{subass:ecorr2}. Therefore, for the third addendum, we have
	\begin{equation}
		 \left\lvert\frac{1}{Tp_1} \sum\limits_{t=1}^{T} \sum\limits_{i=1}^{p_1} e^2_{tij} - \frac{1}{p_1} \sum\limits_{i=1}^{p_1} \mathbb{E}\left[e^2_{tij}\right]  \right\rvert = O_p\left(\frac{1}{\sqrt{Tp_1}}\right)
	\label{eq:k.h_III}
	\end{equation}
	Combining \eqref{eq:k.h_I}, \eqref{eq:k.h_II}, and \eqref{eq:k.h_III} yields the desired result. Part (ii) follows by similar steps. 
	
	Consider (iii), we have that
	\[
	%
	\]	
	Consider the first addendum, we have that
	\begin{equation}
	%
	\]
	because of Assumption \ref{subass:factor} and Lemmas \ref{lemma:tildeBound}(iii), \ref{lemma:boundRC}(i) and E.1 in \cite{yu2021projected},
	\[
	%
	\]
	because of Theorem 3.1 in \cite{yu2021projected} and Assumption \ref{subass:factor}. 
	Consider the second addendum, we have that
	\begin{equation}
	%
	\]
	by Lemma \ref{lemma:boundRC}(ii) and A.1 in \cite{yu2021projected}. Consider the third addendum, since the result in \eqref{eq:k.h_III_exp} does not depend on $j$, we have that
	\begin{equation}
	%
	\label{eq:sk.h_III}
	\end{equation}
	Combining \eqref{eq:sk.h_I}, \eqref{eq:sk.h_II} and \eqref{eq:sk.h_III} yields the desired result. The proof of (iv) follows similar steps.
\end{proof}

	\begin{lemma}
		Under Assumptions \ref{ass:common}-\ref{ass:idio}, we have that as $\min\left\{T,p_1,p_2\right\}\rightarrow\infty$
		\begin{itemize}
			\item[(i)] $\left\lVert\widehat{\mathbf{K}}^{(0)-1}\right\rVert=O_p(1)$ 
			\item[(ii)] $\left\lVert\widehat{\mathbf{H}}^{(0)-1}\right\rVert=O_p(1)$
			\item[(iii)] $\frac{1}{\sqrt{p_2}}\left\lVert \widehat{\mathbf{C}}^{(0)\prime}\widehat{\mathbf{K}}^{(0)-1} -  \mathbf{C}' \dg{\mathbf{K}}^{-1} \right\rVert=O_p\left(\max\left\{\frac{1}{\sqrt{Tp_1}},\frac{1}{p_1p_2},\frac{1}{Tp_2}\right\}\right)$
			\item[(iv)]$\frac{1}{\sqrt{p_1}}\left\lVert  \widehat{\mathbf{R}}^{(0)\prime}\widehat{\mathbf{H}}^{(0)-1} -  \mathbf{R}'\dg{\mathbf{H}}^{-1} \right\rVert=O_p\left(\max\left\{\frac{1}{\sqrt{Tp_2}},\frac{1}{p_1p_2},\frac{1}{Tp_1}\right\}\right)$
			\item[(v)] $p_2\left\lVert\left(\widehat{\mathbf{C}}^{(0)\prime}\widehat{\mathbf{K}}^{(0)-1}\widehat{\mathbf{C}}^{(0)}\right)^{-1} \right\rVert=O_p(1)$ 
			\item[(vi)] $p_1\left\lVert\left(\widehat{\mathbf{R}}^{(0)\prime}\widehat{\mathbf{H}}^{(0)-1}\widehat{\mathbf{R}}^{(0)}\right)^{-1} \right\rVert=O_p(1)$
			\item[(vii)]  $p_2$$\left\lVert \left(\widehat{\mathbf{C}}^{(0)\prime}\widehat{\mathbf{K}}^{(0)-1}\widehat{\mathbf{C}}^{(0)}\right)^{-1} - \left( \mathbf{C}'\dg{\mathbf{K}}^{-1}\mathbf{C} \right)^{-1}
			 \right\rVert=O_p\left(\max\left\{\frac{1}{\sqrt{Tp_1}},\frac{1}{p_1p_2},\frac{1}{Tp_2}\right\}\right)$
			 \item[(viii)]  $p_1\left\lVert \left(\widehat{\mathbf{R}}^{(0)\prime}\widehat{\mathbf{H}}^{(0)-1}\widehat{\mathbf{R}}^{(0)}\right)^{-1} - \left( \mathbf{R}'\dg{\mathbf{H}}^{-1}\mathbf{R} \right)^{-1} \right\rVert=O_p\left(\max\left\{\frac{1}{\sqrt{Tp_2}},\frac{1}{p_1p_2},\frac{1}{Tp_1}\right\}\right)$
			\item[(ix)]  $\sqrt{p_2}\left\lVert \left(\widehat{\mathbf{C}}^{(0)\prime}\widehat{\mathbf{K}}^{(0)-1}\widehat{\mathbf{C}}^{(0)}\right)^{-1}\widehat{\mathbf{C}}^{(0)\prime}\widehat{\mathbf{K}}^{(0)-1} - \left( \mathbf{C}'\dg{\mathbf{K}}^{-1}\mathbf{C} \right)^{-1} \mathbf{C}'\dg{\mathbf{K}}^{-1} \right\rVert=O_p\left(\max\left\{\frac{1}{\sqrt{Tp_1}},\frac{1}{p_1p_2},\frac{1}{Tp_2}\right\}\right)$ 
			\item[(x)] $\sqrt{p_1}\left\lVert \left(\widehat{\mathbf{R}}^{(0)\prime}\widehat{\mathbf{H}}^{(0)-1}\widehat{\mathbf{R}}^{(0)}\right)^{-1}\widehat{\mathbf{R}}^{(0)\prime}\widehat{\mathbf{H}}^{(0)-1} - \left( \mathbf{R}'\dg{\mathbf{H}}^{-1}\mathbf{R} \right)^{-1} \mathbf{R}'\dg{\mathbf{H}}^{-1} \right\rVert=O_p\left(\max\left\{\frac{1}{\sqrt{Tp_2}},\frac{1}{p_1p_2},\frac{1}{Tp_1}\right\}\right)$
		\end{itemize}
		\label{lemma:est.loadscale} 
	\end{lemma}
	\begin{proof}
		Consider (i) and note that
		\[
		\begin{array}{lcl}
			\left\lVert\widehat{\mathbf{K}}^{(0)-1}\right\rVert &=& \left\{\nu^{(p_2)}\left(\widehat{\mathbf{K}}^{(0)}\right)\right\}^{-1} \\[0.1in]
			&=& \left\{ \min\limits_{j = 1,\dots,p_2} k_{jj} + \widehat{k}^{(0)}_{jj} - k_{jj}  \right\}^{-1} \\[0.1in]
			&=& \left\{ \min\limits_{j = 1,\dots,p_2} k_{jj} - \min\limits_{j = 1,\dots,p_2} \left\lvert \widehat{k}^{(0)}_{jj} - k_{jj} \right\rvert  \right\}^{-1} \\[0.1in]
			&=& \left\{ C^{-1}_K -  \left\lvert \widehat{k}^{(0)}_{jj} - k_{jj} \right\rvert  \right\}^{-1} = C_K + O_p\left(\max\left\{\frac{1}{\sqrt{Tp_1}},\frac{1}{p_1p_2}, \frac{1}{Tp_2}\right\}\right)
		\end{array}	
		\]
		by Assumption \ref{subass:idshocks} and Lemma \ref{lemma:boundkh}(i). The proof for (ii) follows the same steps.  Consider (iii), we have that
		\[
		\begin{array}{l}
			p_2^{-1} \left\lVert \widehat{\mathbf{C}}^{(0)\prime}\widehat{\mathbf{K}}^{(0)-1} -  \mathbf{C}'\dg{\mathbf{K}}^{-1} \right\rVert^2 = p_2^{-1} \left\lVert \sum\limits_{i=j}^{p_2} \widehat{\mathbf{c}}^{(0)}_{j\cdot} \widehat{k}^{(0)-1}_{jj} - \sum\limits_{j=1}^{p_2} {\mathbf{c}}_{j\cdot} {k}^{-1}_{jj}   \right\rVert^2 \\[0.1in]
			= p_2^{-1} \left\lVert \sum\limits_{j=1}^{p_2} \widehat{\mathbf{c}}^{(0)}_{j\cdot} \left(\widehat{k}^{(0)}_{jj} - k_{jj} + k_{jj}\right)^{-1} - \sum\limits_{j=1}^{p_2} {\mathbf{c}}_{j\cdot} {k}^{-1}_{jj}   \right\rVert^2 \\[0.1in]
			= p_2^{-1} \left\lVert \sum\limits_{j=1}^{p_2} \widehat{\mathbf{c}}^{(0)}_{j\cdot} k^{-1}_{jj} \left(1+ (\widehat{k}^{(0)}_{jj} - k_{jj})/k_{jj}\right)^{-1} - \sum\limits_{j=1}^{p_2} {\mathbf{c}}_{j\cdot} {k}^{-1}_{jj}   \right\rVert^2 \\[0.1in]
			\leq p_2^{-1} \left\lVert \left(1+ \min\limits_{j=1,\dots,p_2}(\widehat{k}^{(0)}_{jj} - k_{jj})/k_{jj}\right)^{-1}  \sum\limits_{j=1}^{p_2} \widehat{\mathbf{c}}^{(0)}_{j\cdot} k^{-1}_{jj} - \sum\limits_{j=1}^{p_2} {\mathbf{c}}_{j\cdot} {k}^{-1}_{jj}   \right\rVert^2 \\[0.1in]
			\leq p_2^{-1} \left\lVert \left(1 - \min\limits_{j=1,\dots,p_2}(\widehat{k}^{(0)}_{jj} - k_{jj})/k_{jj} + o\left(\min\limits_{j=1,\dots,p_2}(\widehat{k}^{(0)}_{jj} - k_{jj})/k_{jj}\right)\right)  \sum\limits_{j=1}^{p_2} \widehat{\mathbf{c}}^{(0)}_{j\cdot} k^{-1}_{jj} - \sum\limits_{j=1}^{p_2} {\mathbf{c}}_{j\cdot} {k}^{-1}_{jj}   \right\rVert^2 \\[0.1in]
			\leq p_2^{-1} \left\lVert \widehat{\mathbf{C}}^{(0)\prime}\dg{\mathbf{K}}^{-1} -  \mathbf{C}'\dg{\mathbf{K}}^{-1} \right\rVert^2 + p_2^{-1} \left(\min\limits_{j=1,\dots,p_2} \left\lvert(\widehat{k}^{(0)}_{jj} - k_{jj})/k_{jj} \right\rvert \right)^2 \left\lVert \widehat{\mathbf{C}}^{(0)\prime}\dg{\mathbf{K}}^{-1} \right\rVert^2 \\[0.1in]
			\leq  p_2^{-1} \left\lVert \widehat{\mathbf{C}}^{(0)\prime}\dg{\mathbf{K}}^{-1} -  \mathbf{C}'\dg{\mathbf{K}}^{-1} \right\rVert^2 + p_2^{-1} C^2_K \left(\min\limits_{j=1,\dots,p_2} \left\lvert \widehat{k}^{(0)}_{jj} - k_{jj} \right\rvert \right)^2  \left\lVert \mathbf{C} \widehat{\mathbf{J}}_2 \right\rVert^2 \left\lVert\dg{\mathbf{K}}^{-1} \right\rVert^2 \\[0.1in] 
			\qquad +  p_2^{-1} C^2_K \left(\min\limits_{j=1,\dots,p_2} \left\lvert \widehat{k}^{(0)}_{jj} - k_{jj} \right\rvert \right)^2 \left\lVert \widehat{\mathbf{C}}^{(0)} - \mathbf{C}\widehat{\mathbf{J}}_2 \right\rVert^2 \left\lVert\dg{\mathbf{K}}^{-1} \right\rVert^2 \\[0.1in]
			=O_p\left(\max\left\{\frac{1}{Tp_1},\frac{1}{p^2_1p^2_2},\frac{1}{T^2p^2_2}\right\}\right)
		\end{array}
		\]
		by Lemmas \ref{lemma:tildeBound}(iii), \ref{lemma:boundRC}, \ref{lemma:boundkh}(iii) and Assumption \ref{subass:idshocks}. Part (iv) follows analogously.
		Consider (v) and note that 
		\[
		\begin{array}{rll}
			\det \left(\widehat{\mathbf{C}}^{(0)\prime}\widehat{\mathbf{K}}^{(0)-1}\widehat{\mathbf{C}}^{(0)}\right)^{-1} &=& \prod\limits_{j=1}^{k_2} \nu^{(j)}\left( \widehat{\mathbf{C}}^{(0)\prime}\widehat{\mathbf{K}}^{(0)-1}\widehat{\mathbf{C}}^{(0)}\right) \\[0.1in]
			&\geq& \left(\nu^{(k_2)}\left( \widehat{\mathbf{C}}^{(0)\prime}\widehat{\mathbf{K}}^{(0)-1}\widehat{\mathbf{C}}^{(0)}\right)\right)^{k_2} \\[0.1in]
			&\geq& \left( \nu^{(k_2)}\left( \mathbf{C}'\dg{\mathbf{K}}^{-1}\mathbf{C}\right) - \left| \nu^{(k_2)}\left( \widehat{\mathbf{C}}^{(0)\prime}\widehat{\mathbf{K}}^{(0)-1}\widehat{\mathbf{C}}^{(0)}\right)- \nu^{(k_2)}\left( \mathbf{C}'\dg{\mathbf{K}}^{-1}\mathbf{C}\right)\right| \right)^{k_2}
		\end{array}
		\]
		From Lemma \ref{lemma:tildeBound}(iv), we have that $\lim\limits_{p_2\rightarrow\infty}p^{-1}_2\nu^{(k_2)}\left( \mathbf{C}'\dg{\mathbf{K}}^{-1}\mathbf{C}\right)>0$. Moreover,
		\[
		\begin{array}{rll}
			\frac{1}{p_2}\left| \nu^{(k_2)}\left( \widehat{\mathbf{C}}^{(0)\prime}\widehat{\mathbf{K}}^{(0)-1}\widehat{\mathbf{C}}^{(0)}\right)- \nu^{(k_2)}\left( \mathbf{C}'\mathbf{K}^{-1}\mathbf{C}\right)\right| &\leq& \frac{1}{p_2} \left\lVert \widehat{\mathbf{C}}^{(0)\prime}\widehat{\mathbf{K}}^{(0)-1}\widehat{\mathbf{C}}^{(0)} - \mathbf{C}'\dg{\mathbf{K}}\mathbf{C} \right\rVert \\[0.1in]
			& \lesssim & \frac{1}{p_2} \left\lVert \widehat{\mathbf{C}}^{(0)\prime}\widehat{\mathbf{K}}^{(0)-1} - \mathbf{C}'\dg{\mathbf{K}} \right\rVert \left\lVert \mathbf{C} \right\rVert  \\[0.1in]
			&&  + \frac{1}{p_2}\left\lVert \mathbf{C}'\dg{\mathbf{K}}^{-1} \right\rVert \left\lVert \widehat{\mathbf{C}}^{(0)}-\mathbf{C}\widehat{\mathbf{J}}_2\right\rVert \\[0.1in]
			&=& O_p\left(\max\left\{\frac{1}{\sqrt{Tp_1}},\frac{1}{p_1p_2},\frac{1}{Tp_2}\right\}\right)
		\end{array}
		\]
		by Lemmas \ref{lemma:tildeBound}(iii), \ref{lemma:tildeBound}(v), Proposition \ref{prop:EMloadCons} and term (iii), implyingthat $\det \left(p^{-1}_2\widehat{\mathbf{C}}^{(0)\prime}\widehat{\mathbf{K}}^{(0)-1}\widehat{\mathbf{C}}^{(0)}\right)^{-1} > 0$ with probability tending to one as $\min\{T,p_1,p_2\}$ goes to infinity, i.e. $p_2\left\lVert\left(\widehat{\mathbf{C}}^{(0)\prime}\widehat{\mathbf{K}}^{(0)-1}\widehat{\mathbf{C}}^{(0)}\right)^{-1} \right\rVert=O_p(1)$. The proof for (vi) follows the same steps. 
		Consider (vii), we have that
		\[
		\begin{array}{rll}
			p_2\left\lVert \left(\widehat{\mathbf{C}}^{(0)\prime}\widehat{\mathbf{K}}^{(0)-1}\widehat{\mathbf{C}}^{(0)}\right)^{-1} - \left( \mathbf{C}'\mathbf{K}^{-1}\mathbf{C} \right)^{-1} \right\rVert & \leq & p_2\left\lVert\left(\widehat{\mathbf{C}}^{(0)\prime}\widehat{\mathbf{K}}^{(0)-1}\widehat{\mathbf{C}}^{(0)}\right)^{-1} \right\rVert p_2  \left\lVert \left( \mathbf{C}'\dg{\mathbf{K}}^{-1}\mathbf{C} \right)^{-1}  \right\rVert \\[0.1in]
			&&  \frac{1}{p_2} \left\lVert \widehat{\mathbf{C}}^{(0)\prime}\widehat{\mathbf{K}}^{(0)-1}\widehat{\mathbf{C}}^{(0)}- \mathbf{C}'\dg{\mathbf{K}}^{-1}\mathbf{C}  \right\rVert \\[0.1in]
			&=& O_p\left(\max\left\{ \frac{1}{\sqrt{Tp_1}}, \frac{1}{p_1p_2}, \frac{1}{Tp_2}\right\} \right)
		\end{array}
		\]
		because of Lemma \ref{lemma:tildeBound}(iv), term (v) and since $ \frac{1}{p_2} \left\lVert \widehat{\mathbf{C}}^{(0)\prime}\widehat{\mathbf{K}}^{(0)-1}\widehat{\mathbf{C}}^{(0)}- \mathbf{C}'\mathbf{K}^{-1}\mathbf{C}  \right\rVert=O_p\left(\max\left\{ \frac{1}{\sqrt{Tp_1}}, \frac{1}{p_1p_2}, \frac{1}{Tp_2}\right\} \right)$ by the same steps used in the proof of term (iii). Proof for (viii) follows analogously. The proofs for (ix) and (x) follow directly from (iii) and (vii), and from (iv) and (viii), respectively.
	\end{proof}

	\subsubsection{Results on EM estimators}
	\begin{lemma}
		Under Assumption \eqref{ass:common} through \eqref{ass:dep},  for all $n \in \mathbb{N}_+$, as $\min\left\{p_1,p_2,T\right\} \rightarrow \infty$,
		\begin{itemize}
			\item[(i)] $\left\lvert\widehat{k}^{(n)}_{jj} - k_{jj} \right\rvert = O_p\left(\max\left\{\frac{1}{\sqrt{Tp_1}}, \frac{1}{\sqrt{Tp_2}}, \frac{1}{p_1p_2} \right\}\right)$ uniformly in $j$
			\item[(ii)] $\left\lvert\widehat{h}^{(n)}_{ii} - h_{ii} \right\rvert = O_p\left(\max\left\{\frac{1}{\sqrt{Tp_1}}, \frac{1}{\sqrt{Tp_2}}, \frac{1}{p_1p_2} \right\}\right)$ uniformly in $i$
			\item[(iii)] $\frac{1}{p_2}\left\lvert \sum\limits_{j=1}^{p_2}\left(\widehat{k}^{(n)}_{jj} - k_{jj}\right) \right\rvert = O_p\left(\max\left\{\frac{1}{\sqrt{Tp_1}}, \frac{1}{\sqrt{Tp_2}}, \frac{1}{p_1p_2} \right\}\right)$
			\item[(iv)] $\frac{1}{p_1}\left\lvert \sum\limits_{i=1}^{p_1}\left(\widehat{h}^{(n)}_{ii} - h_{ii}\right) \right\rvert =O_p\left(\max\left\{\frac{1}{\sqrt{Tp_1}}, \frac{1}{\sqrt{Tp_2}}, \frac{1}{p_1p_2} \right\}\right)$
		\end{itemize}
		\label{lemma:boundkh1}
	\end{lemma}

	\begin{proof}
		Consider (ii) and recall that 
		\[

		\]
		For term $I$ we have that
		\[
		%
		\]
		by Assumption \ref{subass:factor}, Lemmas \ref{lemma:starnorm},  \ref{lemma:tildeBound}(iii), \ref{lemma:est.loadscale}(i) and Proposition \ref{prop:EMloadCons}. For terms $II$ and $III$ we have that
		\[
		%
		\]
		by Assumptions \ref{subass:idshocks}, \ref{subass:ecorr2} and Lemma \ref{lemma:boundkh}(iii).	The proof for (i) follows the same steps. Consider (iv), then
		\[
		%
		\]
		For term $I$ we have that
		\[
		%
		\]
		by Assumption \ref{subass:factor}, Lemmas \ref{lemma:tildeBound}(iii), \ref{lemma:est.loadscale}(i), and Proposition \ref{prop:EMloadCons}. For terms $II$ and $III$ we have that
		\[
		%
		\]
		by Assumptions \ref{subass:idshocks}, \ref{subass:ecorr2} and Lemma \ref{lemma:boundkh}(i).	Repeating the same steps for all $n \in \mathbb{N}_+$, replacing Lemmas \ref{lemma:boundkh} and \ref{lemma:est.loadscale} with Lemmas \ref{lemma:boundkh1} and \ref{lemma:est.loadscale.1}, respectively, completes the proof.
	\end{proof}

	\begin{lemma}
		Under Assumptions \ref{ass:common}-\ref{ass:idio}, for all $n \in \mathbb{N}_+$, we have that as $\min\left\{T,p_1,p_2\right\}\rightarrow\infty$
		\begin{itemize}
			\item[(i)] $\left\lVert\widehat{\mathbf{K}}^{(n)-1}\right\rVert=O_p(1)$ 
			\item[(ii)] $\left\lVert\widehat{\mathbf{H}}^{(n)-1}\right\rVert=O_p(1)$
			\item[(iii)] $\frac{1}{\sqrt{p_2}}\left\lVert \widehat{\mathbf{C}}^{(n)\prime}\widehat{\mathbf{K}}^{(n)-1} -  \mathbf{C}' \dg{\mathbf{K}}^{-1} \right\rVert=O_p\left(\max\left\{\frac{1}{\sqrt{Tp_1}},\frac{1}{\sqrt{Tp_2}},\frac{1}{p_1p_2}\right\}\right)$
			\item[(iv)]$\frac{1}{\sqrt{p_1}}\left\lVert  \widehat{\mathbf{R}}^{(n)\prime}\widehat{\mathbf{H}}^{(n)-1} -  \mathbf{R}'\dg{\mathbf{H}}^{-1} \right\rVert=O_p\left(\max\left\{\frac{1}{\sqrt{Tp_1}},\frac{1}{\sqrt{Tp_2}},\frac{1}{p_1p_2}\right\}\right))$
			\item[(v)] $p_2\left\lVert\left(\widehat{\mathbf{C}}^{(n)\prime}\widehat{\mathbf{K}}^{(n)-1}\widehat{\mathbf{C}}^{(n)}\right)^{-1} \right\rVert=O_p(1)$ 
			\item[(vi)] $p_1\left\lVert\left(\widehat{\mathbf{R}}^{(n)\prime}\widehat{\mathbf{H}}^{(n)-1}\widehat{\mathbf{R}}^{(n)}\right)^{-1} \right\rVert=O_p(1)$
			\item[(vii)]  $p_2$$\left\lVert \left(\widehat{\mathbf{C}}^{(n)\prime}\widehat{\mathbf{K}}^{(n)-1}\widehat{\mathbf{C}}^{(n)}\right)^{-1} - \left( \mathbf{C}'\dg{\mathbf{K}}^{-1}\mathbf{C} \right)^{-1}
			\right\rVert=O_p\left(\max\left\{\frac{1}{\sqrt{Tp_1}},\frac{1}{\sqrt{Tp_2}},\frac{1}{p_1p_2}\right\}\right)$
			\item[(viii)]  $p_1\left\lVert \left(\widehat{\mathbf{R}}^{(n)\prime}\widehat{\mathbf{H}}^{(n)-1}\widehat{\mathbf{R}}^{(n)}\right)^{-1} - \left( \mathbf{R}'\dg{\mathbf{H}}^{-1}\mathbf{R} \right)^{-1} \right\rVert=O_p\left(\max\left\{\frac{1}{\sqrt{Tp_1}},\frac{1}{\sqrt{Tp_2}},\frac{1}{p_1p_2}\right\}\right)$
			\item[(ix)]  $\sqrt{p_2}\left\lVert \left(\widehat{\mathbf{C}}^{(n)\prime}\widehat{\mathbf{K}}^{(n)-1}\widehat{\mathbf{C}}^{(n)}\right)^{-1}\widehat{\mathbf{C}}^{(n)\prime}\widehat{\mathbf{K}}^{(n)-1} - \left( \mathbf{C}'\dg{\mathbf{K}}^{-1}\mathbf{C} \right)^{-1} \mathbf{C}'\dg{\mathbf{K}}^{-1} \right\rVert=O_p\left(\max\left\{\frac{1}{\sqrt{Tp_1}},\frac{1}{\sqrt{Tp_2}},\frac{1}{p_1p_2}\right\}\right)$ 
			\item[(x)] $\sqrt{p_1}\left\lVert \left(\widehat{\mathbf{R}}^{(n)\prime}\widehat{\mathbf{H}}^{(n)-1}\widehat{\mathbf{R}}^{(n)}\right)^{-1}\widehat{\mathbf{R}}^{(n)\prime}\widehat{\mathbf{H}}^{(n)-1} - \left( \mathbf{R}'\dg{\mathbf{H}}^{-1}\mathbf{R} \right)^{-1} \mathbf{R}'\dg{\mathbf{H}}^{-1} \right\rVert=O_p\left(\max\left\{\frac{1}{\sqrt{Tp_1}},\frac{1}{\sqrt{Tp_2}},\frac{1}{p_1p_2}\right\}\right)$
		\end{itemize}
		\label{lemma:est.loadscale.1} 
	\end{lemma}
	\begin{proof}
		Consider (i) and note that
		\[
		\begin{array}{lcl}
			\left\lVert\widehat{\mathbf{K}}^{(n)-1}\right\rVert &=& \left\{\nu^{(p_2)}\left(\widehat{\mathbf{K}}^{(n)}\right)\right\}^{-1} \\[0.1in]
			&=& \left\{ \min\limits_{j = 1,\dots,p_2} k_{jj} + \widehat{k}^{(n)}_{jj} - k_{jj}  \right\}^{-1} \\[0.1in]
			&=& \left\{ \min\limits_{j = 1,\dots,p_2} k_{jj} - \min\limits_{j = 1,\dots,p_2} \left\lvert \widehat{k}^{(n)}_{jj} - k_{jj} \right\rvert  \right\}^{-1} \\[0.1in]
			&=& \left\{ C^{-1}_K -  \left\lvert \widehat{k}^{(n)}_{jj} - k_{jj} \right\rvert  \right\}^{-1} = C_K + O_p\left(\max\left\{\frac{1}{\sqrt{Tp_1}},\frac{1}{\sqrt{Tp_2}},\frac{1}{p_1p_2}\right\}\right)
		\end{array}	
		\]
		by Assumption \ref{subass:idshocks} and Lemma \ref{lemma:boundkh1}(i). The proof for (ii) follows the same steps.  Consider (iii), we have that
		\[
		\begin{array}{rll}
			\frac{1}{\sqrt{p_2}}\left\lVert \widehat{\mathbf{C}}^{(n)\prime}\widehat{\mathbf{K}}^{(n)-1} -  \mathbf{C}'\dg{\mathbf{K}}^{-1} \right\rVert &\lesssim& \frac{1}{\sqrt{p_2}}\left\lVert  \widehat{\mathbf{C}}^{(n)} - \mathbf{C}\widehat{\mathbf{J}}_2   \right\rVert \left\lVert \dg{\mathbf{K}}^{-1} \right\rVert + \frac{1}{\sqrt{p_2}} \left\lVert\mathbf{C}' \left(\widehat{\mathbf{K}}^{(n)-1} -\dg{\mathbf{K}}^{-1}\right) \right\rVert \\
			& = & O_p\left(\max\left\{\frac{1}{\sqrt{Tp_1}}, \frac{1}{\sqrt{Tp_2}},\frac{1}{p_1p_2}\right\}\right)
		\end{array}
		\]
		because of Proposition \ref{prop:EMloadCons}, Assumption \ref{subass:idshocks}, and since
		\[
		\begin{array}{rll}
			\frac{1}{\sqrt{p_2}} \left\lVert\mathbf{C}' \left(\widehat{\mathbf{K}}^{(n)-1} -\dg{\mathbf{K}}^{-1}\right) \right\rVert&\leq& \frac{1}{\sqrt{p_2}} \left\{ \sum\limits_{i=1}^{k_2}\sum\limits_{j=1}^{p_2} \left\lvert\left[\mathbf{C}' \left(\widehat{\mathbf{K}}^{(n)-1} -\dg{\mathbf{K}}^{-1}\right) \right]_{ij} \right\rvert^2 \right\}^{\frac{1}{2}} \\[0.1in]
			&\leq& \sqrt{k_2} \max\limits_{ij} \left\lvert\mathbf{c}'_{\cdot i} \left[\widehat{\mathbf{K}}^{(n)-1} -\dg{\mathbf{K}}^{-1} \right]_{\cdot j} \right\rvert \\[0.1in]
			&\leq& \bar{c}\sqrt{k_2}  \max\limits_{j}  \left\lvert \left\{\widehat{k}^{(n)-1}_{jj}k^{-1}_{jj} \left(\widehat{k}^{(n)}_{jj} - k_{jj}\right) \right\} \right\rvert \\[0.1in]
			&\leq&  \bar{c}\sqrt{k_2}  \max\limits_{j}   \left\lvert \left(\min\limits_{j}\widehat{k}^{(n)}_{jj}\right)^{-1} \left(\min\limits_{j}k_{jj}\right)^{-1} \sum\limits_{j=1}^{p_2} \left(\widehat{k}^{(n)}_{jj} - k_{jj}\right) \right\rvert \\[0.1in]
			&\leq& \bar{c}\sqrt{k_2}C^2_K\left\lvert \widehat{k}^{(n)}_{jj} - k_{jj} \right\rvert \\[0.1in] 
			&=& O_p\left(\max\left\{\frac{1}{\sqrt{Tp_1}},\frac{1}{\sqrt{Tp_2}},\frac{1}{p_1p_2},\right\}\right)
		\end{array}
		\]
		by Assumptions \ref{subass:loading}, \ref{subass:idshocks} and Lemma \ref{lemma:boundkh1}(i). The proof for (iv) follows the same steps. 
		Consider (v) and note that 
		\[
		\begin{array}{rll}
			\det \left(\widehat{\mathbf{C}}^{(n)\prime}\widehat{\mathbf{K}}^{(n)-1}\widehat{\mathbf{C}}^{(0)}\right)^{-1} &=& \prod\limits_{j=1}^{k_2} \nu^{(j)}\left( \widehat{\mathbf{C}}^{(n)\prime}\widehat{\mathbf{K}}^{(n)-1}\widehat{\mathbf{C}}^{(n)}\right) \\[0.1in]
			&\geq& \left(\nu^{(k_2)}\left( \widehat{\mathbf{C}}^{(n)\prime}\widehat{\mathbf{K}}^{(n)-1}\widehat{\mathbf{C}}^{(n)}\right)\right)^{k_2} \\[0.1in]
			&\geq& \left( \nu^{(k_2)}\left( \mathbf{C}'\dg{\mathbf{K}}^{-1}\mathbf{C}\right) - \left| \nu^{(k_2)}\left( \widehat{\mathbf{C}}^{(n)\prime}\widehat{\mathbf{K}}^{(n)-1}\widehat{\mathbf{C}}^{(n)}\right)- \nu^{(k_2)}\left( \mathbf{C}'\dg{\mathbf{K}}^{-1}\mathbf{C}\right)\right| \right)^{k_2}
		\end{array}
		\]
		From Lemma \ref{lemma:tildeBound}(iv), we have that $\lim\limits_{p_2\rightarrow\infty}p^{-1}_2\nu^{(k_2)}\left( \mathbf{C}'\dg{\mathbf{K}}^{-1}\mathbf{C}\right)>0$. Moreover,
		\[
		\begin{array}{rll}
			\frac{1}{p_2}\left| \nu^{(k_2)}\left( \widehat{\mathbf{C}}^{(n)\prime}\widehat{\mathbf{K}}^{(n)-1}\widehat{\mathbf{C}}^{(n)}\right)- \nu^{(k_2)}\left( \mathbf{C}'\mathbf{K}^{-1}\mathbf{C}\right)\right| &\leq& \frac{1}{p_2} \left\lVert \widehat{\mathbf{C}}^{(n)\prime}\widehat{\mathbf{K}}^{(n)-1}\widehat{\mathbf{C}}^{(n)} - \mathbf{C}'\dg{\mathbf{K}}\mathbf{C} \right\rVert \\[0.1in]
			& \lesssim & \frac{1}{p_2} \left\lVert \widehat{\mathbf{C}}^{(n)\prime}\widehat{\mathbf{K}}^{(n)-1} - \mathbf{C}'\dg{\mathbf{K}} \right\rVert \left\lVert \mathbf{C} \right\rVert  \\[0.1in]
			&&  + \frac{1}{p_2}\left\lVert \mathbf{C}'\dg{\mathbf{K}}^{-1} \right\rVert \left\lVert \widehat{\mathbf{C}}^{(n)}-\mathbf{C}\widehat{\mathbf{J}}_2\right\rVert \\[0.1in]
			&=& O_p\left(\max\left\{\frac{1}{\sqrt{Tp_1}},\frac{1}{\sqrt{Tp_2}},\frac{1}{p_1p_2}\right\}\right)
		\end{array}
		\]
		by Lemmas \ref{lemma:tildeBound}(iii), \ref{lemma:tildeBound}(v), Proposition \ref{prop:EMloadCons} and term (iii), implying that $\det \left(p^{-1}_2\widehat{\mathbf{C}}^{(n)\prime}\widehat{\mathbf{K}}^{(n)-1}\widehat{\mathbf{C}}^{(n)}\right)^{-1} > 0$ with probability tending to one as $\min\{T,p_1,p_2\}$ goes to infinity, i.e. $p_2\left\lVert\left(\widehat{\mathbf{C}}^{(n)\prime}\widehat{\mathbf{K}}^{(n)-1}\widehat{\mathbf{C}}^{(n)}\right)^{-1} \right\rVert=O_p(1)$. The proof for (vi) follows the same steps. 
		Consider (vii), we have that
		\[
		\begin{array}{rll}
			p_2\left\lVert \left(\widehat{\mathbf{C}}^{(n)\prime}\widehat{\mathbf{K}}^{(n)-1}\widehat{\mathbf{C}}^{(n)}\right)^{-1} - \left( \mathbf{C}'\mathbf{K}^{-1}\mathbf{C} \right)^{-1} \right\rVert & \leq & p_2\left\lVert\left(\widehat{\mathbf{C}}^{(n)\prime}\widehat{\mathbf{K}}^{(n)-1}\widehat{\mathbf{C}}^{(n)}\right)^{-1} \right\rVert p_2  \left\lVert \left( \mathbf{C}'\dg{\mathbf{K}}^{-1}\mathbf{C} \right)^{-1}  \right\rVert \\[0.1in]
			&&  \frac{1}{p_2} \left\lVert \widehat{\mathbf{C}}^{(n)\prime}\widehat{\mathbf{K}}^{(n)-1}\widehat{\mathbf{C}}^{(n)}- \mathbf{C}'\dg{\mathbf{K}}^{-1}\mathbf{C}  \right\rVert \\[0.1in]
			&=& O_p\left(\max\left\{ \frac{1}{\sqrt{Tp_1}},  \frac{1}{\sqrt{Tp_2}} \frac{1}{p_1p_2}\right\} \right)
		\end{array}
		\]
		because of Lemma \ref{lemma:tildeBound}(iv), term (v) and since $ \frac{1}{p_2} \left\lVert \widehat{\mathbf{C}}^{(n)\prime}\widehat{\mathbf{K}}^{(n)-1}\widehat{\mathbf{C}}^{(n)}- \mathbf{C}'\mathbf{K}^{-1}\mathbf{C}  \right\rVert=O_p\left(\max\left\{ \frac{1}{\sqrt{Tp_1}}, \frac{1}{\sqrt{Tp_2}} \frac{1}{p_1p_2}\right\} \right)$. Proof for (viii) follows analogously. The proofs for (ix) and (x) follow directly from (iii) and (vii), and from (iv) and (viii), respectively.
	\end{proof}

	\begin{lemma}
		Under Assumptions \ref{ass:common} through \ref{ass:dep}, there exist matrices $\widehat{\mathbf{J}}_1$ and $\widehat{\mathbf{J}}_2$ satisfying $\widehat{\mathbf{J}}_1\widehat{\mathbf{J}}'_1 \xrightarrow{p} \mathbb{I}_{k_1k_1}$ and $\widehat{\mathbf{J}}_2\widehat{\mathbf{J}}'_2 \xrightarrow{p} \mathbb{I}_{k_2k_2} $, such that, for all $n \in \mathbb{N}$, as $\min\left\{p_1,p_2,T\right\} \rightarrow \infty$,
		\begin{itemize}
			\item[(i)] $\left(\frac{1}{T}\sum\limits_{t=1}^{T}\left(\mathsf{f}^{(n)}_{t|T} - \widehat{\mathbf{J}}^{-1} \mathsf{f}_t\right)  \mathsf{f}'_t\mathbf{\widehat{J}}\right) =  O_p\left(\max \left\{ \frac{1}{\sqrt{Tp_1}}, \frac{1}{\sqrt{Tp_2}}, \frac{1}{p_1p_2} \right\} \right)$
			\item[(ii)] $\left\lVert\frac{1}{Tp_2}\sum\limits_{j=1}^{p_2}\sum\limits_{t=1}^{T}e_{tij} \left(\mathsf{f}^{(n)}_{t|T} - \widehat{\mathbf{J}}^{-1} \mathsf{f}_t\right) \right\rVert = O_p\left(\max \left\{ \frac{1}{\sqrt{Tp_1}}, \frac{1}{\sqrt{Tp_2}}, \frac{1}{p_1p_2} \right\} \right)$
			\item[(iii)] $\left\lVert \frac{1}{Tp_1} \sum\limits_{i=1}^{p_1}\sum\limits_{t=1}^{T} e_{tij}(\mathsf{f}^{(n)}_{t|T}-\widehat{\mathbf{J}}^{-1}\mathsf{f}_t) \right\rVert = O_p\left(\max \left\{ \frac{1}{\sqrt{Tp_1}}, \frac{1}{\sqrt{Tp_2}}, \frac{1}{p_1p_2} \right\} \right)$
			\item[(iv)] $\left\lVert\frac{1}{T\sqrt{p_1}p_2}\sum\limits_{i=1}^{p_2}\sum\limits_{t=1}^{T}\left(\mathsf{f}^{(n)}_{t|T} - \widehat{\mathbf{J}}^{-1} \mathsf{f}_t\right)  \mathbf{e}'_{t\cdot i}\right\rVert = O_p\left(\max \left\{ \frac{1}{\sqrt{Tp_1}}, \frac{1}{\sqrt{Tp_2}}, \frac{1}{p_1p_2} \right\} \right)$
			\item[(v)] $\left\lVert \frac{1}{Tp_1\sqrt{p_2}} \sum\limits_{i=1}^{p_1}\sum\limits_{t=1}^{T} \mathbf{e}_{ti\cdot}(\mathsf{f}^{(n)}_{t|T}-\widehat{\mathbf{J}}^{-1}\mathsf{f}_t)' \right\rVert = O_p\left(\max \left\{ \frac{1}{\sqrt{Tp_1}}, \frac{1}{\sqrt{Tp_2}}, \frac{1}{p_1p_2} \right\} \right)$
		\end{itemize}
		\label{lemma:crossf}
	\end{lemma}
	\begin{proof}
		Let $\mathsf{x}_t= \left\{\mathsf{f}_t\mathbf{\widehat{J}}, \frac{1}{p_2}\sum\limits_{j=1}^{p_2}e_{tij}, \frac{1}{p_1}\sum\limits_{i=1}^{p_1} e_{tij}, \frac{1}{\sqrt{p_1}p_2}\sum\limits^{p_2}_{i=1} \mathbf{e}_{t\cdot i}, \frac{1}{p_1\sqrt{p_2}}\sum\limits_{i=1}^{p_1}\mathbf{e}_{ti\cdot}\right\}$ . From (B.4) in \cite{barigozzi2019quasi}, we have that
		\[
		\begin{array}{rll}
			\left\lVert \frac{1}{T}\sum\limits_{t=1}^{T}\left(\mathsf{f}^{(n)}_t - \widehat{\mathbf{J}}^{-1} \mathsf{f}_t\right)  \mathsf{x}^{\prime}_t \right\rVert &\leq& \left\lVert \frac{1}{T}\sum\limits_{t=1}^{T}\left(\mathsf{f}^{(n)}_{t|T} - \mathsf{f}^{(n)}_{t|t}\right)  \mathsf{x}^{\prime}_t   \right\rVert \\[0.1in]
			&& + \left\lVert \frac{1}{T}\sum\limits_{t=1}^{T}\left(\mathsf{f}^{(n)}_{t|t} - \mathsf{f}^{LS(n)}_{t}\right)  \mathsf{x}^{\prime}_t  \right\rVert \\[0.1in]
			&& + \left\lVert \frac{1}{T}\sum\limits_{t=1}^{T}\left(\mathsf{f}^{LS(n)}_{t}-\widehat{\mathbf{J}}^{-1} \mathsf{f}_t\right)  \mathsf{x}^{\prime}_t\right\rVert \\[0.1in]
			&=& I + II + III
		\end{array}
		\]
		where
		\[
		\mathsf{f}^{LS(n)}_{t} = \left( \left(\left(\widehat{\mathbf{C}}^{(n)\prime}\widehat{\mathbf{K}}^{(n)-1}\widehat{\mathbf{C}}^{(n)}\right)^{-1}\widehat{\mathbf{C}}^{(n)\prime}\widehat{\mathbf{K}}^{(n)-1}\right)\otimes\left(\left(\widehat{\mathbf{R}}^{(n)\prime}\widehat{\mathbf{H}}^{(n)-1}\widehat{\mathbf{R}}^{(n)}\right)^{-1}\widehat{\mathbf{R}}^{(n)\prime}\widehat{\mathbf{H}}^{(n)-1}\right)\right)\mathsf{y}_t
		\]
		Consider the case $n=0$. From Lemmas \ref{lemma:boundRC}, \ref{lemma:boundAB}, \ref{lemma:est.loadscale}, (B.5) and (B.6) in  \cite{barigozzi2019quasi}, it follows that terms I and II are both $O_p\left(\frac{1}{p_1p_2}\right)$. Focusing on the third term, we have
		\[
		\begin{array}{rll}
			III &\leq& \left\lVert \frac{1}{T}\sum^T\limits_{t=1} \left( \left(\left(\widehat{\mathbf{C}}^{(0)\prime}\widehat{\mathbf{K}}^{(0)-1}\widehat{\mathbf{C}}^{(0)}\right)^{-1}\widehat{\mathbf{C}}^{(0)\prime}\widehat{\mathbf{K}}^{(0)-1}\right)\otimes\left(\left(\widehat{\mathbf{R}}^{(0)\prime}\widehat{\mathbf{H}}^{(0)-1}\widehat{\mathbf{R}}^{(0)}\right)^{-1}\widehat{\mathbf{R}}^{(0)\prime}\widehat{\mathbf{H}}^{(0)-1}\right)\right) \right. \\[0.1in]
			&& \hspace*{1.5in}\left. \times \left( \mathbf{C}\widehat{\mathbf{J}}_2 \otimes \mathbf{R}\widehat{\mathbf{J}}_1 - \widehat{\mathbf{C}}^{(0)} \otimes \widehat{\mathbf{R}}^{(0)} \right) \widehat{\mathbf{J}}^{-1}\mathsf{f}_t \mathsf{x}'_t \right\rVert \\[0.1in]
			&& + \left\lVert \frac{1}{T}\sum^T\limits_{t=1} \left( \left(\left(\widehat{\mathbf{C}}^{(0)\prime}\widehat{\mathbf{K}}^{(0)-1}\widehat{\mathbf{C}}^{(0)}\right)^{-1}\widehat{\mathbf{C}}^{(0)\prime}\widehat{\mathbf{K}}^{(0)-1}\right)\otimes\left(\left(\widehat{\mathbf{R}}^{(0)\prime}\widehat{\mathbf{H}}^{(0)-1}\widehat{\mathbf{R}}^{(0)}\right)^{-1}\widehat{\mathbf{R}}^{(0)\prime}\widehat{\mathbf{H}}^{(0)-1}\right)\right)\mathsf{e}_t\mathsf{x}'_t  \right\rVert \\[0.1in]
			&\leq& \left\lVert \left( \left(\left({\mathbf{C}}^{\prime}\dg{\mathbf{K}}^{-1}\mathbf{C}\right)^{-1}\mathbf{C}^{\prime}\mathbf{K}^{-1}\right)\otimes\left(\left(\mathbf{R}^{\prime}\dg{\mathbf{H}}^{-1}\mathbf{R}\right)^{-1}\mathbf{R}^{\prime}\dg{\mathbf{H}}^{-1}\right)\right) \right. \\[0.1in]
			&& \hspace*{1.5in}\left. \times \left( \mathbf{C}\widehat{\mathbf{J}}_2 \otimes \mathbf{R}\widehat{\mathbf{J}}_1 - \widehat{\mathbf{C}}^{(0)} \otimes \widehat{\mathbf{R}}^{(0)} \right) \right\rVert \left\lVert \frac{1}{T}\sum^T\limits_{t=1}  \widehat{\mathbf{J}}^{-1}\mathsf{f}_t \mathsf{x}'_t \right\rVert \\[0.1in]
			&& + \left\lVert \left( \left(\left(\widehat{\mathbf{C}}^{(0)\prime}\widehat{\mathbf{K}}^{(0)-1}\widehat{\mathbf{C}}^{(0)}\right)^{-1}\widehat{\mathbf{C}}^{(0)\prime}\widehat{\mathbf{K}}^{(0)-1}\right)\otimes\left(\left(\widehat{\mathbf{R}}^{(0)\prime}\widehat{\mathbf{H}}^{(0)-1}\widehat{\mathbf{R}}^{(0)}\right)^{-1}\widehat{\mathbf{R}}^{(0)\prime}\widehat{\mathbf{H}}^{(0)-1}\right)\right)  \right. \\[0.1in]
			&& \left. \hspace*{1.5in} - \left(\left(\mathbf{C}'\mathbf{K}^{-1}\mathbf{C}\right)^{-1}\mathbf{C}'\mathbf{K}^{-1}\right) \otimes \left(\left(\mathbf{R}^{\prime}\mathbf{H}^{-1}\mathbf{R}\right)^{-1}\mathbf{R}^{\prime}\mathbf{H}^{-1} \right) \right\rVert \\[0.1in]
			&& \hspace*{0.5in} \times  \left\lVert  \mathbf{C}\widehat{\mathbf{J}}_2 \otimes \mathbf{R}\widehat{\mathbf{J}}_1 - \widehat{\mathbf{C}} \otimes \widehat{\mathbf{R}} \right\rVert \left\lVert \frac{1}{T}\sum^T\limits_{t=1} \widehat{\mathbf{J}}^{-1} \mathsf{f}_t \mathsf{x}'_t \right\rVert \\[0.1in]
			&& + \left\lVert  \left(\left(\mathbf{C}'\mathbf{K}^{-1}\mathbf{C}\right)^{-1}\right)\otimes\left(\left(\mathbf{R}'\mathbf{H}^{-1}\mathbf{R}\right)^{-1}\right) \right\rVert  \left\lVert\frac{1}{T}\sum\limits_{t=1}^T \left(\mathbf{C}'\mathbf{K}^{-1}\right)\otimes\left(\mathbf{R}'\mathbf{H}^{-1} \right) \mathsf{e}_t\mathsf{x}'_t  \right\rVert \\[0.1in]
			&& +  \left\lVert \left( \left(\left(\widehat{\mathbf{C}}^{(0)\prime}\widehat{\mathbf{K}}^{(0)-1}\widehat{\mathbf{C}}^{(0)}\right)^{-1}\widehat{\mathbf{C}}^{(0)\prime}\widehat{\mathbf{K}}^{(0)-1}\right)\otimes\left(\left(\widehat{\mathbf{R}}^{(0)\prime}\widehat{\mathbf{H}}^{(0)-1}\widehat{\mathbf{R}}^{(0)}\right)^{-1}\widehat{\mathbf{R}}^{(0)\prime}\widehat{\mathbf{H}}^{(0)-1}\right)\right)  \right. \\[0.1in]
			&& \left. \hspace*{1.5in} - \left(\left(\mathbf{C}'\mathbf{K}^{-1}\mathbf{C}\right)^{-1}\mathbf{C}'\mathbf{K}^{-1}\right) \otimes \left(\left(\mathbf{R}^{\prime}\mathbf{H}^{-1}\mathbf{R}\right)^{-1}\mathbf{R}^{\prime}\mathbf{H}^{-1} \right) \right\rVert \left\lVert \frac{1}{T}\sum^T\limits_{t=1} \mathsf{e}_t \mathsf{x}'_t \right\rVert \\[0.1in]
			&& = III_a + III_{b} + III_{c} + III_{d}
		\end{array}
		\]
		Since,
		\[
		\begin{array}{rll}
			III_a &\leq& \left\lVert p_2 \left(\mathbf{C}'\dg{\mathbf{K}}^{-1}\mathbf{C}\right)^{-1}\right\rVert \left\lVert\frac{\mathbf{C}' \dg{\mathbf{K}}^{-1}}{\sqrt{p_2}}\right\rVert \left\lVert p_1 \left(\mathbf{R}^{\prime}\dg{\mathbf{H}}^{-1}\mathbf{R}\right)^{-1}\right\rVert \left\lVert\frac{\mathbf{R}^{\prime}\dg{\mathbf{H}}^{-1}}{\sqrt{p_1}} \right\rVert \\[0.1in]
			&& \hspace*{0.5in}\times \frac{1}{\sqrt{p_1p_2}} \left\lVert  \mathbf{C}\widehat{\mathbf{J}}_2 \otimes \mathbf{R}\widehat{\mathbf{J}}_1 - \widehat{\mathbf{C}}^{(0)} \otimes \widehat{\mathbf{R}}^{(0)} \right\rVert \left\lVert \frac{1}{T}\sum^T\limits_{t=1}  \widehat{\mathbf{J}}^{-1} \mathsf{f}_t \mathsf{x}'_t \right\rVert \\[0.1in]
			&=& O_p\left(\max\left\{\frac{1}{\sqrt{Tp_1}},\frac{1}{\sqrt{Tp_2}},\frac{1}{p_1p_2}\right\}\right)  \left\lVert \frac{1}{T}\sum^T\limits_{t=1} \widehat{\mathbf{J}}^{-1}  \mathsf{f}_t \mathsf{x}'_t \right\rVert
		\end{array}
		\]
		by Lemmas \ref{lemma:tildeBound}(iv),  \ref{lemma:tildeBound}(v), \ref{lemma:boundRC}(iii),
		\[
		\begin{array}{rll}
			III_b &\lesssim&  \sqrt{p_2}\left\lVert\left(\widehat{\mathbf{C}}^{(0)\prime}\widehat{\mathbf{K}}^{(0)-1}\widehat{\mathbf{C}}^{(0)}\right)^{-1}\widehat{\mathbf{C}}^{(0)\prime}\widehat{\mathbf{K}}^{(0)-1}-\left(\mathbf{C}'\dg{\mathbf{K}}^{-1}\mathbf{C}\right)^{-1}\mathbf{C}'\dg{\mathbf{K}}^{-1} \right \rVert \\[0.1in]
			&& \hspace*{1in} \times \left\lVert \left(\frac{\mathbf{R}^{\prime}\dg{\mathbf{H}}^{-1}\mathbf{R}}{p_1}\right)^{-1}\right\rVert \frac{\left\lVert\mathbf{R}^{\prime}\dg{\mathbf{H}}^{-1}  \right\rVert}{\sqrt{p_1}}  \frac{\left\lVert  \mathbf{C}\widehat{\mathbf{J}}_2 \otimes \mathbf{R}\widehat{\mathbf{J}}_1 - \widehat{\mathbf{C}} \otimes \widehat{\mathbf{R}} \right\rVert}{\sqrt{p_1p_2}}  \left\lVert\frac{1}{T}\sum^T\limits_{t=1} \widehat{\mathbf{J}}^{-1}  \mathsf{f}_t \mathsf{x}'_t \right\rVert \\[0.1in]
			&& + \sqrt{p_1} \left\lVert \left(\widehat{\mathbf{R}}^{(0)\prime}\widehat{\mathbf{H}}^{(0)-1}\widehat{\mathbf{R}}^{(0)}\right)^{-1}\widehat{\mathbf{R}}^{(0)\prime}\widehat{\mathbf{H}}^{(0)-1} - \left(\mathbf{R}^{\prime}\dg{\mathbf{H}}^{-1}\mathbf{R}\right)^{-1}\mathbf{R}^{\prime}\dg{\mathbf{H}}^{-1}\right\rVert \\[0.1in]
			&& \hspace*{1in} \times \left\lVert \left(\frac{ \mathbf{C}'\dg{\mathbf{K}}^{-1}\mathbf{C}}{p_2}\right)^{-1} \right\rVert \frac{\left\lVert \mathbf{C}'\dg{\mathbf{K}}^{-1} \right \rVert}{\sqrt{p_1}} \left\lVert \frac{ \mathbf{C}\widehat{\mathbf{J}}_2 \otimes \mathbf{R}\widehat{\mathbf{J}}_1 - \widehat{\mathbf{C}} \otimes \widehat{\mathbf{R}}}{\sqrt{p_1p_2}} \right\rVert \left\lVert \frac{1}{T}\sum^T\limits_{t=1} \widehat{\mathbf{J}}^{-1}  \mathsf{f}_t \mathsf{x}'_t \right\rVert \\[0.1in]
			&=& o_p\left(\max\left\{\frac{1}{\sqrt{Tp_1}},\frac{1}{\sqrt{Tp_2}},\frac{1}{p_1p_2}\right\}\right)  \left\lVert \frac{1}{T}\sum^T\limits_{t=1}  \widehat{\mathbf{J}}^{-1} \mathsf{f}_t \mathsf{x}'_t \right\rVert
		\end{array}
		\]
		by Lemmas \ref{lemma:tildeBound}(iv),  \ref{lemma:tildeBound}(v),  \ref{lemma:boundRC}(iii),  \ref{lemma:est.loadscale}(ix) and \ref{lemma:est.loadscale}(x),
		\[
		\begin{array}{rll}
			III_c &\leq& \left\lVert \left(\frac{\mathbf{C}'\mathbf{K}^{-1}\mathbf{C}}{p_2}\right)^{-1}\right\rVert \left\lVert\left(\frac{\mathbf{R}'\mathbf{H}^{-1}\mathbf{R}}{p_1}\right)^{-1} \right\rVert  \frac{1}{p_1p_2} \left\lVert\frac{1}{T}\sum\limits_{t=1}^T \left(\mathbf{C}'\mathbf{K}^{-1}\right)\otimes\left(\mathbf{R}'\mathbf{H}^{-1} \right) \mathsf{e}_t\mathsf{x}'_t  \right\rVert \\[0.1in]
			&\lesssim&  \frac{1}{p_1p_2} \left\lVert\frac{1}{T}\sum\limits_{t=1}^T \left(\mathbf{C}'\mathbf{K}^{-1}\right)\otimes\left(\mathbf{R}'\mathbf{H}^{-1} \right) \mathsf{e}_t\mathsf{x}'_t  \right\rVert
		\end{array}
		\]
		by Lemma \ref{lemma:tildeBound}(iv), and
		\[
		\begin{array}{rll}
			III_d &\lesssim& \sqrt{p_2}\left\lVert\left(\widehat{\mathbf{C}}^{(0)\prime}\widehat{\mathbf{K}}^{(0)-1}\widehat{\mathbf{C}}^{(0)}\right)^{-1}\widehat{\mathbf{C}}^{(0)\prime}\widehat{\mathbf{K}}^{(0)-1}-\left(\mathbf{C}'\dg{\mathbf{K}}^{-1}\mathbf{C}\right)^{-1}\mathbf{C}'\dg{\mathbf{K}}^{-1} \right \rVert \\[0.1in]
			&& \hspace*{1in} \times \left\lVert \left(\frac{\mathbf{R}^{\prime}\dg{\mathbf{H}}^{-1}\mathbf{R}}{p_1}\right)^{-1}\right\rVert \frac{\left\lVert\mathbf{R}^{\prime}\dg{\mathbf{H}}^{-1}  \right\rVert}{\sqrt{p_1}}  \frac{1}{\sqrt{p_1p_2}}  \left\lVert\frac{1}{T}\sum^T\limits_{t=1} \mathsf{e}_t \mathsf{x}'_t \right\rVert \\[0.1in]
			&& + \sqrt{p_1} \left\lVert \left(\widehat{\mathbf{R}}^{(0)\prime}\widehat{\mathbf{H}}^{(0)-1}\widehat{\mathbf{R}}^{(0)}\right)^{-1}\widehat{\mathbf{R}}^{(0)\prime}\widehat{\mathbf{H}}^{(0)-1} - \left(\mathbf{R}^{\prime}\dg{\mathbf{H}}^{-1}\mathbf{R}\right)^{-1}\mathbf{R}^{\prime}\dg{\mathbf{H}}^{-1}\right\rVert \\[0.1in]
			&& \hspace*{1in} \times \left\lVert \left(\frac{ \mathbf{C}'\dg{\mathbf{K}}^{-1}\mathbf{C}}{p_2}\right)^{-1} \right\rVert \frac{\left\lVert \mathbf{C}'\dg{\mathbf{K}}^{-1} \right \rVert}{\sqrt{p_1}} \left\lVert \frac{ 1}{\sqrt{p_1p_2}} \right\rVert \left\lVert \frac{1}{T}\sum^T\limits_{t=1} \mathsf{e}_t \mathsf{x}'_t \right\rVert \\[0.1in]
			&=&O_p\left(\max\left\{\frac{1}{\sqrt{Tp_1}},\frac{1}{\sqrt{Tp_2}},\frac{1}{p_1p_2}\right\}\right) \frac{1}{\sqrt{p_1p_2}}\left\lVert\frac{1}{T}\sum\limits_{t=1}^T \mathsf{e}_t\mathsf{x}'_t  \right\rVert
		\end{array}
		\]
		by Lemmas \ref{lemma:tildeBound}(iv)-(v), \ref{lemma:boundRC}, \ref{lemma:est.loadscale}(ix)-(x), we obtain
		\[
		\begin{array}{rll}
			III &=& O_p\left(\max\left\{\frac{1}{\sqrt{Tp_1}},\frac{1}{\sqrt{Tp_2}},\frac{1}{p_1p_2}\right\}\right) \left\{   \left\lVert \frac{1}{T}\sum^T\limits_{t=1} \widehat{\mathbf{J}}^{-1}  \mathsf{f}_t \mathsf{x}'_t \right\rVert +  \frac{1}{\sqrt{p_1p_2}}\left\lVert\frac{1}{T}\sum\limits_{t=1}^T \mathsf{e}_t\mathsf{x}'_t  \right\rVert \right\} \\[0.1in]
			&& + \frac{1}{p_1p_2} \left\lVert\frac{1}{T}\sum\limits_{t=1}^T \left(\mathbf{C}'\mathbf{K}^{-1}\right)\otimes\left(\mathbf{R}'\mathbf{H}^{-1} \right) \mathsf{e}_t\mathsf{x}'_t  \right\rVert
		\end{array}
		\] 
		From Lemma \ref{lemma:dmfmt}, we have that the stochastic behavior of $\widehat{\mathbf{J}}^{-1}\mathsf{f}_t$ is equivalent to that of $\mathsf{f}_t$. Set $\mathsf{x}_t= \mathsf{f}_t$, we have that 
		\[
		\left\lVert \frac{1}{T}\sum^T\limits_{t=1}  \mathsf{f}_t \mathsf{f}'_t   \right\rVert = O_p(1)
		\]
		by Assumption \ref{subass:factor},
		\[
		\frac{1}{\sqrt{Tp_1p_2}}\left\lVert \frac{1}{\sqrt{T}}\sum^T\limits_{t=1}  \mathsf{e}_t \mathsf{f}'_t \right\rVert = O_p\left(\frac{1}{\sqrt{T}}\right)
		\]
		by Assumption \ref{subass:fecorr}
		\[
		\begin{array}{rll}
			\frac{1}{p_1p_2}\left\lVert\frac{1}{T}\sum\limits_{t=1}^T \left(\mathbf{C}'\mathbf{K}^{-1}\right)\otimes\left(\mathbf{R}'\mathbf{H}^{-1} \right) \mathsf{e}_t\mathsf{f}'_t  \right\rVert = O_p\left(\max\left\{\frac{1}{\sqrt{Tp_1p_2}} \right\}\right) 
		\end{array}
		\]
		by Assumption \ref{subass:fecorr} and Lemma \ref{lemma:tildeBound}(iii). This concludes the proof for (i). Let $\mathsf{x}_t= \frac{1}{p_2}\sum\limits_{j=1}^{p_2}e_{tij}$, we have that 
			\[
			\left\lVert \frac{1}{Tp_2}\sum\limits_{j=1}^{p_2}\sum^T\limits_{t=1}  \mathsf{f}_t e_{tij} \right\rVert = O_p(1)
			\]
			by Assumption \ref{subass:fecorr},
			\[
			\begin{array}{rll}
				\frac{1}{\sqrt{p_1p_2}}\left\lVert \frac{1}{Tp_2}\sum^{p_2}\limits_{i=1}\sum^T\limits_{t=1}  \mathsf{e}_t e_{tij} \right\rVert &=& O_p\left(\frac{1}{\sqrt{Tp_2}}\right)
			\end{array}
			\]
			since
			\[
			\begin{array}{rll}
				\mathbb{E}\left[\frac{1}{{p_1p_2}}\left\lVert \frac{1}{Tp_2}\sum^{p_2}\limits_{i=1}\sum^T\limits_{t=1}  \mathsf{e}_t e_{tij} \right\rVert^2_F \right] &=& \mathbb{E}\left[  \frac{1}{{p_1p_2}} \sum\limits_{l}^{p_1}\sum\limits_{h}^{p_2} \left\lvert \frac{1}{Tp_2}\sum\limits_{j=1}^{p_2}\sum^T\limits_{t=1}  e_{tlh} e_{tij} \right\rvert^2 \right]\\[0.1in]
				&=& \max\limits_{l}\max\limits_{h}  \mathbb{E}\left[  \left\lvert \frac{1}{Tp_2}\sum\limits_{j=1}^{p_2}\sum^T\limits_{t=1}  e_{tlh} e_{tij} \right\rvert^2 \right]\\[0.1in]
				&=& \max\limits_{l}\max\limits_{h} \frac{1}{T^2p^2_2}   \sum\limits_{j_1,j_2=1}^{p_2}\sum^T\limits_{s,t=1}  \mathbb{E}\left[ e_{tlh} e_{tij_1}e_{slh} e_{sij_2} \right]\\[0.1in]
				&=& O_p\left(\frac{1}{Tp_2}\right)
			\end{array}
			\]
			by Assumption \ref{subass:ecorr2}, and
					\[
					\begin{array}{rll}
							\frac{1}{p_1p_2}\left\lVert\frac{1}{Tp_2}\sum\limits_{j=1}^{p_2}\sum\limits_{t=1}^T \left(\mathbf{C}'\mathbf{K}^{-1}\right)\otimes\left(\mathbf{R}'\mathbf{H}^{-1} \right) \mathsf{e}_te_{tij} \right\rVert = O_p\left(\max\left\{\frac{1}{\sqrt{Tp_1p_2}} \right\}\right) 
						\end{array}
					\]
			This concludes the proof for (ii). The results for  (iii), (iv) and (v) can be established analogously. Repeating the same steps for all $n\in \mathbb{N}$ using Proposition \ref{prop:EMloadCons}, Proposition 1 (a.4)-(a.5) in \cite{barigozzi2019quasi}, and Lemma \ref{lemma:est.loadscale.1}  in place of Lemmas \ref{lemma:boundRC}, \ref{lemma:boundAB},  \ref{lemma:est.loadscale} completes the proof.		
	\end{proof}
	
	\section{Cointegrated factors and common trends}\label{sec:gianni}
	To prove \eqref{eq:gianni}, we must find a  $k_1\times k_1$ invertible matrix $\bm{\mathcal R}$ and a $k_2\times k_2$ invertible matrix $\bm{\mathcal C}$ such that
	 \begin{align}\nonumber
	 \bm{\mathcal R}\mathbf F_t \bm{\mathcal C}' =  \left(\begin{array}{cc}
	 \mathbf G_{1t}& \mathbf 0_{r_1,q_2}\\
	 \mathbf 0_{q_1,r_2}& \mathbf G_{0t}
	 \end{array}
	 \right), \quad  \mathbf R \bm{\mathcal R}^{-1} = \left[ \mathbf R_1\;\;\mathbf R_0\right], \quad  \mathbf C\bm{\mathcal C}^{-1} = \left[ \mathbf C_1\;\;\mathbf C_0\right].
\end{align}

Here, as an illustration, we provide one possible choice.
Let $\bm\beta_1$ be $k_1\times q_1$ such that $\bm\beta_1'\bm\beta_1=\mathbb I_{q_1}$ and $\text{vec}(\bm\beta_1'\mathbf F_t)\sim I(0)$, which means that all columns of $\mathbf F_t$ have the same cointegration relations. Similarly, let $\bm\beta_2$ be $k_2\times q_2$ such that $\bm\beta_2'\bm\beta_2=\mathbb I_{q_2}$ and $\text{vec}(\mathbf F_t\bm\beta_2)\sim I(0)$, which means that all rows of $\mathbf F_t$ have the same cointegration relations. Let also $\bm\beta_{i\perp}$ be $k_i\times k_i-q_i$ such that 
$\bm\beta_{i\perp}'\bm\beta_{i\perp}=\mathbb I_{k_i-q_i}$ and $\bm\beta_{i\perp}'\bm\beta_i=\mathbf 0_{k_i-q_i,q_i}$, for $i=1,2$. Let us also assume that $\bm\beta_{i\perp}'\bm\beta_j=\mathbf 0_{k_i-q_i,q_j}$ for $i\ne j$. Then,
\[
\bm{\mathcal R} = \left(
\begin{array}{c}
\bm\beta_1'\\
\bm\beta_{1\perp}'\\
\end{array}
\right)\quad \text{and} \quad \bm{\mathcal C} = \left(
\begin{array}{c}
\bm\beta_2'\\
\bm\beta_{2\perp}'\\
\end{array}
\right).
\]
	
\section{Additional simulation results}\label{app:sim}

\subsection{Separate estimation of $\mathbf{A}$ and $\mathbf{B}$}

	We conduct a Monte Carlo simulation to evaluate the finite-sample performance of the proposed EM estimator when the autoregressive matrices $\mathbf{A}$, $\mathbf{B}$, and the innovation covariance matrices $\mathbf{P}$, $\mathbf{Q}$ are estimated separately using the procedures outlined in Appendix~\ref{ap:em}. Table~\ref{tab:EMvsPCA_fm} reports a comparison between the EM estimator and the PE approach in terms of their accuracy in recovering the factor and loading matrices, under stationary conditions. Across all scenarios considered, the EM algorithm consistently outperforms PE.
	
	\begin{table}[h!]
	\caption{Average and standard deviation (in parenthesis) of the ratio between the performance  of the EM estimator and PE over 100 replications, for each of $\mathcal{D}(\mathbf{R}, \widehat{\mathbf{R}})$, $\mathcal{D}(\mathbf{C}, \widehat{\mathbf{C}})$ and $\textrm{MSE}_{\textbf{S}}$. }
	\centering
	\begin{tabular}{cccccccccccc}
		\toprule
		& & & & & & \multicolumn{3}{c}{$T=100$} & \multicolumn{3}{c}{$T=400$}   \\
		\midrule
		$\mu$ & $\delta$ & $\tau$ & $\mathfrak{D}$ & $p_1$ & $p_2$ &  $\mathcal{D}(\mathbf{R}, \widehat{\mathbf{R}})$ & $\mathcal{D}(\mathbf{C}, \widehat{\mathbf{C}})$ & $\textrm{MSE}_{\textbf{S}}$ & $\mathcal{D}(\mathbf{R}, \widehat{\mathbf{R}})$ & $\mathcal{D}(\mathbf{C}, \widehat{\mathbf{C}})$ & $\textrm{MSE}_{\textbf{S}}$   \\ 
		\midrule
		\multirow{2}{*}{$0.7$} & \multirow{2}{*}{$0$} & \multirow{2}{*}{$0$} & \multirow{2}{*}{N}& 20 & 20 &  0.98 & 0.97 & 0.92 & 0.98 & 0.96 & 0.91 \\
		&&&& &   & (0.05) & (0.05) & (0.03)  & (0.05) & (0.06) & (0.01)\\
		&&&& 10 & 30 &  0.98 & 0.96 & 0.9 & 0.96 & 0.96 & 0.9\\
		&&&& &   & (0.11) & (0.05) & (0.01) & (0.09) & (0.05) & (0.03) \\
		\midrule
		\multirow{2}{*}{$0.7$} & \multirow{2}{*}{$0.7$} & \multirow{2}{*}{$0.5$} & \multirow{2}{*}{N} & 20 & 20 & 0.8 & 0.71 & 0.73 & 0.74 & 0.65 & 0.75  \\
		&&&& &   &  (0.07) & (0.08) & (0.04) & (0.06) & (0.06) & (0.02) \\
		&&&& 10 & 30 & 0.87 & 0.68 & 0.7  & 0.82 & 0.63 & 0.75 \\
		&&&& &   & (0.1) & (0.09) & (0.05) & (0.1) & (0.06) & (0.03)\\
		\bottomrule
	\end{tabular}
	\label{tab:EMvsPCA_fm}
\end{table}

\subsection{Handling missing data: Initialization from balanced subpanels}
As an alternative initialization strategy for datasets with missing observations, we consider using starting values derived by applying our EM algorithm to a fully observed subset of the original matrix $\mathbf{Y}_t$. Because this approach necessitates excluding any rows and columns with missing values, we focus on the block missing data pattern. For comparison, we continue to use the PE estimator as a benchmark, applied to the original matrix after imputation using the method proposed by \cite{cen2024tensor}. Table \ref{tab:EMvsPCAmiss2} reports summary statistics for the ratio of the EM estimator’s performance relative to that of the PE. The results further confirm that the EM algorithm yields improved estimates compared to the PE. 
	
\begin{table}[htbp]
	\caption{Average and standard deviation (in parenthesis) of the ratio between the performance of PE and of the EM algorithm over 100 replications, for each of $\mathcal{D}(\mathbf{R}, \widehat{\mathbf{R}})$, $\mathcal{D}(\mathbf{C}, \widehat{\mathbf{C}})$, $\textrm{MSE}_{\textbf{S}}$, and $\textrm{MSE}_{\textbf{Y}^{(0)}}$. }
	\centering
		\begin{tabular}{ccccccccccccc}
			\toprule
			& & & & & \multicolumn{4}{c}{$T=100$} & \multicolumn{4}{c}{$T=400$}   \\
			\midrule
			$\mu$ & $\mathfrak{D}$ & $\pi$ & $p_1$ & $p_2$ &  $\mathcal{D}(\mathbf{R}, \widehat{\mathbf{R}})$ & $\mathcal{D}(\mathbf{C}, \widehat{\mathbf{C}})$ & $\textrm{MSE}_{\textbf{S}}$ & $\textrm{MSE}_{\textbf{Y}^{(0)}}$& $\mathcal{D}(\mathbf{R}, \widehat{\mathbf{R}})$ & $\mathcal{D}(\mathbf{C}, \widehat{\mathbf{C}})$ & $\textrm{MSE}_{\textbf{S}}$ & $\textrm{MSE}_{\textbf{Y}^{(0)}}$   \\ 
			\midrule
			\multirow{2}{*}{0.7} & \multirow{2}{*}{N} & \multirow{2}{*}{$25\%$} & 20 & 20 &  0.82 & 0.92 & 0.92 & 0.99 & 0.65 & 0.87 & 0.94 & 1.00  \\
			&&& &  & (0.17) & (0.06) & (0.06) & (0.01) & (0.14) & (0.07) & (0.02) & (0.00) \\
			&&& 10 & 30 &   0.86 & 0.97 & 0.91 & 1 .00  & 0.67 & 0.95 & 0.9 & 1.00  \\
			&&& &   & (0.15) & (0.05) & (0.03) & (0.00) & (0.14) & (0.05) & (0.01) & (0.00) \\
			\midrule
			\multirow{2}{*}{0.7} & \multirow{2}{*}{N} &\multirow{2}{*}{$50\%$} & 20 & 20 &  0.92 & 0.69 & 0.74 & 0.98 &  0.73 & 0.7 & 0.87 & 0.99    \\
			&&& &  &  (0.07) & (0.15) & (0.12) & (0.02)  & (0.10) & (0.12) & (0.04) & (0.00) \\
			&&& 10 & 30 &   0.76 & 0.88 & 0.82 & 0.98   &  0.54 & 0.84 & 0.85 & 0.99 \\
			&&& & &  (0.14) & (0.12) & (0.07) & (0.02) & (0.09) & (0.10) & (0.02) & (0.00)\\
			\midrule
			\multirow{2}{*}{0.7} & \multirow{2}{*}{St} &\multirow{2}{*}{$25\%$} & 20 & 20 &  0.83 & 0.98 & 0.92 & 0.99  &  0.73 & 1.10 & 0.96 & 1.00  \\
			&&& &  & (0.17) & (0.12) & (0.14) & (0.05) &(0.14) & (0.08) & (0.04) & (0.00) \\
			&&& 10 & 30 & 0.87 & 0.92 & 0.88 & 0.99   &  0.71 & 0.96 & 0.9 & 1.00   \\
			&&& &   & (0.21) & (0.09) & (0.09) & (0.02) &(0.17) & (0.02) & (0.03) & (0.00) \\
			\midrule
			\multirow{2}{*}{0.7} & \multirow{2}{*}{St} &\multirow{2}{*}{$50\%$} & 20 & 20 &  0.88 & 0.69 & 0.7 & 0.93  & 0.71 & 0.8 & 0.83 & 0.98   \\
			&&& &  & (0.09) & (0.23) & (0.23) & (0.14) & (0.08) & (0.2) & (0.11) & (0.02) \\
			&&& 10 & 30 &  0.81 & 0.77 & 0.74 & 0.95 &  0.62 & 0.85 & 0.79 & 0.98  \\
			&&&& & (0.21) & (0.21) & (0.2) & (0.08) & (0.15) & (0.13) & (0.11) & (0.03) \\
			\midrule
			\multirow{2}{*}{1} &\multirow{2}{*}{N} &\multirow{2}{*}{$25\%$} & 20 & 20 & 0.71 & 0.6 & 0.51 & 0.94  &  0.36 & 0.25 & 0.25 & 0.86    \\
			&&& &  & (0.21) & (0.23) & (0.25) & (0.08) & (0.18) & (0.09) & (0.18) & (0.08) \\
			&&& 10 & 30 & 0.67 & 0.84 & 0.67 & 0.97  &  0.3 & 0.61 & 0.43 & 0.92 \\
			&&& &   & (0.22) & (0.14) & (0.19) & (0.04)& (0.15) & (0.22) & (0.22) & (0.08) \\
			\midrule
			\multirow{2}{*}{1} & \multirow{2}{*}{N} &\multirow{2}{*}{$50\%$} & 20 & 20 & 0.9 & 0.18 & 0.12 & 0.5  &  0.79 & 0.07 & 0.07 & 0.34 \\
			&&& &  & (0.08) & (0.16) & (0.17) & (0.27) & (0.14) & (0.08) & (0.14) & (0.25) \\
			&&& 10 & 30 & 0.91 & 0.58 & 0.29 & 0.79 &  0.74 & 0.31 & 0.14 & 0.59 \\
			&&&& & (0.16) & (0.28) & (0.21) & (0.21) & (0.24) & (0.26) & (0.17) & (0.31)\\
			\bottomrule 
		\end{tabular}
	\label{tab:EMvsPCAmiss2}
\end{table}

\end{document}